\PassOptionsToPackage{colorlinks}{hyperref}

\documentclass[aos]{imsart}

\usepackage{algorithm}
\usepackage{algpseudocode}
\usepackage{amsmath,amssymb,amsfonts,latexsym,amsbsy,epsfig, mathrsfs, multirow, ifthen,extarrows}
\usepackage{ifthen}
\usepackage{comment}
\usepackage{pgf,tikz}
\usepackage{amsmath,amssymb,mathtools}
\usepackage{stmaryrd}
\usepackage[mathscr]{euscript}
\usepackage{multirow}
\usepackage{enumerate}
\usepackage{subfig}
\usepackage{xr}

\doi{10.1214/22-AOS2253}% Updated by VTEXPTS2LaTeX.exe, 27.01.2023 08:03
\volume{51}% Updated by VTEXPTS2LaTeX.exe, 14.03.2023 12:29
\pubyear{2023}% Updated by VTEXPTS2LaTeX.exe, 14.03.2023 12:29
\firstpage{334}% Updated by VTEXPTS2LaTeX.exe, 14.03.2023 12:29
\lastpage{361}% Updated by VTEXPTS2LaTeX.exe, 14.03.2023 12:29
\issue{1}

\newcommand{\intgr}{\mathfrak{I}}
\newcommand\civtex{\protect\mathpalette{\protect\independenT}{\perp}}
\def\independenT#1#2{\mathrel{\rlap{$#1#2$}\mkern3mu{#1#2}}}

\usetikzlibrary{arrows, decorations.markings,positioning}

\usetikzlibrary{shapes}

\usepackage[authoryear]{natbib}

\RequirePackage[OT1]{fontenc}
\RequirePackage{amsthm,amsmath,natbib,refcount}
\RequirePackage{hypernat}

\newtheorem{theorem}{Theorem}
\newtheorem{corollary}[theorem]{Corollary}
\newtheorem{proposition}[theorem]{Proposition}
\newtheorem{lemma}[theorem]{Lemma}
\newtheorem{definition}[theorem]{Definition}
\newtheorem{example}[theorem]{Example}

\newcommand{\indm}[2]{\ensuremath{{\mathfrak I}_{\kern-1pt\scriptstyle#1}({\mathcal
#2})}} % for independence models

\newcommand{\ind}{\mbox{$\perp \kern-5.5pt \perp$}}

\newcommand{\uned}{\hbox{\kern3pt\raise2.5pt\vbox{\hrule
width9pt height 0.3pt}\kern3pt}}

\newcommand{\dashed}{\hbox{\kern3.05pt\raise2.5pt\vbox{\hrule%dashed edges
width1.7pt height 0.3pt}\kern1.8pt\raise2.5pt\vbox{\hrule
width1.7pt height 0.3pt}\kern1.8pt\raise2.5pt\vbox{\hrule
width1.7pt height 0.3pt}\kern1.8pt\raise2.5pt\vbox{\hrule
width1.7pt height 0.3pt}\kern3.05pt}}
\newcommand{\lhead}{\ensuremath{\prec}}

\newcommand{\head}{\ensuremath{\succ}}

\newcommand{\pedg}[2]{\ensuremath{{\kern0.5pt
\scriptstyle{\ifthenelse{\equal{\head}{#1}}{\lhead\kern0.5pt}{#1\kern0.5pt}}\joinrel\relbar
\negthinspace\relbar\joinrel{\kern0.5pt #2}\kern0.5pt}}}
\newcommand{\pdots}{\hbox{\kern2.5pt\raise1.5pt\hbox{\ensuremath{\ldots}}\ke
rn2.5pt}}  %between PAG edges

\def\independenT#1#2{\mathrel{\rlap{$#1#2$}\mkern3mu{#1#2}}}

\DeclareMathOperator{\dis}{dis}
\DeclareMathOperator{\pavtex}{pa}
\DeclareMathOperator{\ndvtex}{nd}
\DeclareMathOperator{\devtex}{de}
\DeclareMathOperator{\chvtex}{ch}
\DeclareMathOperator{\anvtex}{an}
\DeclareMathOperator{\mbvtex}{mb}
\DeclareMathOperator{\pre}{pre}
\DeclareMathOperator{\Dovtex}{do}
\DeclareMathOperator{\fmvtex}{fam}

\def\textbff{\emph}

\newcommand{\dotcup}{\mathbin{\dot\cup}}
\newcommand{\cmid}{\,|\,}

\newcommand{\G}{{\mathscr G}}

\begin{document}

\begin{frontmatter}

%\dochead{}
\title{Nested Markov properties for\\
acyclic directed mixed graphs}
\runtitle{Nested Markov properties}
%
%% If relates to Discussion/Comment or Rejoinder fill in \thanksref in \title and \relateddoi:
%\relateddois{T1}{Discussed in \relateddoi[title={Child Title},type={see-also}]{}{\kaka{Child DOI}}%
%%; rejoinder at \relateddoi[title={Child Title},type={see-also}]{}{\kaka{Child DOI}}.
%}

%\begin{aug}
%% {\fnms{}~\snm{}\thanksref{}\ead[label=e?]{}\orcid{}}
%%
%%% e-mail is mandatory for each author
%%%% initials in fnms (if any) with spaces
%%
%\author[A]{\fnms{Thomas S.}~\snm{Richardson}\ead[label=e1]{thomasr@uw.edu}},
%\author[B]{\fnms{Robin J.}~\snm{Evans}\ead[label=e2]{evans@stats.ox.ac.uk}},
%\author[C]{\fnms{James M.}~\snm{Robins}\ead[label=e3]{robins@hsph.harvard.edu}}
%\and
%\author[D]{\fnms{Ilya}~\snm{Shpitser}\ead[label=e4]{ilyas@cs.jhu.edu}}
%%%\runauthor{} %% auto
%%
%%\dedicated{}
%%
%%%%%%% printead order as in MS:
%\address[A]{Department of Statistics, University of Washington\printead[presep={,\ }]{e1}}
%\address[B]{Department of Statistics, University of Oxford\printead[presep={,\ }]{e2}}
%\address[C]{Department of Epidemiology, Harvard T.H. Chan School of Public Health, Harvard University\printead[presep={,\ }]{e3}}
%\address[D]{Department of Computer Science, Johns Hopkins University\printead[presep={,\ }]{e4}}
%%%%%%%
%\end{aug}

\begin{aug}
\author{\fnms{Thomas S.} \snm{Richardson}\thanksref{m1}
\ead[label=e1]{tsr@stat.washington.edu}
}
%\and
\author{\fnms{Robin J.} \snm{Evans}\thanksref{m2}
\ead[label=e2]{evans@stats.ox.ac.uk}
}
%\and
\author{\fnms{James M.} \snm{Robins}\thanksref{m3}
\ead[label=e3]{robins@hsph.harvard.edu}
}
\and
\author{\fnms{Ilya} \snm{Shpitser}\thanksref{m4}
\ead[label=e4]{ilyas@cs.jhu.edu}
}
%\runauthor{T.S.\ Richardson, R.J.\ Evans, J.M.\ Robins and I.\ Shpitser}
\runauthor{Richardson, Evans, Robins and Shpitser}

\affiliation{
University of Washington\thanksmark{m1},
University of Oxford\thanksmark{m2},
Harvard University\thanksmark{m3},
Johns Hopkins University\thanksmark{m4}
}

\address{
Department of Statistics,\\
University of Washington\\
B313 Padelford Hall\\
Northeast Stevens Way\\
Seattle, WA 98195\\
\printead{e1}
}

\address{
Department of Statistics,\\
University of Oxford\\
24-29 St Giles'\\
Oxford OX1 3LB, UK\\
\printead{e2}
}

\address{
Department of Biostatistics,\\
Harvard T.H. Chan School of Public Health\\
Harvard University\\
677 Huntington Avenue\\
Kresge Building\\
Boston, Massachusetts 02115\\
\printead{e3}
}

\address{
Department of Computer Science,\\
Johns Hopkins University\\
Baltimore, MD 21218\\
\printead{e4}
}
\end{aug}

%% Move support info to funding

% \renewcommand{{\mathscr G}}{{\mathscr G}}

% HISTORY:
%\received{\smonth{4} \syear{2022}} % Updated by VTEXPTS2LaTeX.exe, 27.01.2023 08:03
%\revised{\smonth{11} \syear{2022}} % Updated by VTEXPTS2LaTeX.exe, 27.01.2023 08:03

\maketitle

% ABSTRACT
\begin{abstract}
Conditional independence models associated with directed acyclic graphs
(DAGs) may be characterized in at least three different ways: via a factorization,
the global Markov property (given by the d-separation criterion), and the
local Markov property. Marginals of DAG models also imply equality constraints
that are not conditional independences; the well-known ``Verma constraint''
is an example. Constraints of this type are used for testing edges, and
in a computationally efficient marginalization scheme via variable elimination.

We show that equality constraints like the ``Verma constraint'' can be
viewed as conditional independences in kernel objects obtained from joint
distributions via a fixing operation that generalizes conditioning and
marginalization. We use these constraints to define, via ordered local
and global Markov properties, and a factorization, a graphical model associated
with acyclic directed mixed graphs (ADMGs). We prove that marginal distributions
of DAG models lie in this model, and that a set of these constraints given
by Tian provides an alternative definition of the model. Finally, we show
that the fixing operation used to define the model leads to a particularly
simple characterization of identifiable causal effects in hidden variable
causal DAG models.
\end{abstract}

% KEYWORDS
\begin{keyword}[class=MSC]
\kwd[Primary ]{62H99}
\kwd[; secondary ]{60E05}
\end{keyword}
\begin{keyword} %% 1 u.c.
\kwd{graphical models}
\kwd{hidden variable models}
\kwd{conditional independence}
\kwd{causal inference}
\end{keyword}

\end{frontmatter}

%\tableofcontents%% if 50 pages and more
%%%%%%%%%%%%%%%%%%%%%%%%%%%%%%%%%%%%%%%%%%%%%%%%%%%%%%%%%%%%%%%%%%%%%%%%%
%%%% Main text entry area:

%s1 #&#
\section{Introduction}

Graphical models provide a principled way to take advantage of independence
constraints for probabilistic modeling, learning and inference, while giving
an intuitive graphical description of qualitative features useful for these
tasks. A popular graphical model represents a joint distribution by means
of a directed acyclic graph (DAG), where each vertex in the graph corresponds
to a random variable. The popularity of DAG models, also known as Bayesian
network models, stems from their well-understood theory and from the fact
that they admit an intuitive causal interpretation (under the assumption
that there are no unmeasured common causes; see
\citet{spirtes93causation}). An arrow from a variable $A$ to a variable
$B$ in a DAG model can be interpreted, in a way that can be made precise,
to mean that $A$ is a ``direct cause'' of $B$.

Starting from a causally interpreted DAG, the consequences of intervention
in the system under study can be understood by modifying the graph via
removing certain edges, and modifying the corresponding joint probability
distribution via reweighting
(\citet{strotz:60}, \citet{spirtes93causation}, \citet{pearl00causality}). For example, the
DAG in Figure~\ref{fig:mut}(i) represents distributions that factorize
as
\begin{align*}
p(x_{0}, x_{1}, x_{2}, x_{3},
x_{4}) = p(x_{0}) p(x_{1}) p(x_{2}
\cmid x_{0}, x_{1}) p(x_{3} \cmid
x_{1}, x_{2}) p(x_{4} \cmid x_{0},
x_{3}).
\end{align*}
If the model is interpreted causally, an experiment to externally set the
value of $X_{3}$ will break the dependence of $X_{3}$ on $X_{1}$ and
$X_{2}$; however, the dependence of $X_{4}$ upon $X_{3}$ will be preserved;
see Figure~\ref{fig:mut}(ii). This is represented graphically by severing
incoming edges to $3$, an operation some authors call ``mutilation,'' and
probabilistically by removing the factor
$p(x_{3} \cmid x_{1}, x_{2})$ from the factorization of $p$ to yield a
new distribution:
%
%e1 #&#
\begin{align}
p^{*}(x_{0}, x_{1}, x_{2},
x_{4} \cmid x_{3}) = p(x_{0})
p(x_{1}) p(x_{2} \cmid x_{0}, x_{1})
p(x_{4} \cmid x_{0}, x_{3}). \label{eqn:mut}
\end{align}
In certain contexts, we will later call these operations ``fixing.''

The functional in (\ref{eqn:mut}) is sometimes called the g-formula
\citep{robins86new}, the manipulated distribution
\citep{spirtes93causation}, or the truncated factorization
\citep{pearl00causality}. The distribution
$p^{*}(x_{0}, x_{1}, x_{2}, x_{4} \cmid x_{3})$ is commonly denoted by
$p(x_{0}, x_{1}, x_{2}, x_{4} \cmid \Dovtex(x_{3}))$. In a causal model given
by a DAG where all variables are observed, any interventional probability
distribution can be identified by this method.

Often not all common causes are measured, or there is no way to know
\emph{a priori} whether this is the case. This motivates the study of DAG
models containing latent variables \citep{allman15parameter}, and the constraints
they imply. Existing theoretical machinery based on DAGs can be applied
to such settings, simply by treating the unobserved variables as missing
data. However, this creates a number of problems that are particularly
severe when the structure of the underlying DAG model with latents is unknown.
First, there are, in general, an infinite number of DAG models with latent
variables that imply the (independence) constraints holding in a given
observed distribution. Second, assumptions concerning the state-space or
distribution of latent variables may have a profound effect on the model.
This is problematic if prior knowledge about latent variables is scarce.

%f1 #&#
\begin{figure}
\begin{center}
%  \begin{tikzpicture}[>=stealth, node distance=1.2cm]
\begin{tikzpicture}[>=stealth, node distance=1.0cm,
pre/.style={->,>=stealth,ultra thick,line width = 1.4pt},
    format/.style={circle, draw, very thick, circle, minimum size=5mm, inner sep=.3mm}]
  \begin{scope}
    \node[format] (1) {$1$}; %{$\vrt{x_1}$};
    \node[format, right of=1] (2) {$2$}; %{$\vrt{x_2}$};
    \node[format, right of=2] (3) {$3$}; %{$\vrt{x_3}$};
    \node[format, right of=3] (4) {$4$}; %{$\vrt{x_4}$};
    \node[format, color=red, above of=3, yshift=-3mm] (5) {$0$}; %{$\vrt{x_0}$};
    
                 \draw (1) edge[pre, blue] (2)
                  (2) edge[pre, blue] (3)
                  (1.315) edge[pre, blue, bend right] (3.225)
                  (3) edge[pre, blue] (4)
                  (5) edge[pre, red] (2)
                  (5) edge[pre, red] (4);
  \node[below of=3, xshift=-6mm, yshift=1mm] {(i)};
  \end{scope} 
  
   \begin{scope}[xshift=4.5cm]
%    \tikzstyle{format} = [circle, inner sep=2pt,draw, thick, circle, line width=1.4pt, minimum size=6mm]
    \node[format] (1) {$1$}; %{$\vrt{x_1}$};
    \node[format, right of=1] (2) {$2$}; %{$\vrt{x_2}$};
    \node[format, rectangle, right of=2] (3) {$3$}; %{$\vrt{x_3}$};
    \node[format, right of=3] (4) {$4$}; %{$\vrt{x_4}$};
    \node[format, color=red, above of=3, yshift=-3mm] (5) {$0$}; %{$\vrt{x_0}$};
    
                 \draw (1) edge[pre, blue] (2)
                  (3) edge[pre, blue] (4)
                  (5) edge[pre, red] (2)
                  (5) edge[pre, red] (4);
  \node[below of=3, xshift=-6mm, yshift=1mm] {(ii)};
  \end{scope}
  
    \begin{scope}[xshift=9cm]
%    \tikzstyle{format} = [circle, inner sep=2pt,draw, thick, circle, line width=1.4pt, minimum size=6mm]
    \node[format] (1) {$1$}; %{$\vrt{x_1}$};
    \node[format, right of=1] (2) {$2$}; %{$\vrt{x_2}$};
    \node[format, right of=2] (3) {$3$}; %{$\vrt{x_3}$};
    \node[format, right of=3] (4) {$4$}; %{$\vrt{x_4$}};
    \node[format, color=red, above of=3, yshift=-3mm] (5) {$0$}; %{$\vrt{x_0}$};
    
                 \draw (1) edge[pre, blue] (2)
                  (2) edge[pre, blue] (3)
                  (1.315) edge[pre, blue, bend right] (3.225)
                  (3) edge[pre, blue] (4)
                  (5) edge[pre, red] (2)
                  (5) edge[pre, red] (4)
                  (1.315) edge[pre, blue, bend right=30] (4)
                  ;
  \node[below of=3, xshift=-6mm, yshift=1mm] {(iii)};
  \end{scope} 
  
  \end{tikzpicture}
\end{center}
\caption{(i) A DAG on five variables and (ii) a DAG representing the model
after an experiment to externally fix $X_{3}$. (iii) A DAG on five variables
representing a statistical model distinguishable from the model represented
by the DAG in (i) by a generalized independence constraint.}
\label{fig:mut}
\end{figure}

An alternative approach considers a supermodel defined by taking a subset
of the constraints implied by a DAG model with latent variables on the
observed marginal distribution. More specifically, we consider models defined
by equality constraints that are implied by the factorization of a DAG
with latents, but that do not depend on assumptions regarding the state-space
or distribution of the latent variables. Models defined by these constraints
are naturally represented by mixed graphs, that is, graphs containing directed
($\rightarrow $) and bidirected ($\leftrightarrow $) edges, obtained from
DAGs via a \emph{latent projection} operation \citep{verma90equiv}; see
the graph in Figure~\ref{fig:verma1}(i) for the latent projection of the
DAG in Figure~\ref{fig:mut}(i).

Much previous work
\citep{wermuth:94, richardson:2002, ali09equiv, evans10maximum, wermuth:11, sadeghi14markov}
has defined Markov models for mixed graphs via independence constraints
implied by hidden variable DAGs on the observed margin. It is well known,
however, that DAG models with latent variables imply nonparametric constraints
that are not conditional independence constraints. For example, consider
the DAG shown in Figure~\ref{fig:mut}(i), and take the vertex $0$ as hidden.
This DAG implies no conditional independence restrictions on the observed
margin $p(x_{1},x_{2},x_{3},x_{4})$. This is because all vertex sets that
d-separate pairs of observed variables---that is, the pairs $(x_{1},x_{4})$ and $(x_{2},x_{4})$---include the unobserved variable
$x_{0}$. However, it may be shown that the
$p(x_{1},x_{2},x_{3},x_{4})$ margin of any distribution
$p(x_{0},x_{1},x_{2},x_{3},x_{4})$, which factorizes according to the DAG
in Figure~\ref{fig:mut}(i), obeys the constraint that
%
%e2 #&#
\begin{equation}
\label{eq:verma} \sum_{x_{2}} p(x_{4} \cmid
x_{1},x_{2},x_{3}) p(x_{2} \cmid
x_{1}) \quad\text{is a function of only }x_{3}\text{ and
}x_{4};
\end{equation}
see \citet{verma90equiv,robins86new,wermuth08distortion}. In Robins (\citeyear{robins:1999}),
it is shown that this constraint is equivalent to the requirement that
$X_{4}$ is independent of $X_{1}$ given $X_{3}$ in the distribution obtained
from $p(x_{1},x_{2},x_{3},x_{4})$ after dividing by the conditional
$p(x_{3} \cmid x_{2},x_{1})$. Note that this is the same manipulation performed
in (\ref{eqn:mut}), but the operation, which we later call ``fixing,''
is purely probabilistic and can be performed \emph{without} requiring that
the model has any causal interpretation. If we interpret the original DAG
as causal, then (\ref{eq:verma}) is an (identifiable) dormant independence
constraint \citep{shpitser08dormant, shpitser:behaviormetrika:2014}.

Since, as we have seen, the DAG in Figure~\ref{fig:mut}(i) implies no
conditional independence restrictions on the joint
$p(x_{1},x_{2},x_{3},x_{4})$, the set of distributions obeying these independence
relations is (trivially) saturated. Consequently, a structure learning
algorithm such as FCI \citep{spirtes93causation} that learns a Markov equivalence
class of DAGs with latent variables, under the assumption of faithfulness,
will return a (maximally uninformative) unoriented complete graph. The
assumption of ``faithfulness'' implies that if
$X_{A} \civtex X_{B} \mid X_{C}$ in the observed distribution then $A$ is d-separated
from $B$ given $C$ in the underlying DAG with latent variables.

Indeed, as originally pointed out by \citet{robins:1999}, if we assume
a certain generalization of faithfulness, if (\ref{eq:verma}) holds, we
may rule out the model in Figure~\ref{fig:mut}(iii). More generally,
\cite{shpitser09edge} used pairwise constraints of this form to test for
the presence of certain directed edges (in the context of a specific graph).
Further, \cite{tian02on} presented a general algorithm for finding nonparametric
constraints from DAGs with latent variables.

In this paper, we introduce a (statistical) model called the
\emph{nested Markov model}, defined by these nonparametric constraints,
and associated with a mixed graph called an acyclic directed mixed graph
(ADMG). We give equivalent characterizations of the model in terms of global
and ordered local Markov properties, and a factorization. We next prove
that the set of marginal distributions given by a latent DAG model is a
submodel of the nested Markov model associated with the corresponding ADMG
(obtained by latent projection). We show that our results lead to a particularly
simple characterization of identifiable causal effects in hidden variable
causal DAG models. Finally, we relate the nested Markov model to a set
of constraints obtained by \cite{tian02on}.

%s1.1 #&#
\subsection{Applications of the nested Markov model}

Nested models are defined by constraints that correspond to the absence
of direct effects \citep{robins97estimation, robins:1999}. For example,
consider the causal graph in Figure~\ref{fig:mut}(i), where vertices
$1$ and $3$ correspond to treatments $X_{1}$ and $X_{3}$, vertices
$2$ and $4$ are responses $X_{2},X_{4}$ and vertex $0$ is an unobserved
confounding variable $X_{0}$. In this graph, $X_{1}$ has no direct effect
on $X_{4}$, and as noted above, this implies (\ref{eq:verma}). The nested
model shown in Figure~\ref{fig:verma1}(i) is defined by this constraint
(see Section~\ref{sec:nested}).

Such restrictions have been increasingly exploited in the analysis of longitudinal
observational data in medicine and public health. For example, consider
the nested Markov constraint that a diagnostic or screening test (e.g.,~a
mammogram) has no direct effect on a patient's clinical outcome of interest
(death from breast cancer), except through the effect of the test results
on the choice of treatment (no treatment if the mammogram was negative
and, if positive, biopsy possibly followed by breast surgery and/or chemotherapy).
Several recent papers
\citep{neugebauer17identification, caniglia19emulating, kreif21exploiting}
have leveraged this no direct effect of screening constraint to construct
novel highly efficient estimators of an optimal joint testing and treatment
regime. What was surprising and, indeed, unprecedented is that, in actual
medical studies, the new estimators have provided a 50-fold increase in
efficiency (and thus a 50-fold reduction in the required sample size) compared
to estimators that fail to leverage the no direct effect constraint
\citep{caniglia19emulating}.

Nested Markov constraints can also be used for structure learning. The
review paper \citep{shpitser:behaviormetrika:2014} includes all nested
Markov equivalence classes over four variables; classes with nested constraints
can be very small, so they are potentially extremely informative about
causal structure. This reflects the fact that whereas conditional independences
involve three sets, nested constraints are relations among more groups
of variables (see Section~\ref{subsec:kernel-properties}). A general theory
of nested Markov equivalence and search for ADMGs remains a subject for
future work. However, distributions consistent with these constraints have
already been observed in applications; see
\citet{bhattacharya21differentiable}.

The nested constraints also facilitate the construction of more computationally
efficient marginalization methods in some causal graphs
\citep{shpitser11eid}, for categorical variables. This is achieved via
Theorem~\ref{thm:1-line-id} and the parameterization of
\citet{evans:richardson:19}.

This work is of separate interest to people constructing causal models
for quantum mechanics, since it \emph{does not} impose the inequalities
implied by latent variable models \citep{navascues20inflation}; see Example~\ref{ex:chsh}.

\citet{evans:complete} shows that for categorical variables the nested
Markov model is the closest approximation to the margin of a DAG model
(ignoring inequalities), in that it incorporates all equality constraints,
and in this sense is ``complete.''

%s1.2 #&#
\subsection{Overview of nested Markov models}

We now outline our strategy for defining the nested Markov model in terms
of ordinary conditional independence (in derived distributions) by analogy
to a way of defining DAG models in terms of undirected graphical models.
We give a specific example of a nested Markov model, outlining the key
concepts, while providing references to the formal definitions within the
paper.

A nested Markov model is represented by an ADMG, a mixed graph naturally
derived from DAGs with latent variables via an operation called
\emph{latent projection}. Intuitively, the ADMG does not contain latent
variables, but indicates the presence of such variables by the inclusion
of bidirected ($\leftrightarrow $) edges. Earlier work
(\citet{richardson:2002,richardson03sjs}) established Markov properties for
ordinary independence models defined by ADMGs. Such an independence model
is defined by fewer constraints than the nested Markov model represented
by the same ADMG, and hence the former is a supermodel of the latter. The
global Markov property for these independence models simply corresponds
to the natural extension of d-separation \citep{pearl88probabilistic} to
ADMGs. This extension, which is sometimes called m-separation, consists
of allowing colliders to involve bidirected edges. Latent projection is
defined in Section~A.3; ADMGs and m-separation in
Sections~\ref{ssec:mixedg} and A.2, respectively.

We also consider conditional ADMGs (CADMGs) where certain ``fixed'' vertices
represent nonrandom variables that index distributions over the random
variables. Such vertices are treated similarly to the so-called ``strategy
nodes'' in influence diagrams \citep{dawid02influence}. The Markov property
for CADMGs is a simple extension of m-separation that takes into account
fixed nodes. CADMGs are defined formally in Section~\ref{subsec:kernels}; the corresponding global Markov property for a CADMG
is given in Section~\ref{subsec:cadmg-global}. Note that an ADMG is
the special case of a CADMG with no fixed vertices.

CADMGs and their associated Markov models characterize the nested Markov
model in much the same way that undirected graphs and their associated
Markov models can be used to describe a DAG model. We first briefly review
the characterization of DAGs via undirected models.

The global Markov property for DAGs may be obtained from the (union of
the) Markov properties associated with undirected graphs derived from the
DAG by the moralization operation
(\citep{lauritzen96graphical}, Theorem~3.27); the resulting property is
equivalent to d-separation \citep{pearl88probabilistic}. More precisely,
the DAG Markov property corresponds to (the union of) the Markov properties
associated with undirected graphs representing ``ancestral'' margins. Likewise,
the set of distributions corresponding to the DAG is the intersection of
the sets of distributions obeying the factorization properties associated
with these undirected graphs; this is equivalent to the characterization
that the joint distribution factors into the product of each variable given
its parents in the graph. As an example, consider the DAG in Figure~\ref{fig:moral}(i). Undirected graphs associated with some ancestral margins,
and their factorizations, are shown in Figure~\ref{fig:moral}(ii), (iii),
and (iv).

%f2 #&#
\begin{figure}
\begin{center}
\begin{tikzpicture}[>=stealth, node distance=1.2cm,
pre/.style={->,>=stealth,ultra thick,line width = 1.4pt}]
\tikzstyle{vecArrow} = [>=latex, thick, decoration={markings,mark=at position
   1 with {\arrow[semithick]{open triangle 60}}},
   double distance=1.4pt, shorten >= 5.5pt,
   preaction = {decorate},
   postaction = {draw,line width=1.4pt, white,shorten >= 4.5pt}]
\tikzstyle{innerWhite} = [semithick, white,line width=1.4pt, shorten >= 4.5pt]
    \tikzstyle{format} = [circle, draw, very thick, minimum size=5mm, inner sep=.3mm]

  \begin{scope}
%    \tikzstyle{format} = [circle, draw, ultra thick, line width=1.4pt, minimum size=5mm, inner sep=.5mm]
    \path[->]	 
    		node[format] (x1) {$1$} %{$\vrt{x_1}$}
    		node[format, right of=x1] (x2) {$2$} %{$\vrt{x_2}$}
                  (x1) edge[pre, blue] (x2)
		node[format, right of=x2] (x3) {$3$} %{$\vrt{x_3}$}
                  %(x2) edge[pre, blue] (x3)
                  %(x1.315) edge[pre,->, bend right,blue] (x3.225)
		node[format, right of=x3] (x4) {$4$} %{$\vrt{x_4}$}
                  (x3) edge[pre, blue] (x4)
                  (x2.45) edge[pre, ->, blue, bend left] (x4.135)   
	         ;
	         \node[left of=x1, yshift = 0cm, xshift= 0cm] (label) {(i)};
	         \node[right=1cm of x4, align=left] (p21) {$p(x_1,x_2,x_3,x_4)$};
	\end{scope}
	
		\draw[ultra thick] (-0.3,-1.15) -- (10.5,-1.15);

  \begin{scope}[yshift=-2.5cm]
%    \tikzstyle{format} = [circle, draw, ultra thick, line width=1.4pt, minimum size=5mm, inner sep=.5mm]
    \path[->]	node[format] (x1) {$1$} %{$\vrt{x_1}$}
    		node[format, right of=x1] (x2) {$2$} %{$\vrt{x_2}$}
                  (x1) edge[pre, -] (x2)
		node[format, right of=x2] (x3) {$3$} %{$\vrt{x_3}$}
                  (x2) edge[pre, -] (x3)
		node[format, right of=x3] (x4) {$4$} %{$\vrt{x_4}$}
                  (x3) edge[pre, -] (x4)
                  (x2.45) edge[pre, -, bend left] (x4.135)
   
	        ;
   \node[left of=x1, yshift = 0cm, xshift=0cm] (label) {(ii)};
    \node[right=1cm of x4, align=left] (p22) {$p(x_1,x_2,x_3,x_4)=\varphi_{12}(x_1,x_2)\varphi_{234}(x_2,x_3,x_4)$};
  \end{scope}
  
    \begin{scope}[yshift=-4.5cm]
      \path[->]
      		node[format] (x1) {$1$} %{$\vrt{x_1}$}
    		node[format, right of=x1] (x2) {$2$} %{$\vrt{x_2}$}
                  (x1) edge[pre,-] (x2)
		node[format, right of=x2] (x3) {$3$} %{$\vrt{x_3}$}
                  	node[format, rectangle, right of=x3] (x4) {$4$} %{$\vrt{x_4}$}
                  ;            
                     \node[left of=x1, yshift = 0cm, xshift=0cm] (label) {(iii)};
                     \node[right=1cm of x4, align=left] (p23) 	{$\displaystyle\sum_{x_4}p(x_1,x_2,x_3,x_4)\!=\!\varphi_{12}(x_1,x_2)\varphi_{3}(x_3 )$};
    \end{scope}
    
            \begin{scope}[yshift=-6.5cm, xshift=0.0cm]

%        \tikzstyle{format} = [circle, draw, ultra thick, line width=1.4pt, minimum size=5mm, inner sep=.5mm]
    \path[->]	
    		node[format] (x1) {$1$} %{$\vrt{x_1}$}
    		node[format, right of=x1, rectangle] (x2) {$2$} %{$\vrt{x_2}$}
		node[format, right of=x2] (x3) {$3$} %{$\vrt{x_3}$}
		node[format, right of=x3, rectangle] (x4) {$4$} %{$\vrt{x_4}$}                  
                  ;
           \node[left of=x1, yshift = 0cm, xshift=0cm] (label) {(iv)};
             \node[right=1cm of x4, align=left] (p23) 	{$\displaystyle\sum_{x_2,x_4} p(x_1,x_2,x_3,x_4)=\varphi_{1}(x_1)\varphi_{3}(x_3)$};
                        \end{scope}
  \end{tikzpicture}

\end{center}
\caption{Reduction of a DAG model to a set of undirected models via marginalization
and moralization: (i)~The original DAG ${\mathscr G} (\{1,2,3,4\})$. Undirected graphs
representing the factorization of different ancestral margins: (ii)~$p(x_{1},x_{2},x_{3},x_{4})$; (iii) $p(x_{1},x_{2},x_{3})$; (iv)
$p(x_{1},x_{3})$. Note that we have also included the marginalized variables
on the graph in square nodes (since, as we later show, marginalization
is a special case of fixing). The DAG model may be characterized by (the
union of) the conditional independence properties implied by the undirected
graphs for all ancestral margins.}
\label{fig:moral}
\end{figure}

Likewise, the set of distributions in the nested Markov model associated
with an ADMG corresponds to the intersection of the sets of distributions
obeying factorization properties encoded by specific CADMGs obtained from
the original ADMG. However, whereas the undirected graphs corresponding
to a DAG may be seen as representing specific (ancestral) \emph{margins},
the CADMGs obtained from an ADMG represent ``kernel'' distributions obtained
by sequentially applying a new ``fixing'' operation on distributions, one
that generalizes conditioning and marginalizing. This fixing operation
has a natural causal interpretation, as do the kernels that form CADMG
factorizations of the nested Markov model. Specifically, in the context
of a latent variable causal model whose projection is a given ADMG, kernels
can be viewed as (identified) interventional distributions. Not all variables
are fixable. From a causal perspective this is natural since in the presence
of latent variables, not all interventional distributions are identifiable.
The fixing operation and the set of fixable vertices are defined in Section~\ref{subsec:fixing}.

As a specific example, consider the graph shown in Figure~\ref{fig:verma1}(i).
In this ADMG, the vertex $3$ may be fixed to give the CADMG shown in
Figure~\ref{fig:verma1}(ii), where the corresponding distribution and factorization
are also shown, with
\begin{align*}
q_{1}(x_{1})= p(x_{1}),\qquad q_{24}(x_{2},
x_{4} \cmid x_{1}, x_{3}) = \frac{p(x_{1}, x_{2}, x_{3}, x_{4})}{p(x_{1})p(x_{3} \cmid x_{2}, x_{1})}.
\end{align*}
Note that, although the original graph implied no conditional independences,
the graph in Figure~\ref{fig:verma1}(ii) implies the independence
$X_{1} \civtex X_{4} \mid X_{3}$ via m-separation.

Whereas the undirected graphs associated with a DAG correspond to distributions
obtained by \emph{specific} marginalizations (namely those that remove vertices
that have no children), CADMGs correspond to certain \emph{specific} ordered
sequences of fixing operations. Not all such sequences are allowed: in
some cases a vertex may be fixable only after another vertex has already
been fixed. The global nested Markov property corresponds to the (union
of the) m-separation relations encoded in the CADMGs derived via allowed
sequences of these fixing operations, which we call ``valid.'' Valid fixing
sequences are defined in Section~\ref{subsec:reachable}. These fixing sequences
are closely related to a particular identification strategy for interventional
distributions consisting of recursively applying the g-formula to an
already identified intervention distribution to obtain the result of further
interventions; this connection is explored further in Section~\ref{subsec:causal}.

%f3 #&#
\begin{figure}
\begin{center}
\begin{tikzpicture}[>=stealth, node distance=1.2cm, 
pre/.style={->,>=stealth,ultra thick,line width = 1.4pt}]
\tikzstyle{vecArrow} = [>=latex, thick, decoration={markings,mark=at position
   1 with {\arrow[semithick]{open triangle 60}}},
   double distance=1.4pt, shorten >= 5.5pt,
   preaction = {decorate},
   postaction = {draw,line width=1.4pt, white,shorten >= 4.5pt}]
\tikzstyle{innerWhite} = [semithick, white,line width=1.4pt, shorten >= 4.5pt]
    \tikzstyle{format} = [circle, draw, very thick, minimum size=5mm, inner sep=.3mm,node distance=6mm]

  \begin{scope}
 %   \tikzstyle{format} = [circle, draw, ultra thick, line width=1.4pt, minimum size=5mm, inner sep=.5mm]
    \path[->]	 
    		node[format] (x1) {$1$} %{$\vrt{x_1}$}
    		node[format, right= of x1] (x2) {$2$} %{$\vrt{x_2}$}
                  (x1) edge[pre, blue] (x2)
		node[format, right= of x2] (x3) {$3$} %{$\vrt{x_3}$}
                  (x2) edge[pre, blue] (x3)
                  (x1.315) edge[pre,->, bend right,blue] (x3.225)
		node[format, right= of x3] (x4) {$4$} %{$\vrt{x_4}$}
                  (x3) edge[pre, blue] (x4)
                  (x2.45) edge[pre, <->, red, bend left] (x4.135)
                 ; 
                 \node[left=of x1] (label) {(i)};
	         \node[right=0.8cm of x4, align=left] (p21) {$p(x_1,x_2,x_3,x_4)$};
	\end{scope}
	
		\draw[ultra thick] (-0.3,-1.15) -- (9.5,-1.15);
	
  \begin{scope}[yshift=-2.5cm]
%    \tikzstyle{format} = [circle, draw, ultra thick, line width=1.4pt, minimum size=5mm, inner sep=.5mm]
    \path[->]	
    		node[format] (x1) {$1$} %{$\vrt{x_1}$}
    		node[format, right=of x1] (x2) {$2$} %{$\vrt{x_2}$}
                  (x1) edge[pre, blue] (x2)
		node[format, right=of x2, rectangle] (x3) {$3$} %{$\vrt{x_3}$}
%                  (x2) edge[pre, blue] (x3)
%                  (x1.315) edge[pre,->, bend right,blue] (x3.225)
		node[format, right= of x3] (x4) {$4$} %{$\vrt{x_4}$}
                  (x3) edge[pre, blue] (x4)
   (x2.45) edge[pre, <->, red, bend left] (x4.135)
   
	        ;
%                  (x2) edge[pre, <->, bend left, red] (x4);
   \node[left=of x1] (label) {(ii)};
    \node[right=0.8cm of x4, align=left] (p22) {$\frac{p(x_1,x_2,x_3,x_4)}{p(x_3\cmid x_1,x_2)}= q_{1}(x_1)q_{24}(x_2,x_4\cmid x_1,x_3)$};
  \end{scope}
%     \begin{scope}[yshift=-4.5cm]
%%    \tikzstyle{format} = [circle, draw, ultra thick, line width=1.4pt, minimum size=5mm, inner sep=.5mm]
%  %                  (x2) edge[pre, <->, bend left, red] (x4);
%    \path[->]	
%    		node[format] (x1) {$x_1$}
%    		node[format, right=of x1] (x2) {$x_2$}
%                  (x1) edge[pre, blue] (x2)
%		node[format, right=of x2,rectangle] (x3) {$x_3$}
%           %      (x2) edge[pre, blue] (x3)
%%                  (x1.315) edge[pre,->, bend right,blue] (x3.225)
%		node[format, right= of x3, rectangle] (x4) {$x_4$}
%          %        (x3) edge[pre, blue] (x4)
% %  (x2.45) edge[pre, <->, red, bend left] (x4.135)
%	        ;
%   \node[left=of x1] (label) {(iii)};
%    \node[right=0.8cm of x4, align=left,yshift=-0.4cm] (p22) 
%    {${\displaystyle \sum_{x_4}}\frac{p(x_1,x_2,x_3,x_4)}{p(x_3\cmid x_1,x_2)}=
%    {\displaystyle \sum_{x_3,x_4}}{p(x_1,x_2,x_3,x_4)}$\\[5pt]
%    \kern70pt $=q_{1}(x_1)q_{2}(x_2\cmid x_1)$};
%  \end{scope}
     \begin{scope}[yshift=-4.5cm]
%    \tikzstyle{format} = [circle, draw, ultra thick, line width=1.4pt, minimum size=5mm, inner sep=.5mm]
    \path[->]	
    		node[format,rectangle] (x1) {$1$} %{$\vrt{x_1}$}
    		node[format, right=of x1] (x2) {$2$} %{$\vrt{x_2}$}
                 (x1) edge[pre, blue] (x2)
		node[format, right=of x2, rectangle] (x3) {$3$} %{$\vrt{x_3}$}
%                  (x2) edge[pre, blue] (x3)
%                  (x1.315) edge[pre,->, bend right,blue] (x3.225)
		node[format, right= of x3] (x4) {$4$} %{$\vrt{x_4}$}
                  (x3) edge[pre, blue] (x4)
   (x2.45) edge[pre, <->, red, bend left] (x4.135)
   
	        ;
%                  (x2) edge[pre, <->, bend left, red] (x4);
   \node[left=of x1] (label) {(iii)};
    \node[right=0.8cm of x4, align=left] (p22) {$\frac{p(x_1,x_2,x_3,x_4)}{p(x_1)p(x_3\cmid x_1,x_2)}= q_{24}(x_2,x_4\cmid x_1, x_3)$};
  \end{scope}
   \begin{scope}[yshift=-6.5cm]
%    \tikzstyle{format} = [circle, draw, ultra thick, line width=1.4pt, minimum size=5mm, inner sep=.5mm]
    \path[->]	
    		node[format,rectangle] (x1) {$1$} %{$\vrt{x_1}$}
    		node[format, right=of x1,rectangle] (x2) {$2$} %{$\vrt{x_2}$}
     %             (x1) edge[pre, blue] (x2)
		node[format, right=of x2, rectangle] (x3) {$3$} %{$\vrt{x_3}$}
%                  (x2) edge[pre, blue] (x3)
%                  (x1.315) edge[pre,->, bend right,blue] (x3.225)
		node[format, right= of x3] (x4) {$4$} %{$\vrt{x_4}$}
                  (x3) edge[pre, blue] (x4)
 %  (x2.45) edge[pre, <->, red, bend left] (x4.135)
   
	        ;
%                  (x2) edge[pre, <->, bend left, red] (x4);
   \node[left=of x1] (label) {(iv)};
    \node[right=0.8cm of x4, align=left] (p22) {
    %{\displaystyle \sum_{x_1,x_2}}\frac{p(x_1,x_2,x_3,x_4)}{p(x_3\cmid x_1,x_2)}= 
%    \left.\sum_{x_2}\frac{p(x_1,x_2,x_3,x_4)}{p(x_3\cmid x_1,x_2)} \right/ \sum_{x_2,x_4}\frac{p(x_1,x_2,x_3,x_4)}{p(x_3\cmid x_1,x_2)} =    
${\displaystyle \sum_{x_2}}\frac{p(x_1,x_2,x_3,x_4)}{p(x_1)p(x_3\cmid x_1,x_2)}= 
{\displaystyle \sum_{x_2}}p(x_2\cmid x_1)p(x_4\cmid x_1,x_2,x_3) $\\[5pt]
\kern85pt$= q_{4}(x_4\cmid x_3)$};
  \end{scope}

   \end{tikzpicture}

\end{center}
\caption{Reduction of the nested Markov model for an ADMG to a set of ordinary
Markov models associated with CADMGs: (i) The ADMG
$\mathscr{G} (\{1,2,3,4\})$, which is the latent projection of the graph
${\mathscr G} $ from Figure~\protect\ref{fig:mut}(i). CADMGs, representing the Markov structure
of derived distributions, resulting from sequences of fixing operations
in ${\mathscr G} (\{1,2,3,4\})$: (ii) $\langle 3\rangle $; (iii)
$\langle 3, 1\rangle $ or alternatively $\langle 1, 3\rangle $; (iv) any
of the sequences $\langle 1, 3, 2\rangle $,
$\langle 3, 1, 2\rangle $, $\langle 3, 2, 1\rangle $; it is not valid to
fix $2$ before $3$. The nested Markov model may be defined via the conditional
independence properties for all CADMGs and associated kernels obtained
(via valid fixing sequences) from the original ADMG and distribution. See
also Figure~\protect\ref{fig:moral} and text.}
\label{fig:verma1}
\end{figure}

Returning to the example, given the CADMG in Figure~\ref{fig:verma1}(ii),
the vertex $1$ may be fixed to give the CADMG in Figure~\ref{fig:verma1}(iii).
Further, given the CADMG in Figure~\ref{fig:verma1}(iii), we may fix $2$.
The kernel in this graph is
%
%e3 #&#
\begin{align}
\label{eq:verma-constraint-intro} q_{4}(x_{4} \cmid x_{3}) &= \sum
_{x_{2}} p(x_{4} \cmid x_{3},
x_{2}, x_{1}) p(x_{2} \cmid x_{1}).
\end{align}
The quantity on the RHS of (\ref{eq:verma-constraint-intro}) is a function
only of $x_{3}$ and $x_{4}$, and not $x_{1}$. This is precisely the constraint
(\ref{eq:verma}) implied by the original latent variable DAG model in Figure~\ref{fig:mut}(i).
This constraint characterizes the nested Markov model corresponding to
the ADMG in Figure~\ref{fig:verma1}(i).

In the original ADMG in Figure~\ref{fig:verma1}(i), we could also have
fixed $1$ and $4$. Had we chosen to fix $1$, we could then subsequently
have fixed $3$, and would have arrived at the same CADMG and distribution
as shown in Figure~\ref{fig:verma1}(iii). Thus, the operations of fixing
$3$ and $1$ commute. However, this is not always the case: for example,
$2$ may only be fixed after $3$.

Like the DAG model, the nested Markov model may be characterized via a
factorization property, as well as by local and global Markov properties
described in Section~\ref{sec:nested}. In each case these properties are
defined via the set of CADMGs that are ``reached'' via valid fixing sequences.
Just as the factorization property of the DAG model leads naturally to
a parameterization in terms of Markov factors, likewise the factorization
property of the nested Markov model leads to a parameterization in terms
of its factors in multivariate Gaussian and finite categorical cases
\citep{uai:18, evans:richardson:19}. The usual DAG parameters are such
that the choice of value for one parameter
$p(x_{v}\cmid x_{\pavtex (v)})$ does not restrict the set of possible values
for another, $p(x_{w} \cmid x_{\pavtex (w)})$; in other words, they are
\emph{variation independent}, but this is not true for the nested Markov
model.

The rest of the paper is organized as follows. In Section~\ref{ssec:mixedg}, we introduce graph theoretic preliminaries. In Section~\ref{subsec:cadmgs}, we define conditional mixed graphs, which are used
to construct our model. In Section~\ref{subsec:kernels}, we introduce a
generalization of conditional distributions which we call
\emph{kernels}, and in Section~\ref{subsec:kernel-indep} we generalize the
notion of conditional independence from standard distributions to these
kernels. In Sections~\ref{subsubsec:kernel-semi-graphoid}--\ref{subsec:kernel-preserve},
we provide further properties of kernels, including a proof that conditional
independence in kernels satisfies the semigraphoid axioms (Proposition~\ref{prop:semigraphoid}). In Section~\ref{subsec:markov-cadmg} and Section~\ref{subsec:fact-cadmg}, we give the Markov properties and factorization
of kernels associated with CADMGs; we prove they are equivalent in Theorem~\ref{thm:one-step-markov} in Section~\ref{subsec:markov}. We define the
graph-based fixing operation in Section~\ref{subsec:fixing}, give some
of its properties in Section~\ref{subsec:fixing-fact}, and define sets
that can be reached by successive applications of this operation in Section~\ref{subsec:reachable}.

We define the nested model in Section~\ref{sec:nested}. Theorem~\ref{thm:invariant} in Section~\ref{subsec:invariance} shows a key result,
that graphs and kernels are invariant to the precise order of fixing in
a sequence, as long as each fixing operation is valid. In Section~\ref{subsec:global}, Section~\ref{subsec:fact}, and Section~\ref{subsec:local}, respectively, we give definitions of the nested Markov model in terms
of the global Markov property, the factorization, and the ordered local
Markov property. Corollary~\ref{cor:complete-saturated} in Section~\ref{subsec:saturated} shows that the nested Markov model corresponding
to any complete ADMG gives the saturated model. In Section~\ref{subsec:latent-dag}, we show in Theorem~\ref{thm:dags_in_nested}, that
any marginal distribution in a hidden variable DAG model lies in the nested
Markov model associated with the latent projection. A simple characterization
of identifiable causal effects in hidden variable causal DAG models, based
on the fixing operation, is given in Theorem~\ref{thm:1-line-id} of Section~\ref{subsec:id}. Finally, a result stating that the nested Markov model
may be defined by the set of constraints found by the algorithm in
\citet{tian02on} is given as Theorem~\ref{thm:tian_equals_nested} in Section~\ref{subsec:tian}. All proofs of claims are in the Supplementary Material (\citet{supp}).

%s2 #&#
\section{Latent variable DAG models}
\label{sec:mixed}

We assume familiarity with standard graphical definitions relating to DAGs,
latent projections, and mixed graphs. See Section~A.1 in the Supplementary Material for a review.

%s2.1 #&#
\subsection{Basic concepts}
\label{ssec:mixedg}

The motivation for introducing mixed graphs is twofold. First, by removing
latent variables and replacing them with bidirected edges we simplify the
representation. For example, to perform a search, instead of considering
a potentially infinite class of DAGs with arbitrarily many latent variables,
we need only consider a finite set of mixed graphs. Second, although the
statistical models that we associate with mixed graphs capture many of
the constraints implied by latent variable models, the resulting model
will still, in general, be a superset of the set of distributions over
the observables that are implied by the original DAG with latents. The
use of a mixed graph to represent our model serves to emphasize that, in
spite of this connection, the set of distributions we are constructing
is nonetheless \emph{not} a latent variable model (see Example~\ref{ex:chsh}).

We will make use of standard genealogic relations in graphs. The sets of
parents, children, ancestors, descendants, and nondescendants of $a$ in
${\mathscr G} $ are written $\pavtex _{{\mathscr G}}(a)$, $\chvtex _{{\mathscr G}}(a)$, $\anvtex _{{\mathscr G}}(a)$,
$\devtex _{{\mathscr G}}(a)$, and $\ndvtex _{{\mathscr G}}(a)$, respectively. By convention, every vertex
is its own ancestor and descendant, so
$\anvtex _{{\mathscr G}}(a) \cap \devtex _{{\mathscr G}}(a) = \{a\}$. An ordering $\prec $ of nodes
in ${\mathscr G} $ is said to be \emph{topological} if for any pair of vertices
$a,b$, if $a \prec b$, then $a \notin \devtex _{{\mathscr G}}(b)$; note that this implies
that if $a \prec b$ then $a \neq b$. We define the set
$\pre _{{\mathscr G} , \prec}(b) \equiv \{ a   |  a \prec b \}$. We apply these
definitions disjunctively to sets, for example,~$\anvtex _{{\mathscr G}}(A) =
\bigcup_{a\in A} \anvtex _{{\mathscr G}}(a)$. A set of vertices $A$ in ${\mathscr G} $ is called
\emph{ancestral} if $\anvtex _{{\mathscr G}}(a) \subseteq A$ whenever $a \in A$.

Given a DAG ${\mathscr G} $ with vertex set $V\dotcup L$ where $L$ is the set of
latent variables, we associate with it a mixed graph, denoted
${\mathscr G} (V)$, via the standard operation of latent projection
$\sigma _{L}$ \citep{verma91theory}; see the Supplementary Material, Section~A.3 for the definition \citep{supp}. Here and elsewhere, we use
standard $\dotcup $ notation when we wish to emphasize that the sets in
the union are disjoint.

The latent projection ${\mathscr G} (V) = \sigma _{L}({\mathscr G} (V\dot\cup L))$ represents
the set of d-separation relations holding among the variables in $V$ in
${\mathscr G} $.

%p1 #&#
\begin{proposition}%
\label{prop:msep}
Let ${\mathscr G} $ be a DAG with vertex set $V\dotcup L$. For disjoint subsets
$A, B, C \subseteq V$ (where $C$ may be empty), $A$ is d-separated from
$B$ given $C$ in ${\mathscr G} $ if and only if $A$ is m-separated from $B$ given
$C$ in ${\mathscr G} (V)$.
\end{proposition}

However, as we will see later, the latent projection ${\mathscr G} (V)$ captures
much more than simply the d-separation relations holding in $V$. As suggested
by Figures~\ref{fig:mut}(i) and \ref{fig:verma1}(i), ${\mathscr G} (V)$ also represents
constraints such as (\ref{eq:verma}), and further all those found by the
algorithm in \citet{tian02on}, as well as others.

%s2.2 #&#
\subsection{Conditional ADMGs}
\label{subsec:cadmgs}

We define a \emph{conditional} acyclic directed mixed graph%\noxml{\break}
(CADMG)
${\mathscr G} (V,W)$ to be an ADMG with two disjoint sets of vertices $V$ and
$W$, which we call \textit{random} and \textit{fixed}, respectively. We require
that for all $w \in W$, $\pavtex _{{\mathscr G}}(w) = \emptyset $ and there are no bidirected
edges involving $w$.

Thus, in a CADMG ${\mathscr G} (V,W)$ there are no edges connecting vertices in
$W$, and all edges connecting $w\in W$ and $v \in V$ take the form
$w \to v$. Vertices in $V$ will represent random variables $X_{V}$, as
in a standard graphical model. The rationale for excluding edges between
vertices in $W$ or with arrowheads in $W$ is that the CADMG will represent
objects where the vertices in $W$ correspond to (nonrandom) values that index distributions
over $V$.

We note that CADMGs represent kernels that are not, in general, formed
by standard conditioning from the original observed distribution. For example,
a set of interventional distributions over $X_{V}$ indexed by possible
interventions on $X_{W}$; note, however, that we do not require a causal
interpretation of the graph. We also introduce operators
$\mathbb{V}({{\mathscr G}})$ and $\mathbb{W}({{\mathscr G}})$ that return, respectively, the
sets of random and fixed nodes associated with a CADMG ${\mathscr G} $. We will use
circular nodes to indicate the random vertices $\mathbb{V}({{\mathscr G}})$, and
square nodes to indicate the fixed vertices $\mathbb{W}({{\mathscr G}})$; see, for
instance, the CADMGs in Figures~\ref{fig:mut}(ii), \ref{fig:verma1}(ii)--(iv) and
\ref{fig:cadmg}(i). When the vertex sets are clear from context, we will
abbreviate ${{\mathscr G}}(V,W)$ as ${\mathscr G} $.

%f4 #&#
\begin{figure}
\begin{center}
%  \begin{tikzpicture}[>=stealth, node distance=1.2cm]
\begin{tikzpicture}[>=stealth, node distance=1.2cm,
pre/.style={->,>=stealth,ultra thick,line width = 1.4pt}]
    \tikzstyle{format} = [circle, draw, very thick, minimum size=5mm, inner sep=.3mm]
     \tikzstyle{square} = [draw, very thick, rectangle, minimum size=5mm]
   \begin{scope}
    \path[->]	node[square] (x1) {$1$} %{$\vrt{x_1}$}
    		node[format, right of=x1] (x2) {$2$} %{$\vrt{x_2}$}
                  (x1) edge[pre, blue] (x2)
		node[square, right of=x2] (x3) {$3$} %{$\vrt{x_3}$}
            %      (x2) edge[pre, blue] (x3)
            %      (x1) edge[pre,<->, bend right,red] (x3)
		node[format, right of=x3] (x4) {$4$} %{$\vrt{x_4}$}
                  (x3) edge[pre, blue] (x4)
                  (x2.45) edge[pre, <->, bend left, red] (x4.135);
   \node[below of=x2, xshift=.6cm, yshift = .2cm] {(i)};
  \end{scope}
  
    \begin{scope}[xshift=6cm]
%    \tikzstyle{format} = [circle, draw, ultra thick, line width=1.4pt, minimum size=5mm, inner sep=.5mm]
%     \tikzstyle{square} = [draw, ultra thick, rectangle, minimum size=5mm]
    \path[->]			node[format] (x1) {$1$} %{$\vrt{x_1}$}
    % node[format, rectangle] (x1) {$1$} %{$\vrt{x_1}$}
    		node[format, right of=x1] (x2) {$2$} %{$\vrt{x_2}$}
                  (x1) edge[pre, blue] (x2)
		node[format, right of=x2] (x3) {$3$} %{$\vrt{x_3}$}
%		node[format, rectangle, right of=x2] (x3) {$3$} %{$\vrt{x_3}$}
            %      (x2) edge[pre, blue] (x3)
            %      (x1) edge[pre,<->, bend right,red] (x3)
		node[format, right of=x3] (x4) {$4$} %{$\vrt{x_4}$}
                  (x3) edge[pre, blue] (x4)
                  (x2.45) edge[pre, <->, bend left, red] (x4.135)
                  (x1.315) edge[pre, <->, bend right, red] (x3.225);
   \node[below of=x2, xshift=.6cm, yshift = .2cm] {(ii)};
  \end{scope}

  \end{tikzpicture}

\end{center}
\caption{(i) A conditional mixed graph
${\mathscr G} (V = \{2,4\},W = \{1,3\})$ describing the structure of a kernel
$q_{24}(x_{2},x_{4} \cmid x_{1},x_{3})$. (ii) The corresponding graph
${\mathscr G} ^{|W}$ from which the conditional Markov property given by ${\mathscr G} $ may
be obtained by applying m-separation. The edge $1 \leftrightarrow 3$ is
added.}
\label{fig:cadmg}
\end{figure}

We will apply the standard genealogical definitions to CADMGs. Conversely,
since an ADMG ${\mathscr G} (V)$ may be seen as a CADMG in which
$W=\emptyset $, all subsequent definitions for CADMGs will also apply to
ADMGs.

The \emph{induced subgraph of a CADMG ${\mathscr G} (V,W)$ on a set $A$}, denoted
${\mathscr G} _{A}$, is a CADMG with $\mathbb{V}({{\mathscr G} _{A}})= V\cap A$ and
$\mathbb{W}({{\mathscr G} _{A}})=W\cap A$, and precisely those edges from
${\mathscr G} $ that have both endpoints in $A$. Note that in forming
${\mathscr G} _{A}$, the status of the vertices in $A$ with regard to whether they
are in $V$ or $W$ is preserved.

Given a CADMG ${\mathscr G} (V \dotcup L,W)$ and a set of latent variables $L$ that
are all random, we define the \emph{latent projection} ${\mathscr G} (V,W)$ to be
the CADMG obtained by applying the standard definition of latent projection
for ADMGs (see the Supplementary Material, Section~A.3).

%p2 #&#
\begin{proposition}
In a CADMG ${\mathscr G} (V\cup L,W)$ if $V\cup W$ is ancestral, then
${{\mathscr G}}_{V\cup W} \equiv ({\mathscr G} (V \cup L, W))_{V \cup W} = {{\mathscr G}}(V, W)$.
\end{proposition}
Thus, the induced subgraph on an ancestral set $V\cup W$ is the same as
the latent projection onto $V\cup W$.

%d3 #&#
\begin{definition}%
\label{def:district}
A set of vertices $C$ in a CADMG ${\mathscr G} $ is called
\emph{bidirected-connected} if for every pair of vertices $c,d \in C$ there
is a bidirected path in ${\mathscr G} $ between $c$ and $d$ with every vertex on
the path in $C$. A maximal bidirected-connected set of vertices in
$\mathbb{V}({\mathscr G} )$ is referred to as a \emph{district} in ${\mathscr G} $. Let
\begin{equation*}
{\mathcal D}({\mathscr G} ) \equiv \{ D \mid D \text{ is a district in } {\mathscr G} \}
\end{equation*}
be the set of districts in ${\mathscr G} $. For $v \in \mathbb{V}({\mathscr G} )$, let
$\dis _{{\mathscr G}}(v)$ be the district containing $v$ in ${\mathscr G} $. We write
$\dis _{A}(v)$ as a shorthand for $\dis _{{\mathscr G} _{A}}(v)$, the district of
$v$ in the induced subgraph ${\mathscr G} _{A}$.
\end{definition}

\citet{tian02on} refer to districts in ADMGs as ``c-components.'' In models
associated with ADMGs, districts are used to specify variable partitions
that define terms in the factorization of observed random variables. For
this reason, districts in CADMGs include only random vertices. In a DAG
${\mathscr G} (V)$, the districts ${\mathcal D}({\mathscr G} ) = \{ \{v\} \mid v\in V\}$ are
the singleton subsets of $V$, these are naturally associated with the usual Markov
factors $p(x_{v} \cmid x_{\pavtex _{{\mathscr G}}(v)})$.

%s2.3 #&#
\subsection{Kernels}
\label{subsec:kernels}

We consider collections of real-valued random variables
$(X_{v})_{ v\in V}$ taking values in $({\mathfrak X}_{v})_{v\in V}$. For
$A\subseteq V$, we let
${\mathfrak X}_{A}\equiv \bigtimes _{v\in A} ({\mathfrak X}_{v})$, and
$X_{A}\equiv (X_{v})_{v\in A}$. Here, $v$ denotes a vertex and
$X_{v}$ the corresponding random variable, likewise $A$ denotes a vertex
set and $X_{A}$ is the vector $(X_{v} : v \in A)$.

Following \citeauthor{lauritzen96graphical} (\citeyear{lauritzen96graphical}), p.~46, we define a \emph{kernel} over
$\mathfrak{X}_{V}$ and indexed by $\mathfrak{X}_{W}$ to be a nonnegative
function $q_{V}(x_{V}\cmid x_{W})$ satisfying
%
%e4 #&#
\begin{equation}
\label{eq:cond} \sum_{x_{V} \in {\mathfrak{X}_{V}}} q_{V}(x_{V}
\cmid x_{W}) = 1,\quad\text{for all }x_{W} \in
\mathfrak{X}_{W}.
\end{equation}
We use the term ``kernel'' and write $q_{V}(\cdot | \cdot )$ (rather than
$p(\cdot | \cdot )$) to emphasize that these functions, though they satisfy
(\ref{eq:cond}), and thus most properties of conditional densities, will
not, in general, be formed via the usual operation of conditioning on the
event $X_{W}=x_{W}$. Following standard conventions for densities, we will
use $q_{V}(x_{V} \cmid x_{W})$ to refer either to the function itself,
or to its realization under specific values of $x_{V},x_{W}$, with the
precise meaning being clear from context. To conform with standard notation
for densities, we define for every $A \subseteq V$
%
%e5 #&#
\begin{align}
\label{eqn:kernel-marg-cond} q_{V}(x_{A} \cmid x_{W}) \equiv
\sum_{x_{V\setminus A} \in {
\mathfrak{X}_{V\setminus A}}} q_{V}(x_{V} \cmid
x_{W});\qquad %
 q_{V}(x_{V \setminus A} \cmid
x_{W \cup A}) \equiv \frac{q_{V}(x_{V} \cmid x_{W})}{q_{V}(x_{A} \cmid x_{W})}.
\end{align}
In general, for $R,S \subseteq V$ and $T \subset W$, the density
$q_{V}(x_{R} \cmid x_{S\cup T})$ may not be defined, since in the absence
of a distribution over $X_{W}$, we cannot integrate out the variables
$X_{W\setminus T}$ (though see Definition~\ref{def:ind} and subsequent
comments below). For disjoint $V_{1} \dot{\cup} V_{2} = V$ and
$W_{1} \dot{\cup} W_{2} = W$, we will sometimes write
$q_{V}(x_{V_{1}}, x_{V_{2}} \cmid x_{W_{1}}, x_{W_{2}})$ to mean
$q_{V}(x_{V_{1} \cup V_{2}} \cmid x_{W_{1} \cup W_{2}})$.

%s2.4 #&#
\subsection{Independence in kernels}
\label{subsec:kernel-indep}

We extend the notion of conditional independence to kernels over
$\mathfrak{X}_{V}$ indexed by $\mathfrak{X}_{W}$. A rigorous treatment
of conditional independence in settings where not all variables are random
was given in \citet{naiya13ci}.

%d4 #&#
\begin{definition}%
\label{def:ind}
For disjoint subsets $A, B, C \subseteq V\cup W$, we define $X_{A}$ to
be conditionally independent of $X_{B}$ given $X_{C}$ in a kernel
$q_{V}$, written
\begin{equation*}
X_{A} \civtex X_{B} \mid X_{C}
\quad[q_{V}]
\end{equation*}
if \textbff{either}:
\begin{itemize}%[\textup{(a)}]
\item[\textup{(a)}] $A \cap W = \emptyset $ and
$q_{V}(x_{A} \cmid x_{B}, x_{C}, x_{W\setminus (B\cup C)})$ is a function
only of $x_{A}$ and $x_{C}$ (whenever this kernel is defined). In this
case, we define
$q_{V}(x_{A}\cmid x_{C}) \equiv q_{V}(x_{A} \cmid x_{B}, x_{C}, x_ {W
\setminus (B\cup C)})$,

{\textbff{or}}

\item[\textup{(b)}] $B \cap W = \emptyset $ and
$q_{V}(x_{B} \cmid x_{A}, x_{C}, x_ {W\setminus (A \cup C)})$ is a function
only of $x_{B}$ and $x_{C}$ (whenever this kernel is defined). In this
case, we define
$q_{V}(x_{B}\cmid x_{C}) \equiv q_{V}(x_{B} \cmid x_{A}, x_{C}, x_ {W
\setminus (B\cup C)})$.
\end{itemize}
%was: longlist environment
%
\end{definition}
 See Figure~\ref{fig:venn} for an illustration of cases; it is necessary
for either $A$ or $B$ not to intersect $W$ because there is no joint distribution
over elements in $W$. This definition reduces to ordinary conditional independence
in the case $W = \emptyset $. The condition that the kernel should be defined
simply addresses the situation where the conditioning event has zero probability
density. We remark that if $(A \cup B) \cap W=\emptyset $, then both definitions
hold (or fail to hold) at the same time; the conditions become equivalent
to saying that the kernel
$q_{V}(x_{A}, x_{B} \cmid x_{C}, x_{W \setminus C})$ should factorize into
a piece that depends only upon $x_{A},x_{C}$ and a piece that depends only
upon $x_{B},x_{C}$.

%f5 #&#
\begin{figure}
\begin{center}
%  \begin{tikzpicture}[>=stealth, node distance=1.2cm]
\begin{tikzpicture}[>=stealth, node distance=1.0cm,
pre/.style={->,>=stealth,ultra thick,line width = 1.4pt}]
  \begin{scope}
    \tikzstyle{format} = [inner sep=2pt,draw, thick, ellipse, line width=1.4pt, minimum size=6mm, rounded
        corners=3mm]
	  \draw[thick, rounded corners] (3.5,0) rectangle (0,2);
	  \draw[thick] (2.2,0) to (1.8,2);
	  \node[ellipse, draw, minimum width=1.0cm, xshift=0.8cm,yshift=0.9cm] (a) {$A$};
	  \node[ellipse, draw, minimum width=1.0cm, xshift=1.7cm,yshift=1.5cm] (b) {$B$};
	  \node[ellipse, draw, minimum width=1.0cm, xshift=1.9cm,yshift=0.5cm] (c) {$C$};
	  \node[xshift=0.2cm,yshift=1.8cm] (v) {$V$};
	  \node[xshift=3.3cm,yshift=1.8cm] (v) {$W$};

  \path[]	node[yshift = -0.5cm, xshift=1.75cm] (label) {(a)};
  \end{scope}
  \begin{scope}[xshift=4.0cm]
	  \draw[thick, rounded corners] (3.5,0) rectangle (0,2);
	  \draw[thick] (2.2,0) to (1.8,2);
	  \node[ellipse, draw, minimum width=1.0cm, xshift=0.8cm,yshift=0.9cm] (b) {$B$};
	  \node[ellipse, draw, minimum width=1.0cm, xshift=1.7cm,yshift=1.5cm] (a) {$A$};
	  \node[ellipse, draw, minimum width=1.0cm, xshift=1.9cm,yshift=0.5cm] (c) {$C$};
	  \node[xshift=0.2cm,yshift=1.8cm] (v) {$V$};
	  \node[xshift=3.3cm,yshift=1.8cm] (v) {$W$};

  \path[]	node[yshift = -0.5cm, xshift=1.75cm] (label) {(b)};
  \end{scope}
\end{tikzpicture}
\end{center}
\caption{Illustration of cases in Definition~\protect\ref{def:ind}:%
 (a) $A \cap W = \emptyset $; (b) $B \cap W = \emptyset $.}
\label{fig:venn}
\end{figure}

Note that the kernels appearing in (a) and (b) specify values for all of
the variables $X_{W}$, and are defined via conditioning in (a margin of)
the original kernel $q_{V}$. For example, in (a),
\begin{equation*}
q_{V}(x_{A} \cmid x_{B}, x_{C},
x_ {W\setminus (B\cup C)}) = q_{V}(x_{V
\cap (A\cup B\cup C)} \cmid x_{W} ) /
q_{V}(x_{V\cap (B\cup C)} \cmid x_{W}).
\end{equation*}

%p5 #&#
\begin{proposition}
\label{prop:add-w}
In a kernel $q_{V}(x_{V} \cmid x_{W})$, $X_{A} \civtex X_{B} \mid X_{C}$ if
and only if either $X_{A} \civtex X_{B\cup (W\setminus C)} \mid X_{C}$ or
$X_{B} \civtex X_{A\cup (W\setminus C)} \mid X_{C}$.
\end{proposition}

%s2.5 #&#
\subsection{Semigraphoid axioms in kernels}
\label{subsubsec:kernel-semi-graphoid}

Classical conditional independence constraints may logically imply other
such constraints. Though no finite axiomatization of these connections
is possible \citep{studeny92conditional}, deductive derivations of conditional
independence constraints in DAGs can be restricted to the semigraphoid
axioms of symmetry and the ``chain rule'' \citep{dawid79cond}, which we
reproduce here:
\begin{align*}
(X_{A} \civtex X_{B} \mid X_{C}) \quad&
\Leftrightarrow \quad(X_{B} \civtex X_{A} \mid X_{C}),
\\
(X_{A} \civtex X_{B} \mid X_{C \cup D}) \land
(X_{A} \civtex X_{D} \mid X_{C}) \quad&\Leftrightarrow\quad
(X_{A} \civtex X_{B \cup D} \mid X_{C}).
\end{align*}
(The chain rule axiom is sometimes written as the three separate axioms
of contraction, decomposition, and weak union.) We now show that, unsurprisingly,
conditional independence constraints defined for kernels also obey these
axioms. An additional set of axioms called \emph{separoids} has been shown
to apply to versions of conditional independence involving nonstochastic
variables like $X_{W}$ \citep{naiya13ci}.

%p6 #&#
\begin{proposition}
\label{prop:semigraphoid}
The semigraphoid axioms are sound for kernel independence.
\end{proposition}

%s2.6 #&#
\subsection{Constructing kernels}
\label{subsec:kernel-properties}

We will typically construct new kernels via the operation of dividing either
a distribution $p(x_{V})$ by $p(x_{H} \cmid x_{T})$ or an existing kernel
$q_{V}(x_{V} \cmid x_{W})$ by $q_{V}(x_{H} \cmid x_{T \cup W})$, where
$H\dotcup T \subseteq V$. For the results in the remainder of this section,
we will consider a kernel $q_{V}(x_{V} \cmid x_{W})$ where
$V = R \dotcup H \dotcup T$, and a new object
%
%e6 #&#
\begin{align}
q^{*}_{V \setminus H}(x_{V \setminus H} \cmid x_{H},
x_{W}) = q^{*}_{V
\setminus H}(x_{R},
x_{T} \cmid x_{H}, x_{W})&\equiv
\frac{q_{V}(x_{R}, x_{H}, x_{T} \cmid x_{W})}{
q_{V}(x_{H} \cmid x_{T}, x_{W})}. \label{eq:invariance-one}
\end{align}
See Figure~\ref{fig:venn2} for an illustration.

%f6 #&#
\begin{figure}
\begin{center}
%  \begin{tikzpicture}[>=stealth, node distance=1.2cm]
\begin{tikzpicture}[>=stealth, node distance=1.0cm,
pre/.style={->,>=stealth,ultra thick,line width = 1.4pt}]
  \begin{scope}
    \tikzstyle{format} = [inner sep=2pt,draw, thick, ellipse,
line width=1.4pt, minimum size=6mm, rounded corners=3mm ]
          \draw[thick, %rounded corners
] (4.5,0) rectangle (0,2);
          \draw[thick] (3.2,0) to (2.8,2);
\fill[gray] (3.2,0) to (2.8,2) to (4.5,2) to (4.5,0);
            \draw[thick] (2.2,0) to (1.8,2);
             \draw[thick] (1.2,0) to (0.8,2);
%         \node[ellipse, draw, minimum width=1.0cm, xshift=0.8cm,yshift=0.9cm] (a) {$A$};
%         \node[ellipse, draw, minimum width=1.0cm, xshift=1.7cm,yshift=1.5cm] (b) {$B$};
%         \node[ellipse, draw, minimum width=1.0cm, xshift=1.9cm,yshift=0.5cm] (c) {$C$};
          \node[xshift=0.2cm,yshift=1.8cm] (r) {$R$};
           \node[xshift=1.2cm,yshift=1.8cm] (h) {$H$};
             \node[xshift=2.2cm,yshift=1.8cm] (t) {$T$};
          \node[xshift=4.2cm,yshift=1.8cm, %red
          	white] (w) {$W$};
\path[] node[yshift = -0.5cm, xshift=1.75cm] (label) {$q_V(x_V \cmid x_W)$};
  \end{scope}
   \begin{scope}[xshift=6cm]
    \tikzstyle{format} = [inner sep=2pt,draw, thick, ellipse, line width=1.4pt, minimum
size=6mm, rounded
        corners=3mm]
          \draw[thick, %rounded corners
] (4.5,0) rectangle (0,2);
          \draw[thick] (3.2,0) to (2.8,2);
            \draw[thick] (2.2,0) to (1.8,2);
             \draw[thick] (1.2,0) to (0.8,2);

\fill[gray] (3.2,0) to (2.8,2) to (4.5,2) to (4.5,0);

\fill[gray] (1.8,2) to (0.8,2) to (1.2,0) to (2.2,0);

%         \node[ellipse, draw, minimum width=1.0cm, xshift=0.8cm,yshift=0.9cm] (a) {$A$};
%         \node[ellipse, draw, minimum width=1.0cm, xshift=1.7cm,yshift=1.5cm] (b) {$B$};
%         \node[ellipse, draw, minimum width=1.0cm, xshift=1.9cm,yshift=0.5cm] (c) {$C$};
          \node[xshift=0.2cm,yshift=1.8cm] (r) {$R$};
           \node[xshift=1.2cm,yshift=1.8cm, %red
           	white] (h) {$H$};
             \node[xshift=2.2cm,yshift=1.8cm] (t) {$T$};
          \node[xshift=4.2cm,yshift=1.8cm, %red
          	white] (w) {$W$};
 \path[]        node[yshift = -0.5cm, xshift=1.75cm] (label) {$q^*_{V\setminus
H}(x_{V\setminus H} \cmid x_{W\cup H})$};
  \end{scope}

%  \begin{scope}[xshift=4.0cm]
%         \draw[thick, rounded corners] (3.5,0) rectangle (0,2);
%         \draw[thick] (2.2,0) to (1.8,2);
%         \node[ellipse, draw, minimum width=1.0cm, xshift=0.8cm,yshift=0.9cm] (b) {$B$};
%         \node[ellipse, draw, minimum width=1.0cm, xshift=1.7cm,yshift=1.5cm] (a) {$A$};
%         \node[ellipse, draw, minimum width=1.0cm, xshift=1.9cm,yshift=0.5cm] (c) {$C$};
%         \node[xshift=0.2cm,yshift=1.8cm] (v) {$V$};
%         \node[xshift=3.3cm,yshift=1.8cm] (v) {$W$};
%
%  \path[]      node[yshift = -0.5cm, xshift=1.75cm] (label) {$(b)$};
%  \end{scope}
\end{tikzpicture}
\end{center}
\caption{Structure of sets for invariance properties considered in Section~\protect\ref{subsec:kernel-properties}; $V = R\dotcup H \dotcup T$; shaded sets
are fixed; see equation (\protect\ref{eq:invariance-one}). Note that here vertices
earlier in the ordering are on the right, so
$R\succ H\succ T \succ W$.}
\label{fig:venn2}
\end{figure}

%l7 #&#
\begin{lemma}
\label{lem:fixing-yields-kernel}
$q^{*}_{V\setminus H}(x_{R}, x_{T}\cmid x_{H}, x_{W})$ is a kernel.
\end{lemma}

%l8 #&#
\begin{lemma}
\label{lem:invariance-kernel}
For the kernel constructed in \textup{(\ref{eq:invariance-one})},
%
%e7 #&#
\begin{align}
q^{*}_{V\setminus H}(x_{R}, x_{T}\cmid
x_{H}, x_{W}) &= q_{V}(x_{R} \cmid
x_{H}, x_{T}, x_{W}) q_{V}(x_{T}
\cmid x_{W}), \label{eq:invariance-two}
\end{align}
and hence
%
%e8 #&#
%e9 #&#
\begin{align}
q^{*}_{V\setminus H}(x_{R}\cmid x_{H},
x_{T}, x_{W}) &= q_{V}(x_{R} \cmid
x_{H}, x_{T}, x_{W}); \label{eq:invariance-cond}
\\
q^{*}_{V\setminus H}(x_{T} \cmid x_{H},
x_{W}) &= q_{V}(x_{T} \cmid x_{W}) =
q^{*}_{V\setminus H}(x_{T} \cmid x_{W}) .
\label{eq:invariance-marg}
\end{align}
\end{lemma}
 Note that if $R = \emptyset $,
$q^{*}_{V \setminus H}(x_{V \setminus H} \cmid x_{H}, x_{W}) = \sum_{x_{H}}
q_{V}(x_{V} \cmid x_{W}) = q_{V}(x_{T} \cmid x_{W})$.

%p9 #&#
\begin{proposition}[Separation]
\label{prop:constructed}
For the kernel constructed in \textup{(\ref{eq:invariance-one})}:
\begin{itemize}%[{\textup{(iii)}}]
\item[{\textup{(i)}}] $(X_{H} \civtex X_{T} \mid X_{W})$ holds in
$q^{*}_{V\setminus H}(x_{V \setminus H} \cmid x_{H}, x_{W})$.
\item[{\textup{(ii)}}] If
$X_{V} \civtex X_{W \setminus W_{1}} \mid X_{W_{1}}$ in $q_{V}$, then
$X_{H \cup (W \setminus W_{1})} \civtex X_{T} \mid X_{W_{1}}$ in
$q^{*}_{V \setminus H}$.
\item[{\textup{(iii)}}] If $R = \emptyset $,
$X_{H} \civtex X_{V \setminus H} \cmid X_{W}$ in $q^{*}_{V \setminus H}$.
\end{itemize}
%was: longlist environment
 %
%
\end{proposition}

%s2.7 #&#
\subsection{Preservation of existing independences in a kernel}
\label{subsec:kernel-preserve}

We now state two important properties that transfer conditional independence
statements to a new kernel formed via the operation (\ref{eq:invariance-one}).

%p10 #&#
\begin{proposition}[Ordering]
\label{prop:ind-preserve-1}
Given disjoint sets $A, B,C$, where $C$ may be empty, if
$A,B,C \subseteq T \cup W$, then
%
%e10 #&#
\begin{align}
X_{A} \civtex X_{B} \mid X_{C} \quad [q_{V}]\quad \quad\Leftrightarrow\quad\quad X_{A} \civtex X_{B} \mid
X_{C}\quad \bigl[q^{*}_{V\setminus H}\bigr]. \label{eqn:ordering}
\end{align}
\end{proposition}
 This result, which follows directly from (\ref{eq:invariance-marg}), states
that, given any ordering of variables in which $W \cup T$ precedes
$H$, which in turn precedes $R$, performing the operation in (\ref{eq:invariance-one})
on $X_{H}$ preserves conditional independence statements among variables
that precede $H$ in the ordering; see Figure~\ref{fig:venn3}(i). Note that,
by Definition~\ref{def:ind}, both independences in (\ref{eqn:ordering})
imply either $A \subseteq T$ or $B \subseteq T$.

%f7 #&#
\begin{figure}
\begin{center}
%  \begin{tikzpicture}[>=stealth, node distance=1.2cm]
\begin{tikzpicture}[>=stealth, node distance=1.0cm,
pre/.style={->,>=stealth,ultra thick,line width = 1.4pt}]
  \begin{scope}
    \tikzstyle{format} = [inner sep=2pt,draw, thick, ellipse,
line width=1.4pt, minimum size=6mm, rounded corners=3mm ]
          \draw[thick, %rounded corners
] (4.5,0) rectangle (0,2);
          \draw[thick] (3.4,0) to (3.0,2);
%\fill[gray]  (3.2,0) to (2.8,2) to (4.5,2) to (4.5,0);
            \draw[thick] (2.2,0) to (1.8,2);
             \draw[thick] (1.2,0) to (0.8,2);
         \node[ellipse, draw, minimum width=1.0cm, xshift=2.6cm,yshift=1.6cm, scale=0.6] (b) {$A$};
          \node[ellipse, draw, blue, minimum width=2.0cm, xshift=3.1cm,yshift=1.0cm, scale=0.6] (a) {$B\quad\quad\;$};
         \node[ellipse, draw, green, minimum width=2.0cm, xshift=3.2cm,yshift=0.4cm, scale=0.6] (c) {$C\quad\quad\;$};
          \node[xshift=0.2cm,yshift=1.8cm] (r) {$R$};
           \node[xshift=1.15cm,yshift=1.8cm,red] (h) {$H$};
             \node[xshift=2.1cm,yshift=1.8cm] (t) {$T$};
          \node[xshift=4.2cm,yshift=1.8cm, red] (w) {$W$};
\path[] node[yshift = -0.5cm, xshift=1.75cm] (label) {(i) Independence in $q_V(x_T \cmid x_W)$.};
  \end{scope}
   \begin{scope}[xshift=7cm]
    \tikzstyle{format} = [inner sep=2pt,draw, thick, ellipse, line width=1.4pt, minimum
size=6mm, rounded
        corners=3mm]
          \draw[thick, %rounded corners
] (4.5,0) rectangle (0,2);
          \draw[thick] (3.8,0) to (3.4,2);
            \draw[thick] (2.8,0) to (2.4,2);
             \draw[thick] (1.8,0) to (1.4,2);
              \draw[very thick, blue] (0.8,1.97) to (4.2,1.97) to (4.3,1.01) to (0.6,1.01) to (0.8,1.97);
               \draw[very thick, green] (1.2,0.97) to (4.2,0.97) to (4.3,0.03) to (1.0,0.03) to (1.2,0.97);
%\fill[gray]  (3.2,0) to (2.8,2) to (4.5,2) to (4.5,0);
%\fill[gray] (1.8,2) to (0.8,2) to (1.2,0) to (2.2,0);
          \node[xshift=0.2cm,yshift=1.7cm] (r) {$R$};
           \node[xshift=1.7cm,yshift=1.7cm, red] (h) {$H$};
             \node[xshift=2.8cm,yshift=1.7cm] (t) {$T$};
          \node[xshift=4.2cm,yshift=1.7cm, red] (w) {$W$};
           \node[xshift=0.8cm,yshift=1.2cm, scale=0.6, blue] (t) {$B$};
             \node[ellipse, draw, thick, minimum width=1.0cm, xshift=0.6cm,yshift=0.5cm, scale=0.6] (a) {$A$};
              \node[xshift=1.2cm,yshift=0.2cm, scale=0.6, green] (t) {$C$};
 \path[]        node[yshift = -0.5cm, xshift=2.3cm] (label) {(ii) Independence in $q_V(x_R \cmid x_H, x_T, x_W)$.};
  \end{scope}

%  \begin{scope}[xshift=4.0cm]
%         \draw[thick, rounded corners] (3.5,0) rectangle (0,2);
%         \draw[thick] (2.2,0) to (1.8,2);
%         \node[ellipse, draw, minimum width=1.0cm, xshift=0.8cm,yshift=0.9cm] (b) {$B$};
%         \node[ellipse, draw, minimum width=1.0cm, xshift=1.7cm,yshift=1.5cm] (a) {$A$};
%         \node[ellipse, draw, minimum width=1.0cm, xshift=1.9cm,yshift=0.5cm] (c) {$C$};
%         \node[xshift=0.2cm,yshift=1.8cm] (v) {$V$};
%         \node[xshift=3.3cm,yshift=1.8cm] (v) {$W$};
%
%  \path[]      node[yshift = -0.5cm, xshift=1.75cm] (label) {$(b)$};
%  \end{scope}
\end{tikzpicture}
\end{center}
\caption{Independences that are preserved after fixing in a kernel described
in (i) Proposition~\protect\ref{prop:ind-preserve-1} (where $A\subseteq T$), and
(ii) Proposition~\protect\ref{prop:ind-preserve-2}.}
\label{fig:venn3}
\end{figure}

%p11 #&#
\begin{proposition}[Modularity]
\label{prop:ind-preserve-2}
Given disjoint sets $A, B,C$, where $C$ may be empty, if
$A \subseteq R$ and $(B \cup C) \supseteq H \cup T$, then
\begin{equation*}
X_{A} \civtex X_{B} \mid X_{C}\quad [q_{V}] \quad\quad\Leftrightarrow\quad\quad X_{A} \civtex X_{B} \mid
X_{C}\quad \bigl[q^{*}_{V\setminus H}\bigr].
\end{equation*}
\end{proposition}
 This result, which follows directly from (\ref{eq:invariance-cond}), is
illustrated in Figure~\ref{fig:venn3}(ii). In words, it says that if we
express a probability distribution as a set of factors via the chain rule
of probability, and a conditional independence statement can be stated
exclusively in terms of one of the factors, then applying the operation
in (\ref{eq:invariance-one}) such that another factor is dropped does not
affect this conditional independence statement. In other words, ``factors
are modular.''

%s2.8 #&#
\subsection{Markov properties for CADMGs}
\label{subsec:markov-cadmg}

As described earlier, a CADMG ${\mathscr G} (V,W)$ represents the independence structure
of a kernel $q_{V}(x_{V}\cmid x_{W})$. We now introduce a number of Markov
properties for a given kernel, the equivalence of which we will state in
Section~\ref{subsec:markov}.%

%s2.8.1 #&#
\subsubsection{The CADMG global Markov property}
\label{subsec:cadmg-global}

The global Markov property for CADMGs may be derived from m-separation
via the following simple construction.

%d12 #&#
\begin{definition}
Given a CADMG ${\mathscr G} (V,W)$, we define ${\mathscr G} ^{|W}$ to be a mixed graph with
vertex set $V^{*}=V \cup W$, and edge set
\begin{equation*}
E^{*} \equiv E \cup \{ w\leftrightarrow w^{\prime }\mid w,w^{\prime }
\in W\}.
\end{equation*}
\end{definition}
 In words, the graph ${\mathscr G} ^{|W}$ is formed from ${\mathscr G} $ by adding bidirected
edges between all pairs of vertices $w,w^{\prime }\in W$, and then eliminating
the distinction between vertices in $V$ and $W$; see Figure~\ref{fig:cadmg}(ii)
for an example.

%d13 #&#
\begin{definition}
A kernel $q_{V}$ satisfies the \textit{global Markov property} for a CADMG
${\mathscr G} (V,W)$ if for arbitrary disjoint sets $A,B,C$, ($C$ may be empty),
\begin{align*}
A \text{ is m-separated from }B\text{ given }C\text{ in }{\mathscr G} ^{|W}
\quad\implies \quad X_{A}\civtex X_{B} \mid X_{C}
\quad[q_{V}].
\end{align*}
\end{definition}
 We denote this set of kernels ${\mathcal{P}}^{\mathrm{c}}_{\mathrm{m}}({\mathscr G} )$, where ``\textup{c}''
and ``\textup{m}'' denote \emph{C}ADMG and \textit{m}-separation, respectively.

%s2.8.2 #&#
\subsubsection{The CADMG local Markov property}

The (ordered) local Markov property for a DAG states that each vertex $v$ is independent
of vertices prior to $v$ (under a topological ordering) conditional on
the parents of $v$. In a CADMG, the Markov blanket plays a role analogous
to that of the set of parents.

If $t\in V$, then the \emph{Markov blanket} of $t$ in ${\mathscr G} (V,W)$ is defined
as
%
%e11 #&#
\begin{equation}
\label{eq:mb2} \mbvtex _{{\mathscr G}}({t}) \equiv \pavtex _{{\mathscr G}}
\bigl( \dis _{{\mathscr G}}(t) \bigr) \cup \bigl(\dis _{{\mathscr G}}(t)\setminus
\{t\} \bigr).
\end{equation}

Given a vertex $t \in V$ such that $\chvtex _{{\mathscr G}}(t) = \emptyset $, a kernel
$q_{V}$ obeys the \emph{local Markov property for ${\mathscr G} $ at $t$} if
%
%e12 #&#
\begin{equation}
\label{eq:local} X_{t} \civtex X_{(V\cup W)\setminus (\mbvtex _{{\mathscr G}}({t})\cup \{t\})} \mid X_{\mbvtex _{{\mathscr G}}({t})}
\quad[q_{V}].
\end{equation}

 If $\prec $ is a topological total ordering on the vertices in ${\mathscr G} $, then
for a subset $A$ define $\max_{\prec}(A)$ to be the $\prec $-greatest
vertex in $A$.

We define the set of kernels obeying the \emph{ordered local Markov property
for the CADMG ${\mathscr G} (V,W)$ under the ordering $\prec $} as follows:
%
%e13 #&#
\begin{align}\label{eq:local-one-step-model}
{\mathcal{P}}^{\mathrm{c}}_{\mathrm{l}}({\mathscr G} ,\prec ) \equiv {}&\Bigl\{
q_{V}(x_{V} \cmid x_{W}) %\mid
\;\left|\; \vphantom{f_A^D} \right.
\text{for every
ancestral set }A, \text{with } { \max_{
\prec}}(A) \in V,
\nonumber
\\
& q_{V}(x_{V\cap A}\cmid x_{W}) \text{ obeys the
local Markov property for }
\\
& {\mathscr G} (V\cap A,W)\text{ at } { \max_{
\prec}}(A)\Bigr\}.
\nonumber
\end{align}

We will make use of the following extension of the Markov blanket notation: for an ancestral set
$A$ in a CADMG ${\mathscr G}$ and a vertex $t \in V \cap A$ such that
$\chvtex _{{\mathscr G}}(t) \cap A = \emptyset $, let
%
%e14 #&#
\begin{equation}
\label{eq:mb-in-A} \mbvtex _{{\mathscr G}}(t,A) \equiv \mbvtex _{{\mathscr G} _{A}}(t).
\end{equation}

%p14 #&#
\begin{proposition}%
\label{prop:mb}
Given a CADMG ${\mathscr G} $, an ancestral set $A$ and a random vertex
$t \in A$ such that $\chvtex _{{\mathscr G}}(t) \cap A = \emptyset $:
\begin{itemize}%[\textup{(ii)}]
\item[\textup{(i)}]
$\mbvtex _{{\mathscr G}}(t,A) \subseteq \mbvtex _{{\mathscr G}}(t) \subseteq D \cup \pavtex _{{\mathscr G}}(D)$,
where $t \in D \in {\mathcal{D}}({\mathscr G} )$;
\item[\textup{(ii)}] if $A^{*}$ is an ancestral set and
$t \in A^{*} \subseteq A$, then
$\mbvtex _{{\mathscr G}}(t,A^{*}) \subseteq \mbvtex _{{\mathscr G}}(t,A)$.
\end{itemize}
%was: longlist environment
%
\end{proposition}

%s2.8.3 #&#
\subsubsection{The CADMG augmented Markov property}

The following is the analog of moralization in DAGs for CADMGs. For a CADMG
${\mathscr G} (V,W)$, the \textit{augmented graph derived from ${\mathscr G} $}, denoted
$({\mathscr G} )^{a}$, is an undirected graph with vertex set $V \cup W$ such that
$c \text{ --- } d \text{ in } ({\mathscr G} )^{a}$ if and only if $c$ and $d$ are connected
by a path in ${\mathscr G} ^{|W}$ containing only colliders (or a single edge).

A kernel $q_{V}$ obeys the
\textit{augmented Markov property for a CADMG ${\mathscr G} (V, W)$} if
$X_{A}\civtex X_{B} \mid X_{C}$ in $q_{V}$ whenever: for arbitrary disjoint
sets $A,B,C$ ($C$ may be empty), every path in
$({\mathscr G} _{\anvtex (A\cup B\cup C)})^{a}$ from any $a \in A$ to any
$b \in B$ passes through a vertex in $C$. We denote the set of such kernels
by ${\mathcal{P}}^{\mathrm{c}}_{\mathrm{a}}({\mathscr G} )$.

%s2.9 #&#
\subsection{Tian factorization for CADMGs}
\label{subsec:fact-cadmg}

The joint distribution under a DAG model may be factorized into univariate
densities. In DAG models, these factors take the form
$p(x_{a} \cmid x_{\pavtex _{{\mathscr G}}(a)})$. This factor is a conditional distribution
for a singleton variable $X_{a}$, given the set of variables corresponding
to parents of $a$ in the graph. The factorization property may be generalized
to CADMGs by requiring factorization of $q_{V}$ into kernels over districts.

We define the set of kernels that \emph{Tian factorize} with respect to
a CADMG:
%
%e15 #&#
\begin{align}\label{eq:tian-factor}\begin{split}
{\mathcal{P}}^{\mathrm{c}}_{\mathrm{f}}({\mathscr G} ) \equiv{}& \biggl\{
q_{V} \;\left|\;\, \text{for every ancestral set }A, \text{ there exist
kernels }f_{D}^{A}( \cdot | \cdot ) \right.
\\
& \text{s.t. } q_{V}(x_{V\cap A} \cmid x_{W}) =
\prod_{D \in {\mathcal{D}}({\mathscr G} _{A})} f_{D}^{A}(x_{D}
\cmid x_{\pavtex _{
{\mathscr G}}(D)\setminus D}) \biggr\}.
\end{split}
\end{align}
In the next lemma, we show that the terms
$f_{D}^{A}(\cdot \cmid \cdot )$ arising in (\ref{eq:tian-factor}) are equal
to products of univariate conditional densities, that is,~instances of
the g-formula of \citet{robins86new}, and thus $f_D^A(\cdot \cmid \cdot)$ does not depend on $A$ other than through $D$.

%l15 #&#
\begin{lemma}%
\label{lem:markovfact}
Let ${\mathscr G} $ be a CADMG with topological ordering $\prec $. If
$q_{V} \in {\mathcal{P}}^{\mathrm{c}}_{\mathrm{f}}({\mathscr G} )$, then for every ancestral set
$A$ and every $D \in {\mathcal{D}}({\mathscr G} _{A})$, we have
%
%e16 #&#
\begin{equation}
\label{eq:district-factor-ident} f_{D}^{A}(x_{D} \cmid
x_{\pavtex (D) \setminus D}) = \prod_{d \in D}
q_{V}(x_{d} \cmid x_{T^{(d,D)}_{\prec}}),
\end{equation}
where
$T^{(d,D)}_{\prec }\equiv \mbvtex _{{\mathscr G}}( d,  \anvtex _{{\mathscr G}}(D) \cap
{\color{blue}(\pre _{{\mathscr G} ,\prec}(d) \cup \{ d \})} )$, so that
%
%e17 #&#
\begin{equation}
\label{eq:markovb} q_{V}(x_{d} \cmid x_{T^{(d,D)}_{\prec}}) =
q_{V}(x_{d} \cmid x_{A
\cap \pre _{{\mathscr G} ,\prec}(d)}, x_{W}).
\end{equation}
Conversely, if \textup{(\ref{eq:markovb})} holds for all $d \in A$ and ancestral
sets $A$, and the $f_{D}^{A}$ functions are defined by \textup{(\ref{eq:district-factor-ident})},
then $q_{V} \in {\mathcal{P}}^{\mathrm{c}}_{\mathrm{f}}({\mathscr G} )$.
\end{lemma}

Note that by Proposition~\ref{prop:mb},
$\mbvtex _{{\mathscr G}}( d, \anvtex _{{\mathscr G}}(D) \cap \pre _{{\mathscr G} ,\prec}(d)) \subseteq D
\cup \pavtex _{{\mathscr G}}(D)$. Lemma~\ref{lem:markovfact} has the following important
consequence:
%
%e18 #&#
\begin{align}\begin{split}
{\mathcal{P}}^{\mathrm{c}}_{\mathrm{f}}({\mathscr G} ) = {}&\biggl\{ q_{V}
\;\left|\;\, \text{for every ancestral set }A, \vphantom{f_A^D} \right.
\\
& q_{V}(x_{V\cap A} \cmid x_{W}) = \prod
_{D \in {
\mathcal{D}}({\mathscr G} _{A})} q_{D}(x_{D} \cmid
x_{\pavtex (D)\setminus D})  \biggr\}, \label{eq:tian-factor2}
\end{split}\end{align}
where $q_{D}(x_{D} \cmid x_{\pavtex (D)\setminus D})$ is defined via the right-hand
side of (\ref{eq:district-factor-ident}) under any topological ordering.
In a context where we refer to $q_{V}$ and $q_{D}$ for
$D \in {\mathcal{D}}({\mathscr G} )$, it is implicit that $q_{D}$ is derived from
$q_{V}$ in this way.  In Section \ref{sec:nested} we will %subsequently
extend this notation to include other ``reachable'' sets.

%s2.10 #&#
\subsection{Equivalence of factorizations and Markov properties for CADMGs}
\label{subsec:markov}

The above definitions describe the same set of kernels due to the following
result.
%
%t16 #&#
\begin{theorem}%
\label{thm:one-step-markov}
\label{cadmg-equiv}
 ${\mathcal{P}}^{\mathrm{c}}_{\mathrm{f}}({\mathscr G} )   = {\mathcal{P}}^{\mathrm{c}}_{\mathrm{l}}({\mathscr G} ,
\prec )   =   {\mathcal{P}}^{\mathrm{c}}_{\mathrm{m}}({\mathscr G} )   =   {\mathcal{P}}^{\mathrm{c}}_{
\mathrm{a}}({\mathscr G} )$.
\end{theorem}
We thus use ${\mathcal{P}}^{\mathrm{c}}({\mathscr G} )$ to denote the set of such kernels,
and simply refer to a kernel $q_{V} \in {\mathcal{P}}^{\mathrm{c}}({\mathscr G} )$ as being
\emph{Markov} with respect to a CADMG ${\mathscr G} $.

%s2.11 #&#
\subsection{The fixing operation and fixable vertices}
\label{subsec:fixing}

We now introduce a ``fixing'' operation on an ADMG or CADMG that has the
effect of transforming a random vertex into a fixed vertex, thereby changing
the graph. However, we define this operation only for a subset of the vertices
in the graph, which we term the set of (potentially) fixable vertices.
%
%d17 #&#
\begin{definition}
Given a CADMG ${{\mathscr G}}(V,W)$, the set of \emph{fixable vertices} is
\begin{equation*}
{\mathbb F}({{\mathscr G}}) \equiv \bigl\{ v \in V \left|\; \dis _{{\mathscr G}}(v) \cap
\devtex _{{\mathscr G}}(v) = \{v\} \right.\bigr\}.
\end{equation*}
\end{definition}
In words, a vertex $v$ is fixable in ${\mathscr G} $ if there is no other vertex
$t$ that is both a descendant of $v$ and in the same district as $v$ in
${\mathscr G} $. For details on the causal interpretation of fixing, see Section~\ref{sec:conn_causal}.

%p18 #&#
\begin{proposition}%
\label{prop:exists-fixable-descendant}
For any $v \in V$, we have
$\devtex _{{\mathscr G}}(v) \cap \dis _{{\mathscr G}}(v) \cap {\mathbb F}({{\mathscr G}}) \neq
\emptyset $.
\end{proposition}
Thus, if a vertex in a district is not fixable then there is a descendant
of the vertex within the district that is fixable. In particular, every
district contains at least one vertex that is fixable.

Recall that $\mbvtex _{{\mathscr G}}(t)$, defined in (\ref{eq:mb2}), is the set of vertices
$v \in (V \setminus \{t\}) \cup W$, which can be reached via paths of the
form: $t\leftrightarrow \cdots \leftrightarrow v$ or
$t \leftrightarrow \cdots \leftrightarrow \leftarrow v$ or
$t \leftarrow v$. Suppose that $q_{V}$ is Markov with respect to the CADMG
${\mathscr G} $; if $\chvtex _{{\mathscr G}}(t) = \emptyset $, then (\ref{eq:local}) follows by
the CADMG local Markov property. More generally, if $t$ is fixable then
%
%e19 #&#
\begin{align}
\label{eq:mb-when-t-fixable} X_{t} \civtex X_{\ndvtex _{{\mathscr G}}(t) \setminus \mbvtex _{{\mathscr G}}(t)} \cmid
X_{\mbvtex _{{\mathscr G}}(t)}\quad [q_{V}]
\end{align}
follows from m-separation in ${\mathscr G} ^{|W}$; thus, in addition,
$X_{t} \civtex X_{W \setminus \mbvtex _{{\mathscr G}}(t)} \cmid X_{\mbvtex _{{\mathscr G}}(t)}$ in
$q_{V}$. These hold even if $\chvtex _{{\mathscr G}}(t) \neq \emptyset $.

%d19 #&#
\begin{definition}
Given a CADMG ${{\mathscr G}}(V,W)$, and a kernel $q_{V}(x_{V} \cmid x_{W})$, for
every $r \in {\mathbb F}({{\mathscr G}})$, we associate a \emph{fixing transformation
$\phi _{r}$ on the graph} ${{\mathscr G}}$ defined as follows:
\begin{equation*}
\phi _{r} ({{\mathscr G}}) \equiv %
{{\mathscr G}}^{*}\bigl(V
\setminus \{r\}, W \cup \{r\}\bigr),
\end{equation*}
where ${\mathscr G} ^{*}(V\setminus \{r\}, W \cup \{r\})$ has precisely the subset
of edges in ${\mathscr G} (V,W)$ that do not have arrowheads at $r$. With slight
abuse of notation, we define a \emph{fixing transformation
$\phi _{r}$ on the pair $(q_{V}(x_{V} \cmid x_{W}),{{\mathscr G}})$}:
%
%e20 #&#
\begin{equation}
\label{eq:fix-kernel} \phi _{r} \bigl(q_{V}(x_{V}
\cmid x_{W}); {{\mathscr G}}\bigr)\equiv %
\frac{q_{V}(x_{V} \cmid x_{W})}{q_{V}(x_{r} \cmid x_{\mbvtex _{{\mathscr G}}(r)})}.
\end{equation}
\end{definition}

Note that
$\mathbb{V}(\phi _{r} ({{\mathscr G}})) = \mathbb{V}({{\mathscr G}})\setminus \{r\}$ and
$\mathbb{W}(\phi _{r} ({{\mathscr G}})) = \mathbb{W}({{\mathscr G}})\cup \{r\}$, so that
$\phi _{r} ({{\mathscr G}})$ is a new CADMG in which the status of $r$ changes from
random to fixed, while $\phi _{r} (q_{V}; {{\mathscr G}})$ forms a new kernel, per
Lemma~\ref{lem:fixing-yields-kernel}. Although the CADMG
$\phi _{r} ({{\mathscr G}})$ is determined solely by the graph ${\mathscr G} $ given as input,
the transformation $\phi _{r} (q_{V}(x_{V} \cmid x_{W}); {{\mathscr G}})$ on the
kernel $q_{V}(x_{V} \cmid x_{W})$ is a function of both the graph and the
kernel itself.

%p20 #&#
\begin{proposition}
\label{prop:barren-fixable}
If $q_{V}$ is Markov with respect to a CADMG ${{\mathscr G}}(V,W)$ and
$\chvtex _{{\mathscr G}}(r) = \emptyset $ (and hence $r \in \mathbb{F}({\mathscr G} )$), then
\begin{equation*}
\phi _{r} \bigl(q_{V}(x_{V} \cmid
x_{W}); {{\mathscr G}}\bigr) = \sum_{x_{r}}
q_{V}(x_{V} \cmid x_{W}) = q_{V}(x_{V\setminus \{r\}}
\cmid x_{W}).
\end{equation*}
\end{proposition}
 Thus, if $\chvtex _{{\mathscr G}}(r) = \emptyset $, then $\phi _{r}$ simply marginalizes
over $X_{r}$: the conditioning on $X_{r}$ in
$\phi _{r} (q_{V}(x_{V} \cmid x_{W});{\mathscr G} )$ is vacuous in the sense that
the resulting kernel does not depend on the value of $X_{r}$. Though it
may at first appear unnatural, it greatly simplifies our subsequent analyses
to unify marginalization and fixing in this way.

%p21 #&#
\begin{proposition}%
\label{prop:fix-with-mb}
If $q_{V}$ is Markov w.r.t. ${{\mathscr G}}(V,W)$ and
$r \in \mathbb{F}({\mathscr G} )$, then
%
%e21 #&#
\begin{align}
\label{eq:alt-fix-def} \phi _{r} \bigl(q_{V}(x_{V}
\cmid x_{W});{{\mathscr G}}\bigr) = q_{V}(x_{V} \cmid
x_{W}) / q_{V}(x_{r} \cmid x_{\ndvtex _{{\mathscr G}}(r)}).
\end{align}
\end{proposition}

%l22 #&#
\begin{lemma}%
\label{lem:no-backtracking}
If $r \in {\mathbb F}({{\mathscr G}})$, then
${\mathbb F}({{\mathscr G}}) \setminus \{r\} \subseteq {\mathbb F}(\phi _{r} ({
{\mathscr G}}))$.
\end{lemma}
That is, any vertex $s$ that was fixable before $r$ was fixed is still
fixable after $r$ has been fixed (with the obvious exception of $r$ itself).
Thus, when fixing vertices, although our choices may be limited at various
stages, we are never faced with a choice between fixing $r$ and
$r^{\prime}$, whereby choosing $r$ precludes subsequently fixing
$r^{\prime}$.

%p23 #&#
\begin{proposition}
\label{prop:sub}
If ${\mathscr G}$ is a subgraph of ${{\mathscr G}}^{*}$ with the same random and fixed vertex
sets, then ${\mathbb F}( {{\mathscr G}}^{*}) \subseteq {\mathbb F}( {{\mathscr G}})$.
\end{proposition}

For conciseness, we use $D^{r}$ to denote $\dis _{{\mathscr G}}(r)$ when the graph
${\mathscr G} $ is clear from the context.

%p24 #&#
\begin{proposition}%
\label{prop:districts-after-fixing}
Let ${\mathscr G} $ be a CADMG, with $r \in {\mathbb F}({{\mathscr G}})$. If
$r \in D^{r} \in {\mathcal{D}}({\mathscr G} )$, then
\begin{equation*}
{\mathcal D}\bigl(\phi _{r}({\mathscr G} )\bigr) = \bigl( {\mathcal{D}}(
{\mathscr G} )\setminus \bigl\{{D}^{r} \bigr\} \bigr) \cup {\mathcal{D}}(
{\mathscr G} _{{D}^{r}\setminus \{r\}}),
\end{equation*}
where the sets on the right are disjoint. Thus, if
$D \in {\mathcal{D}}(\phi _{r}({{\mathscr G}}))$ then $D\subseteq D^{*}$ for some
$D^{*} \in {\mathcal{D}}({{\mathscr G}})$; further if $D \neq D^{*}$, then
$r\in D^{*}$.
\end{proposition}
In words, the set of districts in $\phi _{r}({\mathscr G} )$, the graph obtained
by fixing $r$, consists of the districts in ${\mathscr G} $ that do not contain
$r$, together with new districts that are subsets of $D^{r}$, the district
in ${\mathscr G} $ that contains $r$. The new districts are bidirected-connected
subsets of $D^{r}$ after removing $r$.

%s2.12 #&#
\subsection{Fixing and factorization}
\label{subsec:fixing-fact}

%p25 #&#
\begin{proposition}
\label{prop:fix_gform}
\label{prop:fix-dist-marg}
Take a CADMG ${\mathscr G} (V,W)$ with kernel
$q_{V} \in {\mathcal{P}}^{\mathrm{c}}({\mathscr G} )$ with associated district factorization:
%
%e22 #&#
\begin{equation}
\label{eq:district-factorize} q_{V}(x_{V} \cmid x_{W}) = \prod
_{D \in {\mathcal{D}}({\mathscr G} )} q_{D}(x_{D} \cmid
x_{\pavtex _{{\mathscr G}}(D) \setminus D}),
\end{equation}
where the kernels $q_{D}(x_{D} \cmid x_{\pavtex _{{\mathscr G}}(D) \setminus D}) $ are
defined via the right-hand side of \textup{(\ref{eq:district-factor-ident})}.
If $r\in {\mathbb F}({{\mathscr G}})$ and $D^{r} \in {\mathcal{D}}({\mathscr G} )$ is the district
containing $r$, then
\begin{equation*}
\phi _{r} \bigl(q_{V}(x_{V} \cmid
x_{W}); {{\mathscr G}}\bigr) = q_{D^{r}} (x_{D^{r}
\setminus \{r\}} \cmid
x_{\pavtex _{{\mathscr G}}(D^{r} ) \setminus D^{r} }) \prod_{D \in {\mathcal{D}}({\mathscr G} )\setminus \{D^{r} \}} q_{D}(x_{D}
\cmid x_{\pavtex _{{\mathscr G}}(D) \setminus D}).
\end{equation*}
\end{proposition}
The proof is found in the Supplementary Material, Section A.6 in \citet{supp}. In words,
the result of a fixing operation is solely to marginalize the variable
$X_{r}$ from the density $q_{D^{r}}$ associated with the district
$D^{r}$ in which the vertex $r$ occurs, while leaving unchanged all of
the other terms $q_{D}$ in the factorization.

%s2.13 #&#
\subsection{Reachable graphs derived from an ADMG}
\label{subsec:reachable}

A sequence of distinct vertices ${\mathbf{w}}$ is said to be \emph{valid} in
${\mathscr G} $ if either ${\mathbf{w}} = \langle \rangle $, or
${\mathbf{w}} = \langle w_{1}, w_{2}, \ldots , w_{k} \rangle $,
$w_{1} \in \mathbb{F}({\mathscr G} )$ and
$\langle w_{2}, \ldots , w_{k} \rangle $ is valid in
$\phi _{w_{1}}({\mathscr G} )$. Given a valid ${\mathbf{w}}$, we define the fixing operator
$\phi _{\mathbf{w}}({\mathscr G} )$ on graphs, and the fixing operator
$\phi _{\mathbf{w}}(q_{V}; {\mathscr G} )$ on kernels inductively as follows:
\begin{align*}
\phi _{\mathbf{\langle \rangle}}({\mathscr G} ) &\equiv {\mathscr G} ;
\\
\phi _{\langle { w_{1},\ldots ,w_{k}}\rangle}({\mathscr G} ) &\equiv \phi _{
\langle { w_{2},\ldots ,w_{k}}\rangle}\bigl(\phi
_{w_{1}}({\mathscr G} )\bigr);
\\
\phi _{\langle \rangle}(q_{V}; {\mathscr G} ) &\equiv q_{V};
\\
\phi _{\langle { w_{1},\ldots ,w_{k}}\rangle}(q_{V}; {\mathscr G} ) &\equiv \phi _{\langle { w_{2},\ldots ,w_{k}}\rangle}
\bigl(\phi _{w_{1}}(q_{V}; {\mathscr G} ); \phi _{w_{1}}(
{\mathscr G} )\bigr).
\end{align*}

%d26 #&#
\begin{definition}
A CADMG ${{\mathscr G}}(V,W)$ is \emph{reachable} from an ADMG
${{\mathscr G}}^{*}(V\cup W)$ if there exists a valid fixing sequence $\mathbf{w}$ of
the vertices in $W$ such that ${\mathscr G} = \phi _{\mathbf{w}}({{\mathscr G} ^{*}})$.
\end{definition}

In words, a graph is reachable from ${\mathscr G} ^{*}$ if there exists an ordering
on vertices in $W$ such that the first element $w_{1}$ in the ordering
may be fixed in ${\mathscr G} ^{*}$, the second element $w_{2}$ in
$\phi _{\langle w_{1} \rangle}({\mathscr G} ^{*})$, the third element $w_{3}$ in
$\phi _{\langle w_{1}, w_{2} \rangle}({\mathscr G} ^{*})$, and so on. Note that by
definition ${\mathscr G} $ is reachable from itself. If a CADMG ${{\mathscr G}}(V,W)$ is reachable
from ${{\mathscr G}}^{*}(V \cup W)$, we say that \emph{$V$ is reachable in
${\mathscr G} ^{*}$}.

A key result, which we will show later as Theorem~\ref{thm:invariant},
is that under the nested Markov model reachable CADMGs and their associated
kernels are invariant with respect to any valid fixing sequence. It is
not hard to see that if there are two valid fixing sequences
${\mathbf{w}}$ and ${\mathbf{u}}$ for $W$ then
$\phi _{\mathbf{w}}({{\mathscr G}}) = \phi _{{\mathbf{u}}}({{\mathscr G}})$. However, it requires more
work to show that
$\phi _{\mathbf{w}}(q_{V};{{\mathscr G}}) = \phi _{{\mathbf{u}}}(q_{V}; {{\mathscr G}})$.

%s3 #&#
\section{Nested Markov models}
\label{sec:nested}

In this section, we define a set of recursive Markov properties and a factorization,
and show their equivalence, in Proposition \ref{prop:nested_global_factorization} and Theorem \ref{thm:global_local}. The models, which obey these properties, will
be called ``nested'' Markov models. For the rest of this section, we will
fix an ADMG ${\mathscr G} (V)$ and a density $p(x_{V})$.

%d27 #&#
\begin{definition}
\label{dfn:global-definition}
We say that a distribution $p(x_{V})$ is
\emph{globally nested Markov} with respect to ${\mathscr G} (V)$ if for all fixing
sequences ${\mathbf{w}}$ valid in ${\mathscr G} $, the kernel
$\phi _{{\mathbf{w}}}(p(x_{V});{\mathscr G} )$ obeys the global Markov property for
$\phi _{{\mathbf{w}}}({\mathscr G} )$.
\end{definition}
Note that the same graph may be reached by more than one sequence, that
is, $\phi _{{\mathbf{w}}_{1}}({\mathscr G} ) = \phi _{{\mathbf{w}}_{2}}({\mathscr G} )$ for two distinct
valid sequences ${\mathbf{w}}_{1},{\mathbf{w}}_{2}$.

%d28 #&#
\begin{definition}
\label{dfn:factorization-definition}
We say that a distribution $p(x_{V})$ \emph{nested Markov factorizes} with
respect to ${\mathscr G} (V)$ if, for all valid fixing sequences ${\mathbf{w}}$ in
${\mathscr G} $, there exist kernels
$\{ f^{{\mathbf{w}}}_{D}(x_{D} \cmid x_{\pavtex _{{\mathscr G}}(D) \setminus D}) : D
\in {\mathcal{D}}(\phi _{{\mathbf{w}}}({\mathscr G} )) \}$ such that
%
%e23 #&#
\begin{equation}
\phi _{{\mathbf{w}}}\bigl(p(x_{V});{\mathscr G} \bigr) = \prod
_{D \in {\mathcal{D}}(\phi _{{\mathbf{w}}}(
{\mathscr G} ))} f^{{\mathbf{w}}}_{D}(x_{D} \cmid
x_{\pavtex _{{\mathscr G}}(D) \setminus D}).
\end{equation}
\end{definition}

%p29 #&#
\begin{proposition}
\label{prop:nested_global_factorization}
With respect to an ADMG ${\mathscr G} (V)$, a distribution $p(x_{V})$ is globally
nested Markov if and only if it nested Markov factorizes.%
\end{proposition}

%s3.1 #&#
\subsection{Invariance to the order of fixing in an ADMG}
\label{subsec:invariance}

In this section, we show that, given a distribution that obeys the nested
Markov property with respect to an ADMG, any two valid fixing sequences
that fix the same vertices will lead to the same reachable graph and kernel.
For marginal distributions obtained from a hidden variable DAG, this claim
follows by results in \citet{tian02on} on identification of causal effects
in hidden variable DAG models. However, for distributions which obey the
nested Markov property for an ADMG, but which are not derived from any
hidden variable DAG, the claim is far less obvious; see Example~\ref{ex:chsh}
below. For instance in the ADMG in Figure~\ref{fig:two_fix}, the fixing
sequence $\langle 4, 3, 1 \rangle $, which leads to the kernel
\begin{equation*}
q^{1}_{2,5}(x_{2},x_{5} \cmid
x_{4}, x_{3}, x_{1}) \equiv \frac{\sum_{x_{3}} p(x_{5} \cmid x_{4}, x_{3}, x_{2}, x_{1}) p(x_{3}, x_{2}, x_{1})}{
\sum_{x_{3},x_{2},x_{5}} p(x_{5} \cmid x_{4}, x_{3}, x_{2}, x_{1}) p(x_{3}, x_{2}, x_{1})
}
\end{equation*}
and the fixing sequence $\langle 3, 4, 1 \rangle $, which leads to the
kernel
\begin{equation*}
q^{2}_{2,5}(x_{2},x_{5} \cmid
x_{4}, x_{3}, x_{1}) \equiv \frac{p(x_{5} \cmid x_{4},x_{3}, x_{2}, x_{1}) p(x_{2}, x_{1})}{
\sum_{x_{5},x_{2}} p(x_{5} \cmid x_{4},x_{3}, x_{2}, x_{1}) p(x_{2}, x_{1})
}
\end{equation*}
are both valid, and these two kernels are therefore the same, in the context
of our model. This is not entirely obvious from inspecting these expressions.
In addition, $q^{1}_{2,5}$ and $q^{2}_{2,5}$ are not functions of
$x_{3}$ in our model; this is clear for $q^{1}_{2,5}$ since $x_{3}$ is
summed out, but not so obvious for $q^{2}_{2,5}$.

%f8 #&#
\begin{figure}
\begin{center}
%  \begin{tikzpicture}[>=stealth, node distance=1.2cm]
\begin{tikzpicture}[>=stealth, node distance=1.2cm,
pre/.style={->,>=stealth,ultra thick,line width = 1.4pt}]
    \tikzstyle{format} = [circle, draw, very thick, minimum size=5mm, inner sep=.5mm]
  
  \begin{scope}

    \path[->]
    
    		node[format] (x1) {$1$} %{$\vrt{x_1}$}
    		node[format, right of=x1] (x2) {$2$} %{$\vrt{x_2}$}
                  (x1) edge[pre, blue] (x2)
		node[format, right of=x2] (x3) {$3$} %{$\vrt{x_3}$}
                  (x2) edge[pre, blue] (x3)
                  (x1.315) edge[pre,->, bend right,blue] (x3.225)
		node[format, right of=x3] (x4) {$4$} %{$\vrt{x_4}$}
                  (x3) edge[pre, blue] (x4)
   		
		%(x2.45) edge[pre, <->, red, bend left] (x4.135)
		
		node[format, right of=x4] (x5) {$5$} %{$\vrt{x_5}$}
		
		(x4) edge[pre, blue] (x5)
		
		(x2) edge[pre,<->, bend left=40, red] (x5)
		(x1) edge[pre,<->, bend left=40, red] (x4)

		(x2) edge[pre,->, bend right, blue] (x5)
		%(x3) edge[pre,->, bend left=25, blue] (x5)
		;

   %\node[below of=x2, yshift = .1cm, xshift=.6cm] (label) {(i)};

  \end{scope}

  \end{tikzpicture}

\end{center}
\caption{A graph where $\langle 4, 3, 1\rangle $ and
$\langle 3, 4, 1\rangle $ are valid fixing sequences.}
\label{fig:two_fix}
\end{figure}

%l30 #&#
\begin{lemma}%
\label{lem:permute}
Let ${\mathscr G} (V,W)$ be a CADMG with $r,s \in {\mathbb F}({\mathscr G} )$ and let
$q_{V}$ be a kernel Markov w.r.t. ${\mathscr G} $. Then
\begin{align*}
\phi _{\langle r, s \rangle}({\mathscr G} ) = \phi _{\langle s, r \rangle}({\mathscr G} ) \quad \text{and} \quad
\phi _{\langle r, s \rangle}(q_{V}; {\mathscr G} ) = \phi _{
\langle s, r \rangle}(q_{V};
{\mathscr G} ).
\end{align*}
\end{lemma}

In words, if we have a choice to fix two vertices in ${\mathscr G} $ then the order
in which we do this does not affect the resulting graph, or kernel, provided
that the original kernel is Markov with respect to ${\mathscr G} $.
%
%t31 #&#
\begin{theorem}%
\label{thm:invariant}
\label{invariance}
 Let ${p}(x_{V})$ be a distribution that is nested Markov with respect to
an ADMG ${\mathscr G} $ (in either the sense of Definitions
\ref{dfn:global-definition} or \ref{dfn:factorization-definition}). Let
${\mathbf{u}},{\mathbf{w}}$ be different valid fixing sequences for the same set
$W\subset V$. Then $ \phi _{{\mathbf{u}}}({\mathscr G} ) = \phi _{{\mathbf{w}}}({\mathscr G} )$ and
%
%e24 #&#
\begin{equation}
\label{eq:invar-thm} \phi _{{\mathbf{u}}}\bigl({p}(x_{V});{\mathscr G} \bigr) =
\phi _{{\mathbf{w}}}\bigl({p}(x_{V});{\mathscr G} \bigr).
\end{equation}
\end{theorem}
Due to this theorem, our fixing operations $\phi _{{\mathbf{w}}}$, which were
defined for a specific fixing sequence ${\mathbf{w}}$, can, under the model, be defined purely
in terms of the set $W$ of nodes that were fixed; the order does not matter
(provided that at least one valid fixing sequence exists). Consequently,
we will subscript the fixing operator $\phi $ by a set rather than a sequence.
That is, we write $\phi _{V \setminus R}({\mathscr G} )$ and
$\phi _{V \setminus R}(p(x_{V}); {\mathscr G} )$ to mean ``apply the fixing operator
$\phi _{\mathbf{w}}$, for any valid sequence ${\mathbf{w}}$ of elements in
$V \setminus R$, to ${\mathscr G} $ and the pair $(p(x_{V});   {\mathscr G} )$, respectively.''
For conciseness, and consistency with notation in \citet{tian02on}, we
will also denote $\phi _{V \setminus R}({\mathscr G} )$ by ${\mathscr G} [R]$.

%c32 #&#
\begin{corollary}
\label{cor:nested-factorize}
If a distribution $p(x_{V})$ obeys the global nested Markov property for
${\mathscr G} (V)$, then for all reachable sets $R$, the kernel
\begin{align*}
q_{R}(x_{R} \cmid x_{\pavtex (R) \setminus R}) \equiv \phi
_{V \setminus R}\bigl(p(x_{V}); {\mathscr G} \bigr)
\end{align*}
obeys the global Markov property for $\phi _{V \setminus R}({\mathscr G} )$.
\end{corollary}

We will subsequently see that if we assume the existence of a latent variable
DAG model (with observed variables $V\cup W$) that has latent projection
${\mathscr G} $, then if $W$ is fixable, the kernel
$\phi _{W}(p(x_{V\cup W}); {\mathscr G} )$ can be interpreted as the intervention
distribution $p(x_{V}\cmid \Dovtex _{{\mathscr G}}(x_{W}))$; see Lemma~D.3 in the Supplementary Material. In this context, a valid fixing
sequence corresponds to a sequence of steps in the ID algorithm of
\cite{tian02on} that identify this intervention distribution; see Section~\ref{subsec:id}. Consequently, were we to assume the existence of a DAG
with latent variables, then the soundness of the ID algorithm would directly
imply the equality (\ref{eq:invar-thm}). However, since we are \emph{not}
assuming such a DAG exists, $\phi _{W}(p(x_{V\cup W}); {\mathscr G} )$ may not correspond
to an intervention distribution, and hence a separate proof is required;
see Example~\ref{ex:chsh} for an inequality constraint that is implied
by the existence of a latent variable, but does not follow from the nested
Markov property applied to the latent projection.

%s3.2 #&#
\subsection{Intrinsic sets}
\label{subsec:global}

We introduce the following two definitions that will prove useful.

%d33 #&#
\begin{definition}
A set $C$ is \emph{intrinsic} if it is a district in a reachable graph derived
from ${\mathscr G} $. The set of intrinsic sets in an ADMG ${\mathscr G} $ is denoted by
${\mathcal{I({\mathscr G} )}}$.
\end{definition}

%d34 #&#
\begin{definition}
For a set $R$ reachable in ${\mathscr G}$ and a vertex $v \in R$, with
$\chvtex _{ \phi _{V \setminus R}({\mathscr G} )}(v)=\emptyset $, we define \emph{the
Markov blanket of $v$ in $R$} to be
%
%e25 #&#
\begin{equation}
\label{eq:mb-in-r} \mbvtex _{{\mathscr G}}(v, R) \equiv \mbvtex _{\phi _{V \setminus R}({\mathscr G} )}(v).
\end{equation}
\end{definition}
 Since every ancestral set $A$ is reachable in ${\mathscr G} $, this is a natural
extension of our previous definition (\ref{eq:mb-in-A}). The next two subsections
give different characterizations of the nested Markov model.

%s3.3 #&#
\subsection{The simplified nested factorization property}
\label{subsec:fact}

%t35 #&#
\begin{theorem}%
\label{thm:global-reachable-factorization}
If $p(x_{V})$ nested Markov factorizes with respect to ${\mathscr G} $, then for
every reachable $R$ in ${\mathscr G} $,
%
%e26 #&#
\begin{align}
\phi _{V \setminus R}\bigl(p(x_{V}); {\mathscr G} \bigr) = \prod
_{D \in {\mathcal{D}}(\phi _{V
\setminus R}({\mathscr G} ))} \phi _{V \setminus D}\bigl(p(x_{V}); {\mathscr G}
\bigr). \label{eqn:simp_nest_fact}
\end{align}
\end{theorem}

Since all the sets $D$ quantified in the product are districts in a reachable
graph derived from ${\mathscr G} $, it follows that in a nested Markov model every
kernel corresponding to a reachable set can be constructed by combining
kernels corresponding to intrinsic sets. We call (\ref{eqn:simp_nest_fact})
the \emph{simplified} \emph{nested factorization}.

%s3.4 #&#
\subsection{The ordered local nested property}
\label{subsec:local}

Notwithstanding Theorem~\ref{thm:invariant}, in general we may not know
\emph{a priori} that a distribution is Markov with respect to a graph.
It is therefore helpful to have a canonical order in which variables should
be fixed for the purposes of verifying that a distribution is in the nested
Markov model.

%d36 #&#
\begin{definition}
We define the \emph{intrinsic power DAG} for an ADMG ${\mathscr G} $ and topological
ordering to be the graph ${\frak I}({\mathscr G} )$ whose vertices are intrinsic
sets ${\mathcal{I}}({\mathscr G} )$, and such that there is an edge from $D$ to $C$ if
one can obtain $C$ as the district of the maximal vertex in $D$ by fixing
some \emph{other} vertex in $D$.
\end{definition}

To simplify the next definition, we set
$\fmvtex _{{\mathscr G}}(C) := C \cup \pavtex _{{\mathscr G}}(C)$.

%d37 #&#
\begin{definition}
Let ${\mathscr G} $ be an ADMG with arbitrary topological order $\prec $; define
$T_{v} = \{w : w \prec v\} \cup \{v\}$ as the initial segment for
$v$. A distribution $p(V)$ is \emph{ordered local nested Markov} with respect
to $\prec $, if:
\begin{itemize}%[(ii)]
\item[(i)] for each $v \in V$, we have the independences
%
%e27 #&#
\begin{align}
X_{v} \civtex X_{T_{v} \setminus \fmvtex (C)} \mid X_{\fmvtex (C) \setminus \{v\}}\quad [p],
\label{eqn:localknew0}
\end{align}
where $C = \dis _{T_{v}}(v)$;
\item[(ii)] for every edge $D \rightarrow C$ in ${\intgr}({\mathscr G} )$ (obtained by
fixing $w \in D \setminus \{v\}$) and $v$ being maximal under
$\prec $ in $D$, we have
%
%e28 #&#
\begin{align}
X_{v} \civtex X_{\fmvtex (D) \setminus (\fmvtex (C) \cup \{w\})} \mid X_{\fmvtex (C)
\setminus \{v\}}
\quad
\bigl[%
\phi _{w}\bigl({q}_{D}; {\mathscr G} [D]\bigr)
\bigr], \label{eqn:localknew}
\end{align}
where $q_{D}$ is the unique kernel resulting from some valid fixing sequence
for $D$.
\end{itemize}
%was: longlist environment
%
\end{definition}

Note that in this definition, we assume that we have reached each node
$D$ by traversals from a root node (i.e.,~a set $C$ in (i) for some
$v$). In this way, crossing the edge $D \to C$ will (potentially) introduce
a new conditional independence, as well as defining $q_{C}$.

%t38 #&#
\begin{theorem}
\label{thm:global_local}
$p(x_{V})$ is globally nested Markov with respect to ${\mathscr G} $ if and only
if $p(x_{V})$ is ordered local nested Markov with respect to ${\mathscr G} $ for
any topological ordering $\prec $.
\end{theorem}

Proposition~\ref{prop:nested_global_factorization} and Theorem~\ref{thm:global_local} thus yield three equivalent ways of defining the
set ${\mathcal{P}}^{\mathrm{n}}({{\mathscr G}})$ of distributions in the nested Markov
model of ${{\mathscr G}}$.

%s3.5 #&#
\subsection{Nested Markov models for complete graphs are saturated}
\label{subsec:saturated}

It is known that any distribution is Markov relative to a complete DAG
or ADMG (i.e., a graph with at least one edge between every pair of vertices).
This also holds in the nested Markov case.

%t39 #&#
\begin{theorem}
\label{thm:saturated_a}
Let ${\mathscr G} $ be an ADMG. The model ${\mathcal{P}}^{\mathrm{n}}({\mathscr G} )$ is saturated if
and only if for every valid fixing sequence $r_{1}, \ldots , r_{k}$, and
every $i,j \in 1,\ldots , k$, either
$r_{i} \in \mbvtex _{{\mathscr G} ^{(j)}}(r_{j})$ or
$r_{j} \in \mbvtex _{{\mathscr G} ^{(i)}}(r_{i})$. Here, ${\mathscr G} ^{(1)} \equiv {\mathscr G} $ and
${\mathscr G} ^{(\ell +1)} \equiv \phi _{r_{\ell}}({\mathscr G} ^{(\ell )})$.
\end{theorem}

%c40 #&#
\begin{corollary}
\label{cor:complete-saturated}
Let ${\mathscr G} $ be a complete ADMG; then ${\mathcal{P}}^{\mathrm{n}}({\mathscr G} )$ is saturated.
\end{corollary}

%s4 #&#
\section{Connections with causal inference}
\label{sec:conn_causal}

The fixing operation is closely related to the identification of intervention
distributions.

%s4.1 #&#
\subsection{Latent variable DAG models are in the nested Markov model}
\label{subsec:latent-dag}

We first show that if $p(x_{L \cup V})$ is Markov relative to a DAG
${\mathscr G} (L \cup V)$, then $p(x_{V})$ is in the nested Markov model of
${\mathscr G} (V)$. We extend the latent projection operation in the natural way
from ADMGs to CADMGs, and denote the operation of creating a latent projection
${\mathscr G} (V,W)$ of a CADMG ${\mathscr G} (L \cup V, W)$ onto the subset $V$ as
$\sigma _{L}$. That is, $\sigma _{L}({\mathscr G} (L \cup V, W)) = {\mathscr G} (V, W)$, and
${\mathscr G} (V,W)^{|W}$ encodes the m-separation relations holding among
$V \cup W$ in ${\mathscr G} (L \cup V, W)^{|W}$. Precise statements appear in the
Supplementary Material, Section A.3.

We will call a CADMG, which does not contain bidirected arrows, a
\emph{conditional acyclic directed graph} (CDAG). It is a corollary of the
definition of ${\mathcal{P}}^{\mathrm{c}}_{\mathrm{f}}$ that if ${\mathscr G} (V, W)$ is a CDAG,
then $q_{V}(x_{V} \cmid x_{W}) \in {\mathcal{P}}^{\mathrm{c}}_{\mathrm{f}}({\mathscr G} )$ if
\begin{equation*}
q_{V}(x_{V} \cmid x_{W}) = \prod
_{a \in V} q_{V}(x_{a} \cmid
x_{\pavtex _{
{\mathscr G}}(a)}).
\end{equation*}

%l41 #&#
\begin{lemma}%
\label{lem:reachable-factorizes}
\label{lem:fixing_in_dags}
 Let ${\mathscr G} $ be a DAG with a vertex set $V$. Then every nonempty subset
$S$ of $V$ is reachable, and if $p(x_{V})$ is Markov with respect to
${\mathscr G} $,
\begin{equation*}
\phi _{V \setminus S}\bigl(p(x_{V}); {\mathscr G} \bigr) =
q_{S}(x_{S} \cmid x_{\pavtex _{{\mathscr G}}(S)
\setminus S}) = \prod
_{a \in S} p(x_{a} \cmid x_{\pavtex _{{\mathscr G}}(a)}).
\end{equation*}
In other words,
$\phi _{V \setminus S}(p(x_{V}); {\mathscr G} ) \in {\mathcal{P}}^{\mathrm{c}}(\phi _{V
\setminus S}({\mathscr G} ))$.
\end{lemma}
For a DAG ${\mathscr G} $, let ${\mathcal{P}}^{\mathrm{d}}({\mathscr G} )$ denote the set of distributions
Markov with respect to ${\mathscr G} $ (see the Supplementary Material, Section~A.1, for details).
%
%c42 #&#
\begin{corollary}
\label{cor:dag_in_nested}
For a DAG ${\mathscr G} $, ${\mathcal{P}}^{\mathrm{d}}({\mathscr G} ) = {\mathcal{P}}^{\mathrm{n}}({\mathscr G} )$.
\end{corollary}

 That ${\mathcal{P}}^{\mathrm{d}}({\mathscr G} ) \subseteq {\mathcal{P}}^{\mathrm{n}}({\mathscr G} )$ is a special
case of Theorem~\ref{thm:dags_in_nested} below when $L$ is empty.

%l43 #&#
\begin{lemma}
\label{lem:fix_marg_commute_graph}
Let ${\mathscr G} (L \cup V, W)$ be a CDAG. Assume $v \in V$ is fixable in
${\mathscr G} (V, W) = \sigma _{L}({\mathscr G} (L \cup V, W))$. Then
$\sigma _{L} ( \phi _{v}({\mathscr G} (L \cup V, W) ) ) = \phi _{v} ( \sigma _{L}(
{\mathscr G} (L \cup V, W) ) )$. That is, the commutative diagram in Figure \textup{\ref{fig:commute-diagram}(a)} holds.
\end{lemma}
Note that $v \in V$ is fixable in ${\mathscr G} (V,W)$ by assumption, while the fact
that every element of $V$ is fixable in ${\mathscr G} (L \cup V,W)$ follows since
${\mathscr G} (L \cup V,W)$ is a CDAG, and has no bidirected edges.
%
%f9 #&#
\begin{figure}
\centering
\begin{tikzpicture}
	\begin{scope}
            \node (og) {$\G(L \cup V, W)$};
            \node (fg) [below of=og, yshift=-1.2cm] {$\G((L \cup V)
            	\setminus \{ v \}, W \cup \{ v \})$};
            \node (pg) [right of=og, xshift=5.0cm] {$\G(V, W)$};
            \node (fpg) [right of=fg, xshift=5.0cm] {$\G(V \setminus \{ v \},
            	W \cup \{ v \})$};
            
            \draw[->] (og) to node[left] {$\phi_v$} (fg);
            \draw[->] (og) to node[above] {$\sigma_L$} (pg);
            \draw[->] (fg) to node[below] (l) {$\sigma_L$} (fpg);
            \node[below of=l, yshift=0.4cm] (alab) {$(a)$};
            \draw[->] (pg) to node[right] {$\phi_v$} (fpg);
            
	\end{scope}
	\begin{scope}[yshift=-4.5cm]
    	    \node (og) {$q_{L \cup V}(x_{L \cup V} \cmid x_W)$};
            \node (fg) [below of=og, yshift=-1.2cm] {$q_{(L \cup V)\setminus\{v\}}
            	(x_{(L \cup V)\setminus \{v\}} \cmid x_{W\cup\{v\}})$};
            %\node (fg-p) [below of=fg, yshift=0.5cm] {$\equiv \phi_v(q_{L \cup V}(x_{L \cup V} \cmid x_W); \G(L\cup V,W))$};
            
            \node (pg) [right of=og, xshift=5.0cm] {$q_{L \cup V}(x_V \cmid x_W)$};
            \node (fpg) [right of=fg, xshift=5.0cm] {$q_{V\setminus\{v\}}
            	(x_{V\setminus \{v\}} \cmid x_{W\cup\{v\}})$};
            %\node (fpg-p) [right of=fg-p, yshift=5.0cm] {$\equiv \phi_v(q_{L \cup V}(x_V \cmid x_W); \G(V,W))$};
            
            \draw[->] (og) to node[left] {$\phi_v(.;\G(L\cup V, W))$} (fg);
            \draw[->] (og) to node[above] {$\sum_{x_L}$} (pg);
            \draw[->] (fg) to node[below] (l) {$\sum_{x_L}$} (fpg);
            \node[below of=alab, yshift=-3.7cm] (l2) {$(b)$};
            \draw[->] (pg) to node[right] {$\phi_v(.;\G(V, W))$} (fpg);
	\end{scope}
\end{tikzpicture}
\caption{Commutativity diagrams for: (a) Lemma~\protect\ref{lem:fix_marg_commute_graph}; (b) Lemma~\protect\ref{lem:fix_marg_commute_kernel}.}
\label{fig:commute-diagram}
\end{figure}
In fact, an inspection of the proof of this lemma (found in the Supplementary Material)
shows it does not rely on the vertex $v$ being fixable in
${{\mathscr G}}(V,W)$, only on the specific way edges are removed by $\phi $.
In Lemma \ref{lem:fix_marg_commute_kernel},
we give a more general version of this result, that is useful for deriving properties of causal models, which we discuss later in Section~\ref{subsec:causal}.

Given a DAG ${\mathscr G} (V \dot\cup L)$, there is always a well-defined intervention
distribution associated with any vertex $v \in V$; however, if $v$ is not
fixable in ${\mathscr G} (V)$, then this distribution may not be identifiable from
$p(x_{V})$; see (\ref{eq:g-formula}) below. We use $\phi ^{*}_{v}$ to denote
the graphical operation that corresponds to this intervention, and note
that this applies to all vertices in ${\mathscr G} (V)$, as opposed to
$\phi _{v}$, which presupposes identifiability from $p(x_{V})$, and hence
applies only to elements in $\mathbb{F}({\mathscr G} (V))$. The resulting graph
$\phi _{v}^{*}({\mathscr G} (V))$, without distinguishing fixed and random vertices,
was denoted ${\mathscr G} _{\overline{v}}$ by \cite{pearl00causality}. We use
$\phi ^{*}_{v}$ to make the connection to fixing more explicit.

For any $r,s \in V$,
$\phi ^{*}_{\langle r, s \rangle}({\mathscr G} ) = \phi ^{*}_{\langle s, r
\rangle}({\mathscr G} )$, and so for any $Z \subseteq V$, we define
$\phi ^{*}_{Z}({\mathscr G} )$ inductively under any order as in Lemma~\ref{lem:permute}.

%c44 #&#
\begin{corollary}
\label{cor:commute}
\label{cor:mut_marg_commute_graph}
Let ${\mathscr G} (L \cup V, W)$ be a CDAG. Then for any $v \in V$,
$\sigma _{L} ( \phi ^{*}_{v}({\mathscr G} (L \cup V, W) ) ) = \phi ^{*}_{v} (
\sigma _{L}({\mathscr G} (L \cup V, W) ) )$.
\end{corollary}
Corollary~\ref{cor:commute} is essentially equivalent to Proposition~8
of \citet{evans:mdags:16}.

%l45 #&#
\begin{lemma}
\label{lem:fix_marg_commute_kernel}
Let ${\mathscr G} (L \cup V, W)$ be a CDAG, and assume
$q_{L \cup V}(x_{L \cup V} \cmid x_{W}) \in {\mathcal{P}}^{\mathrm{c}}_{\mathrm{f}}(
{\mathscr G} (L \cup V, W))$. Assume $v \in V$ is fixable in
${\mathscr G} (V, W) = \sigma _{L}({\mathscr G} (L \cup V, W))$. Then
\begin{equation*}
\sum_{x_{L}}\phi _{v}\bigl(q_{L \cup V}(x_{L \cup V}
\cmid x_{W}); {\mathscr G} (L \cup V, W)\bigr) = \phi _{v}
\bigl(q_{L \cup V}(x_{V} \cmid x_{W}); \sigma
_{L}\bigl( {\mathscr G} (L \cup V, W)\bigr)\bigr)
\end{equation*}
That is, the commutative diagram in Figure \textup{\ref{fig:commute-diagram}(b)} holds.\vadjust{\goodbreak}
\end{lemma}

 \citet{robins86new} proves a similar result that he calls the ``collapse
of the g-formula.'' We now have enough to prove the main result of this
section.

%t46 #&#
\begin{theorem}%
\label{thm:dags_in_nested}
Let ${\mathscr G} (V \cup L)$ be a DAG. Then
\begin{align*}
p(x_{V \cup L}) \in {\mathcal{P}}^{\mathrm{d}}\bigl({\mathscr G} (V \cup L)
\bigr)\quad \Rightarrow\quad p(x_{V}) \in {\mathcal{P}}^{\mathrm{n}}\bigl(
{\mathscr G} (V)\bigr).
\end{align*}
\end{theorem}
Thus, the constraints implied by the nested model for the latent projection
of ${\mathscr G} (V \cup L)$ also hold in the hidden variable CDAG model. Note that
the converse is not true in general. There are distributions
$p(x_{V}) \in {\mathcal{P}}^{\mathrm{n}}({\mathscr G} (V))$ for which there exists no joint
distribution $p(x_{V\cup L}) \in {\mathcal{P}}^{\mathrm{d}}({\mathscr G} (V\cup L))$. This
is because marginals of hidden variable DAGs may induce additional
\emph{inequality constraints}, which are not satisfied by all elements of
the associated nested Markov model. The best known of these are the
\emph{instrumental inequalities} of \citet{pearl:95}, which were generalized
by \citet{bonet}, \citet{evans:12} and \citet{kedagni2020generalized}.
Building on the above result, \citet{evans:complete} showed further that
the discrete latent variable DAG model does not imply any additional
\emph{equality} constraints not implied by the discrete nested Markov model.

%e47 #&#
\begin{example}%
\label{ex:chsh}
Consider variables $X_{0},\ldots ,X_{4}$ under a distribution, which is
Markov with respect to the graph in Figure~\ref{fig:mut}(i). Then the
marginal distribution over $X_{1},\ldots ,X_{4}$ satisfies the nested Markov
property with respect to the graph in Figure~\ref{fig:verma1}(i). However,
if the observed variables are binary (and regardless of the state-space
of $X_{0}$) their marginal distribution also satisfies the following inequality
constraints not implied by the nested Markov property:
\begin{align*}
0 \leq{}& \bigl\{ q_{24}(0_{2} \cmid x_{1}) +
q_{24}(0_{4} \cmid x_{3}) + q_{24}(0_{2},
0_{4} \cmid 1 - x_{1}, 1 - x_{3})
\\
&{} - q_{24}(0_{2}, 0_{4} \cmid x_{1},
x_{3}) - q_{24}(0_{2}, 0_{4} \cmid
x_{1}, 1-x_{3}) - q_{24}(0_{2},
0_{4} \cmid 1 - x_{1}, x_{3}) \bigr\} \leq 1
\end{align*}
 for each $x_{1}, x_{3} \in \{0,1\}$; here, for example, $0_{2}$ is a
shorthand for the event $\{X_{2} = 0\}$. These are related to the CHSH
inequalities of \citet{clauser:69}, and are sometimes referred to as
\emph{Bell inequalities} after \citet{bell:64}.
\end{example}

%s4.2 #&#
\subsection{Causal model of a DAG}
\label{subsec:causal}

The statistical model of a DAG ${\mathscr G} $ with vertices $V$, described earlier,
is a set of distributions $p(x_{V})$ defined by the factorization (A.1)
formulated on ${\mathscr G} $. We define a causal model of a DAG ${\mathscr G} $ by a set of
similar factorizations that yield joint distributions under an
\emph{intervention} operation, which corresponds to setting values of variables
from outside the system.

Specifically, for a DAG ${\mathscr G} $ with vertices $V$, and any
$A \subseteq V$, the causal model for ${\mathscr G} $ defines the kernel resulting
from intervention on $x_{A}$ to be
%
%e29 #&#
\begin{align}
p\bigl(x_{V \setminus A} \mid \Dovtex _{{\mathscr G}}(x_{A})\bigr)
\equiv \prod_{v \in V
\setminus A} p(x_{v} \cmid
x_{\pavtex _{{\mathscr G}}(v)}). \label{eq:g-formula}
\end{align}
This is known as the
\emph{g-formula, truncated factorization or manipulated distribution.} Note
that since for any DAG ${\mathscr G} $,
${\mathcal{P}}^{\mathrm{d}}({\mathscr G} ) = {\mathcal{P}}^{\mathrm{n}}({\mathscr G} )$, we have
%
%e30 #&#
\begin{align}
p\bigl(x_{V \setminus A} \mid \Dovtex _{{\mathscr G}}(x_{A})\bigr)
\equiv \prod_{v \in V
\setminus A} p(x_{v} \cmid
x_{\pavtex _{{\mathscr G}}(v)}) = \phi _{A}\bigl(p(x_{V}); {\mathscr G} \bigr).
\label{eq:g-formula-fix}
\end{align}
We will drop the ${\mathscr G} $ subscript from $\Dovtex (\cdot)$ when the graph is obvious.
Thus, (\ref{eq:g-formula-fix}) provides a causal interpretation of the
fixing operation in a DAG.

In this light the earlier propositions relating to the preservation of
independence in kernels may also be reinterpreted. Specifically, Proposition~\ref{prop:constructed} (Separation) states that a variable that is intervened
on no longer depends on its direct causes. Proposition~\ref{prop:ind-preserve-1} (Ordering) states that interventions in the future
cannot causally affect the past. While Proposition~\ref{prop:ind-preserve-2} (Modularity) says that in the setting given,
intervening and conditioning are interchangeable.

%s4.3 #&#
\subsection{Reformulation of the ID algorithm via fixing}
\label{subsec:id}

Identification of causal effects is a more difficult problem in causal
DAGs where some variables are unobserved. In particular, not every distribution
$p(x_{Y} \mid \Dovtex (x_{A}))$ is identified in every ${\mathscr G} (L \cup V)$. Given
a causal model ${\mathscr G} (L \cup V)$, the goal is to find an identifying functional
for $p(x_{Y} \mid \Dovtex (x_{A}))$ in terms of the observed marginal distribution
$p(x_{V})$ \emph{or}  to show that no such functional exists.

The problem may be formalized by considering a latent projection ADMG
${\mathscr G} (V)$ in place of the original causal DAG with hidden variables,
${\mathscr G} (L \cup V)$. A well-known ``folklore'' result in causal inference states
that any two hidden variable DAGs ${\mathscr G} ^{1}(L^{1} \cup V)$ and
${\mathscr G} ^{2}(L^{2} \cup V)$ with the same latent projection ${\mathscr G} (V)$ will share
all identifying functionals, and so there is no loss of generality in reasoning
on ${\mathscr G} (V)$. We prove this folklore result later in this section as Corollary~\ref{cor:admg-id}.

The ID algorithm, introduced by \citet{tian02on} solves the identification
problem; the algorithm, which applies to ADMGs, was proved to be complete
by \citet{shpitser06id} and also independently by \citet{huang06do}. Here,
``complete'' means that whenever the algorithm fails to find an expression
for $p(x_{Y} \mid \Dovtex (x_{A}))$ in terms of $p(x_{V})$ in the causal model
given by ${\mathscr G} (L \cup V)$, no other algorithm is able to do so without making
more assumptions. In this section, we use the fixing operation to give
a simple constructive characterization (via an algorithm) of all causal
effects identifiable by the ID algorithm, and thus of all causal effects
identifiable in any hidden variable causal DAG ${\mathscr G} (L \cup V)$. We can
view this characterization as using the fixing operation to simplify the
ID algorithm from its original two-page formulation down to the single
formula (\ref{eqn:1-line-id}).

%t48 #&#
\begin{theorem}
\label{thm:1-line-id}
Let ${\mathscr G} (L \cup V)$ be a causal DAG with latent projection ${\mathscr G} (V)$. For
$A\dotcup Y\subseteq V$, let
$Y^{*} = \anvtex _{{\mathscr G} (V)_{V \setminus A}}(Y)$. Then if
${\mathcal{D}}({\mathscr G} (V)_{Y^{*}}) \subseteq {\mathcal{I}}({\mathscr G} (V))$,
%
%e31 #&#
\begin{align}\begin{split}
p\bigl(x_{Y} \cmid \Dovtex _{{\mathscr G} (L \cup V)}(x_{A})\bigr)
&= \sum_{x_{Y^{*}
\setminus Y}} \prod_{D \in {\mathcal{D}}({\mathscr G} (V)_{Y^{*}})}
p\bigl(x_{D} \cmid \Dovtex _{{\mathscr G} (L \cup V)}(x_{\pavtex _{{\mathscr G} (V)}(D)\setminus D})\bigr)
\\
&= \sum_{x_{Y^{*} \setminus Y}} \prod_{D \in {\mathcal{D}}({\mathscr G} (V)_{Y^{*}})}
\phi _{V \setminus D}\bigl(p(x_{V});{\mathscr G} (V)\bigr). \label{eqn:1-line-id}
\end{split}
\end{align}
If not, there exists $D \in {\mathcal{D}}({\mathscr G} (V)_{Y^{*}})$ that is not intrinsic
in ${\mathscr G} (V)$, and $p(x_{Y} \mid \Dovtex _{{\mathscr G} (L \cup V)}(x_{A}))$ is not identifiable
in ${\mathscr G} (L \cup V)$.
\end{theorem}
Note that $Y^{*}$ is the set of vertices $v \in V\setminus A$ for which,
for some $y\in Y$, there is a directed path
$v \rightarrow \cdots \rightarrow y$, with no vertex on the path in
$A$; such paths are called ``proper causal paths'' in
\citet{perkovic2015complete}. Also note that since, by Theorem~\ref{thm:dags_in_nested},
$X_{D} \civtex X_{V\setminus (D\cup \pavtex _{{\mathscr G} (V)}(D))} \mid X_{\pavtex _{{\mathscr G} (V)}(D)
\setminus D}$ in $\phi _{V \setminus D}(p(x_{V});{\mathscr G} (V))$, it follows that
$\phi _{V \setminus D}(p(x_{V});{\mathscr G} (V))$ is a function solely of
$x_{D}$ and $x_{\pavtex _{{\mathscr G} (V)}(D) \setminus D}$. Thus, the product on the
RHS of (\ref{eqn:1-line-id}) is a function of the ``bound'' variables
$x_{Y^{*}\setminus Y}$ present in the sum and (a subset of) the ``input''
variables on the LHS, $x_{Y}$, $x_{A^{*}}$ where
$A^{*} = A \cap \bigcup_{D \in \mathcal{D}({\mathscr G} (V)_{Y*})} \pavtex _{{\mathscr G} (V)}(D)$.

%c49 #&#
\begin{corollary}
\label{cor:admg-id}
Let ${\mathscr G} ^{1}(L^{1} \cup V)$ and ${\mathscr G} ^{2}(L^{2} \cup V)$ be two causal DAGs,
with the same latent projection, so ${\mathscr G} ^{1}(V) = {\mathscr G} ^{2}(V)$. Then for
any $A\dot{\cup}Y \subseteq V$:
\begin{itemize}%[(ii)]
\item[(i)] $p(Y \mid \Dovtex _{{\mathscr G} ^{1}}(A))$ is identified if and only if
$p(Y \mid \Dovtex _{{\mathscr G} ^{2}}(A))$ is identified;
\item[(ii)] if $p(Y \mid \Dovtex _{{\mathscr G} ^{1}}(A))$ is identified, then
$p(Y \mid \Dovtex _{{\mathscr G} ^{1}}(A)) = p(Y \mid \Dovtex _{{\mathscr G} ^{2}}(A))$.
\end{itemize}
%was: longlist environment
 %
%
\end{corollary}

%e50 #&#
\begin{example}
Given some hidden variable DAG ${\mathscr G} (V\cup L)$, where
$V=\{x_{1},\ldots ,x_{4}\}$ with latent projection ${\mathscr G} (V)$ given by the
ADMG in Figure~\ref{fig:verma1}(i), consider the problem of identifying
$p(x_{4} \mid \Dovtex _{{\mathscr G}}(x_{2}))$. Mapping this problem to the notation
of Theorem~\ref{thm:1-line-id}, we have $Y = \{ 4 \}$,
$A = \{ 2 \}$, $Y^{*} = \{ 4, 3, 1 \}$. The districts of
${{\mathscr G}}_{Y^{*}}$ are $\{ 4 \}$, $\{ 3 \}$ and $\{ 1 \}$. In fact, these
three sets are intrinsic in ${\mathscr G} $, and thus a fixing sequence exists for
each corresponding kernel:
\begin{align*}
\phi _{{\langle 2, 3, 4 \rangle}}\bigl(p(x_{1},x_{2},x_{3},x_{4});
{\mathscr G} \bigr) &= \phi _{{\langle 2, 3 \rangle}} \biggl( \frac{p(x_{1},x_{2},x_{3},x_{4})}{p(x_{4} \cmid x_{3},x_{2},x_{1})}; {\mathscr G}
_{\{1,2,3\}} \biggr)
\\
&= \phi _{{\langle 2 \rangle}} \biggl( \frac{p(x_{1},x_{2},x_{3})}{p(x_{3} \cmid x_{2},x_{1})}; {\mathscr G} _{\{1,2\}}
\biggr)
\\
&= \frac{p(x_{1},x_{2})}{p(x_{2} \cmid x_{1})} = p(x_{1}),
\end{align*}
\begin{align*}
\phi _{{\langle 1, 2, 4 \rangle}}\bigl(p(x_{1},x_{2},x_{3},x_{4});
{\mathscr G} \bigr) &= \phi _{{\langle 1, 2 \rangle}} \biggl( \frac{p(x_{1},x_{2},x_{3},x_{4})}{p(x_{4} \cmid x_{3},x_{2},x_{1})}; {\mathscr G}
_{\{1,2,3\}} \biggr)
\\
&= \phi _{{\langle 1 \rangle}} \biggl( \frac{p(x_{3},x_{2},x_{1})}{p(x_{2} \cmid x_{1})}; \phi _{2}({\mathscr G}
_{\{1,2,3
\}}) \biggr)
\\
&= \frac{p(x_{3} \cmid x_{2}, x_{1})p(x_{1})}{p(x_{1})} = p(x_{3} \cmid x_{2},
x_{1}),
\end{align*}
\begin{align*}
\phi _{{\langle 2, 3, 1 \rangle}}\bigl(p(x_{1},x_{2},x_{3},x_{4});
{\mathscr G} \bigr) &= \phi _{{\langle 2, 3 \rangle}} \biggl( \frac{p(x_{1},x_{2},x_{3},x_{4})}{p(x_{1})}=
p(x_{2},x_{3},x_{4} \cmid x_{1});
\phi _{1}({\mathscr G} ) \biggr)
\\
&= \phi _{{\langle 2 \rangle}} \biggl( \frac{p(x_{2},x_{3},x_{4}\cmid x_{1})}{p(x_{3}\cmid x_{2},x_{1})} \equiv q_{24}(x_{2},x_{4}
\cmid x_{1},x_{3}); \phi _{31}({\mathscr G} ) \biggr)
\\
&= \frac{q_{24}(x_{2},x_{4}\cmid x_{1},x_{3})}{q_{24}(x_{2}\cmid x_{1},x_{3},x_{4})}
\\
&= \sum_{x_{2}} p(x_{4}\cmid
x_{3},x_{2},x_{1}) p(x_{2} \cmid
x_{1}).
\end{align*}
The last step here follows because
$q_{24}(x_{2},x_{4}\cmid x_{1},x_{3}) = p(x_{2} \cmid x_{1}) p(x_{4}
\cmid x_{1},x_{2},x_{3})$, and
$q_{24}(x_{2}\cmid x_{1},x_{3},x_{4}) = q_{24}(x_{2},x_{4}\cmid x_{1},x_{3})
  /   ( \sum_{x_{2}} q_{24}(x_{2},x_{4}\cmid x_{1},x_{3})
  )  $. Combining these kernels as in Theorem~\ref{thm:1-line-id} yields the same identifying functional as the one obtained
by the ID algorithm applied to ${\mathscr G} $:
\begin{equation*}
p\bigl(x_{4} \mid \Dovtex _{{\mathscr G}}(x_{2})\bigr) =
\sum_{x_{3},x_{1}} p(x_{1}) p(x_{3}
\cmid x_{2},x_{1}) \sum_{x'_{2}} p
\bigl(x_{4}\cmid x_{3},x'_{2},x_{1}
\bigr) p\bigl(x'_{2} \cmid x_{1}\bigr),
\end{equation*}
where we relabel $x_{2}$ as $x'_{2}$ in the last kernel to avoid confusion
between free and summation quantifier-captured versions of the variable
$x_{2}$ in the final expression.
\end{example}

%s4.4 #&#
\subsection{Connections with Tian's constraint algorithm}
\label{subsec:tian}

An algorithm for enumerating constraints on kernels in marginals of DAG
models was given in \citet{tian02on}. Tian's algorithm may be viewed as
implementing fixing for both graphs and kernels, with three important differences
from our formalization. First, unlike CADMGs, subgraphs obtained by fixing
in \citet{tian02on} do not show fixed nodes explicitly. This makes it difficult
to graphically represent constraints that may involve fixed nodes. Second,
the kernel objects obtained in intermediate steps of the algorithm, called
``q-factors'' and written as $Q[V]$ where $V$ is the set of nodes not yet
fixed, do not explicitly show the dependence on nodes already fixed. This
makes it hard to explicitly write down restrictions on q-factors, since
these restrictions often state that dependence on some already fixed nodes
does not exist. Third, there is no unified fixing operation on kernels,
instead the algorithm in \citet{tian02on} alternates between steps corresponding
to the application of the g-formula (division by a conditional density),
and steps corresponding to marginalization.

For a given DAG ${\mathscr G} (V \cup L)$ and a density $p(x_{V \cup L})$ Markov
relative to ${\mathscr G} (V \cup L)$, a subset of observable nodes $V$, and a topological
order $\prec $ on ${\mathscr G} $, Tian's constraint algorithm gives a list of constraints
of the form ``a kernel corresponding to a q-factor $Q[C]$ obtained by some
set of applications of the g-formula and marginalization on
$p(x_{V})$ does not functionally depend on a set $X_{D}$, for some
$D \subseteq V$.''

The algorithm in \citet{tian02on} was developed in the context of hidden
variable DAG models only. We have reformulated this algorithm, using the
language of CADMGs and kernels, to yield an algorithmic specification of
a set of generalized independence constraints implied by the larger nested
model, which does not rely on the existence of an underlying hidden variable
DAG; see Algorithm~1 in the Supplementary Material, Section~E. As noted before, this is not a trivial reformulation,
since the issue of invariance to choice of valid fixing sequences arises
if we do not assume the existence of an underlying hidden variable DAG.
We now state a key result relating this algorithm and the nested Markov
model.

%t51 #&#
\begin{theorem}
\label{thm:tian_equals_nested}
For an ADMG ${\mathscr G} (V)$, let ${\mathcal{P}}_{\mathrm{t}}({\mathscr G} , V, \prec )$ be the set
of densities $p(x_{V})$ in which the list of constraints found by Algorithm~1 holds. Then
${\mathcal{P}}_{\mathrm{t}}({\mathscr G} , V, \prec ) = {\mathcal{P}}^{\mathrm{n}}({\mathscr G} (V))$.
\end{theorem}
Thus, the set of constraints given by Algorithm~1 implicitly
defines the nested Markov model.

%s4.5 #&#
\subsection{Connections with r-factorization}

\citet{shpitser11eid} used constraints in causal DAG models with latent
variables to construct a variable elimination (VE) algorithm for evaluating
causal queries $p(x_{Y} \cmid \Dovtex (x_{A}))$ in a computationally efficient
manner. This algorithm used an older definition called the ``r-factorization
property.'' The nested Markov model r-factorizes, which implies that the
VE algorithm applies to these models as well.

%t52 #&#
\begin{theorem}
\label{r-fact}
If $p(x_{V}) \in {\mathcal{P}}^{\mathrm{n}}({\mathscr G} (V))$, then $p(x_{V})$ r-factorizes
with respect to ${\mathscr G} $ and
$\{ \phi _{V \setminus C}(p(x_{V}); {\mathscr G} ) \mid C \in {\mathcal{I}}({\mathscr G} ) \}$.
\end{theorem}

%s5 #&#
\section{Summary}

We have introduced a novel statistical model defined by the equality constraints
holding in marginals of DAG models, such as the Verma constraint. Though
this model represents constraints found in marginal distributions, it does
not itself model latent variables explicitly; indeed the existence of a
latent variable may imply additional inequality constraints not captured
by our model; see Example~\ref{ex:chsh}. We call this model the nested
Markov model, and it is represented by means of an acyclic directed mixed
graph (ADMG). Our model is ``nested'' because it is recursively defined.
Specifically, just as a DAG model links a graph and a distribution, the
nested model links sets of graphs derived from the original ADMG by a graphical
fixing operation with sets of kernels obtained from the original distribution
by an analogous fixing operation. This operation unifies certain marginalizations,
conditioning operations and applications of the g-formula. Central to our
model definition is the fact that any two valid sequences of fixing operations
that fix the same set of nodes give the same result. We have characterized
our model via Markov properties and a factorization. We have also shown
a close connection between our model and a constraint enumeration algorithm
for marginals of causal DAG models given in \citet{tian02on}, and used
the fixing operation to characterize all identifiable causal effects in
hidden variable DAG models using a one line formula (\ref{eqn:1-line-id}).

%%%%%%%%%%%%%%%%%%%%%%%%%%%%%%%%%%%%%%%%%%%%%%%%%%%%%%%%%%%%%%%%%%%%%%%%%
%\begin{appendix}
%%\section{}
%\end{appendix}

% Change title according to manuscript
%\begin{acks}[Acknowledgments]
\subsection*{Acknowledgements}
We thank Zhongyi Hu for pointing out an error in the proof of Proposition~\ref{prop:semigraphoid}, and the Associate Editor and referees for their
helpful comments. The authors completed work on this paper while visiting
the American Institute for Mathematics and the Simons Institute, Berkeley, California.
%\end{acks}
%
%\begin{funding}%% 1. Put support info here; 2. Support of each author in new paragraph
The first author was supported in part by ONR Grants N00014-19-1-2446 and
N00014-15-1-2672, and NIH Grant R01 AI032475.

The third author was supported
in part by ONR Grant N00014-19-1-2446, and NIH Grant R01 AI032475.

The fourth author was supported in part by ONR Grant N00014-21-1-2820,
NSF Grants 2040804 and 1942239, and NIH Grant R01 AI127271-01A1.

{
\small

%\begin{thebibliography}{1}
%
%\bibitem[\protect\citeauthoryear{Richardson et al.}{2023}]{supp}
%\begin{bmisc}[author]
%\bauthor{\bsnm{Richardson},~\binits{T.~S.}},
%\bauthor{\bsnm{Evans},~\binits{R.~J.}},
%\bauthor{\bsnm{Robins},~\binits{J.~M.}} \AND
%\bauthor{\bsnm{Shpitser},~\binits{I.}}
%(\byear{2023}).
%\bhowpublished{Supplement to ``Nested Markov properties for acyclic directed mixed graphs.''
%\doiurl{10.1214/22-AOS2253SUPP}}
%\end{bmisc}
%\bptok{imsref}%
%\endbibitem
%
%\end{thebibliography}

\bibliographystyle{chicago}
\bibliography{references}
}

\makeatletter\@input{xx_supp.tex}\makeatother

%\begin{thebibliography}{99}
%\bibitem[\protect\citeauthoryear{}{}]{r1}
%\bibitem{r1}
%\end{thebibliography}
\end{document}

\setcounter{theorem}{35}
\setcounter{figure}{9}

%\jmlrheading{?}{2010}{?-?}{10/10}{?}{Ilya Shpitser, Thomas Richardson,
%\jmlrheading{}{2010}{}{}{}{Ilya Shpitser, Thomas Richardson,
%Robin Evans, and James Robins}

\title{Supplementary Materials for\\
Nested Markov Properties for Acyclic Directed Mixed Graphs}

%\author{\name Ilya Shpitser \email ishpitse@hsph.harvard.edu\\
%\name Thomas Richardson \email tsr@stat.washington.edu\\
%\name James Robins \email robins@hsph.harvard.edu\\
%\name Robin Evans \email rje42@stat.washington.edu
%}

%\editor{}

%TODO: fix dangling lemma ref.

\maketitle

%\begin{keywords}
%graphical models, hidden variable models, structure learning, search and score,
%causality, discrete models, r-factorization
%\end{keywords}
\appendix

\section{Graphical Model Definitions And Results}

\subsection{Graphical Definitions} \label{sec:app:definitions}

An undirected graph $\G(V)$ is a pair consisting of a finite set of vertices
$V$ and a finite set of unordered pairs of vertices corresponding to undirected edges,
which we denote by $E \subseteq \{\{a,b\} \subseteq V : a \neq b\}$.
We denote undirected edges $\{a,b\}$ by $a-b$.

A directed graph $\G(V)$ is similar but the edge set consists of
\emph{ordered} pairs: $E \subseteq \{(a,b) \in V \times V : a \neq b\}$; 
if $(a,b) \in E$ we write $a \rightarrow b$.  A directed acyclic graph
(DAG) is subject to the restriction that there are no directed 
cycles $v\rightarrow \cdots \rightarrow v$.

A \emph{directed mixed graph} $\G(V)$ is a graph with a set of
vertices $V$, and a set of edges $E$ which are each either directed
($\to$) or bidirected ($\leftrightarrow$).  
A {\it path} in any of our graphs $\G$ is a sequence of distinct, 
adjacent edges, of any type or orientation, between distinct vertices.  
The first and last vertices on the path are the \emph{endpoints}.
% (paths consisting of a single vertex are permitted 
%for the purpose of simplifying proofs, in which case 
% the endpoints are the same)
 In mixed graphs it is necessary to specify a path as a sequence 
 of edges rather than vertices because it is possible that there 
 is both a directed and a bidirected edge between the same
 pair of vertices.
 A path of the form $a\rightarrow\cdots\rightarrow b$ is a {\em 
 directed path} from $a$ to $b$; similarly, paths of the form 
$a-\cdots- b$ or 
 $a\leftrightarrow\cdots\leftrightarrow b$ are respectively 
 {\em undirected} or {\em bidirected} paths
 between $a$ and $b$.

Let $a$, $b$ and $d$ be vertices in a directed mixed graph $\G$. If $b\to a$ then
we say that $b$ is a {\em parent} of $a$, and $a$ is a {\em child} of $b$.
% If $a\leftrightarrow b$ then $a$ is said to be a {\em spouse} of $b$.
A vertex $a$ is said to be an {\it ancestor} of a vertex $d$
if {\it either} there is a directed path $a\rightarrow\cdots\rightarrow
d$ from $a$ to $d$, {\it or} $a=d$; similarly $d$ is said to be
a {\em descendant} of $a$.  If this is not the case we say that 
$d$ is a {\em non-descendant} of $a$.  

If $a \leftrightarrow b$ we say that $a$ is a \emph{sibling} of $b$ (and vice
versa).  The \emph{district} of $a$ is the set of vertices that are connected
to $a$ by a bidirected path (including $a$ itself).  
The set of parents, children, ancestors, descendants, 
non-descendants and siblings of $a$ in $\G$ are written $\pa_\G(a)$,
$\ch_\G(a)$, $\an_\G(a)$,
$\de_\G(a)$, $\nd_\G(a)$ and $\sib_\G(a)$ respectively; the 
district of $a$ is denoted $\dis_\G(a)$. 
We apply these definitions disjunctively to sets, 
e.g.~$\an_\G(A) = \bigcup_{a\in A} \an_\G(a)$.
%\]
A set of vertices $A$ in $\G$ is called \emph{ancestral} if $\an_\G(a) \subseteq A$ whenever $a \in A$.

An ordering $\prec$ of nodes in $\G$ is said to be {\em topological}
if for any vertex pair $a,b \in \G$, if $a \prec b$, then
$a \not\in \de_\G(b)$; 
note that this definition is the same as that for a DAG.  
We define the set
$\pre_{\G, \prec}(b) \equiv \{ a \; |\; a \prec b \}$.
%\in \mathbb{V}(\G) %removed as V(G) not defined yet. Also not right for CADMGs

%We define
%the {\em parents} of $v$ to be $\pa_{\G}(v) \equiv \{ x \mid x \rightarrow v\}$.

\begin{definition}\label{def:dag-markov} A distribution $p(x_V)$ is said to be {\em Markov relative
to a DAG} $\G$ if
\begin{equation}\label{eq:dagfactor}
p(x_V) = \prod_{v\in V} p(x_v \cmid x_{\pa_{\G}(v)}).
\end{equation}
\end{definition}
\noindent We denote the set of distributions that are Markov relative to a DAG
$\G$ by ${\cal P}^{\rm d}(\G)$.

A {\em directed cycle} is a
path of the form $v \to \cdots \to w$ along with an edge $w \to v$.  An
\emph{acyclic} directed mixed graph (ADMG) is a mixed graph containing no
directed cycles.  For any $T \subset V$, the \emph{induced subgraph} $\G_T$
of $\G$ contains the vertex set $T$, and the subset of edges in $E$ that
have both endpoints in $T$.

%When the edge set or vertex sets are clear from context we will abbreviate 
%a graph ${\G}(V,E)$ as ${\G}(V)$ or $\G$.

%The set of parents, children, ancestors, descendants, and
%non-descendants of $a$ in $\G$ are written $\pa_\G(a)$,
%$\ch_\G(a)$, $\an_\G(a)$,
%$\de_\G(a)$, and $\nd_\G(a)$ respectively.
%An ordering $\prec$ of nodes in $\G$ is said to be {\em topological}
%if for any vertex pair $a,b \in \G$, if $a \prec b$, then
%$a \not\in \de_\G(b)$; 
%note that this definition is the same as that for a DAG.  
%We define the set
%$\pre_{\G, \prec}(b) \equiv \{ a \; |\; a \prec b \}$.
%%\in \mathbb{V}(\G) %removed as V(G) not defined yet. Also not right for CADMGs
%We apply these definitions disjunctively to sets, e.g.
%%\[
%$\an_\G(A) = \bigcup_{a\in A} \an_\G(a)$.
%%\]
%A set of vertices $A$ in $\G$ is called \emph{ancestral} if $a \in A \Rightarrow \an_\G(a) \subseteq A$.

\subsection{The m-separation criterion}
\label{ssec:msep}

We introduce the natural
extension of d-separation to mixed graphs.
A non-endpoint 
vertex $z$ on a path is a
\emph{collider on the path} if the edges preceding and succeeding $z$ on the path both 
have an arrowhead at $z$, i.e. $\rightarrow  z \leftarrow$, $\leftrightarrow
z\leftrightarrow$,
$\leftrightarrow  z\leftarrow$, $\rightarrow z\leftrightarrow$.  A non-endpoint
vertex $z$ on a path which is not a collider is a 
\emph{non-collider on the path}, i.e. $\leftarrow z \rightarrow$, $\leftarrow
z\leftarrow$, $\rightarrow z \rightarrow$, $\leftrightarrow z 
\rightarrow$, $\leftarrow z \leftrightarrow$.  A path between vertices 
$a$ and $b$ in a mixed graph $\G$ is said to be \emph{m-connecting given a set} $C$ \emph{in} 
 $\G$
if 
%\begin{itemize}
%\item[(i)]
every non-collider on the path is not in $ C$, and %\\[-16pt]
%\item[(ii)]
every collider on the path is an ancestor of $ C$ in $\G$. 
%\end{itemize}
If there is no path 
 m-connecting $a$ and $b$
given $ C$, then $a$  and $b$ are said to be {\it m-separated} given $ C$. 
Sets
$ A$ and $ B$ are said to be 
\emph{m-separated} given $ C$, if for all $a$, $b$, with $a \in {
A}$ and $b\in { B}$, $a$ and $b$ are  m-separated given $ C$. 
 Note that if $\G$ is a DAG then the above
definition is identical to Pearl's d-separation
criterion; see \citep{pearl88probabilistic}.

%\begin{comment}
%A probability measure $P$ on ${\mathfrak X}_V$ satisfies the {\it 
%global Markov property} for ADMG $\G(V,E)$ if for arbitrary 
%disjoint sets $A,B,C$ ($C$ may be empty),
%%\begin{align*}
%$A \hbox{ m-separated from }B\hbox{ given }C\hbox{ in }\G\hbox{ implies }
%%$X_A$ is conditionally independent of $X_B$ given $X_C$ in $P$ (written
%X_A\ci X_B  \mid X_C \;[P]$.
%%\end{align*}
%%The set of probability measures that satisfy the global Markov property 
%%criterion with respect to $\G$ is denoted ${\cal P}_{ m} (\G)$.
%\end{comment}

%\cite{lau:bk, pearl:1988}.

\subsection{Latent Projections}
\label{ssec:latent project}

Given a DAG with latent variables we associate a mixed
graph via the following operation; see \citep{verma91theory}.

\begin{definition}[latent projection]\label{def:proj}
Let  ${\G}$ be an ADMG with vertex set 
$V\dotcup L$ where the vertices in $V$ are observed,
those in $L$ are latent and $\dotcup$ indicates a disjoint union. The {\em latent projection} $\G(V)$
is a directed mixed graph with vertex set $V$, where for every pair of distinct
vertices $a,b \in V$:

\begin{itemize}
\item[\rm (i)] $\G(V)$ contains an edge $a \to b$ if there is a
directed path $a \rightarrow \cdots  \rightarrow b$ on which every
non-endpoint vertex is in $L$.

\item[\rm (ii)] $\G(V)$ contains an edge $a \leftrightarrow b$ if
there exists a path between $a$ and $b$ such that the non-endpoints 
are all non-colliders in $L$, and such that the 
edge adjacent to $a$ and the edge adjacent to $b$ both have arrowheads at those vertices.
For example, $a \leftrightarrow \cdots \rightarrow b$.
\end{itemize}
\end{definition}

%By analogy with the operation of marginalization,

Generalizations of this construction are considered by \citet{wermuth:11},
\citet{koster:02} and \citet{sadeghi14markov} in the context of marginalizing {\em and}\ conditioning.
As an example, the mixed graph in Figure \ref{fig:verma1}(i) is the latent
projection of the DAG shown in Figure \ref{fig:mut}(i).

\begin{proposition}\label{proj-admg}
If $\G$ is an ADMG with vertex set $V\dotcup L$ then $\G(V)$
is also an ADMG.
\end{proposition}
%\begin{proa}{\ref{proj-admg}}
%If $\mathscr{G}(V\dot{\cup} L)$ is a DAG then $\mathscr{G}(V)$ is an ADMG.
%\end{proa}
\begin{proof} 
It follows directly from the construction that if
$v \to v^\prime$ in $\mathscr{G}(V)$ then $v \in \an_{\mathscr{G}}(v^\prime)$.
The presence of a directed cycle in $\mathscr{G}(V)$ would imply
a directed cycle in $\mathscr{G}$, which is a contradiction.
\end{proof}

\bigskip

\begin{definition}[latent projection for CADMGs]
Let  ${\G(V\dotcup L,W)}$ be a CADMG
where $W$ is a set of fixed vertices, the random vertices are $V \cup L$ and
those in $L$ are latent. The {\em latent projection} $\G(V, W)$
is a directed mixed graph with random vertex set $V$, where for every pair of distinct
vertices $a,b \in V \cup W$:

\begin{itemize}
\item[\rm (i)] $\G(V,W)$ contains an edge $a \to b$ if there is a
directed path $a \rightarrow \cdots  \rightarrow b$ on which every
non-endpoint vertex is in $L$.

%% this is wrong!
%\item[\rm (ii)] $\G(V)$ contains an edge $v_i \leftrightarrow v_j$ if
%there exists a path of the form $v_i \gets \cdots \to v_j$, on which every
%non-endpoint vertex is a non-collider and in $H$.

\item[\rm (ii)] $\G(V,W)$ contains an edge $a \leftrightarrow b$ if
there exists a path between $a$ and $b$ such that the non-endpoints 
are all non-colliders in $L$, and such that the 
edge adjacent to $a$ and the edge adjacent to $b$ both have arrowheads at those vertices.
For example, $a \leftrightarrow \cdots \rightarrow b$.

\end{itemize}
\label{def:proj_cadmg}
\end{definition}

\begin{proposition}\label{proj-cadmg}
If $\G(V\dotcup L, W)$ is a CADMG then $\G(V,W)$
is also a CADMG. In particular,  for all $w\in W$, $\pa_{\G(V,W)}(w) = \sib_{\G(V,W)}(w) = \emptyset$.
%so there are no edges between vertices in $W$.
\end{proposition}
\begin{proof}%{\ref{proj-cadmg}}
The absence of directed cycles in $\G(V,W)$ follows from the proof of Proposition
\ref{proj-admg}. Since vertices in $W$ only have children in $\G(V\dotcup L, W)$,
and no vertex in $W$ is latent, it follows from Definition \ref{def:proj_cadmg}
  that there are no additional edges introduced between vertices in $W$.
  Further, any additional edge with a vertex $w \in W$ as an endpoint takes the form 
  $w \rightarrow v$.
%If $\mathscr{G}(V\dot{\cup} L)$ is a DAG then $\mathscr{G}(V)$ is an ADMG.
%\end{proa}
%\begin{proof} 
%It follows directly from the construction that if
%$v \to v^\prime$ in $\mathscr{G}(V)$ then $v \in \an_{\mathscr{G}}(v^\prime)$.
%The presence of a directed cycle in $\mathscr{G}(V)$ would imply
%a directed cycle in $\mathscr{G}$, which is a contradiction.
\end{proof}

\subsection{Additional Results On Graphs And Kernels}

\begin{proa}{\ref{prop:msep}}
Let $\G$ be a DAG with vertex set $V\dotcup L$.  For disjoint subsets $A, B, C \subseteq V$  
(where $C$ may be empty), 
$A$ is d-separated from $B$ given $C$ in $\G$  if and only if
$A$ is m-separated from $B$ given $C$ in $\G(V)$.
\end{proa}
\begin{proof} For every path ${\pi}$ in $\mathscr{G}$, by Definition \ref{def:district}
there is a corresponding path 
${\pi}^*$ in $\mathscr{G}(V)$ consisting of a subsequence of the vertices on $\pi$, such that
if a vertex $v$ is a collider (non-collider) on ${\pi}^*$ then it is a collider (non-collider) on $\pi$.
It follows from this that d-connection in $\mathscr{G}$ implies m-connection in $\mathscr{G}(V)$.
Conversely, by Definition \ref{def:district} for each edge
$\epsilon^*$ with endpoints $e$ and $f$ on $\pi^*$ in $\mathscr{G}(V)$ there is a unique path $\mu_{\epsilon^*}$ with 
endpoints $e$ and $f$ in $\mathscr{G}$ such that there is an arrowhead at $e$ ($f$) on $\epsilon^*$ if and only if the
edge on $\mu_{\epsilon^*}$ with $e$ ($f$) as an endpoint has an arrowhead at $e$ ($f$).
It then follows from Lemma 3.3.1 in \citep{spirtes93causation} that if there is a path m-connecting $a$ and $b$ given $C$
in $\mathscr{G}(V)$ then there is a path d-connecting $a$ and $b$ given $C$ in $\mathscr{G}$. The result
then follows.\end{proof}

\bigskip

\begin{proa}{\ref{prop:add-w}}
In a set of kernels $q_V(x_V \cmid x_W)$, 
$X_A \;\ci\; X_B \mid X_C$ if and only if either
$X_A \;\ci\; X_{B\cup (W\setminus C)} \mid X_C$ or 
$X_B \;\ci\; X_{A\cup (W\setminus C)} \mid X_C$.
\end{proa}
\begin{proof}
This follows directly from Definition \ref{def:ind}.
\end{proof}

\bigskip

\begin{proa}{\ref{prop:semigraphoid}}
The semi-graphoid axioms are sound for kernel independence.
\end{proa}
\begin{proof}
Symmetry follows directly from Definition \ref{def:ind}.

Let $q_V(x_V \cmid x_W)$ be a kernel for which $(X_A \ci X_{B \cup D} \mid X_C)$ holds.  Assume condition (a) for this
independence holds, that is $A \cap W = \emptyset$, and assume $q_V(x_A \cmid x_B, x_C, x_D, x_{W\setminus (B \cup C \cup D)})$
is only a function of $x_A$ and $x_C$.
Then it immediately follows that condition (a) for $(X_A \ci X_B \mid X_{C \cup D})$ also holds.
To see that $(X_A \ci X_D \mid X_C)$ holds, consider the following derivation.
\begin{align*}
\lefteqn{q_V(x_{A} \cmid x_{C}, x_{D}, x_{W \setminus (C \cup D)})}\\
&=
\frac{
\sum\limits_{x_{B \cap V}} q_V(x_{A}, x_{B \cap V}, x_{C \cap V}, x_{D \cap V} \cmid x_W)
}{
\sum\limits_{x_{B \cap V}} q_V(x_{B \cap V}, x_{C \cap V}, x_{D \cap V} \cmid x_W)
}
\\
&=
\frac{
\left(
\sum\limits_{x_{B \cap V}}
q_V(x_{A} \cmid  x_B, x_C, x_D, x_{W \setminus (B \cup C \cup D)}) \cdot
	q_V(x_{B \cap V}, x_{C \cap V}, x_{D \cap V} \cmid x_W)
\right)
}{
\sum\limits_{x_{B \cap V}} q_V(x_{B \cap V}, x_{C \cap V}, x_{D \cap V} \cmid x_W)
}
\\
&=
\frac{
\left(
q_V(x_{A} \cmid  x_B, x_C, x_D, x_{W \setminus (B \cup C \cup D)}) \cdot
\sum\limits_{x_{B \cap V}} 
	q_V(x_{B \cap V}, x_{C \cap V}, x_{D \cap V} \cmid x_W)
\right)
}{
\sum\limits_{x_{B \cap V}} q_V(x_{B \cap V}, x_{C \cap V}, x_{D \cap V} \cmid x_W)
}
\\[6pt]
%&=
%\frac{
%\left(
%\begin{array}{c}
%q_V(x_{A} \mid  x_B, x_C, x_D, x_{W \setminus (B \cup C \cup D)}) \cdot\\
%	q_V(x_{C \cap V}, x_{D \cap V} \cmid x_W)
%\end{array}
%\right)
%}{
%q_V(x_{C \cap V}, x_{D \cap V} \cmid x_W)
%}
&= q_V(x_{A} \cmid  x_B, x_C, x_D, x_{W \setminus (B \cup C \cup D)}).
\end{align*}
Here the first equality follows by (\ref{eqn:kernel-marg-cond}), the second by the chain rule of probability, which applies to
kernels also by (\ref{eqn:kernel-marg-cond}), the third since we established above that
$(X_A \ci X_B \mid X_{C \cup D})$ holds in $q_V(x_V \cmid x_W)$, and the last by cancellation.  Since $(X_A \ci X_{B \cup D} \mid X_{C})$,
the final term does not depend upon $x_B$ or $x_D$, so the independence $(X_A \ci X_D \mid X_C)$ follows.

Now assume $(X_A \ci X_{B \cup D} \mid X_C)$ holds due to condition (b), that is
%$q_V(x_{B \cap V}, x_{D \cap V} \cmid x_{A}, x_{C}, x_{W \setminus (A \cup C)})$ is only a function of $x_B$, $x_C$ and $x_D$.
%Then $q_V(x_{B \cap V} \cmid x_{A}, x_{C}, x_{D}, x_{W \setminus (A \cup C \cup D)})$ is also only a function of
$q_V(x_{B}, x_{D} \cmid x_{A}, x_{C},$ $x_{W \setminus (A \cup C)})$ is only a function of $x_B$, $x_C$ and $x_D$.
Then $q_V(x_{B} \cmid x_{A}, x_{C}, x_{D}, x_{W \setminus (A \cup C \cup D)})$ is also only a function of
$x_B$, $x_C$ and $x_D$, which in turn implies $(X_A \ci X_B \mid X_{C \cup D})$. 
To see that $(X_A \ci X_D \mid X_C)$ holds, %repeat the above argument, but swap $x_A$ and $x_D$,
%and use the fact that $(X_A \ci X_{B \cup D} \mid X_C)$ implies $(X_A \ci X_D \mid X_{B \cup C})$ under either (a) or (b).  This was already shown above.
simply sum over $x_B$, to see that 
$q_V(x_{D} \cmid x_{A}, x_{C}, x_{W \setminus (A \cup C)})$ is only a function of $x_C$ and $x_D$.

To show the converse, let $q_V(x_V \cmid x_W)$ be a kernel in which $(X_A \ci X_B \mid X_{C \cup D})$ and $(X_A \ci X_D \mid X_C)$ hold.
If this is due to condition (a) (where $A \cap W = \emptyset$), then $(X_A \ci X_{B \cup D} \mid X_C)$ follows by the above derivation, and the two assumed independences.  If this is due to condition (b) (where $(B \cup D) \cap W = \emptyset$), then $(X_A \ci X_{B \cup D} \mid X_C)$ follows
by the above derivation where $x_A$ is swapped with $x_B$ and $x_D$ respectively, and the two assumed independences.
\end{proof}

\begin{lema}{\ref{lem:fixing-yields-kernel}}
$q^*_{V\setminus H}(x_{R}, x_{T}\cmid x_{H}, x_{W})$ is a kernel.
\end{lema}

\begin{proof}
Since $q_{V}(x_{R} \cmid x_{H}, x_{T}, x_{W})$ and
$q_{V}(x_{T} \cmid x_{W})$ are derived from $q_{V}(x_V \cmid x_W)$, they are kernels.  Thus,
for every value $x_{H},x_{W}$, $q^*_{V\setminus H}(x_{R}, x_{T}\cmid x_{H}, x_{W}) \geq 0$.
In addition,
\begin{align*}
\sum_{x_R,x_T} q^*_{V\setminus H}(x_{R}, x_{T}\cmid x_{H}, x_{W})
&=
\sum_{x_R,x_T} q_{V}(x_{R} \cmid x_{H}, x_{T}, x_{W}) \, q_{V}(x_{T} \cmid x_{W})\\
&=
\sum_{x_T} q_{V}(x_{T} \cmid x_{W}) \cdot \left( \sum_{x_R} q_{V}(x_{R} \cmid x_{H}, x_{T}, x_{W}) \right)\\
&= \sum_{x_T} q_{V}(x_{T} \cmid x_{W}) \cdot 1 = 1.
\end{align*}
\end{proof}

\begin{lema}{\ref{lem:invariance-kernel}}
For the kernel constructed in {\rm (\ref{eq:invariance-one})}:
%Consider a kernel $q_V(x_V \cmid x_W)$ where $V = R \dotcup H \dotcup T$.
%so $q_V(\cdot \,|\, \cdot) \equiv q_{R\cup H \cup T}(\cdot \,|\,\cdot)$.
%Define a new kernel:\\[-24pt]
%\begin{align}
%q^*_{V\setminus H}(x_{R}, x_{T} \cmid x_{H}, x_{W})&\equiv
%\frac{q_{V}(x_{R}, x_{H}, x_{T} \cmid x_{W})}
%{q_{V}(x_{H} \cmid x_{T}, x_{W})}.\label{eq:invariance-one}
%\end{align}
%Then\\[-24pt]
\begin{align}
q^*_{V\setminus H}(x_{R}, x_{T}\cmid x_{H}, x_{W}) &= q_{V}(x_{R} \cmid x_{H}, x_{T}, x_{W})
q_{V}(x_{T} \cmid x_{W}), \tag{\ref{eq:invariance-two}}
\end{align}
and hence\\[-24pt]
\begin{align}
q^*_{V\setminus H}(x_{R}\cmid x_{H}, x_{T}, x_{W}) &= q_{V}(x_{R}\cmid x_{H}, x_{T}, x_{W});
\tag{\ref{eq:invariance-cond}}\\
q^*_{V\setminus H}(x_{T} \cmid x_{H}, x_{W}) &= q_{V}(x_{T} \cmid x_{W}) = q^*_{V\setminus H}(x_{T} \cmid x_{W}).\tag{\ref{eq:invariance-marg}}
\end{align}
\end{lema}
\begin{proof}
By the chain rule of probability: 
\[
q_{V}(x_{R}, x_{H}, x_{T} \cmid  x_{W}) = 
q_{V}(x_{R} \cmid x_{H}, x_{T}, x_{W}) \, q_{V}(x_{H}  \cmid  x_{T}, x_{W}) \,
q_{V}(x_{T} \cmid  x_{W}).
\]
Hence (\ref{eq:invariance-two}) follows directly from  (\ref{eq:invariance-one}).
Since 
\begin{align}\label{eq:proof-one}
q^*_{V\setminus H}(x_{R}, x_{T}\cmid x_{H}, x_{W}) =
q^*_{V\setminus H}(x_{R}\cmid x_{H}, x_{T}, x_{W}) \,
q^*_{V\setminus H}(x_{T}\cmid x_{H}, x_{W}),
\end{align}
the first equality in (\ref{eq:invariance-marg}) follows by summing the right-hand sides of (\ref{eq:invariance-two}) 
and (\ref{eq:proof-one}) over $x_R$. The second then follows directly since 
$q_{V}(x_{T} \cmid x_{W})$ is not a function of $x_H$.
Finally (\ref{eq:invariance-cond}) follows from (\ref{eq:invariance-marg}) by canceling 
$q^*_{V\setminus H}(x_{T} \cmid x_{H}, x_{W})$ and $q_{V}(x_{T} \cmid x_{W})$
from the right-hand sides of (\ref{eq:invariance-two}) and (\ref{eq:proof-one}).
\end{proof}

\begin{proa}{\ref{prop:constructed}}%[separation]
%Let $q_V(x_V \cmid x_W)$ be a kernel with $T \subset V$,
%$H \subseteq V \setminus T$.
%Then
For the kernel constructed in {\rm (\ref{eq:invariance-one})}:
%\begin{enumerate}[label=(\roman*)] 
\begin{itemize}
\item[{\rm (i)}] $(X_{H} \ci X_{T} \mid X_{W})$ in  $q^*_{V\setminus H}(x_{V \setminus H} \cmid x_{H}, x_{W})$.
% $\equiv \frac{q_V(x_V \cmid x_W)}{q_V(x_H \cmid x_{T \cup W})}$.
\item[{\rm (ii)}] If $X_V \ci X_{W \setminus W_1} \mid X_{W_1}$ in $q_V$, then
$X_{H \cup (W \setminus W_1)} \ci X_T \mid X_{W_1}$ in $q^*_{V \setminus H}$.
\item[{\rm (iii)}] If $R = \emptyset$, $X_H \ci X_{V \setminus H} \,|\, X_W$ in $q^*_{V \setminus H}$.
%\end{enumerate}
\end{itemize}
\end{proa}
\begin{proof}
From (\ref{eq:invariance-marg}), the kernel $q^*_{V\setminus H}(x_T \cmid x_H, x_W)$ is 
a function only of $x_T$ and $x_{W}$, which implies (i).  
Further, if $X_V \ci X_{W \setminus W_1} \mid X_{W_1} \, [q_V]$ then
$q_V(x_V \cmid x_W)$ is not a function of $x_{W \setminus W_1}$, and 
therefore neither is $q_V(x_T \cmid x_W)$.  Hence the previous argument
gives (ii).  (iii) is implied by definition of marginalization.
\end{proof}

\begin{lema}{\ref{lem:markovfact}} 
%For a CADMG $\G$ and any topological ordering $\prec$,
%if $q_V \in {\cal P}^{\rm c}_{\rm f}(\G)$ then
%  for every ancestral set $A$ and every $D \in {\cal D}(\G_{A})$, we have:
%\begin{equation} \tag{\ref{eq:district-factor-ident}}
%f_D^A(x_D \cmid x_{\pa(D) \setminus D}) = \prod_{d \in D} q_V(x_d \cmid x_{T^{(d,D)}_\prec}),
%\end{equation}
%where $T^{(d,D)}_\prec \equiv \mb_\G( d,\; \an_\G(D) \cap \pre_{\G,\prec}(d))$, so that
%\begin{equation} \tag{\ref{eq:markovb}}
%q_V(x_d \cmid x_{T^{(d,D)}_\prec}) = q_V(x_d \cmid x_{A \cap \pre_{\G,\prec}(d)}, x_W).
%\end{equation}
%Conversely if, under some topological ordering $\prec$,  
%{\rm (\ref{eq:district-factor-ident})} and
%{\rm (\ref{eq:markovb})} hold for all ancestral sets $A$ and all $d\in A$ then $q_V \in {\cal P}^{\rm c}_{\rm f}(\G)$.
Let $\G$ be a CADMG with topological ordering $\prec$.
If $q_V \in {\cal P}^{\rm c}_{\rm f}(\G)$, then
  for every ancestral set $A$ and every $D \in {\cal D}(\G_{A})$, we have
\begin{equation}\label{eq:district-factor-ident}
f_D^A(x_D \cmid x_{\pa(D) \setminus D}) = \prod_{d \in D} q_V(x_d \cmid x_{T^{(d,D)}_\prec}),
\end{equation}
where $T^{(d,D)}_\prec \equiv \mb_\G( d,\; \an_\G(D) \cap \pre_{\G,\prec}(d))$, so that
\begin{equation}\label{eq:markovb}
q_V(x_d \cmid x_{T^{(d,D)}_\prec}) = q_V(x_d \cmid x_{A \cap \pre_{\G,\prec}(d)}, x_W).
\end{equation}
Conversely, if
{\rm (\ref{eq:markovb})} holds for all $d \in A$ and ancestral sets $A$, 
and the $f_D^A$ functions are defined by 
{\rm (\ref{eq:district-factor-ident})},
%{\rm (\ref{eq:district-factor-ident})} and
%{\rm (\ref{eq:markovb})} hold for all ancestral sets $A$ and all $d\in A$ 
then $q_V \in {\cal P}^{\rm c}_{\rm f}(\G)$.
\end{lema}

\begin{proof}
(Cf.~proof of Lemma 1 in \citep{tian02general}):\par
\noindent ($\Rightarrow$) The proof is by induction on the size of the set $A$ in (\ref{eq:tian-factor}).
If $|A|=1$, the claim is trivial. 
Suppose that the claim holds for sets $A$ of size less than $n$. Specifically, we assume
that  all factors $f^A_D(\cdot | \cdot )$ occurring in (\ref{eq:tian-factor}) for
sets $A$ such that $|A| < n$, obey (\ref{eq:district-factor-ident}) and (\ref{eq:markovb}).

Now suppose that $A$ contains $n$ variables and that $A \subseteq \{t\} \cup \pre_{\G, \prec}(t)$
for some vertex $t \in A$. Let $D^t \equiv \dis_{\G_{A}}(t)$ be the district containing $t$ 
in $\G_A$, so that by hypothesis:
\begin{equation}\label{eq:indfactor}
q_V(x_{A\cap V} \cmid x_W ) = f^A_{D^t}(x_{D^t} \cmid x_{\pa(D^t)\setminus D^t}) \;
\prod_{D \in {\cal D}(\G_{A})\setminus \{ D^t\}}\!\!\!\! f^A_D(x_D \cmid x_{\pa(D) \setminus D}).
\end{equation}
Since $A\setminus \{t\} \subseteq \pre_{\G, \prec}(t)$, for all
$D \in {\cal D}(\G_{A})\setminus \{ D^t\}$, $t\notin \pa_\G(D)\setminus D$.
Thus summing both sides of (\ref{eq:indfactor}) over $x_t$ leads to:
\begin{align}
q_V(x_{(A\cap V)\setminus \{t\}} \cmid x_W ) &=\left( \sum_{x_t} f^A_{D^t}(x_t, x_{D^t\setminus \{t\}} \cmid x_{\pa(D^t)\setminus D^t})\right)
\nonumber\\
&\quad\times 
\prod_{D \in {\cal D}(\G_{A})\setminus \{D^t\}}\!\!\!\! f^A_D(x_D \cmid x_{\pa(D) \setminus D}).\label{eq:t-margin}
\end{align}
Now since $\prec$ is a topological ordering, $A\setminus \{ t\}$ is an ancestral
set in $\G$; further every district in ${\cal D}(\G_{A})\setminus \{D^t\}$ is also a
district in $\G_{A\setminus \{t\}}$ hence, by the induction hypothesis,
all of the corresponding densities $f^A_D(\cdot | \cdot)$ in
(\ref{eq:t-margin})
obey  (\ref{eq:district-factor-ident}) and (\ref{eq:markovb}). Rearranging (\ref{eq:indfactor}) gives:
\[
 f_{D^t}^A(x_{D^t} \cmid x_{\pa(D^t)\setminus D^t}) =
 \frac{ 
 q_V(x_{A\cap V} \cmid x_W )
 }{
 \prod_{D \in {\cal D}(\G_{A}) \setminus \{D^t\}}\, f^A_D(x_D \cmid x_{\pa(D) \setminus D})
 }.
\]
By the chain rule of probability,
\[
q_V(x_{A\cap V} \cmid x_W) = \prod_{a \in {A\cap V}} q_V(x_a \cmid x_{A \cap \pre_{\G,\prec}(a)},x_W).
\]
Since for every $D \in {\cal D}(\G)\setminus \{D^t\}$, $f^A_D(\cdot | \cdot)$ obeys (\ref{eq:district-factor-ident}) and
 (\ref{eq:markovb})
so 
\begin{equation}\label{eq:firstresult}
f_{D^t}^A(x_{D^t} \cmid x_{\pa(D^t)\setminus D^t}) = 
\prod_{d \in D^t} q_{V}(x_d \cmid x_{A \cap \pre_{\G,\prec}(d)},x_W).
\end{equation}
By the inductive hypothesis applied to $A\setminus \{ t\}$, we have that 
{\rm (\ref{eq:district-factor-ident})} holds for all $D^* \in {\cal D}(\G_{A\setminus \{t\}})$ and 
(\ref{eq:markovb}) holds for all $d \in D^t\setminus \{t\} \subseteq A\setminus \{ t\}$. Consequently:
\begin{equation}\label{eq:firstresult-simplified}
f_{D^t}^A(x_{D^t} \cmid x_{\pa(D^t)\setminus D^t}) = 
q_{V}(x_t \cmid x_{A \cap \pre_{\G,\prec}(t)},x_W)
\prod_{d \in D^t\setminus\{t\}} q_{V}(x_d \cmid x_{T_{\prec}^{(d,A)}}).
\end{equation}
%Consequently, it is sufficient to
%prove that {\rm (\ref{eq:district-factor-ident})} and (\ref{eq:markovb}) also holds for $t$ and thus that 
%{\rm (\ref{eq:district-factor-ident})} holds for $D^t$.
Rearranging (\ref{eq:firstresult-simplified})
%(\ref{eq:firstresult}) 
we obtain:
\begin{eqnarray}
 q_{V}(x_t \cmid x_{A \cap \pre_{\G,\prec}(t)},x_W) &=& 
%  q_{V}(x_t \cmid x_{(A\setminus \{t\})},x_W)\nonumber\\
%  &=&
   \frac{
  f_{D^t}^A(x_{D^t} \cmid x_{\pa(D^t)\setminus D^t})
  }{
   \prod_{d \in D^t\setminus \{t\}} q_{V}(x_d \cmid x_{T_{\prec}^{(d,A)}})
  }.\label{eq:tindep}%\\
%  &=& q_V(x_t \cmid x_{A \cap \pre_{\G,\prec}(t)},x_W)\\
%  &=& q_{V}(x_t \cmid x_{T^{(t,A)}})
\end{eqnarray}
By Proposition \ref{prop:mb}, for all $d \in D^t\setminus \{t\}$ we have $T^{(d,A)}_{\prec} \subseteq (D^t\setminus \{t\}) \cup \pa_\G(D^t)$, so 
the RHS is a function of $x_{D^t \cup \pa (D^t)}$.
%that 
Hence:
\[
X_t \ci X_{(A\cup W) \setminus \left(D^t \cup \pa (D^t)\right)} \mid X_{(D^t\setminus \{t\}) \cup \pa (D^t)}\quad
[q_V]
\]
from which (\ref{eq:markovb}) holds for $t$ and $D^t$; {(\ref{eq:district-factor-ident})} follows from (\ref{eq:firstresult-simplified}).

\noindent ($\Leftarrow$) Follows from construction of the kernels $f^A_D(\cdot | \cdot)$ via (\ref{eq:district-factor-ident}).
\end{proof}

\subsection{Results On Fixing and Fixability} \label{ssec:fixable}

%\begin{proa}{\ref{prop:exists-fixable-in-district}}
%For every district $D$ in a CADMG $\G$, 
%%In a CADMG ${\G}(V,W)$, for every district $D \in {\cal D(G)}$,
%we have $D \cap {\mathbb F}({\G}) \neq \emptyset$.
%\end{proa}
%%That is, in every district there is at least one vertex that is fixable.
%\begin{proof}
%Let $\prec$ be any topological ordering of $\G$. In every district $D$, the 
%$\prec$-maximal vertices in $D$ are fixable in $\G$.
%\end{proof}

\begin{proa}{\ref{prop:exists-fixable-descendant}}
%If  $v\in D \in {\cal D}({\G})$ %, $ D$ %but $v \notin {\mathbb F}({\G})$,
% then $\de_{\G}(v) \cap D \cap {\mathbb F}({\G}) \neq \emptyset$.
%If  $D \in {\cal D}({\G})$, $v\in D$ but $v \notin {\mathbb F}({\G})$,
%then $\de_{\G}(v) \cap D \cap {\mathbb F}({\G}) \neq \emptyset$.
For any $v \in V$, we have $\de_{\G}(v) \cap \dis_\G(v) \cap {\mathbb F}({\G}) \neq \emptyset$.
\end{proa}
%Thus, if a vertex in a district is not fixable then there is a descendant of the vertex within the district that is fixable.
\begin{proof}
Let $\prec$ be any topological ordering of $\G$, and 
%As in the proof above, %of Proposition \ref{prop:exists-fixable-in-district} 
consider the maximal vertex in the set %$\de_{\G}(v) \cap D$.
$\de_{\G}(v) \cap \dis_\G(v) \supseteq \{v\}$.
\end{proof}

\begin{proa}{\ref{prop:barren-fixable}}
If $q_V$ is Markov with respect to a CADMG ${\G}(V,W)$ and $\ch_\G(r) = \emptyset$ 
{\rm (}and hence $r \in \mathbb{F}(\G)${\rm )}, 
%$r \in {\mathbb F}({\G})$ and  $\ch_{\G}(r) = \emptyset$ 
then 
% for all assignments to  $X_r$,
\[
%\left. \vphantom{\vbox{\kern20pt}}
\phi_r (q_V(x_V \cmid x_W); {\G})
%\;\right|_{X_r=x_r}
= \sum_{x_r} q_V(x_V \cmid x_W)
= q_V(x_{V\setminus \{r\}} \cmid x_W).
\]
%If $q_V$ is Markov with respect to a CADMG ${\G}(V,W)$, $r \in {\mathbb F}({\G})$ and  $\ch_{\G}(r) = \emptyset$ 
%then 
%% for all assignments to  $X_r$,
%\[
%%\left. \vphantom{\vbox{\kern20pt}}
%\phi_r (q_V(x_V \cmid x_W); {\G})
%%\;\right|_{X_r=x_r}
%= \sum_{x_r} q_V(x_V \cmid x_W)
%= q_V(x_{V\setminus \{r\}} \cmid x_W).
%\]
%
%%In a CADMG ${\G}(V,W)$ i
%If $r\in V$ and $\ch_{\G}(r) = \emptyset$ then $r \in {\mathbb F}({\G})$, and
%$\phi_r(q_V(x_V \mid x_W); \G) = \sum_{x_r} q_V(x_V \mid x_W)$.
\end{proa}
\begin{proof}
This follows by definition of ${\mathbb F}(\G)$, $\phi_r$ and (\ref{eq:mb-when-t-fixable}).
\end{proof}

Note that, in this Appendix, we will sometimes write simply $\phi_r(q_V)$ rather than
 $\phi_r(q_V; \G)$, provided that the graph with respect to which the fixing is being 
 performed is clear.

\begin{proa}{\ref{prop:fix-with-mb}}
 If $q_V$ is Markov with respect to $\G(V,W)$ and $r \in \mathbb{F}(\G)$, then
 \begin{align}
  \phi_r (q_V(x_V \cmid x_W);{\G}) = q_V(x_V \cmid x_W) / q_V(x_r \cmid x_{\nd_{\G}(r)}). \tag{\ref{eq:alt-fix-def}}
\end{align}
\end{proa}
\begin{proof}
This follows from Theorem \ref{thm:one-step-markov}
 and
 (\ref{eq:mb-when-t-fixable})
 with $t=r$.
 % and $V= \an_\G(\dis_\G(r))$.
\end{proof}

\begin{lema}{\ref{lem:no-backtracking}}
If $r \in {\mathbb F}({\G})$, then ${\mathbb F}({\G}) \setminus \{r\} \subseteq {\mathbb F}(\phi_r ({\G}))$.
\end{lema} 

That is, any vertex $s$ that was fixable before $r$ was
fixed is still fixable after $r$ has been fixed (with the obvious exception of
$r$ itself).  Thus when fixing
vertices, although our choices may be limited at various stages, we are never faced with a choice
between fixing $r$ and $r^\prime$, whereby choosing $r$ precludes subsequently
fixing $r^\prime$.
\begin{proof}
This follows from the definition of ${\mathbb F}({\G})$ and $\phi_r ({\G})$.
Since the set of edges in $\phi_r(\G)$ is a subset of the set of edges in $\G$, any vertex
$t\in V\setminus \{r\}$ that is in ${\mathbb F}({\G})$ is also in ${\mathbb F}(\phi_r({\G}))$.
\end{proof}

\begin{proa}{\ref{prop:sub}}
  If $\G$ is a subgraph of ${\G}^*$ with the same random and fixed
  vertex sets, then
  ${\mathbb F}( {\G}^*) \subseteq {\mathbb F}( {\G})$.
\end{proa}

\begin{proof} 
If $r$ has no descendant within the district containing it in ${\G}^*$ then this also holds in ${\G}$.
\end{proof}

\begin{proa}{\ref{prop:districts-after-fixing}}
Let $\G$ be a CADMG, with $r \in {\mathbb F}({\G})$.  If $r \in D^r \in {\cal D}(\G)$, then
%where $D^r = \dis_\G(r)$, then
\[
 {\mathcal D}(\phi_r(\G)) = \left( {\cal D}(\G)\setminus \{{D}^r \}\right) \cup
\;{\cal D}(\G_{{D}^r\setminus \{r\}}),
\]
where the sets on the right are disjoint. 
Thus, if $D \in {\cal D}(\phi_r({\G}))$ then $D\subseteq D^*$ for some $D^* \in {\cal D}({\G})$;
further if $D \neq D^*$, then $r\in D^*$.
%Let $\G$ be a CADMG, with $r \in {\mathbb F}({\G})$. If $r \in D^r \in {\cal D}(\G)$ then
%\[
% {\mathcal D}(\phi_r(\G)) = \left( {\cal D}(\G)\setminus \{{D}^r \}\right) \dotcup
%\;{\cal D}(\G_{{D}^r\setminus \{r\}}).
%\]
%Thus, if $D \in {\cal D}(\phi_r({\G}))$ then $D\subseteq D^*$ for some $D^* \in {\cal D}({\G})$;
%further if $D \neq D^*$, then $r\in D^*$.
\end{proa}

%In words, the set of districts in $\phi_r(\G)$, the graph obtained by
%fixing $r$, consists of the districts in $\G$ that do not contain $r$,
%together with new districts that are subsets of $D^r$, the district in
%$\G$ that contains $r$. The new districts are bidirected-connected
%subsets of $D^r$ after removing $r$.

\begin{proof}
Fixing $r$ does not affect bidirected edges that are not
adjacent to $r$, so it is clear that districts other than $D^r$
will be preserved.  The districts that replace $D^r$ are just 
those in the induced subgraph over $D^r \setminus \{r\}$, which 
gives the result.
\end{proof}

%\begin{proa}{\ref{prop:parents-after-a-single-fix}}
%Let $t \in {\mathbb F}({\G})$ and $v \in V$; then
%$\de_{\phi_{t}({\G})}(v)\subseteq \de_{\G}(v)$
%and (for $v \neq t$) $\pa_{\phi_{t}({\G})}(v)= \pa_{\G}(v)$.
%\end{proa}
\begin{proposition}
%{\ref{prop:parents-after-a-single-fix}}
Let $t \in {\mathbb F}({\G})$ and $v \in V\cup W$; then
$\de_{\phi_{t}({\G})}(v)\subseteq \de_{\G}(v)$
and (for $v \neq t$) $\pa_{\phi_{t}({\G})}(v)= \pa_{\G}(v)$.
\end{proposition}
\begin{proof}
The directed edges removed by forming $\phi_t(\G)$ are precisely the edges
into $t$ (and none are added), from which both results follow.
\end{proof}
This result implies, in particular, that the fixing operation preserves the set of parents of all vertices other than the one being fixed, meaning that the set of parents of any random vertex in any CADMG resulting from a valid fixing sequence %corresponding to a reachable set
is equal to the set of parents in the original ADMG.

\subsection{Fixing and factorization} \label{ssec:fixfactor}

\begin{proa}{\ref{prop:fix-dist-marg}}
%{\ref{prop:fix_gform}}
Take a CADMG $\G(V,W)$ with kernel
$q_V \in  {\cal P}^{\rm c}(\G)$ %that has
with the associated
district factorization:
\begin{equation}
q_V(x_V \cmid x_W) =
\!
\prod_{D \in {\cal D}(\G)}
\!
q_D(x_D \cmid x_{\pa(D) \setminus D}), \tag{\ref{eq:district-factorize}}
\end{equation}
where the kernels $q_D(x_D \cmid x_{\pa(D) \setminus D}) $ are defined via the right-hand side of 
{\rm (\ref{eq:district-factor-ident})}.
%\[
%f_D(x_D \cmid x_{\pa(D) \setminus D}) \equiv 
%\prod_{D \in {\cal D}(\G_{A \cup W})}
%\prod_{d \in D}
%q_D(x_d \cmid x_{\mb(d, A \cap \pre_{\G,\prec}(d))}).
%\]
If $r\in {\mathbb F}({\G})$ and $D^r \in {\cal D}(\G)$ is the district
containing $r$, 
then \begin{equation}\label{eq:effect-on-one-district}
\phi_r (q_V(x_V \cmid x_W); {\G}) \;=\;
q_{D^r} (x_{D^r \setminus \{r\}} \cmid x_{\pa_{\G}(D^r ) \setminus D^r })
\!\!\!
\!\!\!
\!\!\!
\prod_{D \in {\cal D}(\G)\setminus \{D^r \}}
\!\!\!
\!\!\!
\!\!\!
q_D(x_D \cmid x_{\pa_{\G}(D) \setminus D}).
\end{equation}
\end{proa}

\begin{proof}
\begin{align*}
% \MoveEqLeft[4]{\phi_r (q_V(x_V \cmid x_W); {\G})}\\[4pt]
\lefteqn{\phi_r (q_V(x_V \cmid x_W); {\G})}\\[4pt]
  &=\left. \left(
 \prod_{D \in {\cal D}(\G_{V \cup W})}
 \!\!\!
\!\!\!
 q_D(x_D \cmid x_{\pa(D) \setminus D})\right) \right/ 
 q_V(x_r \cmid x_{\mb(r)})\\[8pt]
 &= \frac{q_{D^r}(x_{D^r} \cmid x_{\pa(D^r ) \setminus D^r })}
  { q_V(x_r \cmid x_{\mb(r)})}
  \prod_{D \in {\cal D}(\G)\setminus \{D^r \}}
\!\!\!
\!\!\!
\!\!\!
\!\!
q_D(x_D \cmid x_{\pa(D) \setminus D}).
\end{align*}
%% TSR NOTE: removed reference to this proposition as I believe it was relevant
%% when we defined fixing via dividing by p(r \cmid nd(r)), but now we define this as p(r \cmid mb(r))
%by Proposition \ref{prop:fix-with-mb}.
Now consider a topological ordering $\prec$ on the vertices in $\G$ under 
which $r$ is the last vertex in $D^r$, so that 
$D^r\setminus \{r\} \subseteq \pre_{\G, \prec}(r)$; since $r \in \mathbb{F}(\G)$,  such an ordering exists. By (\ref{eq:district-factor-ident}), and taking $A=V$, we have that:
\begin{equation}
q_{D^r}(x_{D^r} \cmid x_{\pa({D^r}) \setminus {D^r}}) = \prod_{d \in {D^r}} q_{D^r}(x_d \cmid x_{T_\prec^{(d,A)}}),
\end{equation}
where $T_\prec^{(d,A)} \equiv \mb_\G\!\left( d,\; A \cap \pre_{\G,\prec}(d)\right)
\subseteq {D^r} \cup \pa ({D^r})$ by Proposition \ref{prop:mb}. 
Finally, $q_V(x_r \,|\, x_{\mb(r)}) = q_V(x_r \,|\,  x_{\nd(r)}) = q_{D^r}(x_r \,|\,  x_{T_\prec^{(r,A)}})$, by the local Markov property and (\ref{eq:district-factor-ident}). Hence these terms cancel as required. 
%*** Need to say what $T^{(d,A)}$ and $mb(...)$ are here.
\end{proof}

\subsection{Invariance to the order of fixing in an ADMG}

\begin{lema}{\ref{lem:permute}}
Let $\G(V,W)$ be a CADMG with $r,s \in {\mathbb F}(\G)$ and 
let $q_V$ be a kernel Markov w.r.t. $\G$. 
%such that both
%{\rm (a)} $\phi_{r}(q_V; \G)$ is Markov w.r.t.\ $\phi_{r}(\G)$; 
%and {\rm (b)} $\phi_{s}(q_V; \G)$ is Markov w.r.t.\ $\phi_{s}(\G)$.
Then
%\begin{eqnarray*}
%\phi_{r}\circ  \phi_{s}(\G) &=& \phi_{s}\circ  \phi_{r}(\G);\\
%\phi_{r}\circ  \phi_{s}(q_V; \G) &=& \phi_{s}\circ \phi_{r}(q_V; \G).
%\end{eqnarray*}
%\begin{eqnarray*}
\[
\phi_{\langle r, s \rangle}(\G) = \phi_{\langle s, r \rangle}(\G)\quad \hbox{ and }\quad
\phi_{\langle r, s \rangle}(q_V; \G) = \phi_{\langle s, r \rangle}(q_V; \G).
\]
%\end{eqnarray*}
\end{lema}
%{\color{red}
%[In fact, in many cases (notably if $r$ is an ancestor of $s$) 
%it is sufficient to have 
%to have $X_r \ci X_{\nd(r) \setminus \mb(r)} \mid X_{\mb(r)}$ 
%under $q_{V}$ in $\G$, and the same for $s$.]
%}

\begin{proof}
That the resulting graphs are the same is immediate since 
$\phi_r$ removes edges into $r$, while $\phi_s$ removes edges into $s$.

To show that the resulting kernels are the same we will show that if 
$r,s \in {\mathbb F}(\G)$ then the product 
of the two divisors arising in (\ref{eq:fix-kernel})  in $\phi_{r}(q_V(x_V \cmid x_W);\G)$ and 
$\phi_{s}(\phi_r(q_V(x_V \cmid x_W);\G); \; \phi_r(\G))$, are the same as
the product of the divisors  in $\phi_{s}(q_V(x_V \cmid x_W);\G)$ and 
$\phi_{r}(\phi_s(q_V(x_V \cmid x_W);\G);  \;\phi_s(\G))$.

Let $D^r \in {\mathcal D}(\G)$ be the district containing $r$ in $\G$.
The divisor when first fixing $r$ is given by:
\begin{equation}\label{eq:first-divisor}
%q_V(x_r \cmid x_{\nd_\G(r)}) = 
q_V(x_r \cmid x_{\mb_{\G}(r)}) = 
q_{D^r} (x_{r} \cmid x_{\mb_\G(r)}),
\end{equation}
where $q_{D^r}$ is given by (\ref{eq:district-factor-ident}) and
(\ref{eq:markovb}).

Further, by (\ref{eq:effect-on-one-district}), the resulting kernel is given by:
\begin{equation}\label{eq:phi-r-kernel}
\phi_r (q_V(x_V \cmid x_W); {\G})
=q_{D^r} (x_{D^r \setminus \{r\}} \cmid x_{\pa_\G(D^r ) \setminus D^r})
\!\!\!
\!\!\!
\prod_{D \in {\cal D}(\G)\setminus \{D^r \}}
\!\!\!
\!\!\!
q_D(x_D \cmid x_{\pa_\G(D) \setminus D}).
\end{equation}
Here, and in the remainder of the proof,
we use $q_{D}$, with $D \in {\cal D}(\G)$, to refer to terms in the decomposition (\ref{eq:district-factorize}) associated
with $\G$.

Set $\tilde{\G} = \phi_r(\G)$.
If $r,s \in {\mathbb F}(\G)$ then either (i) 
 $\dis_\G(r) = \dis_\G(s)$, but $r \notin \de_\G(s)$ and
$s \notin \de_\G(r)$,
or (ii) $\dis_\G(r) \neq \dis_\G(s)$.
We now consider each case in turn:
\begin{itemize}
\item[(i)]  
In this case, $s\in D^r$ since $r$ and $s$ are in the same district in $\G$.
By definition, the divisor when fixing $s$, having already fixed $r$, is given by:
\[
(\phi_r(q_V))(x_s \cmid x_{\mb_{\tilde{\G}}(s)}).
\]
Now,
$\{s\} \cup \mb_{\tilde{\G}}(s)$
is a subset of $\{s\} \cup \mb_{\G}(s) = D^r \cup \pa_\G(D^r)$, because fixing removes edges.  
%Hence by Proposition \ref{prop:mb}(i)
%\[
%\{s\} \cup \mb_{\tilde{\G}}(s) \subseteq
% (D^r \setminus \{r\}) \cup \pa_{\G}(D^r) = \{s\} \cup \mb_{\G}(s) \setminus \{r\}.
%\]

It follows from applying m-separation to $\G^{|W}$ by Definition \ref{dfn:global-definition}
(or alternatively the factorization given in Definition \ref{dfn:factorization-definition})
 that the independence
 $
 X_s \ci X_{(\mb_{\G}(s) \setminus \{r\}) \setminus \mb_{\tilde\G}(s)} \mid X_{\mb_{\tilde\G}(s)}
 $
  holds in $q_V$. Since all the
sets in this independence are subsets of $\mb_\G(r)$, it follows from Proposition \ref{prop:ind-preserve-1}
%Lemma \ref{lem:invariance-kernel} (specifically from (\ref{eq:invariance-marg})), 
that this independence also holds in 
$\phi_r(q_{V})$.  Hence
\begin{align*}
\phi_r(q_V)(x_s \cmid x_{\mb_{\tilde\G}(s)}) = \phi_r(q_V)(x_s \cmid x_{\mb_{\G}(s) \setminus \{r\}}) 
= q_V(x_s \cmid x_{\mb_{\G}(s) \setminus \{r\}}),
\end{align*}
where we have used (\ref{eq:invariance-marg}) in the second equality.
The product of the two divisors is then
\begin{align*}
q_V(x_r \cmid x_{\mb_\G(r)}) \cdot \phi_r(q_V)(x_s \cmid x_{\mb_{\tilde\G}(s)}) 
&= q_V(x_r \cmid x_{\mb_\G(r)}) \cdot q_V(x_s \cmid x_{\mb_{\G}(s) \setminus \{r\}})\\
&= q_V(x_r, x_s \cmid x_{\mb_\G(r) \setminus \{s\}}).
\end{align*}

%It follows from the Markov property for $\tilde{\G}=\phi_r(\G)$,  that
%\[
%s \;\ci \; ((D^r \setminus \{r,s\}) \cup \pa_{\G}(D^r))\setminus 
%\mb_{\tilde{\G}}(s) \mid
%\mb_{\tilde{\G}}(s)\quad [\phi_r(q_V)].
%\]
%(Here, for brevity, we have replaced $X_s$ with $s$, and similarly for the other
%variables.)
%It then follows from (\ref{eq:phi-r-kernel}) that:
%\begin{equation}\label{eq:case-one-second-divisor}
%(\phi_r(q_V))(x_s \cmid x_{\mb_{\tilde{\G}}(s)}) = 
%q_{D^r} (x_{s} \cmid x_{\mb_\G(r) \setminus \{s\}}).
%\end{equation}
%Thus the product of the two divisors  (\ref{eq:first-divisor}) and (\ref{eq:case-one-second-divisor}) is: $q_{D^r} (x_{\{r,s\}} \cmid x_{\mb_\G(r) \setminus \{s\}})$. 
Note that since, by hypothesis, $r$ and $s$ are in the same district in $\G$, this last expression is symmetric in $r$ and $s$.

\item[(ii)] Let $D^s$ be the district in ${\cal D}(\G)$ that contains $s$. Since, by assumption,
$D^s \neq D^r$, by Proposition 
\ref{prop:districts-after-fixing} it follows that $s\in D^s \in {\cal D}(\tilde{\G})$. It then follows from
(\ref{eq:phi-r-kernel}) that 
\begin{equation}
(\phi_r(q_V))(x_s \cmid x_{\mb_{\tilde{\G}}(s)})
= q_{D^s} (x_{s} \cmid x_{\mb_{\G}(s)}).
\end{equation}
Thus the product of divisors is given by 
\[
q_{D^r} (x_{r} \cmid x_{\mb_\G(r)}) \cdot q_{D^s} (x_{s} \cmid x_{\mb_\G(s)}).
\]
\end{itemize}
Hence in both cases, the product of the divisors is symmetric in $r$ and $s$, and 
a symmetric argument shows that the same divisor is obtained when fixing $s$ first, and $r$ second.
\end{proof}

\begin{thma}{\ref{invariance}}
%It follows from the previous two Lemmas that the order in which vertices are fixed is irrelevant,
%provided that at every stage the vertices are removable.
%Let ${p}(x_V)$ be a distribution that is nested Markov with respect to an ADMG $\G$.
%Let ${\bf u},{\bf w}$ be different valid fixing sequences for the same set $W\subset V$.  Then
%$
%\phi_{{\bf u}}(\G) = \phi_{{\bf w}}(\G)
%$ and
%\begin{equation}
%\phi_{{\bf u}}({p}(x_V);\G) = \phi_{{\bf w}}({p}(x_V);\G).  \tag{\ref{eq:invar-thm}}
%\end{equation}
Let ${p}(x_V)$ be a distribution that is nested Markov with respect to an ADMG $\G$ {\rm (}in either the sense of Definitions \ref{dfn:global-definition} or \ref{dfn:factorization-definition}{\rm )}.
Let ${\bf u},{\bf w}$ be different valid fixing sequences for the same set $W\subset V$.  Then
$
\phi_{{\bf u}}(\G) = \phi_{{\bf w}}(\G)
$ and
\begin{equation}\label{eq:invar-thm}
\phi_{{\bf u}}({p}(x_V);\G) = \phi_{{\bf w}}({p}(x_V);\G).
\end{equation}
\end{thma}

\begin{proof}
We perform a direct proof. 
%{\color{red}The proof is by induction on the cardinality of $W$. \emph{[Rob: Seems to me more like an induction on the sequence, where the base case is that they're the same.]}} The base case is trivial. Suppose that the result holds for sets $|W^*|< |W|$. 
Let $u_i, w_i$ denote the $i$th vertices in sequences ${\bf u}, {\bf w}$ respectively. Further, let $k$ be the smallest $i$ such that $u_i \neq  w_i$, and let $v \equiv u_k$, so that
${\bf u}$ and ${\bf w}$ agree in the first $k-1$ fixing operations.
By definition of $k$, 
\[
\phi_{\langle u_1,\ldots , u_{k-1}\rangle}(\G)
=
\phi_{\langle w_1,\ldots , w_{k-1}\rangle}(\G).
\]
 %, and $v_2 = w_k$.
Since ${\bf u},{\bf w}$ both contain the same vertices, there exists $l > k$ such that $w_l = v$. %Call this graph $\G^*$; if $i=1$ then let $\G = \G^*$. 
Since, by hypothesis, ${\bf u},{\bf w}$ are both valid fixing sequences, it follows that
$v \in \mathbb{F}(\phi_{\langle w_1,\ldots , w_{k-1}\rangle}(\G))$.
It further follows by Lemma \ref{lem:no-backtracking} that 
\[
v \in \mathbb{F}(\phi_{\langle w_1,\ldots , w_{i-1}\rangle}(\G)), \quad
\hbox{ for }k-1 \leq i \leq l.
\] 
Then by Lemma \ref{lem:permute}, we have that:
\begin{eqnarray*}
\phi_{\langle w_1,\ldots , w_{l-1}, v=w_{l}\rangle}(\G) &=& 
\phi_{\langle w_1,\ldots ,  w_{l-2}, v, w_{l-1}\rangle}(\G)\\
\phi_{\langle w_1,\ldots , w_{l-1}, v=w_{l}\rangle}(p(x_V) ; \G) &=& 
\phi_{\langle w_1,\ldots ,  w_{l-2},v, w_{l-1}\rangle}(p(x_V) ;\G).
\end{eqnarray*}
By further applications of Lemma \ref{lem:permute}, we may show that 
both the graphs and kernels resulting from the fixing sequences 
\[
\langle w_1,\ldots , w_{l-1}, v\!=\!w_{l}\rangle\;\; \hbox{ and }\;\;
\langle w_1,\ldots , w_{k-1}, v, w_{k}, \ldots , w_{l-1}\rangle.
\]
are the same.  It further follows that the whole sequence ${\bf w}$ leads to the same graph and kernel as
$\langle w_1,\ldots , w_{k-1}, v, w_{k},\ldots, w_{l-1}, w_{l+1}, \ldots, w_{|W|}\rangle$.
This latter sequence now agrees with ${\bf u}$ in the first $k$ fixing operations. By repeating the argument
we may thus show that ${\bf u}$ and ${\bf w}$ lead to the same graph and kernel.
\end{proof}

As discussed in the main body of the paper, the above result implies that, under the model, if there are multiple fixing sequences for $V \setminus R$ valid in $\G(V)$,
we can define $\phi_{V \setminus R}(p(V); \G(V))$ without loss of generality to be the result of fixing elements in $V \setminus R$ in any valid sequence.
Furthermore if $(S_1,S_2)$ is a partition of $V \setminus R$, and there exists a fixing sequence for $V \setminus R$ valid in $\G(V)$ that fixes all elements
in $S_1$ before any element in $S_2$, we define
$\phi_{\langle S_1, S_2 \rangle}(\G(V))$ to be equal to $\phi_{S_2}(\phi_{S_1}(\G(V)))$, and
$\phi_{\langle S_1, S_2 \rangle}(p(V); \G(V))$ to be equal to $\phi_{S_2}(\phi_{S_1}(p(V); \G(V)), \phi_{S_1}(\G(V)))$.

\section{Equivalence Of CADMG Markov Properties and Factorization} \label{ssec:cadmg_mps}

\begin{thma}{\ref{cadmg-equiv}}
${\cal P}^{\rm c}_{\rm f}(\G)
\!
=
\!
{\cal P}^{\rm c}_{\rm l}(\G,\prec)
\!
=
\!
{\cal P}^{\rm c}_{\rm m}(\G)
\!
=
\!
{\cal P}^{\rm c}_{\rm a}(\G)$.
\end{thma}

To prove this result we will need a number of intermediate results.
The argument for the last two equalities follows that given by \citet{richardson03sjs}.

\begin{lemma}\label{ncsareadjacent} 
For a CADMG $\G(V,W)$, suppose
${\mu}$ is a path which 
m-connects $t$ and $y$ given $Z$ in $\G^{|W}$.  Then the sequence of
non-colliders on ${\mu}$ form a path connecting $t$ and $y$ in $(\G_{\an(\{t,y\} \cup Z)})^{a}$.
\end{lemma}

This proof follows that of Lemma 3 in \cite{richardson03sjs}.
\begin{proof}  
Every vertex on an
m-connecting path is either an ancestor of a collider, and hence
of some element of $Z$, or an ancestor of an endpoint. Thus all the 
vertices on ${{\mu}}$
are in
$\G_{\an (\{t,y\} \cup Z )}$.
Suppose that $v_i$ and $v_{i+1}$ $(1\!\leq\! i\!\leq\! k\!-\!1)$ are the
successive non-colliders on ${\mu}$. The subpath
${\mu}(v_i,v_{i+1})$ consists entirely of colliders, hence
$v_i$ and
$v_{i+1}$ are adjacent in $(\G_{\an (\{t,y\} \cup
Z )})^{a}$. Similarly $v_1$ and $v_k$ are adjacent to
$t$ and $y$, respectively, in $(\G_{\an (\{t,y\} \cup
Z )})^{a}$. 
\end{proof}

\begin{theorem}\label{thm:aug-msep}
If $\mathscr G$ is a CADMG then
\[
{\cal P}^{\rm c}_{\rm m}(\G) ={\cal P}^{\rm c}_{\rm a}(\G).
\]
\end{theorem}
This proof follows that of Theorem 1 in \cite{richardson03sjs}.

\begin{proof}
(${\cal P}^{\rm c}_{\rm m}(\G) \subseteq{\cal P}^{\rm c}_{\rm a}(\G)$)\\
We proceed by showing that if ${T}$ and ${Y}$ are m-connected given 
${Z}$ in $\G^{|W}$ then $T$ and $Y$ are not separated by $Z$ in $(\G_{\an(T\cup Y\cup Z)})^{a}$.
If $T$ and $Y$ are m-connected given $Z$ in $\G^{|W}$ then 
there are vertices $t\in T$,
$y\in Y$ such that there is a path ${\mu}$ which m-connects
$t$ and
$y$ given $Z$ in $\G^{|W}$.  By Lemma \ref{ncsareadjacent}
the non-colliders on
${\mu}$ form a  path ${\mu}^{*}$ connecting $t$ and $y$ in
$(\G_{\an (T\cup Y \cup  Z)})^{a}$.  Since ${\mu}$ is m-connecting,
no non-collider is in $Z$  hence no vertex on ${\mu}^{*}$ is in $Z$.  Thus
$T$ and $Y$ are not separated by $Z$ in $(\G_{\an (T\cup Y \cup  Z)})^{a}$.\\[-12pt]

\noindent(${\cal P}^{\rm c}_{\rm a}(\G) \subseteq{\cal P}^{\rm c}_{\rm m}(\G)$)\\
We show that if $T$ and $Y$ are not separated by $Z$ in 
$(\G_{\an (T\cup Y \cup  Z)})^{a}$ then
$T$ and $Y$ are m-connected given $Z$ in $\G^{|W}$.
If $T$ and $Y$ are not separated by $Z$ in $(\G_{\an(T\cup Y\cup 
Z)})^{a}$ 
then there are vertices $t\in T$,
$y\in Y$ such that there is a minimal path ${\pi}$
between
$t$ and $y$ in $(\G_{\an(T\cup Y\cup Z)})^{a}$ on which no vertex 
 is in $Z$. Our strategy is to
replace each augmented edge on ${\pi}$ with a corresponding
collider path in $\G^{|W}$ and replace the other edges on
${\pi}$ with the corresponding edges in $\mathscr G$
 (choosing arbitrarily if there is more than one). It follows from
Lemma 2 in \cite{richardson03sjs} %\ref{intersectonminimal}
 that the resulting sequence of edges forms a path
from $t$ to $y$ in $\G^{|W}$, which we denote ${\nu}$. Further, any
non-collider on
${\nu}$ is a vertex on
${\pi}$ and hence not in $Z$. Finally, since all vertices in ${\nu}$
are in
$\G_{\an(T\cup Y\cup Z)}$ it follows that every collider is in
$\an_\G(T\cup Y\cup Z)$. Thus by
Lemma 1 in \cite{richardson03sjs} %\ref{antofxyorz}
there exist vertices $t^{*}\in T$ and $y^{*}\in Y$ which are 
m-connected given $Z$ in $\G^{|W}$,  hence $T$ and $Y$ are 
m-connected given $Z$.
\end{proof}

\begin{lemma} \label{lem:add-w}
If $\G(V,W)$ is a CADMG and $T$ and $Y$ are m-separated given $Z$ in 
$\G^{|W}$ then: 
\begin{itemize}
\item[{\rm (i)}] Either $T \cap W = \emptyset$ or $Y \cap W = \emptyset$ (or both);
\item[{\rm (ii)}]  Either $T$ is m-separated from $Y \cup (W\setminus Z)$ given $Z$ in $\G^{|W}$,
or  $Y$ is m-separated from $T \cup (W\setminus Z)$ given $Z$ in $\G^{|W}$ (or both).
\end{itemize}
\end{lemma}

\begin{proof} (i) Suppose that there are vertices $w, w^* \in W$ such that $w \in T$, $w^*\in Y$.
Since $w \leftrightarrow w^*$ in $\G^{|W}$, it follows that $T$ and $Y$ are m-connected given $Z$ 
in $\G^{|W}$, which is a contradiction.\par
(ii) On similar lines, suppose that there exist vertices $w, w^* \in W\setminus Z$ such that in $\G^{|W}$ some vertex $t\in T$ is m-connected to $w$ given $Z$ by a path ${\nu}_{tw}$, but, in addition, some $y \in Y$ is m-connected to $w^*$ given $Z$ by a path ${\nu}_{w^*y}$. In this case a path m-connecting $t$ and $y$ given $Z$ may be constructed by concatenating 
${\nu}_{tw}$, the edge $w\leftrightarrow w^*$ and ${\nu}_{w^*y}$.
\end{proof}

\begin{lemma}
If $\G(V,W)$ is a CADMG, $t\in V$ is a vertex in an ancestral set
$A \subseteq V\cup W$, and $\ch_{\G}(t) \cap A = \emptyset$,
then the induced subgraph of the
augmented graph $(\G_A)^a$ on the set $\{ t \} \cup \mb_\G(t, A)$ is always
a clique.  In addition, if $y \uned t$ in $(\G_A)^a$ then
$y \in \mb_\G(t, A)$.
\label{lem:clique}
\end{lemma}

\begin{proof}  
(Cf.~proof of Theorem 4 in \cite{richardson03sjs})
 By construction, $y\uned t$ in $(\G_{A})^a$ 
if and only if $t$ is collider-connected to $y$ in $(\G_{A})^{|W\cap A}$.
Since $t \in V$ and $\ch_\G(t)\cap A=\emptyset$, the vertex adjacent to $t$ on 
any collider path is in $\sib_{\G_{A}}(t)\cup \pa_{\G_{A}}(t)$. 
Consequently, a collider path to $t$ in $(\G_{A})^{|W\cap A}$ takes one of three
forms:
\begin{itemize}
\item[(a)]\(y\rightarrow t\quad\Leftrightarrow\quad 
y\in\pa_{\G_{A}}(t)\); 

\item[(b)] \(y\leftrightarrow u\leftrightarrow \cdots 
\leftrightarrow t\quad\Leftrightarrow\quad y\in\dis_{{\G}_{A}}(t)\setminus\{t\}\);

\item[(c)] \(y\rightarrow u\leftrightarrow \cdots 
\leftrightarrow t\quad\Leftrightarrow\quad y\in\pa_{\G_{A}}
(\dis_{\G_{A}}(t)\setminus\{t\})\).

\end{itemize}
It then follows from the definition of a Markov blanket and $\G^{|W}$ that if $t \in V$ then 
$y$ is collider-connected to $t$ in $(\G_{A})^{|W\cap A}$ if and only if
$y\in \mb_{\G}(t,A)$. 
%This establishes that if $y\uned t$ in $(\G_{A})^{a}$ then $y\in 
%\mb_{\G}(t,A)$. 

Suppose that
$u,v\in \mb_{\G}(t,A)$, with $u\!\neq\! v$. Then there are collider paths
${\nu}_{ut},
{\nu}_{vt}$ in $(\G_{A})^{|W\cap A}$
%;  
%note
that
%since $t\in V$, ${\nu}_{ut}, {\nu}_{vt}$
do not contain any bidirected edge between vertices in $W$.
Traversing the path ${\nu}_{ut}$ from $u$
 to $t$, let $s$ be the first vertex which is also on 
 ${\nu}_{vt}$; such a vertex is guaranteed to exist since $t$ is 
 common to both paths. Concatenating the subpaths ${\nu}_{ut}(u,s)$ and 
 ${\nu}_{vt}(v,s)$ forms a collider path connecting $u$ and $s$ 
 in $(\G_{A})^{|W\cap A}$.
 (If $s\!=\!t$, this follows from ${\ch_{\G}(t)\cap A}=\emptyset$.) Hence
  $u\uned v$ in  $(\G_{A})^{a}$, proving 
the first claim.
%Finally, if $A\subseteq W$ then $\G_A$ is a complete graph containing bidirected
%edges and so it is easy to see that the claim holds for $t \in A \subseteq W$.
\end{proof}

\begin{theorem}
If $\G(V,W)$ is a CADMG and $\prec$ is a topological ordering then
\[
{\cal P}^{\rm c}_{\rm l}(\G, \prec) = {\cal P}^{\rm c}_{\rm f}(\G).
\]
\end{theorem}

\begin{proof}
First we show that ${\cal P}^{\rm c}_{\rm l}(\G, \prec)  \subseteq  {\cal P}^{\rm c}_{\rm f}(\G)$.
Fix an ancestral set $A \subseteq V \cup W$.  We have:
\begin{eqnarray*}
q_V(x_{A \cap V} \cmid x_W) &=&
%\!\!\!\!\!
\!\!\!\!\!
\!\!\!\!\!
	\prod_{d \in D \in {\cal D}(\G_{A})}
%\!\!\!\!\!
\!\!\!\!\!
	q_V(x_d \cmid x_{A \cap \pre_{\G,\prec}(d)%\setminus \{d\}
	},x_W)\\
&=&
%\!\!\!\!\!
\!\!\!\!\!
\!\!\!\!\!
\prod_{d \in D \in {\cal D}(\G_{A})}
%\!\!\!\!\!
\!\!\!\!\!
	q_V(x_d \cmid x_{\mb\left({d},A \cap \pre_{\G,\prec}(d)\right)}).
%&=& \prod_{D \in {\cal D}(\G)} \prod_{d \in D}
%	p^(x_d \cmid x_{T^{(d,A)}})\\
\end{eqnarray*}
The first line is by the chain rule of probabilities and the fact that
${\cal D}(\G_{A}) = {\cal D}(\G_{A\cup W}) $ is a partition of random nodes in
$\G_{A}$, the second by the ordered local Markov property.
% and the third by Proposition \ref{mb}.
This is sufficient for the conclusion.

Now we show that
${\cal P}^{\rm c}_{\rm f}(\G)\subseteq {\cal P}^{\rm c}_{\rm l}(\G, \prec)$.
Let $V=\{v_{1},\ldots, v_{n}\}$ be a numbering of the vertices in $V$ such 
that $v_{i}\prec v_{j}$ if and only if $i<j$, so $\pre_{\G,\prec}(v_{k})\cap V=\{v_{1},\ldots ,v_{k-1}\}$; the proof is by induction on this sequence of vertices.

For $k = 1$, if $A$ is an ancestral set in which $v_1$ is the maximal vertex in $A \cap V$ then 
$A\cup W = \{v_1\} \cup W$. The assumed Markov factorization for the graph $\G(\{v_1\},W)$ then implies
$q_V(x_{v_1} \mid x_W) = q_V(x_{v_1} \mid x_{\pa(v_1)}) = q_V(x_{v_1} \mid x_{\mb(v_1,A)})$, as required by the ordered local Markov property.

Now assume the inductive hypothesis
holds for all $j < k$, so that for all ancestral sets $A^*$ in which $v_j$ is maximal, $q_V$ obeys the ordered local Markov property for $\G(V\cap A^*, W)$ at $v_j$.
Now fix an ancestral set
$A$ with $\max_{\prec}(A \cap V) = v_k$.  
By the chain rule of probabilities and the fact that
${\cal D}(\G_A)$ is a partition of the random nodes in $\G_A$, we have:
\[
q_V(x_{A \cap V} \cmid x_W) =
\!\!\!\!\!
%\!\!\!%\!\!
%\!\!\!\!\!
	\prod_{d \in D \in {\cal D}(\G_{A})}
\!\!\!\!\!
%\!\!\!\!\!
	q_V(x_d \cmid x_{A \cap \pre_{\G,\prec}(d)},x_W).
\]
However, by the assumed Markov factorization we also have:
\[
q_V(x_{A \cap V} \cmid x_W) 
=
\!\!\!\!\!
	\prod_{d \in D \in {\cal D}(\G_{A})}
\!\!\!\!\!
%\!\!\!\!\!
	q_V(x_d \cmid x_{\mb\left( d,\; A \cap \pre_{\G,\prec}(d)\right)}).
\]
%	p(X_d \cmid X_{\mb(d,A \cap \pre_{\G,\prec}(d))})\\
%\prod_{D \in {\cal D}(\G_{A \cup W})} \prod_{d \in D}
%	p(X_d \cmid X_{(A \cap \pre_{\G,\prec}(d))})\\
%
%The first line is, again,  , and the second
%is by the Markov factorization.
%
%, and the third by Proposition \ref{mb}.  Thus
%
Thus:
\begin{align*}
\prod_{D \in {\cal D}(\G_{A})} \prod_{d \in D}
	q_V(x_d \cmid x_{A \cap \pre_{\G,\prec}(d)},x_W) =
\!\!\!\!\!
%\!\!\!\!\!
\prod_{D \in {\cal D}(\G_{A})} \prod_{d \in D}
	q_V(x_d \cmid x_{\mb\left( d,\; A \cap \pre_{\G,\prec}(d)\right)}).
%	q_V(x_d \cmid x_{\mb(d,A \cap \pre_{\G,\prec}(d))})\\
\end{align*}
As the inductive hypothesis holds for all vertices $d \prec v_k$, 
we can cancel all terms from the above equality, except
for those involving ${v_k}$:
$q_V(x_{v_k} \cmid x_{A \cap \pre_{\G,\prec}(v_k)},x_W)$ 
 and 
% \[
$q_V(x_{v_k} \cmid x_{\mb\left( v_k,\; A \cap \pre_{\G,\prec}(v_k)\right)})$.
%\]
Since $v_k$ is the maximal vertex in $A\cap V$, it follows that  
$\mb\left( v_k, A \cap \pre_{\G,\prec}(v_k)\right) = \mb(v_k,A)$,
which establishes the ordered local property for $A$ in $\G_{A}$ at $v_k$.
\end{proof}

\begin{theorem}
If $\G(V,W)$ is a CADMG and $\prec$ is a topological ordering then
\[
%{\m(\G) = {\cal P}^{\rm c}_{\rm a}(\G) = {\cal P}^{\rm c}_{\rm l}(\G, \prec) =
%{\f(\G)
{\cal P}^{\rm c}_{\rm a}(\G) = {\cal P}^{\rm c}_{\rm l}(\G, \prec).
\]
\end{theorem}

\begin{proof}
%We have already established the first equality in Theorem \ref{mmstar}.
We first show that
${\cal P}^{\rm c}_{\rm l}(\G,\prec) \subseteq{\cal P}^{\rm c}_{\rm a}(\G)$.
The proof is similar to 
Proposition 5 in \cite{lauritzen1990independence}, and Theorem 2 in
\cite{richardson03sjs}.

Number the vertices in $V=\{v_{1},\ldots, v_{n}\}$ such 
that $v_{i}\prec v_{j}$ if and only if $i<j$, so $\pre_{\G,\prec}(v_{k+1})=\{v_{1},\ldots ,v_{k}\}$.
Let $q_V\in {\cal P}^{\rm c}_{\rm l}(\G,\prec)$. 
The proof is by induction on the sequence of ordered vertices in $V$. The inductive 
hypothesis is that if $T\dot\cup Y\dot\cup Z\subseteq \{v_{1},\ldots 
,v_{k}\}\cup W$ and $T$ is separated from $Y$ by $Z$ in $(\G_{\an 
(T\cup Y\cup Z)})^{a}$ then $X_T\ind X_Y \mid X_Z$ in $q_V$.

By Theorem \ref{thm:aug-msep} and Lemma \ref{lem:add-w} either $T\cap W = \emptyset$ and $T$ is m-separated from $W\setminus Z$ given $Z$ or $Y\cap W = \emptyset$ and $Y$ is m-separated from $W\setminus Z$ given $Z$ in $\G^{|W}$.
Without loss of generality we suppose the former; the other case is symmetric. 
It follows that we may extend $Y$ such that $W \subseteq Y\cup Z$, so that $T$ is m-separated from $Y$ by $Z$ in $\G^{|W}$, and thus the corresponding separation holds in $(\G_{\an 
(T\cup Y\cup Z)})^{a}$ by Theorem \ref{thm:aug-msep}.

(Base case) For $k=1$, it follows that we have
%from 
%Lemma \ref{lem:add-w}(i)  and Theorem \ref{thm:aug-msep} that
%either $T=\{v_1\}$ and $Y\cup Z \subseteq W$, or
%$Y=\{v_1\}$ and $T\cup Z \subseteq W$.
%Without much loss of generality suppose the former; the other case is symmetric.
%It follows from Lemma \ref{lem:add-w}(ii) 
 $T=\{v_1\}$ and $Y=W\setminus Z$.
Let $A\equiv T \cup Y \cup Z=\{v_1\}\cup W$. Note that $A=\an_\G(A)$ and $\mb_\G(v_1,A) = \pa_\G(v_1) \subseteq Z\cap W$.
By the ordered local Markov property for $A$ we have:
\begin{equation}\label{eq:local-proof}
 X_{v_1}\, \ci \,X_{W\setminus (\pa(v_1)\cup \{v_1\})} \cmid X_{\pa(v_1)}\;\; [q_V].
\end{equation}
Since $Y \subseteq W\setminus Z$ and $ \pa_\G(v_1) \subseteq Z$, (\ref{eq:local-proof}) implies $X_{v_1} \ci X_{Y} \mid X_{Z}$  $[q_V]$.

(Inductive case) Suppose that the induction 
hypothesis holds 
for $j<k$.  Let $\Hg \equiv (\G_{\an(T\cup Y\cup Z)})^{a}$.  If $T$ is separated from $Y$ by $Z$ in 
$\Hg$ and $v\in \an_\G(T\cup Y\cup Z)\setminus (T\cup Y\cup Z)$ then 
in $\Hg$
either $v$ is separated from $Y$ by $Z$, or $v$ is separated from $T$ by $Z$ (or both).  
Hence we may always extend $T$ and $Y$, so that $\an_\G(T\cup Y\cup Z) = 
T\cup Y\cup Z$, and thus need only consider this case.

If $v_{k} \notin 
(T\cup Y\cup Z)$ then $T\cup Y\cup Z \subseteq \{ v_{j} \} \cup \pre_{\G,\prec}(v_{j} ) \cup W$ for some $v_{j} \prec v_{k}$  hence the required 
independence follows directly from the induction hypothesis.  Thus we 
suppose that $v_{k}\in (T\cup Y\cup Z) \subseteq \{ v_{k} \} \cup \pre_{\G,\prec}(v_{k}) \cup W$.  As before, let $A\equiv T\cup Y \cup Z$.  Since $A$ is ancestral 
and $v_{k}$ has no children in $A$, the local Markov 
property for $A=A\cup W$ implies that
\begin{equation}\label{eq:local-markov-proof-b}
X_{v_{k}} \, \ind\, X_{(A\cup W)\setminus(\{v_{k}\}\cup \mb(v_{k},A))}\;\, \mid \;\,
X_{\mb(v_{k},A)}\;\;[q_V].
\end{equation}
There are now three cases to consider: (i) $v_{k}\in T$; (ii) $v_{k}\in Y$; (iii) 
$v_{k}\in Z$.
\begin{itemize}
\item[(i)] Note that $(\G_{\an(Y\cup Z\cup (T\setminus\{v_{k}\}))})^{a}$
contains a  subset of the edges in $\Hg$.
Thus (if non-empty) $T\setminus\{ v_{k}\}$ is separated
from 
$Y$ by $Z$ in $(\G_{\an(Y\cup Z\cup (T\setminus\{v_{k}\}))})^{a}$, hence 
%If $T\neq \{v_{k}\}$ then 
$X_{T\setminus\{ v_{k}\}}\ind X_Y\mid
X_Z$ in $q_V$ by the induction hypothesis. It is thus sufficient to prove that
$X_{v_{k}}\ind X_{Y} \mid X_{Z\cup (T\setminus\{ v_{k}\})}$ in $q_V$; this also covers 
the case where $T=\{v_{k}\}$.
Since, by Lemma \ref{lem:clique}, the vertices in $\{ v_{k}\}\cup \mb_\G(v_{k},A)$ form a clique
in $\Hg$ it follows that $\{ v_{k}\}\cup \mb_\G(v_{k},A) \subseteq
T\cup Z$, so $Y\subseteq A\setminus(\{v_{k}\}\cup \mb_\G(v_{k},A))$.
Thus by the local Markov property for $A=A\cup W$,
\[
X_{v_{k}}\,\ind\, X_{A\setminus(\{v_{k}\}\cup 
\mb(v_{k},A))} \,|\,
 X_{\mb(v_{k},A)}\,[q_V] \;\Rightarrow \; X_{v_{k}} \ind X_Y 
\,|\, X_{Z\cup (T\setminus \{v_{k} \})}\,[q_V].
\]

%{\small
%\[\begin{array}{r}
%\multicolumn{1}{l}{X_{v_{k}}\,\ind\, X_{A\setminus(\{v_{k}\}\cup 
%\mb(v_{k},A))} \;\mid\;
% X_{\mb(v_{k},A)}\;[q_V] \quad\quad\quad\quad}\\[8pt]
%\Rightarrow \;\; X_{v_{k}} \ind X_Y 
%\;\mid\; X_{Z\cup (T\setminus \{v_{k} \})}\; [q_V]
%\end{array}
%\]
%}
\item[(ii)] Similar to case (i).

\item[(iii)] By hypothesis, $A=\an_\G(A)\subseteq \{v_k\}\cup \pre_{\G,
\prec}(v_{k})\cup W$, so $v_{k}\notin \an_\G(T\cup Y\cup (Z\setminus\{v_{k}\}))$.  Thus
the  vertex $v_{k}$ is not in $(\G_{\an(T\cup Y\cup (Z\setminus\{v_{k}\}))})^{a}$, and this graph contains a
subset  of the edges in $\Hg$. Hence $T$ is separated from $Y$ given 
$Z\setminus\{ v_{k}\}$ in $(\G_{\an(T\cup Y\cup (Z\setminus\{v_{k}\}))})^{a}$.  The induction then implies 
$X_T\ind X_Y\mid X_{Z\setminus\{ v_{k}\}}$.  It is then sufficient to prove
that either $X_{v_{k}}
\ind X_{Y} \mid X_{T\cup (Z\setminus\{ v_{k}\})}$ or $X_{v_{k}}\ind X_T\mid X_{Y\cup 
(Z\setminus\{ v_{k}\})}$ in $q_V$.  Since by Lemma \ref{lem:clique}, $\{ v_{k}\}\cup
\mb_\G(v_{k},A)$ forms a clique in 
$(\G_{A})^a$ it follows that either
\[
\{ v_{k}\}\cup \mb_\G(v_{k},A)\;\subseteq\; T\cup Z \quad\hbox{or}\quad \{ v_{k}\}\cup 
\mb_\G(v_{k},A)\;\subseteq\; Y\cup Z. 
\]
%\begin{align*}
%\{ v_{k}\}\cup \mb(v_{k},A)\;\subseteq\; T\cup Z \quad\hbox{or}\quad\\ \{ v_{k}\}\cup 
%\mb(v_{k},A)\;\subseteq\; Y\cup Z. 
%\end{align*}
Suppose the former. In this case, by the ordered local Markov property, 
{%\small
\[\begin{array}{r}
\multicolumn{1}{l}{X_{v_{k}}\;\ind\; X_{A\setminus(\{ v_{k}\}\cup \mb(v_{k},A))}\; \mid\; 
X_{\mb(v_{k},A)}\;[q_V] \quad\quad\quad\quad}\\[8pt]
 \Rightarrow\;\;  X_{v_{k}}\;\ind\; X_Y\;\mid\; X_{T\cup (Z\setminus \{ 
 v_{k}\})}\;[q_V].
\end{array}
\]
}
If the latter then $X_{v_{k}}\ind X_{T} \mid X_{Y\cup (Z\setminus \{ 
v_{k}\})}\;[q_V]$.
\end{itemize}

Now we show that ${\cal P}^{\rm c}_{\rm a}(\G) \subseteq{\cal P}^{\rm c}_{ 
\rm l}(\G,\prec)$.
Let $A$ be an ancestral set with $t = \max_{\prec}(A) \in V$. Note that $A\cup W$ is ancestral.
By Lemma \ref{lem:clique}, every vertex adjacent to $t$ in $(\G_{A\cup W})^{a}$ is in $\mb_\G(t,A\cup W)$. Thus $t$ is separated from 
$(A\cup W)\setminus (\mb_\G(t,A)\cup\{t\})$ by $\mb_\G(t,A)$ in $(\G_{A\cup W})^{a}$. Hence if $q_V\in {\cal P}^{\rm c}_{\rm a}(\G)$ then $q_V\in {\cal P}^{\rm c}_{ 
{\rm l}}(\G,\prec)$.
\end{proof}

\section{Results On Nested Markov Models}

We begin with a proof of the equivalence of the global nested Markov
property and district factorization.

\begin{proa}{\ref{prop:nested_global_factorization}}
With respect to an ADMG $\G(V)$, a 
distribution $p(x_V)$ is globally nested Markov %with respect to an ADMG $\G$ 
if and only if it nested Markov factorizes.% with respect to $\G$.
\end{proa}

\begin{proof}
This follows immediately by equivalence of the CADMG global property and CADMG factorization.
\end{proof}

\begin{thma}{\ref{thm:global-reachable-factorization}}
If $p(x_V)$ nested Markov factorizes with respect to $\G$ then
for every reachable $R$ in $\G$,
\[
\phi_{V \setminus R}(p(x_V); \G) = \prod_{D \in {\cal D}(\phi_{V \setminus R}(\G))} \phi_{V \setminus D}(p(x_V); \G).
\]
%$p(x_V)$ obeys the global nested Markov property with respect to $\G$ if and only if for every reachable $R$ in $\G$,
%\[
%\phi_{V \setminus R}(p(x_V); \G) = \prod_{D \in {\cal D}(\phi_{V \setminus R}(\G))} \phi_{V \setminus D}(p(x_V); \G).
%\]
\end{thma}

\begin{proof}
By Theorem \ref{cadmg-equiv} we know that, for each reachable set $R$ (and 
valid fixing sequence ${\bf w}$ for $R$), $\phi_{\bf w}(p)$ will be a member of  
${\cal P}^{\rm c}_{\rm m}(\G[R])$ if and only if it is a member of 
${\cal P}^{\rm c}_{\rm f}(\G[R])$. 
If $p$ nested Markov factorizes w.r.t.~$\G$ then, by Definition \ref{dfn:factorization-definition} and Corollary \ref{cor:nested-factorize}, 
%we know that % $\phi_{V \setminus R}$ is well defined, and then 
%the equality given is satisfied.
%
%If $p(x_V)$ obeys the global nested Markov property, then by the equivalence of the CADMG factorization and CADMG global Markov properties,
for every $R$ reachable in $\G$, $q_R(x_R \cmid x_{V \setminus R}) \equiv \phi_{V \setminus R}(p(x_V);\G) \in {\cal P}^{\rm c}_{\rm f}(\G[R])$, so
\[
%q_R(x_R \cmid x_{V \setminus R}) \equiv 
\phi_{V \setminus R}(p(x_V);\G) = \!\!\!\! \prod_{D \in {\cal D}(\phi_{V \setminus R}(\G))} \!\!\!\! f^{\bf w}_D(x_D \cmid x_{\pa(D) \setminus D}).
\]
By Lemma \ref{lem:markovfact}, $f^{\bf w}_D = q_D$, where $q_D = \phi_{R \setminus D}(q_R; \G[R]) =
\phi_{R \setminus D}(\phi_{V \setminus R}(p; \G); \phi_{V \setminus R}(\G))$.
%%% TSR 21/9/22 looks good to me!  
However, by Theorem \ref{thm:invariant} we know that this is equivalent to $\phi_{V \setminus D}(p; \G)$.
%But each factor $q_D$ is equal to $\phi_{R \setminus D}(q_R; \phi_{V \setminus R}(\G))$ by
%an inductive application of Proposition \ref{prop:fix-dist-marg}.  By invariance of fixing and definition of $q_R$, 
%We then have
%\[
%\phi_{V \setminus R}(p(x_V);\G) = \prod_{D \in {\cal D}(\phi_{V \setminus R}(\G))}
%\phi_{V \setminus D}(p(x_V); \G).
%\]
This gives the result.
\end{proof}

%\begin{thma}{\ref{thm:global_local}}
%$p(x_V)$ obeys the global nested Markov property with respect to $\G$ if and only if $p(x_V)$ obeys the ordered local
%nested Markov property for $\G$.
%\end{thma}
%
%\begin{proof}
%By earlier results on CADMGs,
%$p(x_V)$ obeys the global nested Markov property if and only if for every reachable set $R$, $\phi_{V \setminus R}(p(x_V); \G)$
%obeys the local Markov property for $\phi_{V \setminus R}(\G)$ at $\max_{\prec}(R)$, the largest element of $R$
%according to $\prec$.   Since every element of ${\cal I}(\G)$ is reachable,
%$p(x_V)$ obeys the ordered local nested Markov property for $\G(V)$.
%
%To see the converse, fix $R$ reachable in $\G$, with vertex $v$ maximal in $R$ according to $\prec$.
%Let $D^v$ be the element in ${\cal D}(\phi_{V\setminus R}(\G))$ containing $v$.  Then $D^v \in {\cal I}(\G)$, and
%therefore, $v \ci V \setminus \mb_{\G}(v,D^v) \mid \mb_{\G}(v,D^v)$ in $\phi_{V \setminus D^v}(p(x_V);\G)$
%is part of the ordered local Markov property for $\G$.
%
%Then $v \ci V \setminus \mb_{\G}(v,D^v) \mid \mb_{\G}(v,D^v)$ holds in $\phi_{V \setminus R}(p(x_V);\G)$ by
%Proposition \ref{prop:ind-preserve-2}.
%Since $\mb_{\G}(v,D^v) = \mb_{\G}(v,R)$, then
%$v \ci V \setminus \mb_{\G}(v,R) \mid \mb_{\G}(v,R)$ holds in $\phi_{V \setminus R}(p(x_V);\G)$.
%\end{proof}

\subsection{Ordered Local Nested Markov Property}

%In order to show that the ordered local nested Markov property implies the global 
%nested Markov property, we will first need to show that it implies the `ordinary' 
%Markov property.  
%The `ordinary' ordered local Markov property can be summarised as requiring that 
%for every ancestral set $A$ with $\prec$-maximal element $v\in V$, 
%\begin{align}
%X_v \ci X_{A \setminus (\{v\} \cup \mb(v; A))} \mid X_{\mb(v; A)} \; [p], \label{eqn:oolmp}
%\end{align}
%where $\mb_\G(v; A) = (\dis_A(v) \setminus \{v\}) \cup \pa_\G(\dis_A(v))$.  Letting
%$C = \dis_A(v)$ and noting that $C$ is an intrinsic set, this can be expressed as 
%\begin{align*}
%X_v \ci X_{A \setminus (C \cup \pa(C))} \mid X_{C \setminus \{v\} \cup \pa(C)} \; [p].
%\end{align*}
%An intrinsic set that can be written as $\dis_A(v)$ for ancestral $A$ is called an 
%\emph{ancestral district}.
%We will first show that the subset of ordered local nested Markov property independence 
%statements pertaining to ancestral districts also imply that the model is ordinary Markov 
%with respect to $\G$.  
The following example shows that we must take care 
in our definition of the ordered local nested Markov property if we do not yet know that
the distribution is Markov with respect to the graph. 

\begin{example}
%When it has not yet been established that a distribution is ordinary Markov
%with respect to a graph, the order of fixing is important.  
It is natural to construct a nested local property in terms of kernels; although under
the model all valid fixing sequences for a given set lead to the same kernel, in terms of \emph{defining} the model
the choice of fixing sequences becomes important.   To see this in a simple example
consider the three node 
graph shown in Figure \ref{fig:simple}(i) and the independence $X_2 \ci X_3$.
That this constraint holds in the kernel that results from fixing $1$ and then $2$ is equivalent to saying that
this independence holds in $p$.  However, if we first fix $2$ this corresponds to marginalization,
and the independence will hold trivially in any kernels derived from it.
%\ref{fig:simple}(iii), as we would have even if the edge $x_2 \leftrightarrow x_3$ were present. 
%To fix 1 in this graph we would divide
%by the full conditional, i.e.\ $p(x_1 \,|\, x_2, x_3)$, and remove edges into 1.  This gives 
%Figure (ii), in which $X_2 \ci X_3$ is clearly witnessed.  If, however, we fix 2 first (because, 
%for example, we choose the topological ordering 1, 2, 3) then we will end up with a
%graph [in (iii)] in which there are no non-trivial independences. 
The `moral' of this example is that we must be careful which fixing sequence we choose
to describe the kernels used to define the local nested Markov property.

\begin{figure}
\begin{center}
%  \begin{tikzpicture}[>=stealth, node distance=1.2cm]
\begin{tikzpicture}[>=stealth, node distance=1.2cm,
pre/.style={->,>=stealth,ultra thick,line width = 1.4pt}]
\tikzstyle{format} = [circle, draw, very thick, minimum size=5mm, inner sep=.5mm]
    \tikzstyle{square} = [rectangle, draw, very thick, minimum size=5mm, inner sep=.5mm]
   \begin{scope}
    \path[<->]	node[format] (x2) {$2$} %{$\vrt{x_2}$}
    		node[format, right of=x2] (x1) {$1$} %{$\vrt{x_1}$}
                  (x1) edge[pre, <->, red] (x2)
		node[format, right of=x1] (x3) {$3$} %{$\vrt{x_3}$}
                  (x1) edge[pre, <->, red] (x3);
   \node[below of=x1, xshift=0cm, yshift = .2cm] {(i)};
  \end{scope}
   \begin{scope}[xshift=7cm]
%    \tikzstyle{format} = [circle, draw, thick, minimum size=5mm, inner sep=.5mm]
        \path[<->]	node[format] (x2) {$2$} %{$\vrt{x_2}$}
    		node[square, right of=x2] (x1) {$1$} %{$\vrt{x_1}$}
%                  (x1) edge[pre, <->, red] (x2)
		node[format, right of=x1] (x3) {$3$}; %{$\vrt{x_3}$};
%                  (x1) edge[pre, <->, red] (x3);
   \node[below of=x1, xshift=0cm, yshift = .2cm] {(ii)};
  \end{scope}
   \begin{scope}[xshift=3.5cm, yshift=-2cm]
%    \tikzstyle{format} = [circle, draw, thick, minimum size=5mm, inner sep=.5mm]
%    \tikzstyle{square} = [rectangle, draw, thick, minimum size=5mm, inner sep=.5mm]
        \path[<->]	node[square] (x2) {$2$} %{$\vrt{x_2}$}
    		node[format, right of=x2] (x1) {$1$} %{$\vrt{x_1}$}
%                  (x1) edge[pre, <->, red] (x2)
		node[format, right of=x1] (x3) {$3$} %{$\vrt{x_3}$}
                  (x1) edge[pre, <->, red] (x3);
   \node[below of=x1, xshift=0cm, yshift = .2cm] {(iii)};
  \end{scope}
  \end{tikzpicture}
\end{center}
\caption{(i) An ADMG, and (ii, iii) the ADMG from (i) after fixing 1 and 2 respectively.}
\label{fig:simple}
\end{figure}

\end{example}

\subsection{Definition of Local Property}

\begin{definition}
Suppose that we have an ADMG with a topological ordering $\prec$.  We
define the \emph{initial segments} to be the $|V|$ sets 
$\Pre_v \equiv \pre_{\G,\prec}(v) \cup \{v\} = \{w : w \prec v\} \cup \{v\}$ consisting of
the first $k$ vertices in the ordering, for each $k=1,\ldots,|V|$.  
The associated \emph{initial segment district} is the district of 
the maximal vertex in each initial segment. 
\end{definition}

Let $\G$ be an ADMG with a topological order $\prec$.  
Let $(C, D)$ be two intrinsic sets, such that 
\begin{enumerate}
\item $\max_\prec(C) = v = \max_\prec(D)$;
\item $C = \dis_{\G[D \setminus \{w\}]}(v)$ for some $w \in D$, such that $w \in \mathbb{F}(\G[D])$.
\end{enumerate}

We define ${\intgr}(\G)$ to be the DAG whose vertices are the intrinsic 
sets of $\G$, and edges $D \rightarrow C$ if and only if the pair $(C,D)$ 
satisfies the two conditions above.  We call $\intgr(\G)$ the \emph{intrinsic 
power DAG} for $\G$.
For any transition $D \rightarrow C$ in ${\intgr}(\G)$ obtained by
fixing (say) $w \in D$; %, define $D^* = D \setminus \{w\}$ to be what remains
%after the initial fixing.  
then $C = \dis_{\G[D \setminus \{w\}]}(v)$, where $v = \max_\prec(D)$.

Note that the power DAG consists of $|V|$ separate connected components, each 
with an initial segment district as its root node.

We are now ready to define the ordered local nested Markov property, which we do 
by using the transitions represented in the power DAG.   
To simplify subscripts we use the notation $\fm_\G(C) := C \cup \pa_\G(C)$.

\begin{definition}
Let $\G$ be an ADMG with an arbitrary topological order $\prec$. %; define  %, and $C_v = \dis_{\G[\Pre_v]}(v)$
%as the initial segment for $v$.
%and the initial segment district respectively. 
A distribution $p$ obeys the \emph{ordered local nested Markov property}
with respect to $(\G,\prec)$, if:
\begin{enumerate}[label={\rm (\roman*)}]
\item for each $v \in V$ we have that the independences
%$p$ factorizes into functions of its districts given their parents:
%\begin{align}
%p(x_V) = \prod_{D \in \mathcal{D}(\G)} q_D(x_D \cmid x_{\pa(D) \setminus D}).
%\label{eqn:localknew0a}
%\end{align}
\begin{align}
%X_v \ci X_{\pre(v; \prec) \setminus \mb(v, \pre(v; \prec)\cup \{v\})} \mid X_{\mb(v, \pre(v; \prec)\cup \{v\})} \;\; [\phi^\prec_{V \setminus (\pre(v;\prec) \cup \{v\})}(p)]\\
X_v \ci X_{\Pre_v \!\setminus\! \fm(C)} \mid X_{\fm(C) \setminus \{v\}} \;\; [\phi_{\Pre_v \!\!\setminus C}(p(x_{\Pre_v}); \G[\Pre_v])]\label{eqn:localknew0}
%X_v \ci X_{\Pre_v \!\!\setminus (C \cup \pa(C))} \mid X_{(C \setminus \{v\}) \cup \pa(C)} \;\; [\phi_{\Pre_v \!\!\setminus C}(p(x_{\Pre_v}); \G[\Pre_v])]\label{eqn:localknew0}
%X_v \ci X_{\Pre_v \!\!\setminus (C_v \cup \pa(C_v))} \mid X_{(C_v \setminus \{v\}) \cup \pa(C_v)} \;\; [\phi_{\Pre_v \!\!\setminus C_v}(p_{\Pre_v}; \G[\Pre_v])]\label{eqn:localknew0}
\end{align}
%\begin{align}
%X_v \ci X_{V \setminus (C \cup \pa(C))} \mid X_{(C \cup \pa(C)) \setminus\{v\}} \;\; [p]\label{eqn:localknew0}
%\end{align}
hold, where $C = \dis_{\Pre_v}(v)$ and  $\phi_{\Pre_v\!\! \setminus C}$
%$\phi_{\Pre_v\!\! \setminus C_v}$
is implemented via some valid fixing sequence;

\item for every transition $D \rightarrow C$ in ${\intgr}(\G)$ 
(obtained by fixing $w \in D$) 
%$D \in {\cal PI}_{\G}(C)$ 
with $v$ being maximal under $\prec$ in $C$, we have:
\begin{align}
X_v \ci X_{\fm(D) \setminus (\fm(C) \cup \{w\})} \mid X_{\fm(C) \setminus\{v\}} \;\; [
%\phi_{D \setminus (C \cup \{w\})} \circ\phi_w({q}_{D}; \G[D])
%\robin{\phi_{D \setminus (C \cup \{w\})}\circ\phi_w({q}_{D}; \G[D])}
 \phi_{\langle \{ w \}, D \setminus (C \cup \{w\}) \rangle}(q_D; \G[D]) ], \label{eqn:localknew}
%X_v \ci X_{(\pa(D) \cup D \setminus \{w\}) \setminus (C \cup \pa(C))} \mid X_{(C \setminus\{v\})\cup \pa(C)} \;\; [\phi_{D \setminus (C \cup \{w\})} \circ\phi_w({q}_{D}; \G[D])], \label{eqn:localknew}
%X_v \ci X_{(\cld(C; D) \cup (\pa(D) \setminus D)) \setminus (C \cup \pa(C))} \mid X_{(C \cup \pa(C)) \setminus\{v\}} \;\; [\tilde\phi_{D \setminus C}({q}_{D}; \G)], \label{eqn:localknew}
\end{align}
where $q_D$ is the unique kernel resulting from some valid fixing sequence for 
$D$, and $\phi_{D \setminus (C\cup \{w\})}$ is implemented via some valid fixing sequence. 
\end{enumerate}
\label{dfn:local-definition}
\end{definition}

\begin{remark}
Note that there is one independence associated with each initial segment 
district, and one associated with each transition in the power DAG.  Any of these 
independences may be null, in that the right hand side of the independence might 
not contain any variables. 
\end{remark}

We will show as part of our proof of Proposition \ref{prop:olmp_implies_df} that 
the kernel $q_D$ in (\ref{eqn:localknew}) %, and that the result
is invariant to the choice of valid fixing sequences in (\ref{eqn:localknew0}) and (\ref{eqn:localknew});
hence, the local nested Markov property itself remains invariant to any choice of valid sequence for any $q_D$ that
appears in the definition.

%Note also that 
%$\phi_{V \setminus (\pre(v;\prec) \cup \{v\})}(p)$ is just the margin over the
%vertices up to $v$ in the topological ordering, by an
%inductive application of those independences.

%Note that $q_D$ is well defined by Theorem \ref{thm:invariant_k}, provided 
%that we can show that the local property up to level $k'$ implies the district
%factorization up to level $k'$ for each $k' \geq 0$.

Note that the above definition is \emph{not} the same as that given in the main body
of the paper, but it follows from the next result that these two sets of
independences are equivalent.  
%Indeed, we will see in Proposition \ref{prop:justfixed} that the set of statements
%(\ref{eqn:localknew}) for intrinsic sets up to level $k$ is equivalent to the same 
%set of statements 
%\begin{align}
%X_v \ci X_{V \setminus (C \cup \pa(C))} \mid X_{(C \cup \pa(C)) \setminus\{v\}} \quad [\tilde\phi_{V \setminus C}(p; \G)] \label{eqn:local2}
%\end{align}
%for intrinsic sets up to level $k$, and where $\tilde\phi_{V \setminus C}$ uses any order of fixing that does not pass through a reachable set of a level higher than $k$.
%%

\begin{lemma} \label{lem:last_guy}
Let $\prec$ be a topological order on an ADMG $\G$, and $p$ an arbitrary distribution.
%It follows that 
Then for a given $v$, {\rm (\ref{eqn:localknew0})} holds under 
%$\phi_{\Pre_v\!\! \setminus C}(p(x_{\Pre_v}); \G[\Pre_v])$ 
$\phi_{T_v \setminus C}(p(x_{T_v}); \G[T_v])$ 
if and only if {\rm (\ref{eqn:localknew0})} holds under $p$,
where $T_v = \Pre_v$.

Let $D \rightarrow C$ be a transition in $\intgr(\G)$ introduced by fixing $w$, and
%such that the maximal vertex $v$ in $C$ under $\prec$ is also the maximal vertex in $D$.
%Let
$q_D(x_D \cmid x_{\pa(D) \setminus D})$ be an arbitrary kernel for $D$. 
Then for $v$ maximal under $\prec$ in $C$, {\rm (\ref{eqn:localknew})} holds for 
 $\phi_{\langle \{ w \}, D \setminus (\{w\} \cup C) \rangle}(q_D; \G[D])$
%\robin{$(\phi_{D \setminus (\{w\} \cup C)} \circ \phi_w)(q_D; \G[D])$}
%\ilya{$\phi_{D \setminus (\{w\} \cup C)}(\phi_w(q_D; \G[D]); \phi_w(\G[D]))$}
if and only 
if the same independence holds under 
$\phi_{w}(q_D; \G[D])$.
\end{lemma}

\begin{proof}
For the first claim, note that this independence will clearly hold under 
$p(x_{\Pre_v})$ if and only if it holds under $p$, since all the variables 
involved are contained in $\Pre_v$.  Then also notice that, since none of the
variables in $\Pre_v \!\setminus \, C$ are descendants of $v$ and they are all in 
 districts that do not contain $v$, by Proposition 
\ref{prop:ind-preserve-2} fixing them will have no effect on the status 
of the conditional independence (\ref{eqn:localknew0}).

For the second claim, note that $C$ becomes a district in $\G[D]$ immediately 
after fixing $w$.  Then
note that any later fixings are within districts other than $C$, and therefore
the Markov blankets of those vertices do not include the maximal 
vertex $v$.  Hence, again by Proposition \ref{prop:ind-preserve-2}, the 
status of the independence involving the conditional distribution of $v$ given $D \setminus \{w,v\}$ is 
unaffected by all these fixings, and so it also holds in $\phi_w(q_D; \G[D])$ if and
only if it holds in
%\robin{$(\phi_{D \setminus (\{w\} \cup C)} \circ \phi_w)(q_D; \G[D])$}.
$\phi_{\langle \{ w \}, D \setminus (\{w\} \cup C) \rangle}(q_D; \G[D])$.
%\ilya{$\phi_{D \setminus (\{w\} \cup C)}(\phi_w(q_D; \G[D]); \phi_w(\G[D]))$.}
%(remembering that $D^* := D \setminus \{w\}$) 
%the conditional distribution 
%of $X_v \mid X_{(D^* \setminus \{v\}) \cup \pa(D^*)}$ is the same in 
%$\phi_w(q_D)$ and $\tilde\phi_{D\setminus C}(q_D)$.  Consequently the 
%independence (\ref{eqn:localknew})
%holds in one if and only if it holds in the other.
\end{proof}

%We remark that this result can be applied to more general transitions:
%see Lemma \ref{lem:last_guy2}.

\subsection{Proof that ordered local nested property implies the global property}

\begin{definition} \label{dfn:fix_by_mar}
Let $v$ be a fixable vertex in a CADMG $\G(V,W)$ such that 
$q_V(x_v \cmid x_{\mb(v)}) = q_V(x_v \cmid x_{V \setminus \{v\}}, x_W)$.  Then we say that
$v$ is \emph{fixed by marginalization} in $q_V$ (and $\G$).
\end{definition}

The terminology is used because, in this case, 
\begin{align*}
\phi_v(q_V; \G) = \frac{q_V(x_V \cmid x_W)}{q_V(x_v \cmid x_{V \setminus \{v\}}, x_W)} = \sum_{x_v} q_V(x_V \cmid x_W),
\end{align*}
and it follows that $\phi_v(q_V; \G)$ does not depend upon $x_v$.
This may occur because 
$\nd_\G(v) = (V \setminus \{v\}) \cup W$ 
and $q_V$ is Markov with respect to $\G$, or equally 
because $\mb_\G(v) = (V \setminus \{v\}) \cup W$ even if $q_V$
is not Markov with respect to $\G$.  This idea will be particularly useful
in proofs where the fact of $q_V$ being Markov with respect to a graph 
has not yet been established.

%For the remainder of this section we will always assume that the 
%topological order in $\G$ is $1,2,\ldots,k$.  
Recall that $\Pre_\ell$ 
denotes $\{\ell\} \cup \{k : k \prec \ell\}$ under a topological ordering $\prec$.
%Recall that $\Pre_v$ 
%denotes $\{v\} \cup \{w : w \prec v\}$ under a topological ordering $\prec$.
%We also introduce
%the notation $\lla{\ell}\rra = \{i \in \mathbb{N} : i \leq \ell\}$ for the set 
%of natural numbers less than or equal to $\ell \in \mathbb{N}$.

\begin{proposition} \label{prop:olmp_implies_df}
Let $p$ be a distribution that is ordered local nested Markov with respect
to $\G$ and the topological order $\prec$.  Then given any valid fixing sequence
for a reachable set $R$ (say ${\bf s}$), the kernel $\phi_{\bf s}(p; \G)$ also district factorizes
according to $\G[R]$.
\end{proposition}

\begin{proof}
Note first that the 
ordered local nested Markov property for $\G$ implies the same property for $\G_{\Pre_k}$,
so we can use the trivial base case with $|V| = 1$ and an inductive argument to assume 
that the global Markov property holds for the margin $p(x_{\pre_{\G,\prec}(k)}) = \phi_k(p(x_{\Pre_{k}}); \G)$;
this identity follows from the independence (\ref{eqn:localknew0}) with $v=k$,
which becomes
\begin{align}
X_k \ci X_{\pre(k) \setminus \mb(k, \Pre_k )} \mid X_{\mb(k, \Pre_k)} \quad\quad [\phi_{\Pre_k \setminus C}(p_{\Pre_k})].
\label{eq:indep-prop-c7-proof}
\end{align} 
We now proceed to prove that the hypothesis also holds for $\Pre_k$.  Note that
an application of Lemma \ref{lem:last_guy} shows that it is equivalent to observe 
the independence (\ref{eq:indep-prop-c7-proof}) in $p_{\Pre_k}$, the marginal distribution over $\Pre_{ k }$.

Using all the independences (\ref{eqn:localknew0}) we obtain
\begin{align*}
p(x_V) &= p(x_k \cmid x_{\mb(k,\Pre_k)}) \cdot \prod_{i : i \prec k} q_{\{i\}}(x_i \cmid x_{\mb(i,\Pre_i)}).
\end{align*}
Note that this gives us the district factorization of $p(x_V)$ w.r.t.~$\G$ since, 
by definition, the Markov blanket within the set of predecessors only ever 
includes vertices in the same district (and its parents).
So in particular, by Lemma \ref{lem:markovfact},  we also have
\begin{align*}
p(x_V) &= \prod_{D \in \mathcal{D}(\G[V])} q_D(x_D \cmid x_{\pa(D) \setminus D}).
\end{align*}

Now, any reachable set $R$ containing $k$ can be obtained by iteratively 
fixing one vertex at a time.  We now proceed to show that $q_R$ district 
factorizes by an inner induction on reachable sets whose maximal vertex is 
$k$.  As a base case, we already 
have the result for $R = V$, so now assume we have the required factorization
for some set $R \subseteq V$.  We will prove that it also holds for sets 
$R' = R\setminus \{\ell\}$ obtained by fixing a vertex 
$\ell \in \mathbb{F}(R) \setminus \{k\}$.  

Let $D \equiv \dis_{\G[R]}(k)$; there are two possibilities: either $\ell$ 
is in $D$, or it is not.  If it is not, then by application of Proposition 
\ref{prop:fix-dist-marg} the fixing will only affect the factors that 
correspond to vertices in the same district.  We therefore know, from the 
original (outer) inductive hypothesis for ${\pre_{\G,\prec}(k)}$, that if we 
only fix vertices not in $D$ we will 
obtain that the appropriate district factorization for the reachable 
graph $\G[R']$ also holds.

If $\ell \in D$, then there will be a transition $D \rightarrow C$ in 
$\intgr(\G)$ corresponding to fixing $\ell$ from $D$, so by (\ref{eqn:localknew}) that yields the independence
\begin{align*}
X_k \ci X_{(\fm(D) \setminus \{\ell\}) \setminus \fm(C)} \mid X_{\fm(C) \setminus \{k\}} \quad %[\tilde\phi_{D \setminus C}(q_D; \G)].
[
\phi_{\langle \{ \ell \}, D \setminus (C\cup\{\ell\}) \rangle}(q_D; \G)
%\robin{\phi_{D \setminus (C\cup\{\ell\})}\circ \phi_\ell(q_D; \G)}
%\ilya{\phi_{D \setminus (C\cup\{\ell\})}(\phi_\ell(q_D; \G); \phi_\ell(\G))}
].
\end{align*}
Note that, from $q_D$ the vertex $\ell$ is fixed by marginalization
by the district factorization; the new independence shows that if $k$ 
is fixed after $\ell$, then $k$ is also fixed by marginalization.

For the reverse order, note that by the inner inductive hypothesis we already 
have the district factorization 
for $R$, so $k$ is clearly also fixed from $D$ by marginalization.  Then
we can apply the outer induction hypothesis on ${\pre_{\G,\prec}(k)}$ to show that 
$\ell$ is also fixed by marginalization from 
its new district in $D \setminus \{k\}$.
%
%We know from Lemma \ref{lem:last_guy} that this holds under $\tilde\phi_{D \setminus C}(q_D; \G)$ 
%if and only if it also holds under $\phi_\ell(q_D; \G)$.  
%In addition we know that $X_k$ and $X_\ell$ are fixed from $D$, so by 
%Proposition \ref{prop:fix-dist-marg} the fixings will only affect $q_D$.
%Additionally, since we know that $q_R$ district factorizes according to 
%$\G[R]$, they are each marginalized from $D$.  In addition, the new independence 
%revealed when we marginalize $\ell$ shows that $k$ will be marginalized if it is
%fixed after $\ell$, and we know that $\ell$ will be fixed by marginalization 
%after $k$ by the outer induction hypothesis applied to $\Pre_{k-1}$.
It follows that $\phi_k$ and $\phi_\ell$ commute when applied to $q_D$, and 
hence $q_{\widetilde{C}}$ is well defined for every district $\widetilde{C}$ of 
$D \setminus \{\ell\}$. 

Hence (letting $D^* \equiv D \setminus \{\ell\}$) we obtain
\begin{align*}
q_{D^*}(x_{D^*} \cmid x_{\pa(D^*) \setminus D^*}) &= q_D(x_k \cmid x_{(C \setminus \{k\}) \cup \pa(C)}) \!\!\! 
\prod_{\widetilde{C} \in \mathcal{D}(\G[D^* \setminus \{k\}])} \!\!\! q_{\widetilde{C}}(x_{\widetilde{C}} \cmid x_{\pa(\widetilde{C}) \setminus \widetilde{C}})\\
&= \prod_{\widetilde{C} \in \mathcal{D}(\G[D^*])} q_{\widetilde{C}}(x_{\widetilde{C}} \cmid x_{\pa(\widetilde{C}) \setminus \widetilde{C}}).
%X_k \ci X_{(D \cup \pa(D)) \setminus (C \cup \pa(C))} \mid X_{(C \setminus \{k\}) \cup \pa(C)} \quad [\tilde\phi_{D \setminus C}(q_D; \G)].
\end{align*}
Consequently we have that for the reachable set $R' = R \setminus \{\ell\}$, 
the distribution $q_{R'} = \prod_{D' \in \mathcal{D}(\G[R'])} q_{D'}$ 
district factorizes according to $\G[R']$.
Hence, by the combination of the two inductions, the same is true for 
every reachable set $R$ in $\G$. 
\end{proof}

\begin{thma}{\ref{thm:global_local}}
$p(x_V)$ is globally nested Markov with respect to $\G$ if and only if $p(x_V)$
is ordered local nested Markov with respect to $\G$ for any topological ordering $\prec$.
%$p(x_V)$ obeys the global nested Markov property with respect to $\G$ if and only if $p(x_V)$ obeys the ordered local
%nested Markov property for $\G$.
%Suppose that $q_V$ is ordered local nested Markov with respect to a CADMG $\G$.  
%Then $q_V$ is globally nested Markov with respect to $\G$.  
\end{thma}

\begin{proof}
That the global property implies the ordered local nested Markov 
property (for any ordering) is obvious from the definitions, since (\ref{eqn:localknew0})
and (\ref{eqn:localknew}) correspond to m-separations in reachable subgraphs.  For the
converse, we simply apply Proposition \ref{prop:olmp_implies_df} and 
Theorem \ref{thm:global-reachable-factorization} to obtain the global
property.
\end{proof}

%\subsection{Equivalence of the Two Versions of the Ordered Local Property}
%
%We note that from the district factorization it is clear that in fact (\ref{eqn:localknew}) 
%(for all transitions in the power DAG) is equivalent to the (apparently) 
%stronger collection of independences
%\begin{align*}
%X_v \ci X_{V \setminus (C \cup \pa(C))} \mid X_{(C \setminus \{v\}) \cup \pa(C)} \quad [\phi_{D \setminus (C \cup \{w\})} \circ\phi_w({q}_{D}; \G[D])]
%\end{align*}
%(for all transitions $D \rightarrow C$ in the power DAG);
%this is because all vertices other than those in $D \cup \pa_\G(D)$ have 
%been fixed already, and these all correspond to marginalizations by 
%application of the proof of Proposition \ref{prop:olmp_implies_df}.

\subsection{Illustrative Examples}

Here we give some examples of the ordered local nested Markov property
associated with a given vertex, so as to illustrate how it works.  

\begin{example}
Consider again the bidirected graph in Figure \ref{fig:simple}(i), under
the topological ordering $1,2,3$.  In 
this case there are three intrinsic sets containing the maximal vertex $3$,
and these are $\{3\}$, $\{2,3\}$ and $\{1,2,3\}$.
Then the power DAG for the initial segment $\{1,2,3\}$ is just 
the complete DAG on these three vertices, and the only transition 
that gives a constraint is $\{1,2,3\} \to \{3\}$ when fixing $1$, which 
tells us that $X_3 \ci X_2 \, [\phi_1(p)]$, or equivalently that 
$X_3$ and $X_2$ are marginally independent.
\end{example}

\begin{example}
Consider the ADMG in Figure \ref{fig:counter1}(i), and consider the 
topological ordering $1,2,3,4,5,6$.  Then the relevant power DAG 
restricted to descendants of $V=\Pre_6$ 
is shown in Figure \ref{fig:powerdag1}.

There are two highlighted arrows, one from $\{1,4,5,6\}$ to $\{6\}$, 
and the other from $\{1,3,5,6\}$ to $\{6\}$.  These are both 
\emph{maximal}, in the sense that there is no superset of 
either of these sets from which one can transition directly to the set 
$\{6\}$.  
%Hence, the two independences that these transitions are
%associated with can \emph{only} be achieved by having these two 
%separate paths to get to 6.  

\begin{figure}
\begin{center}
%  \begin{tikzpicture}[>=stealth, node distance=1.2cm]
\begin{tikzpicture}[>=stealth, node distance=1.2cm,
pre/.style={->,>=stealth,very thick,line width = 1.4pt}]
\tikzstyle{format} = [circle, draw, very thick, minimum size=5mm, inner sep=.5mm]
    \tikzstyle{square} = [rectangle, draw, very thick, minimum size=5mm, inner sep=.5mm]
%    \tikzstyle{format} = [circle, draw, very thick, minimum size=5mm, inner sep=.3mm]
   \begin{scope}
    \path[<->]	node[format] (x4) {$4$} %{$\vrt{x_4}$}
    		node[format, above right of=x4] (x1) {$1$} %{$\vrt{x_1}$}
                  (x4) edge[pre, <->, red] (x1)
                  node[format, below right of=x4] (x2) {$2$} %{$\vrt{x_2}$}
                  (x1) edge[pre, ->, blue] (x2)
		node[format, above of=x1] (x6) {$6$} %{$\vrt{x_6}$}
                  (x1) edge[pre, <->, red] (x6)
                  node[format, below of=x2] (x3) {$3$} %{$\vrt{x_3}$}
                   (x2) edge[pre, ->, blue] (x3)
                  (x4) edge[pre, <->, red] (x2)
                   node[format, left of=x4] (x5) {$5$} %{$\vrt{x_5}$}
                  (x4) edge[pre, <->, red] (x5)
                   (x5) edge[pre, <->, red, out=45, in=180] (x1)
                  (x5) edge[pre, <->, red, out=290, in=160] (x3)
                  ;               
   \node[below of=x5, yshift = -1cm] {$\G$};
  \end{scope}
      \begin{scope}[xshift=4.5cm]
%    \tikzstyle{format} = [circle, draw, thick, minimum size=5mm, inner sep=.3mm]
%      \tikzstyle{square} = [rectangle, draw, very thick, minimum size=5mm, inner sep=.5mm]
   \path[<->]	node[format] (x4) {$4$} %{$\vrt{x_4}$}
    		node[square, above right of=x4] (x1) {$1$} %{$\vrt{x_1}$}
                %  (x4) edge[pre, <->, red] (x1)
                  node[square, below right of=x4] (x2) {$2$} %{$\vrt{x_2}$}
             %     (x1) edge[pre, ->, blue] (x2)
		node[format, above of=x1] (x6) {$6$} %{$\vrt{x_6}$}
                 % (x1) edge[pre, <->, red] (x6)
                  node[square, below of=x2] (x3) {$3$} %{$\vrt{x_3}$}
             %      (x2) edge[pre, ->, blue] (x3)
              %    (x4) edge[pre, <->, red] (x2)
%                  node[format, left of=x2] (x5) {$x_5$}
               %   (x5) edge[pre, <->, red] (x2)
       %           (x5) edge[pre, <->, red] (x3)
                   node[format, left of=x4] (x5) {$5$} %{$\vrt{x_5}$}
                 (x4) edge[pre, <->, red] (x5)
%                   (x5) edge[pre, <->, red, out=45, in=180] (x1)
%                  (x5) edge[pre, <->, red, out=290, in=160] (x3)
                  ;                      
  \node[below of=x5, yshift = -1cm] {$\G[\{4,5,6\}]$};
  \end{scope}  
     \begin{scope}[xshift=9cm]
%    \tikzstyle{format} = [circle, draw, thick, minimum size=5mm, inner sep=.3mm]
%      \tikzstyle{square} = [rectangle, draw, thick, minimum size=5mm, inner sep=.5mm]
    \path[<->]	node[square] (x4) {$4$} %{$\vrt{x_4}$}
    		node[format, above right of=x4] (x1) {$1$} %{$\vrt{x_1}$}
                %  (x4) edge[pre, <->, red] (x1)
                  node[square, below right of=x4] (x2) {$2$} %{$\vrt{x_2}$}
         %         (x1) edge[pre, ->, blue] (x2)
		node[format, above of=x1] (x6) {$6$} %{$\vrt{x_6}$}
                  (x1) edge[pre, <->, red] (x6)
                  node[format, below of=x2] (x3) {$3$} %{$\vrt{x_3}$}
                   (x2) edge[pre, ->, blue] (x3)
           %       (x4) edge[pre, <->, red] (x2)
%                  node[square, left of=x2] (x5) {$x_5$}
          %        (x5) edge[pre, <->, red] (x2)
           %       (x5) edge[pre, <->, red] (x3)
                   node[format, left of=x4] (x5) {$5$} %{$\vrt{x_5}$}
                   (x5) edge[pre, <->, red, out=45, in=180] (x1)
                  (x5) edge[pre, <->, red, out=290, in=160] (x3)
                  ;               
  \node[below of=x5, yshift = -1cm] {$\G[\{1,3,5,6\}]$};
  \end{scope}  
  \end{tikzpicture}
\end{center}
\caption{(i) An ADMG $\G$; (ii) $\G[\{4,5,6\}]$; (iii) $\G[\{1,3,5,6\}]$.}
\label{fig:counter1}
\end{figure}

\begin{figure}
\begin{center}
%  \begin{tikzpicture}[>=stealth, node distance=1.2cm]
\begin{tikzpicture}[>=stealth, node distance=1.5cm,
pre/.style={->,>=stealth,very thick,line width = 1.4pt}]
   \begin{scope}
    \tikzstyle{rv} = [minimum size=5mm, inner sep=.5mm]
\node[rv] (123456) {$123456$};
\node[rv, below of=123456] (12456) {$12456$};
\node[rv, below of=12456] (1456) {$1456$};
\node[rv, left of=1456, xshift=-1cm] (1356) {$1356$};
\node[rv, right of=1456, xshift=1cm] (1246) {$1246$};
\node[rv, below of=1456, xshift=-1.25cm] (156) {$156$};
\node[rv, below of=1456, xshift=1.25cm] (146) {$146$};
\node[rv, below of=156, xshift=1.25cm] (16) {$16$};
\node[rv, below of=16, xshift=0cm] (6) {$6$};
\draw[pre] (123456) -- (12456);
\draw[pre] (12456) -- (1456);
\draw[pre] (12456) -- (1246);
\draw[pre] (123456) -- (1356);
\draw[pre] (123456) -- (1246);
\draw[pre] (12456) -- (156);
\draw[pre] (1456) -- (156);
\draw[pre] (1456) -- (146);
\draw[pre] (1356) -- (156);
\draw[pre] (1246) -- (146);
\draw[pre] (1356) to[bend right] (16);
\draw[pre] (1246) to[bend left] (16);
\draw[pre] (146) -- (16);
\draw[pre] (146) to[bend left] (6);
\draw[pre] (156) -- (16);
\draw[pre, color=green] (1456) to[bend left=20] (6);
\draw[pre] (156) -- (6);
\draw[pre, color=orange] (1356) to[bend right=20] (6);
\draw[pre] (16) -- (6);
  \end{scope}
  \end{tikzpicture}
\end{center}
\caption{The connected component of the power DAG for $\G$ in Figure \ref{fig:counter1} with
intrinsic sets in which 6 is the maximal vertex.}
\label{fig:powerdag1}
\end{figure}

Note that it is necessary for us to use two different fixing 
sequences to get to the intrinsic set 6, because one of them
gives the ordinary independence $X_6 \ci X_4,X_5$, and the other gives
a nested independence $X_6 \ci X_3, X_5 \mid X_2$.  There is 
no way to obtain these \emph{both} from the same fixing 
sequence, and there is no joint independence involving all
of these variables.
\end{example}

\begin{example}
Now consider the graph in Figure \ref{fig:counter2}(a) under a 
numerical topological ordering, and the part of its intrinsic power DAG 
in Figure \ref{fig:counter2}(b).  Again, note that there are two 
distinct edges into $\{6\}$, from $\{2,4,6\}$ and $\{3,4,6\}$, that 
give two distinct independences.  However, there is no intrinsic set 
containing $\{2,3,4,6\}$ from which we could obtain both these 
independences. 

\begin{figure}
\begin{center}
%  \begin{tikzpicture}[>=stealth, node distance=1.2cm]
\begin{tikzpicture}[>=stealth, node distance=1.6cm,
pre/.style={->,>=stealth,ultra thick,line width = 1.4pt}]
   \begin{scope}
    \tikzstyle{format} = [circle, draw, very thick, minimum size=5mm, inner sep=.5mm]
    \path[<->]	node[format] (e) {$4$} %{$\vrt{x_4}$}
    		node[format, below of=e] (f) {$5$} %{$\vrt{x_5}$}
		node[format, below of=f] (y) {$6$} %{$\vrt{x_6}$}
		node[format, left of=f] (a) {$1$} %{$\vrt{x_1}$}
		node[format, left of=a] (d) {$3$} %{$\vrt{x_3}$}
		node[format, left of=y] (b) {$2$} %{$\vrt{x_2}$}
		%node[format, left of=y] (x) {$x$}
		%node[format, right of=e] (d) {$d$}
		%(a) edge[pre, ->, blue] (x)
		(e) edge[pre, ->, blue] (f)
		(f) edge[pre, ->, blue] (y)
		(a) edge[pre, ->, blue] (b)
                (d) edge[pre, <->, red] (e)
                (d) edge[pre, <->, red] (a)
                 (b) edge[pre, <->, red] (e)
                   (a) edge[pre, <->, red] (f)
                 % (f) edge[pre, <->, red] (d)
                  % (y) edge[pre, <->, red] (c)
                    (e) edge[pre, <->, red, out=300,in=60] (y);
                 %    (a) edge[pre, <->, red, out=45,in=135] (d);
%                  (x1) edge[pre, ->, blue] (x2)
%		node[format, right of=x1] (x5) {$x_5$}
%                  (x1) edge[pre, <->, red] (x5)
%                  node[format, below of=x2] (x3) {$x_3$}
%                   (x2) edge[pre, ->, blue] (x3)
%                  (x4) edge[pre, <->, red] (x3);
   \node[left of=b, yshift=-7mm] {(a)};
  \end{scope}
     \begin{scope}[xshift=5cm, yshift=1cm]
    \tikzstyle{rv} = [minimum size=5mm, inner sep=.5mm]
\node[rv] (abdefy) {123456};%{$abde\hspace{-0.1em}f\hspace{-0.1em}y$};
\node[rv, below right of=abdefy] (adefy) {13456};%{$ade\hspace{-0.1em}f\hspace{-0.1em}y$};
\node[rv, below left of=adefy] (bey) {246};%{$bey$};
\node[rv, below right of=adefy] (dey) {346};%{$dey$};
\node[rv, below right of=bey] (ey) {46};%{$ey$};
\node[rv, below of=ey] (y) {6};
%\node[rv, left of=1456, xshift=-1cm] (1356) {$1356$};
%\node[rv, right of=1456, xshift=1cm] (1246) {$1246$};
%\node[rv, below of=1456, xshift=-1.25cm] (156) {$156$};
%\node[rv, below of=1456, xshift=1.25cm] (146) {$146$};
%\node[rv, below of=156, xshift=1.25cm] (16) {$16$};
%\node[rv, below of=16, xshift=0cm] (6) {$6$};
\draw[pre] (abdefy) -- (adefy);
\draw[pre] (abdefy) -- (bey);
\draw[pre] (adefy) -- (dey);
\draw[pre] (adefy) -- (ey);
\draw[pre] (bey) -- (ey);
\draw[pre] (dey) -- (ey);
\draw[pre, color=orange] (bey) -- (y);
\draw[pre, color=green] (dey) -- (y);
\draw[pre] (ey) -- (y);
%\draw[pre] (123456) -- (1356);
%\draw[pre] (123456) -- (1246);
%\draw[pre] (12456) -- (156);
%\draw[pre] (1456) -- (156);
%\draw[pre] (1456) -- (146);
%\draw[pre] (1356) -- (156);
%\draw[pre] (1246) -- (146);
%\draw[pre] (1356) to[bend right] (16);
%\draw[pre] (1246) to[bend left] (16);
%\draw[pre] (146) -- (16);
%\draw[pre] (156) -- (16);
%\draw[pre, color=green] (1456) to[bend left=20] (6);
%\draw[pre] (156) -- (6);
%\draw[pre, color=orange] (1356) to[bend right=20] (6);
%\draw[pre] (16) -- (6);
   \node[left of=y] {(b)};
  \end{scope}
%     \begin{scope}[xshift=8cm]
%    \tikzstyle{format} = [circle, draw, thick, minimum size=5mm, inner sep=.5mm]
%    \path[<->]	node[format] (d) {$d$}
%    		node[format, below of=d] (e) {$e$}
%		node[format, below of=e] (y) {$y$}
%		node[format, left of=e] (f) {$f$}
%		node[format, left of=f] (a1) {$a_1$}
%		node[format, left of=y] (a2) {$a_2$}
%		node[format, below of=a1] (z) {$z$}
%		node[format, left of=z] (w) {$w$}
%		%node[format, left of=y] (x) {$x$}
%		%node[format, right of=e] (d) {$d$}
%		%(a) edge[pre, ->, blue] (x)
%		(d) edge[pre, ->, blue] (e)
%		(e) edge[pre, ->, blue] (y)
%		(f) edge[pre, ->, blue] (a2)
%		(a2) edge[pre, ->, blue] (z)
%		(z) edge[pre, ->, blue] (w)
%		(w) edge[pre, ->, blue] (a1)
%                (a1) edge[pre, <->, red] (d)
%                (a1) edge[pre, <->, red] (f)
%                 (a2) edge[pre, <->, red] (d)
%                   (f) edge[pre, <->, red] (e)
%                 % (f) edge[pre, <->, red] (d)
%                  % (y) edge[pre, <->, red] (c)
%                   (a1) edge[pre, <->, red] (z)
%                    (d) edge[pre, <->, red, out=300,in=60] (y);
%                 %    (a) edge[pre, <->, red, out=45,in=135] (d);
%%                  (x1) edge[pre, ->, blue] (x2)
%%		node[format, right of=x1] (x5) {$x_5$}
%%                  (x1) edge[pre, <->, red] (x5)
%%                  node[format, below of=x2] (x3) {$x_3$}
%%                   (x2) edge[pre, ->, blue] (x3)
%%                  (x4) edge[pre, <->, red] (x3);
%   \node[above of=a1] {(b)};
%  \end{scope}
  \end{tikzpicture}
\end{center}
\caption{(a) Graph with two fixing sequences for $\{6\}$ that lead to different 
conditional independences. (b) The power DAG component with intrinsic sets in which $6$ is the maximal
 vertex.}
\label{fig:counter2}
\end{figure}
\end{example}

A moral of the examples above is that it is also not possible
to have a local Markov property where there is at most  
one independence statement corresponding to each intrinsic set.  
Our local property associates independences with 
initial segment districts $D$, and with pairs $(D,w)$ where $D$ is an 
intrinsic set in which $w$ is fixable.

%\newpage

\subsection{Saturated Nested Models}

%\section{Proof of Theorem \ref{thm:saturated_a}}

In this subsection, we assume a fixing sequence $1, \ldots, k$ for $V$
and denote $\G^{(1)} := \G$ and $\G^{(i+1)} := \phi_i(\G^{(i)})$ 
for $i=1,\ldots,k$. Also let $V^{(i)} \equiv \mathbb{V}(\G^{(i)}) = \{i,\ldots,k\}$. 

Our first lemma will be useful in showing the existence of distributions
that are not Markov with respect to a particular graph.

\begin{lemma} \label{lem:triv-dist}
Let $\G$ be an ADMG and $1, \ldots, i, \ldots, j$ be a valid fixing sequence 
in which $i \notin \mb_{\G^{(j)}}(j)$ and $j \notin \mb_{\G^{(i)}}(i)$.
Also let $p(x_V)$ be a distribution in which $X_{i}, X_{j} \ci X_{V \setminus \{i,j\}}$. 

Then:
\begin{enumerate}[label={\rm (\roman*)}]
\item $X_i, X_j \ci X_{V \setminus \{i,j\}}$ in $q_{V^{(\ell)}}$ for all $\ell \leq j$;
\item $q_{V^{(\ell)}}(x_{j} \cmid x_{i}) = p(x_{j} \cmid x_{i})$ for
all $\ell \leq j$.  In particular, if $X_i \not\ci X_j$ in $p$, then this also holds for each 
$q_{V^{(\ell)}}$.
\end{enumerate}
\end{lemma}

\begin{proof}
For $\ell \neq i,j$, probabilistic fixing involves dividing by 
$q_{V^{(\ell)}}(x_\ell \cmid x_{\mb_{\G^{(\ell)}}(\ell)})$.  By the independence,
this is the same as both 
$q_{V^{(\ell)}}(x_\ell \cmid x_{\mb_{\G^{(\ell)}}(\ell) \setminus \{i,j\}})$ and 
$q_{V^{(\ell)}}(x_\ell \cmid x_{\{i,j\} \cup \mb_{\G^{(\ell)}}(\ell)})$;
and hence by Lemma \ref{lem:no-backtracking} both 
the conditional distribution
of $X_i, X_j$ given $X_{V \setminus \{i,j\}}$ and
the marginal distribution
of $X_i, X_j$ is preserved in $q_{V^{(\ell+1)}}$. 

When fixing $i$, we divide by 
$q_{V^{(i)}}(x_i \cmid x_{\mb_{\G^{(i)}}(i)})$ and, by hypothesis,
$j \notin \mb_{\G^{(i)}}(i)$.  Hence, again by Lemma \ref{lem:invariance-kernel},
the distribution of $X_j$ conditional on $X_{V \setminus \{j\}}$ is 
preserved.  Then since $X_j \ci X_{V \setminus \{i,j\}} \mid X_i$, this 
means that the distribution of $X_j$ conditional on $X_i$ is 
preserved.
\end{proof}

\begin{lemma} \label{lem:fixbymar}
Let $v$ be a childless vertex in a complete ADMG $\G$.  Then, for
any distribution $p$ and valid fixing sequence, $v$ is fixed
by marginalization.\footnote{See Definition \ref{dfn:fix_by_mar}.}  In addition, if any 
$w \in \sib_\G(v)$ is fixed before $v$, it is also fixed by marginalization.
\end{lemma}

\begin{proof}
Note that every other random vertex is a parent or sibling of $v$, and hence
contained in $\mb_\G(v)$.  Certainly, then, if $v$ is fixed from $\G$ then it is
fixed by marginalization.
Let $w$ be a sibling of $v$, and note that if $w$ is fixable in $\phi_{V \setminus R}(\G)$ 
then it is also childless in $\phi_{V \setminus R}(\G)$.  To see this, suppose for contradiction
that $w \rightarrow t$; by completeness, either 
$t \rightarrow v$ or $t \leftrightarrow v$, but in either case $w$ is not 
fixable and we reach our contradiction. 

Now suppose we are fixing $w \in \{v\} \cup \sib_\G(v)$ in some reachable
graph $\phi_{V \setminus R}(\G)$.  Every random 
vertex is joined to $w$, and hence is a parent or sibling of $w$.  Note that,
(if $w \neq v$) any parents of $v$ are also in $\mb_\G(w; R)$.  
The only variables that might be outside
$\mb_\G(w; R)$ are those fixed vertices that were not parents of $w$ (or $v$) 
in $\G$; i.e.\ they were also siblings of $v$.  The first such vertex in the sequence, 
say $a$, is fixed by marginalization because its Markov blanket includes all 
variables.  

For subsequent vertices, the Markov blanket includes all vertices other
than those previously fixed by marginalization; since, by definition, the kernel does not 
depend upon these marginalized vertices, then these subsequent vertices
are also fixed by marginalization.  It follows by a simple induction that $w$ is also
fixed by marginalization. 
\end{proof}

Starting with a complete ADMG, fixing operations remove edges and may lead 
to a CADMG that is no longer complete.  Consequently, the following result is
by no means trivial.

\begin{lemma} \label{lem:commute}
Let $\G$ be a complete ADMG and $p$ an arbitrary distribution over $X_V$ {\rm (}not 
necessarily Markov with respect to $\G${\rm )}.  
If $v$ is a vertex with no children in $\G$, 
then for any reachable set $R \subseteq V \setminus \{v\}$ 
and any $w$ fixable in $\phi_{V \setminus R}(\G)$, we 
have
%$(\phi_w \circ \phi_v)(q_R) = (\phi_v \circ \phi_w)(q_R)$,
{$\phi_{\langle v, w \rangle}(q_R; \phi_{V \setminus R}(\G)) =
\phi_{\langle w, v \rangle}(q_R; \phi_{V \setminus R}(\G))$}
where $q_R = \phi_{V \setminus R}(p; \G)$.
\end{lemma}

\begin{proof}
By Lemma \ref{lem:fixbymar}, $v$ is fixed by marginalization in $\G^*\equiv \phi_{V\setminus R}(\G)$.
If $w$ is a sibling of $v$ then (also by Lemma \ref{lem:fixbymar}) 
it is also fixed by marginalization,
and the operations $\phi_v$ and $\phi_w$ amount to sums that 
clearly commute.

Otherwise $w$ is a parent of $v$ and---since $w$ is fixable---is 
in a different district to $v$ in $\G^*$. % $\phi_{V \setminus R}(\G)$.  
Therefore $\mb_{\G}(w; R) = \mb_{\G}(w; R \setminus \{v\})$ and
also $v \notin \mb_{\G}(w; R)$.  It then follows that 
%from Lemma \ref{lem:invariance-kernel} that 
\begin{align*}
\phi_v(\phi_w(q_R; \G^*); \phi_w(\G^*)) &= \sum_{x_v} \frac{q_R(x_R \cmid x_{V \setminus R})}{q_R(x_w \cmid x_{\mb(w,R)})}\\
&= \frac{\sum_{x_v} q_R(x_R \cmid x_{V \setminus R})}{q_{R \setminus \{v\}}(x_w \cmid x_{\mb(w,R \setminus \{v\})})}\\
&= \frac{q_{R \setminus \{v\}}(x_{R \setminus \{v\}} \cmid x_{V \setminus R})}{q_{R \setminus \{v\}}(x_w \cmid x_{\mb(w,R \setminus \{v\})})}\\
&= \phi_w(\phi_v(q_R; \G^*); \phi_v(\G^*))
\end{align*}
as required.
\end{proof}

\subsubsection{Maximal Arid Graphs}

Let $A \subseteq V$ be a set of vertices, and define  the 
\emph{intrinsic closure} of $A$, denoted $\langle A \rangle_\G$, to be the
smallest reachable set containing $A$.  Note that this set is well-defined by
Theorem \ref{thm:invariant}.  We say that two vertices $v,w$ are 
\emph{densely connected} if either:
\begin{enumerate}[label={\rm (\roman*)}]
\item $v \in \pa_\G(\langle \{w\} \rangle_\G)$; or
\item $w \in \pa_\G(\langle \{v\} \rangle_\G)$; or
\item $\langle \{v,w\} \rangle_\G$ is bidirected-connected.
\end{enumerate}

\begin{lemma}
The condition in Theorem \ref{thm:saturated_a} is equivalent to every pair of vertices being
densely connected.
\end{lemma}

\begin{proof}
First suppose $i$ and $j$ are densely connected.  If $i \in \pa_\G(\langle j \rangle_\G)$, 
then for any fixing sequence $i \in \mb_{\G^{(j)}}(j)$. 
Alternatively, if $\langle \{ i, j \} \rangle_\G$ is bidirected-connected, 
then for any fixing sequence, if $i$ is fixed before
$j$ then $j \in \mb_{\G^{(i)}}(i)$.  % and if $j$ is fixed before $i$, $j \in \mb_{\G^{(i)}}(i)$.

Now suppose that $i$ and $j$ are not densely connected.
%Assume $i \not\in \pa_\G(\langle j \rangle)$, $j \not\in \pa_\G(\langle i \rangle)$, and
%$\langle \{ i, j \} \rangle_\G$ is not bidirected-connected.  
Consider any fixing sequence that fixes
$V \setminus \langle \{ i, j \} \rangle_\G$ first (one exists by definition of $\langle \cdot \rangle_\G$).
If $i,j$ both have children in $\phi_{V \setminus \langle \{ i, j \} \rangle}(\G)$,
then $\phi_{V \setminus \langle \{ i, j \} \rangle}(\G)$ has a childless vertex that is not in $\{ i, j \}$.
Since childless vertices are always fixable, we have a contradiction, since the only vertices fixable in
$\phi_{V \setminus \langle \{ i, j \} \rangle}(\G)$ are in $\{ i, j \}$.

Without loss of generality assume $i$ does not have children in $\phi_{V \setminus \langle \{ i, j \} \rangle}(\G)$.  
Then, since $\langle \{ i, j \} \rangle_\G$ is not bidirected-connected, $i$ is not in
the Markov blanket of $j$ in $\G^{(j)}$ 
for a fixing sequence of $V \setminus \langle \{ i, j \} \rangle_\G$, followed by $j$, 
followed by any remaining vertices in $V \setminus \langle i \rangle_\G$, 
followed by $i$.  Furthermore,
the Markov blanket for $i$ in $\G^{(i)}$ is just $\pa_\G(\langle i \rangle_\G)$,
so by assumption $j \notin \mb_{\G^{(i)}}(i)$.
Thus, we have constructed a fixing sequence for which the condition in Theorem
\ref{thm:saturated_a} fails.
%
%{\color{red} We have now established that the above criterion is equivalent to every vertex pair being densely connected in $\G$.}. 
\end{proof}

In \citet[][Theorem 30]{uai:18} it is shown that every graph $\G$ is nested Markov equivalent to
a graph $\G^{\dag}$, called its \emph{maximal arid projection}, for which 
$\langle \{v\} \rangle_{\G^{\dag}} = \{v\}$.  The 
graph $\G^{\dag}$ is simple (at most one edge between each pair of vertices),
and is obtained from $\G$ by joining pairs of vertices that are 
densely connected.  It follows that the condition in Theorem \ref{thm:saturated_a} 
is satisfied only by graphs that are nested Markov equivalent to a complete, 
simple graph.  

%\begin{thma}{\ref{thm:saturated_a}}
%Let $\G$ be an ADMG.  
%%Given a valid fixing sequence, $r_1, \ldots, r_k$,
%%let $\G^{(i)} \equiv \phi_{r_1,\ldots,r_{i-1}}(\G)$.  
%The model ${\cal P}^{\rm n}(\G)$ is 
%saturated if and only if for every valid fixing sequence $r_1, \ldots, r_k$,
%and every $i,j \in 1,\ldots, k$, either $r_i \in \mb_{\G^{(j)}}(r_j)$ or 
%$r_j \in \mb_{\G^{(i)}}(r_i)$.
%Here $\G^{(1)} \equiv \G$ and $\G^{(i+1)} \equiv \phi_{r_i}(\G^{(i)})$.
%\end{thma}

\begin{thma}{\ref{thm:saturated_a}}
Let $\G$ be an ADMG.  
%Given a valid fixing sequence, $r_1, \ldots, r_k$,
%let $\G^{(i)} \equiv \phi_{r_1,\ldots,r_{i-1}}(\G)$.  
The model ${\cal P}^{\rm n}(\G)$ is 
saturated if and only if for every valid fixing sequence $r_1, \ldots, r_k$,
and every $i,j \in 1,\ldots, k$, either $r_i \in \mb_{\G^{(j)}}(r_j)$ or 
$r_j \in \mb_{\G^{(i)}}(r_i)$.
Here, $\G^{(1)} \equiv \G$ and $\G^{(\ell+1)} \equiv \phi_{r_\ell}(\G^{(\ell)})$.
\end{thma}

\begin{proof}
Choose $\G$ such that the condition is false for a fixing sequence 
$1, \ldots, k$, and a pair of vertices $i, j$ with $i<j$.  Pick a distribution 
$p$ in which $X_{V \setminus \{i,j\}} \ci X_{i}, X_{j}$ but
with $X_{i} \not\ci X_{j}$; we claim $p$ is not nested Markov with respect to $\G$. 
By Lemma \ref{lem:triv-dist} we 
have $X_i,X_j \ci X_{V \setminus \{i,j\}}$ in $q_{V^{(j)}}$, and that 
$q_{V^{(j)}}(x_j \cmid x_i) = p(x_j \cmid x_i)$, so in particular 
$X_i \not\ci X_j$ in $q_{V^{(j)}}$.

By ordinary graphoids
we then obtain $X_i \not\ci X_j \mid X_{C}$ in $q_{V^{(j)}}$ for any 
$C \subseteq V \setminus \{i,j\}$.  But the nested local Markov property
implies that $X_{j} \ci X_{V \setminus (\{j\} \cup \mb_{\G^{(j)}}(j))} \mid X_{\mb_{\G^{(j)}}(j)}$ in 
$q_{V^{(j)}}$ and by hypothesis $i \notin \mb_{\G^{(j)}(j)}$, so in particular 
$X_{j} \ci X_{i} \mid X_{\mb_{\G^{(j)}}(j)}$ in 
$q_{V^{(j)}}$.  This contradicts our construction above, and hence
the distribution $p$ is not nested Markov with respect to $\G$, and
the model is not saturated.

Now, since the condition in the statement of the theorem is equivalent to the resulting 
maximal arid projection $\G^{\dag}$ being 
complete, we will proceed to show that if a graph is simple and complete, then the resulting
nested Markov model is saturated.  Then since every graph $\G$ for which the original 
condition holds is nested Markov equivalent to a complete simple graph $\G^{\dag}$, 
the result will follow.

Suppose all vertex pairs are adjacent in $\G$, take any topological ordering $\prec$,
pick any $S \in {\cal I}(\G)$
and let $i$ be the $\prec$-greatest element of $S$.  We will show that
$q_S$ does not depend upon any $x_j$ for $j \notin S \cup \pa_\G(S)$, and
hence that the local Markov property holds.

%$X_i \ci X_{V \setminus (S \cup \pa_\G(S))} \mid X_{S \setminus \{i\} \cup \pa_\G(S)}$ in $q_S$,
%by showing that $q_S$ does not  

First, note Lemma \ref{lem:commute} implies that we may assume any vertices that are 
not ancestors of $S$ have already been fixed by marginalization by reordering reverse topologically.  
To see this, note that if there is some childless vertex not in $S$, then Lemma \ref{lem:commute}
states that we can fix it first without consequence, and that this corresponds to marginalization.
Hence, if $j$ is a child of $i$, then $q_S$ does not depend upon $x_j$.  
Hereafter we assume without loss of generality that all vertices in $\G$ are ancestors of $S$,
and in particular $i$ does not have any children. 

Since the graph is complete, this leaves only two possibilities.  
If $j \in \pa_\G(i)$, then $j \in \pa_\G(S)$, and there is nothing to show.
Otherwise $j \leftrightarrow i$, and
there is nothing to prove unless $j \notin S \cup \pa_\G(S)$.
In this case, by Lemma \ref{lem:fixbymar}, $j$ is fixed by marginalization. 

We have shown that for any distribution $p$ and any intrinsic set $S$,
the kernel $q_S$ resulting from fixing does not depend upon any 
$x_j$ for $j \notin S \cup \pa_\G(S)$;
hence the local Markov property holds. 
%
%
%[This in fact means $X_j$ is independent of everything else, so will remain true...]
%
%the kernel for $S \setminus \{j\}$ is a function only of the second factor $q_{V^{(j)}}(x_{\mb_{\G^{(j)}}(j)} \cmid x_{W^{(j)}})$.  By definition, this factor does not depend upon $x_j$, and as a result $X_j \ci X_S \mid X_{W^{(j)}}$ trivially.
%
% is trivially independent of anything in the second term, including any function of the second term, such as the kernel for $S$.  
% ({\color{red}Make this more precise}).
%
\end{proof}

\begin{cora}{\ref{cor:complete-saturated}}
Let $\G$ be a complete ADMG; then ${\cal P}^{\rm n}(\G)$ is saturated.
\end{cora}

%\begin{proof}
%We need to show that for any fixable random vertex $v$ in $\phi_W(\G)$ 
%the condition (\ref{eqn:mbnd}) holds.
%Let $\G^* = \phi_W(\G)$; any random vertices $t \in \mathbb{V}(\G^*)$ share an edge with $v$;
%and further since $v \in {\mathbb F}(\G^*)$, $\de_{\G^*}(v) = \emptyset$, so $t \in \nd_{\G^*}(v)$.
%%If $t \in \nd_{\G^*}(v)$ 
%%then that means
%Hence either $t \leftrightarrow v$ or $t \rightarrow v$.  In either case, $t \in \mb_{\G^*}(v)$.
%%By definition, if  $t \in \de_{\G^*}(v)$ then $t \not\in \mb_{\G^*}(v)$; hence 
%%$\mathbb{V}(\G^*) \cap \mb_{\G^*}(v) = \mathbb{V}(\G^*) \cap \nd_{\G^*}(v)$.
%\end{proof}

\begin{proof}
Note that if $i \leftrightarrow j$, then certainly $i$ and $j$ are in each other's
Markov blankets whenever one of these is fixed.  
If $i \rightarrow j$, then $i$ is in $j$'s Markov blanket, and will remain so even if $i$ 
is fixed.  Hence the condition in Theorem \ref{thm:saturated_a} holds.\footnote{The authors have been made aware of an independent proof by Bhattacharya and Nabi of the result that the nested model associated with a 
complete ADMG is nonparametric saturated \citep{bhattacharya21semi}.}
\end{proof}

\section{Connections To Causal Inference}
\label{app:sec:id}

\begin{lemma} \label{lem:proj_preserves_sep}
Let $\G(L \cup V, W)$ be a CADMG.  
The m-separations in $\G(L \cup V, W)^{|W}$ amongst vertices in $V \cup W$ are the same as those in $\G(V, W)^{|W} = \sigma_L(\G(L \cup V, W)^{|W})$.
%For disjoint
%$A,B,C \subseteq V \cup W$, ($C$ may be empty), $A$ is m-separated from $B$ given $C$
%in $\G(H \cup V, W)^{|W}$ if and only if $A$ is m-separated from $B$ given
%$C$ in \G(H \cup V, W)^{|W}$.
\end{lemma}

\begin{proof}
The graphs $\G(V, W)^{|W}$ and $\sigma_L(\G(L \cup V, W)^{|W})$ are ADMGs, and since the 
former is a latent projection of the latter, 
the result follows by standard results on m-connection.
\end{proof}

\begin{cora}{\ref{cor:dag_in_nested}}
%already mention we will be using an ADMG \G(V) with a vertex set V throughout.  Just need to point the ADMG is a DAG here.
%Let $\G$ be a DAG with a vertex set $V$.  Then
For a DAG $\G$, 
% cleaner and more consistent with other definitions
${\cal P}^{\rm d}(\G) = {\cal P}^{\rm n}(\G)$.
\end{cora}

\begin{proof}
($\Rightarrow$) follows by Lemma \ref{lem:reachable-factorizes}.
($\Leftarrow$) follows by Theorem
\ref{thm:global-reachable-factorization} (with $R=V$) since ${\mathcal D}(\G) = \{ \{v\}: v\in V\}$ and 
$\phi_{V\setminus \{v\}}(p(x_V);\G) = p(x_v\,|\,x_{\pa(v)})$.
\end{proof}

\begin{lema}{\ref{lem:fix_marg_commute_graph}}
Let $\G(L \cup V, W)$ be a CDAG.
Assume $v \in V$ is fixable in $\G(V, W) =
\sigma_L(\G(L \cup V, W))$.
Then
%$v$ is fixable in $\G(H \cup V, W)$, and moreover,
$\sigma_L ( \phi_v(\G(L \cup V, W) ) ) =
	\phi_v ( \sigma_L(\G(L \cup V, W) ) )$.\footnote{
Note that $v$ is fixable in $\G(L \cup V, W)$ since this graph
has no bidirected edges, and thus all random vertices are fixable.}
That is, the following commutative diagram holds:
\begin{center}
\begin{tikzpicture}
\node (og) {$\G(L \cup V, W)$};
\node (fg) [below of=og, yshift=-1.2cm] {$\G((L \cup V)
	\setminus \{ v \}, W \cup \{ v \})$};
\node (pg) [right of=og, xshift=5.0cm] {$\G(V, W)$};
\node (fpg) [right of=fg, xshift=5.0cm] {$\G(V \setminus \{ v \},
	W \cup \{ v \})$};

\draw[->] (og) to node[left] {$\phi_v$} (fg);
\draw[->] (og) to node[above] {$\sigma_L$} (pg);
\draw[->] (fg) to node[below] {$\sigma_L$} (fpg);
\draw[->] (pg) to node[right] {$\phi_v$} (fpg);
\end{tikzpicture}
\end{center}
\label{dummy}
\end{lema}

\begin{proof}
%This follows by definition of fixing and latent projections.
Both $\sigma_L ( \phi_v(\G(L \cup V, W) ) )$ and
$\phi_v ( \sigma_L(\G(L \cup V, W) ) )$ have the same set of random
vertices $V \setminus \{ v \}$ and fixed vertices $W \cup \{ v \}$.

Consider the set of edges $E$ in $\sigma_L( \G(L \cup V, W)) = \G(V, W)$.
The set of edges $E'$ in $\phi_v ( \sigma_L(\G(L \cup V, W) ) )$
is a subset of $E$ containing all edges not having an arrowhead at $v$.
Now let $\pi$ be the set of  paths in
$\G(L \cup V, W)$, where both endpoints are in $V \cup W$ and all non-endpoints are non-colliders in $L$.
These paths d-connect marginally (i.e.\ given $\emptyset$).
Similarly, let $\pi'$ be the set of paths in
$\phi_v(\G(L \cup V, W)) = \G((L \cup V) \setminus \{v\},
W\cup\{v\})$, where both endpoints are in $V \cup W$ and all non-endpoints are non-colliders in $L$.  $\pi'$ is the subset of $\pi$ 
formed by removing paths containing an edge with an arrowhead at $v$ 
(note that since $v \notin L$, $v$ can only occur as an endpoint).

By definition of latent projections, there is a bijection that associates each edge $e$
in $E$, with a subset of paths in $\pi$ with the same endpoints as $e$, and the same
starting and ending orientations as $e$.  These subsets partition $\pi$.
Applying $\phi_v$ to $\G(L \cup V, W)$ means that only those  paths in $\pi'$ are left in the resulting
graph. Paths in $\pi'$ are only in subsets of $\pi$ associated with edges in $E'$ (by the bijection).
Applying $\sigma_L$ to the graph then results in the edge set $E'$.  This establishes our
conclusion.
\end{proof}

\begin{lemma}
\label{lem:group-dist}
%\begin{lema}{\ref{lem:group-dist}}
Assume $q_{L \cup V}(x_{L \cup V} \cmid x_W)$ is in ${\cal P}^{\rm c}(\G(L \cup V, W))$ for a CDAG $\G(L \cup V, W)$.
Then
\vspace{-6pt}
\begin{align*}
q_{L \cup V}(x_{V} \cmid x_W)
&=
	\prod_{D \in {\cal D}(\G(V, W))}
\left(
	\sum_{x_{L_D}}
	\prod_{a \in D \cup L_D} q_{L \cup V}(x_a \cmid x_{\pa_{\G(L \cup V,W)}(a)})
\right)\\[2pt]
&=
	\prod_{D \in {\cal D}(\G(V, W))}
\left(
	\prod_{a \in D} q_{L \cup V}(x_a \cmid x_{\pre_{\prec,\G(V,W)}(a)})
\right)
\end{align*}
where $L_D = \an_{\G(L \cup V, W)_{D \cup L}}(D) \cap L$, and $\prec$ is any topological ordering for $\G(V,W)$.
%\end{lema}
\end{lemma}
In other words, the margin over the observed variables
factorizes into separate kernels for each district of
the latent projection.
\begin{proof}
Simple extension of the proof for $W = \emptyset$ found in \citep{tian02on}.
\end{proof}

\begin{lema}{\ref{lem:fix_marg_commute_kernel}}
Let $\G(L \cup V, W)$ be a CDAG, and
assume 
$q_{L \cup V}(x_{L \cup V} \cmid x_W) \in {\cal P}^{\rm c}_{\rm f}(\G(L \cup V, W))$.
Assume $v \in V$ is fixable in $\G(V, W) =
\sigma_L(\G(L \cup V, W))$.
Then
\[
\sum_{x_L}\phi_v(q_{L \cup V}(x_{L \cup V} \cmid x_W); \G(L  \cup V, W)) =
\phi_v(q_{L \cup V}(x_{V} \cmid x_W); \sigma_L(\G(L \cup V, W))).
\]
%$v$ is fixable in $\G(L \cup V, W)$, and moreover,
%%$\sigma_L ( \phi_v(\G(L \cup V, W) ) ) =
%%	\phi_v ( \sigma_L(\G(L \cup V, W) ) )$.
%\footnote{Note that $v$ is fixable in $\G(L \cup V, W)$ since this graph
%has no bidirected edges, and thus all random vertices are fixable.}
%That is, the commutative diagram in Figure {\rm \ref{fig:commute-diagram}(a)} holds.
%Let $\G(L \cup V, W)$ be a CDAG, and assume
%$q_{L \cup V}(x_{L \cup V} \cmid x_W) \in {\cal P}^{\rm c}_{\rm f}(\G(L \cup V, W))$.
%Assume $v \in V$ is fixable in
%$\G(V, W) = \sigma_L(\G(L \cup V, W))$.  Then
%\[
%\sum_{x_L}\phi_v(q_{L \cup V}(x_{L \cup V} \cmid x_W); \G(L  \cup V, W)) =
%\phi_v(q_{L \cup V}(x_{V} \cmid x_W); \sigma_L(\G(L \cup V, W))).
%\]
That is, the following commutative diagram holds:
\begin{center}
\begin{tikzpicture}
\node (og) {$q_{L \cup V}(x_{L \cup V} \cmid x_W)$};
\node (fg) [below of=og, yshift=-1.2cm] {$q_{(L \cup V)\setminus\{v\}}
	(x_{(L \cup V)\setminus \{v\}} \cmid x_{W\cup\{v\}})$};
%\node (fg-p) [below of=fg, yshift=0.5cm] {$\equiv \phi_v(q_{H \cup V}(x_{H \cup V} \cmid x_W); \G(H\cup V,W))$};

\node (pg) [right of=og, xshift=5.0cm] {$q_{L \cup V}(x_V \cmid x_W)$};
\node (fpg) [right of=fg, xshift=5.0cm] {$q_{V\setminus\{v\}}
	(x_{V\setminus \{v\}} \cmid x_{W\cup\{v\}})$};
%\node (fpg-p) [right of=fg-p, yshift=5.0cm] {$\equiv \phi_v(q_{H \cup V}(x_V \cmid x_W); \G(V,W))$};

\draw[->] (og) to node[left] {$\phi_v(.;\G(L\cup V, W))$} (fg);
\draw[->] (og) to node[above] {$\sum_{x_L}$} (pg);
\draw[->] (fg) to node[below] {$\sum_{x_L}$} (fpg);
\draw[->] (pg) to node[right] {$\phi_v(.;\G(V, W))$} (fpg);
\end{tikzpicture}
\end{center}

\end{lema}

\begin{proof}
$\phi_v \to \Sigma_{x_{L}}$ direction:

Since $q_{L \cup V}(x_{L \cup V} \cmid x_W) \in {\cal P}^{\rm c}_{\rm f}(\G(L \cup V, W))$,
we have
\[
q_{L \cup V}(x_{L \cup V} \cmid x_W) =
	\prod_{a \in L \cup V}
	q_{L \cup V}(x_a \cmid x_{\pa_\G(a)}).
\]
This implies by Lemma \ref{lem:fixing_in_dags} that
\[
\phi_v(q_{L \cup V}(x_{L \cup V} \cmid x_W); \G(L \cup V, W)) =
	\prod_{a \in (L \cup V) \setminus \{ v \}}
	q_{L \cup V}(x_a \cmid x_{\pa_\G(a)}),
\]
which implies 
$\phi_v(q_{L \cup V}(x_{L \cup V} \cmid x_W); \G(L \cup V, W)) \in
{\cal P}^{\rm c}_{\rm f}(\phi_v(\G(L \cup V, W)))$.

Then by Lemma \ref{lem:group-dist},
\begin{align}
\sum_{x_L}
\phi_v(q_{L \cup V}(x_{L \cup V} \cmid x_W); \G(L \cup V, W)) =
	\prod_{D \in {\cal D}(\phi_v(\G(V, W)))}
\left(
	\prod_{a \in D} q_{L \cup V}(x_a \cmid x_{\pre_{\prec,\G}(a)})
\right).
\label{eq:fix-sum-terms}
\end{align}

$\Sigma_{x_{L}} \to \phi_v$ direction:

Similarly, by Lemma \ref{lem:group-dist},
\[
\sum_{x_L}
q_{L \cup V}(x_{L \cup V} \cmid x_W) = q_{L \cup V}(x_V \cmid x_W) =
	\prod_{D \in {\cal D}(\G(V, W))}
\left(
	\prod_{a \in D} q_{L \cup V}(x_a \cmid x_{\pre_{\prec,\G}(a)})
\right)
\]

Now let $D^v$ be the element of ${\cal D}(\G(V,W))$ such that $v \in D^v$.  Then by Proposition \ref{prop:fix-dist-marg},
\begin{align}
\notag
\phi_v(q_{L \cup V}(x_V \cmid x_W); \G(V,W)) =&
\left(
	\prod_{a \in D^v \setminus \{ v \}} q_{L \cup V}(x_a \cmid x_{\pre_{\prec,\G}(a)})
\right) \cdot\\
&	\prod_{D \in {\cal D}(\G(V, W)) \setminus \{D^v\}}
\left(
	\prod_{a \in D} q_{L \cup V}(x_a \cmid x_{\pre_{\prec,\G}(a)})
\right)
\label{eq:sum-fix-terms}
\end{align}

Since, by Proposition \ref{prop:districts-after-fixing}, ${\cal D}(\phi_v(\G(V, W))) = ({\cal D}(\G(V, W)) \setminus \{D^v\})\, {\dot{\cup}}\, {\cal D}({\cal G}(V,W)_{D^v \setminus \{ v \}})$,
the right hand sides of (\ref{eq:fix-sum-terms}) and (\ref{eq:sum-fix-terms}) are equal.
%Since a topological order $\prec$ in ${\cal G}$ is also topological in ${\cal G}_{D^r \setminus \{ r \}}$, we can 
\end{proof}

\begin{thma}{\ref{thm:dags_in_nested}}
%\begin{align*}
Let $\G(V \cup L)$ be a DAG.  Then
\begin{align*}
p(x_{V \cup L}) \in {\cal P}^{\rm d}(\G(V \cup L))\quad \Rightarrow \quad
p(x_{V}) \in {\cal P}^{\rm n}(\G(V)).  
\end{align*}
%$p(x_{V \cup L}) \in {\cal P}^{\rm d}(\G(V \cup L)) \Rightarrow
%p(x_{V}) \in {\cal P}^{\rm n}(\G(V))$.
%\end{align*}
\end{thma}
\begin{proof}
Assume $p(x_{V \cup L}) \in {\cal P}^{\rm d}(\G(V \cup L))$, and for a set
$R$ reachable in $\G(V)$ with $A \subseteq R$ and $B,C \subseteq V$ ($C$
possibly empty), suppose that $A$ is m-separated from $B$ given $C$ in
%$\phi_{V \setminus R}(\G(V, W))^{|V\setminus R}$.
$\phi_{V \setminus R}(\G(V))^{|V\setminus R}$.

By an inductive application of Lemma \ref{lem:fix_marg_commute_graph},
%$\phi_{V \setminus R}(\G(V, W))$
$\phi_{V \setminus R}(\G(V))$
is a latent projection of
%$\phi_{(L \cup V)\setminus R}(\G(L \cup V, W))$.
$\phi_{(L \cup V)\setminus R}(\G(L \cup V))$.
Therefore, by Lemma
\ref{lem:proj_preserves_sep}, $A$ is m-separated from $B$ given $C$ in
%$\phi_{(L \cup V) \setminus R}(\G(L \cup V, W))^{|(L \cup V)\setminus R}$.
$\phi_{(L \cup V) \setminus R}(\G(L \cup V))^{|(L \cup V)\setminus R}$.
Our assumption, and Corollary \ref{cor:dag_in_nested} then imply
$X_A \ci X_B \mid X_C$ holds in the kernel $\phi_{(L \cup V)\setminus R}(p(x_{L \cup V}); \G(L \cup V))$.
%\[
%\left(
%A \ci B \mid C
%\right)_{[\phi_{(L \cup V)\setminus R}(p(x_{L \cup V}); \G(L \cup V))]}
%\]
%holds.
By an inductive application of Lemma \ref{lem:fix_marg_commute_kernel},
\[
\sum_{x_L} \phi_{(L \cup V)\setminus R}(p(x_{L \cup V}); \G(L \cup V))
=
\phi_{V \setminus R}(p(x_V); \G(V))
\]
and thus $X_A \ci X_B \mid X_C$ holds in 
$\phi_{V\setminus R}(p(x_{V}); \G(V))$.
%$\left(
%A \ci B \mid C
%\right)_{[\phi_{V\setminus R}(p(x_{V}); \G(V))]}$ also holds.  
Our conclusion follows.
\end{proof}

\begin{lemma} \label{lem:reachable-id}
Let $\G(L \cup V)$ be a hidden variable causal DAG.  For any set $S$ reachable
from $\G(V)$, the interventional distributions $p(x_S \cmid \Do_{\G(L \cup V)}(x_{V \setminus S}))$ 
are identifiable from $p(x_V)$ by the kernel $\phi_{V \setminus S}(p(x_V);\G(V))$, which depends only on
% evaluated at
$x_S$ and $x_{\pa_{\G(V)}(S) \setminus S}$.
\end{lemma}
\begin{proof}
%Fix any DAG $\G^{\dag}$ with a vertex set $V \cup H$ such that the latent projection of $\G^{\dag}$ is $\G$.
%Then 
Our conclusion follows by (\ref{eq:g-formula-fix}) and an inductive application of
Lemma \ref{lem:fix_marg_commute_kernel}.
That the kernel $\phi_{V \setminus S}(p(x_V);\G)$ only depends on $x_S$ and $x_{\pa(S) \setminus S}$
follows by the global nested Markov property, and Theorem \ref{thm:dags_in_nested}.
\end{proof}

%\ilya{[change overbar A, also in the proof later]}
\begin{lemma} \label{lem:canonical-do}
Let $\G(L \cup V)$ be a hidden variable causal DAG.
%For any $A \subseteq V$, let $\G_{\overline{A}}$ be the edge subgraph of $\G$ in which all directed arrows in $\G$ into $A$ are removed.
For any $Y \subseteq V\setminus A$, let
%$A_Y = \an_{\G_{\overline{A}}}(Y) \cap A$.
$A_Y = \an_{\phi^*_{A}(\G)}(Y) \cap A$.
Then
\[
p(x_Y \cmid \Do_{\G}(x_A)) = p(x_Y \cmid \Do_{\G}(x_{A_Y})).
\]
\end{lemma}
\begin{proof}
Follows by (\ref{eq:g-formula}) and the global Markov property for CDAGs.
\end{proof}

\begin{thma}{\ref{thm:1-line-id}}
Let $\G(L \cup V)$ be  a causal DAG with latent projection $\G(V)$.  For $A\dotcup Y\subseteq V$, 
let $Y^* = \an_{\G(V)_{V \setminus A}}(Y)$.  Then if ${\cal D}(\G(V)_{Y^*}) \subseteq {\cal I}(\G(V))$,
\begin{align}
p(x_Y \cmid \Do_{\G(L \cup V)}(x_A)) &= \sum_{x_{Y^* \setminus Y}} \prod_{D \in {\cal D}(\G(V)_{Y^*})} p(x_D \cmid \Do_{\G(L \cup V)}(x_{\pa_{\G(V)}(D)\setminus D})) \nonumber\\
&= \sum_{x_{Y^* \setminus Y}} \prod_{D \in {\cal D}(\G(V)_{Y^*})} \phi_{V \setminus D}(p(x_V);\G(V)).
\tag{\ref{eqn:1-line-id}}
\end{align}
If not, there exists $D \in {\cal D}(\G(V)_{Y^*})$ that is not intrinsic in $\G(V)$, and $p(x_Y \mid \Do_{\G(L \cup V)}(x_A))$ is not identifiable in $\G(L \cup V)$.
\end{thma}

%\begin{thma}{\ref{thm:1-line-id}}
%Let $\G(L \cup V)$ be  a causal DAG with latent projection $\G(V)$.  For $A\dotcup Y\subset V$, 
%let $Y^* = \an_{\G(V)_{V \setminus A}}(Y)$.  Then if ${\cal D}(\G(V)_{Y^*}) \subseteq {\cal I}(\G(V))$,
%\begin{align}
%p(x_Y \cmid \Do_{\G(L \cup V)}(x_A)) = \sum_{x_{Y^* \setminus Y}} \prod_{D \in {\cal D}(\G(V)_{Y^*})} \phi_{V \setminus D}(p(x_V);\G(V)).
%\tag{\ref{eqn:1-line-id}}
%\end{align}
%If not, there exists $D \in {\cal D}(\G(V)_{Y^*}) \setminus {\cal I}(\G(V))$ and $p(x_Y \cmid \Do_{\G(L \cup V)}(x_A))$ is not identifiable in $\G(L \cup V)$.
%\end{thma}

\begin{proof} We first prove (\ref{eqn:1-line-id}).
%Fix any DAG $\G^{\dag}$ with a vertex set $V \cup H$ such that the latent projection of $\G^{\dag}$ is $\G$.
Let $A^* = V \setminus Y^*\supseteq A$.
By Lemma \ref{lem:canonical-do} we have 
\[
p(x_{Y^*} \cmid \Do_{\G(L \cup V)}(x_A)) = p(x_{Y^*} \cmid \Do_{\G(L \cup V)}(x_{A^*})).
\]
%\textcolor{red}{
%\begin{align*}
%p(x_{Y^*} \cmid \Do_{\G(H \cup V)}(x_{A^*})) &= \sum_{x_H} \prod_{v \in H \cup Y^*} p(x_v \cmid x_{\pa_{\G(H \cup V)}(v)})\\
%&= \sum_{x_H} \phi_{A^*}(p(x_{H \cup V}); \G(H \cup V)) \equiv \sum_{x_H} q_{H \cup Y^*}(x_{H \cup Y^*} \cmid x_{A^*})
%\end{align*}}
%, and since $\G(V)_{Y^*} = \sigma_H(\G(H\cup V))_{Y^*} = \phi^*_{Y^*}(\sigma_H(\G(H \cup V)))$,  $(\G^*(Y^*,A^*))_{Y^*} = (\G(V)_{Y^*})$.}
%\textcolor{red}{

Let $\G^*(L \cup (V \setminus A^*), A^*) = \phi_{A^*}(\G(L \cup V))$; note that since $\G(L\cup V)$ is a DAG, $A^*$ is fixable in
$\G(L \cup V)$, and $\G^*(L \cup (V \setminus A^*), A^*)$ is a CDAG.
By Corollary \ref{cor:mut_marg_commute_graph}, $\sigma_L(\phi^*_{A^*}(\G(L \cup V))) = \phi^*_{A^*}(\sigma_L(\G(L \cup V)))$,
where $\sigma_L$ is the latent projection operation, that is $\sigma_L(\G(L \cup V)) = \G(V)$.
Since $\G^*(Y^*,A^*) = \sigma_L(\phi^*_{A^*}(\G(L \cup V))) = \phi^*_{A^*}(\G(V))$. %, and, 
Then, by definition of induced subgraphs and $Y^*$,
$\G(V)_{Y^*} = (\phi^*_{A^*}(\G(V)))_{Y^*}$, and $\G(V)_{Y^*} = \G^*(Y^*,A^*)_{Y^*}$.
%}
%\textcolor{red}{POSSIBLY REPLACE WITH NEW COMMUTING LEMMA}\par
%\textcolor{red}{Note that $(\G^*(Y^*,A^*))_{Y^*} = (\G(V)_{Y^*})$ since, by Definition \ref{def:proj_cadmg},
%since in both $\G^*(Y^*,A^*)$ and $\G(V)$ the presence of an edge between $y_1,y_2 \in Y^*$
%is determined by the presence of paths in $\G^*(H\cup (V\setminus A^*),A^*)$ and $\G(H\cup V)$, while the graphs only differ
%in that the former graph does not contain edges pointing into $A$.  The argument mirrors the proof in Lemma \ref{lem:fix_marg_commute_graph}
%.}
Consequently, ${\cal D}(\G(V)_{Y^*}) = {\cal D}(\G^*(Y^*,A^*))$.

For every $D \in {\cal D}(\G^*(Y^*,A^*))$, define $L_D \equiv L \cap \an_{\G(L \cup V)_{D \cup L}}(D)$, and let
$L^* = \bigcup_{D \in {\cal D}(\G^*(Y^*,A^*))} L_D$. Thus $L_D$ is the set of variables $h \in L$, for which there exists a vertex
$d \in D$ and a directed path $h \rightarrow \cdots \rightarrow d$ in $\G(L\cup V)$ on which, excepting $d$, all vertices are in $L$.

It follows from the construction that:
\begin{itemize}
\item[(a)] if $D,D^\prime \in {\cal D}(\G^*(Y^*,A^*))$, and $D \neq D^\prime$ 
 then $L_D \cap L_{D^\prime}\!=\! \emptyset$; 
 \item[(b)] for each $D \in \mathcal{D}(\G^*(Y^*,A^*))$ we have $\pa_{\G(L \cup V)}(D \cup L_D) \cap L^* = L_D$;
 % if $v \in H_D$, then $\pa_{\G(H \cup V)}(v) \subseteq D \cup \pa_{\G^*(Y^*,A^*)}(D) \cup H_D$; 
 \item[(c)]
 $Y^* \cup L^*$ is ancestral in $\G(L \cup V)_{(L \cup V)\setminus A}$, so if 
$v \in Y^* \cup L^*$,  $\pa_{\G(L \cup V)}(v)\cap L\subseteq L^*$.
\end{itemize}
%It then follows that
%\begin{align}\label{eq:get-rid-of-floating-h}
%\sum_{x_{H}} \phi_{V^*}(p(x_{{H} \cup V}); \G({H} \cup V))
%=xffo
%\sum_{x_{H^*}} \phi_{V^*}(p(x_{{H^*} \cup V}); \G({H^*} \cup V)),
%\end{align}
%for any $V^* \subseteq V$, since we may place a topological order on the vertices in $H\setminus H^*$
%and successively marginalize vertices with no children that have not already been marginalized.
%Note that the latent projection $\G(H^* \cup V)$ after removing the latent variables
%in $H\setminus H^*$  is simply the induced subgraph $\G(H \cup V)_{H^* \cup V}$, which is also therefore a DAG.
We now have:
\begin{align}
\MoveEqLeft{p(x_{Y^*} \cmid \Do_{\G(L \cup V)}(x_{A^*}))}\nonumber\\
&=
\sum_{x_L} \prod_{v \in L \cup Y^*} p(x_v \cmid x_{\pa_{\G(L \cup V)}(v)})\nonumber\\
&=
\sum_{x_{L^*}} \prod_{v \in L^* \cup Y^*} p(x_v \cmid x_{\pa_{\G(L \cup V)}(v)})
\underbrace{\sum_{x_{L \setminus L^*}} \prod_{v \in L\setminus L^*} p(x_v \cmid x_{\pa_{\G(L \cup V)}(v)})}_{=1}
\nonumber\\
&=
\sum_{x_{L^*}}\prod_{D \in {\cal D}(\G^*(Y^*,A^*))} \prod_{v \in D \cup L_D} p(x_v \cmid x_{\pa_{\G(L \cup V)}(v)})\nonumber\\
%&=
%\sum_{x_{H^*\setminus H_{D_0}}}\sum_{x_{H_{D_0}}}\left(\prod_{D \in {\cal D}(\G^*(Y^*))\setminus\{D_0\}} \prod_{v \in D \cup H_D} p(x_v \cmid x_{\pa_{\G(H \cup V)}(v)})\right) \prod_{v \in D_0 \cup H_{D_0}} p(x_v \cmid x_{\pa_{\G(H \cup V)}(v)})\nonumber\\
%&=
%\sum_{x_{H^*\setminus H_{D_0}}}\left(\prod_{D \in {\cal D}(\G^*(Y^*))\setminus\{D_0\}} \prod_{v \in D \cup H_D} p(x_v \cmid x_{\pa_{\G(H \cup V)}(v)})\right) \sum_{x_{H_{D_0}}} \prod_{v \in D_0 \cup H_{D_0}} p(x_v \cmid x_{\pa_{\G(H \cup V)}(v)})\nonumber\\
%&=
%\left( \sum_{x_{H_{D_0}}} \prod_{v \in D_0 \cup H_{D_0}} p(x_v \cmid x_{\pa_{\G(H \cup V)}(v)}) \right)\sum_{x_{H^*\setminus H_{D_0}}}\left(\prod_{D \in {\cal D}(\G^*(Y^*))\setminus\{D_0\}} \prod_{v \in D \cup H_D} p(x_v \cmid x_{\pa_{\G(H \cup V)}(v)})\right)\nonumber\\[4pt]
%&\kern40pt\vdots\kern120pt\vdots\nonumber\\
&=
\prod_{D \in {\cal D}(\G^*(Y^*,A^*))}\left( \sum_{x_{L_D}} \prod_{v \in D \cup L_D} p(x_v \cmid x_{\pa_{\G(L \cup V)}(v)})\right). \label{eq:ystar-into-districts}%
\end{align}
%here $D_0$ is an an arbitrary district.
Here, the first equality follows from (\ref{eq:g-formula-fix}), the second follows from (c), the third from (a), and the fourth from (b).
%
%{\it NOTE: I think the lines involving $D_0$ might be cut, but if so, some text needs to be added to explain why this is kosher.}
%
Now, for any given $D \in {\cal D}(\G^*(Y^*,A^*))$,
\begin{align}
\MoveEqLeft{\sum_{x_{L_D}} \prod_{v \in D \cup L_D} p(x_v \cmid x_{\pa_{\G(L \cup V)}(v)})}\nonumber\\
&= \sum_{x_{L_D}} \prod_{v \in D \cup L_D} p(x_v \cmid x_{\pa_{\G(L \cup V)}(v)})
\underbrace{\sum_{x_{L\setminus L_D}} \prod_{v \in L\setminus L_D} p(x_v \cmid x_{\pa_{\G(L \cup V)}(v)})}_{=1}\nonumber\\[-4pt]
&= \sum_{x_{L}} \prod_{v \in D \cup L} p(x_v \cmid x_{\pa_{\G(L \cup V)}(v)})\nonumber\\
&= \sum_{x_{L}} \phi_{V \setminus D}(p(x_{{L} \cup V}); \G({L} \cup V)).\label{eq:id-soundness-second-step}
\end{align}
Here the second line uses that $\pa_{\G(L \cup V)}(D \cup L_D) \cap (L\setminus L_D)\!=\!\emptyset$, 
which follows (b), (c) and the definition of $L_D$.
Since, by hypothesis, $D \in {\cal D}(\G(V)_{Y^*})={\cal D}(\G^*(Y^*,A^*))\subseteq {\cal I}(\G(V))$, it follows from Lemma \ref{lem:fix_marg_commute_kernel}
that 
\begin{equation}
\sum_{x_{L}} \phi_{V \setminus D}(p(x_{{L} \cup V}); \G({L} \cup V)) = 
\phi_{V \setminus D}(p(x_{V}); \G(V)).\label{eq:finally-we-apply-our-commuting-lemma}
\end{equation}
Hence by (\ref{eq:ystar-into-districts}), (\ref{eq:id-soundness-second-step}) and (\ref{eq:finally-we-apply-our-commuting-lemma}),
\[
p(x_{Y^*} \cmid \Do_{\G(L \cup V)}(x_{A^*})) = \prod_{D \in {\cal D}(\G^*(Y^*,A^*))} \phi_{V \setminus D}(p(x_{V}); \G(V)).
\]
%
%First equality is as before, second equality is because of definition of $H_D$ -- for
%any $D$, we can bring in all other $H_{D'}$ (this I think avoids the issue you are
%worried about, since all "floating H" are now gone from the derivation).? Now we use
%Lemma 54 inductively on the last expression to get rid of $H^*$ entirely.
The conclusion, (\ref{eqn:1-line-id}), then follows since 
\begin{align*}
p(x_Y \cmid \Do_{\G(L \cup V)}(x_A)) = \sum_{x_{Y^* \setminus Y}} p(x_{Y^*} \,|\, \Do_{\G(L \cup V)}(x_{A^*})).
%&= \sum_{x_{Y^* \setminus Y}} \prod_{D \in {\cal D}(\G^*(Y^*))} \phi_{V \setminus D}(p(x_{V}); \G(V)).
\end{align*}

To establish the last claim, fix $D \in {\cal D}(\G(V)_{Y^*}) \setminus {\cal I}(\G(V))$, 
and let $D^* \equiv \langle D \rangle$ be the minimal intrinsic superset of $D$.
Assume, for contradiction, that $D^*$ does not intersect $A^*$.  Then $D^* \subseteq Y^*$.  
But since $D^*$ is intrinsic, it must
be a subset of some $D' \in {\cal D}(\G(V)_{Y^*})$.  But this is impossible since $D \subsetneq D^*$, and
$D \in {\cal D}(\G(V)_{Y^*})$.  Thus $D^*$ intersects $A^*$.

Let $R = \{ v \in D^* \mid \ch_{\phi_{V \setminus D^*}(\G(V))}(v) = \emptyset \}$.  If $R$ is not a subset of $D$,
$D^*$ could not be the minimal intrinsic superset of $D$, since any element in $R \setminus D$ is fixable
in $\phi_{V \setminus D^*}(\G(V))$. Finally, note that by construction, $D \subsetneq D^*$,
$D \cap A^* = \emptyset$, and
%$R \subseteq Y^* \subseteq \an_{\G(V)_{\overline{A}}}(Y)$.
$R \subseteq Y^* \subseteq \an_{\phi^*_{A}(\G(V))}(Y)$.
This implies $D$ and $D^*$ satisfy the definition of a \emph{hedge} for $p(x_{Y} \cmid \Do_{\G(L \cup V)}(x_{A^*})) =
p(x_{Y} \cmid \Do_{\G(L \cup V)}(x_A))$ in $\G(V)$ \citep{shpitser06id}.
%\textcolor{red}{%

Let $Y_D$ be the smallest subset of $Y$ such that $D \subseteq \an_{\G(V)_{V \setminus A}}(Y_D)$.
By Theorem 4 in \citep{shpitser06id}, $p(x_{Y_D} \cmid \Do_{\G(L^\dag \cup V)}(x_{\pa_{\G(V)}(D) \setminus D}))$ is not identifiable
in a canonical hidden variable causal DAG $\G(L^\dag \cup V)$, where $L^\dag$ consists of a hidden variable for
every bidirected arc in $\G(V)$ (see, e.g.~\citet{richardson:2002} \S6 for a formal definition of $\G(L^\dag \cup V)$).
Counterexamples witnessing this non-identification can be easily extended to counterexamples
witnessing non-identification of $p(x_{Y} \cmid \Do_{\G(L \cup V)}(x_{A}))$ in $\G(L \cup V)$ as follows.
First, we restrict attention to
%by restricting to
a submodel of the causal model represented by $\G(L^\dag \cup V)$ where
\[
p(x_{Y} \cmid \Do_{\G(L^\dag \cup V)}(x_{\pa_{\G(V)}(D) \setminus D})) = p(x_{Y} \cmid \Do_{\G(L^\dag \cup V)}(x_{A})).
\]
%Results in the same paper show that $p(x_{Y} \cmid \Do_{\G(H \cup V)}(x_A))$ is not identifiable
%in a canonical hidden variable causal DAG $\G(H^\dag \cup V)$, where $H^\dag$ consists of a hidden variable for every bidirected arc in $\G(V)$ (see, e.g.~ \citet{richardson:2002} \S6 for a formal definition of $\G^\dag(H^\dag \cup V)$).
We then note that it follows from Theorem 2 in \citep{evans:complete} that the model associated with  $\G(L^\dag \cup V)$ is a submodel of that
associated with  $\G(L \cup V)$. 
This immediately implies $p(x_{Y} \cmid \Do_{\G(L \cup V)}(x_{\pa_{\G(V)}(D) \setminus D}))$ is also not identifiable in
$\G(L \cup V)$.
This completes the proof.%
%}
\end{proof}

\begin{cora}{\ref{cor:admg-id}}
Let $\G^1(L^1 \cup V)$ and $\G^2(L^2 \cup V)$ be two causal DAGs, with the same latent projection, so 
$\G^1(V) = \G^2(V)$.  Then for any $A\dot{\cup}Y \subseteq V$:
\begin{itemize}
\item[{\rm (i)}] $p(Y \cmid \Do_{\G^1}(A))$ is identified if and only if $p(Y \cmid \Do_{\G^2}(A))$ is identified;
\item[{\rm (ii)}] if $p(Y \cmid \Do_{\G^1}(A))$ is identified, then $p(Y \cmid \Do_{\G^1}(A)) = %f(p(V)) =
p(Y \cmid \Do_{\G^2}(A))$.
\end{itemize}
\end{cora}
\begin{proof}
Follows directly by Theorem \ref{thm:1-line-id}, since $Y^*$, ${\cal D}(\G(V)_{Y^*})$,  ${\cal I}(\G(V))$ and
the terms on the RHS of (\ref{eqn:1-line-id}) are
defined solely in terms of the latent projection.
\end{proof}

\begin{algorithm}[p!]
\algorithmicrequire{\hspace{0.2cm}\parbox[t]{.7\linewidth}{
$\G(V)$ : an ADMG over a vertex set $V$,\\
$p(x_V)$ : a density over $x_V$,\\
$\prec$ : a total topological ordering on $V$.}}\\

\algorithmicensure{\hspace{0.2cm}\parbox[t]{.7\linewidth}{
A list of constraints on $p(x_V)$ implied by $\G(V)$.}}
\begin{algorithmic}[1]
\Procedure{Find-Constraints}{$\G(V), p(x_V), \prec$}
\State ${\bf L} \gets \{ \}$
\ForAll{$v \in V$}
	\State $\Pre_v \gets \pre_{\G(V), \prec}(v) \cup \{v\}$;
%		, $\G^{(i)} = \G(T)$;
	\State {\bf let} $S \in {\cal D}(
	%\G(T)
	\G[\Pre_v]
	)$ be the unique district s.t.~$v \in S$;
	\If{$\Pre_v \setminus (\mb(v, \Pre_v)\cup\{v\}) \neq \emptyset$}
	\State ${\bf L} \gets {\bf L} \cup
		``
%\begin{array}{cc}
(X_{v} \ci X_{\Pre_v
		\setminus (\mb(v, \Pre_v)\cup\{v\})} \mid
		X_{\mb(v, \Pre_v)} )
		%\text{ in} \\
\, [%q_T
p(x_{\Pre_v})]$'';
%\phi_{T \setminus S}(p(x_{T}); \G(T)) \\
%\end{array}

	\EndIf
	\State ${\bf L} \gets {\bf L} \cup
	\Call{Node-Constraints}{v, 
	\phi_{\Pre_v \!\setminus S}(%\G(T)
	\G[\Pre_v]
	),
	\phi_{\Pre_v \! \setminus S}(p(x_{\Pre_v});
		%\G(T)
		\G[\Pre_v]
		})$;
\EndFor
\State \textbf{return} {\bf L}.
\EndProcedure
\end{algorithmic}
\caption{The constraint-finding algorithm in \citep{tian02on} expressed
in the CADMG and kernel notation used in this manuscript.
$v$ is an element of $V$.}
\label{alg:tian}
\end{algorithm}

\begin{algorithm}[p!]
\algorithmicrequire{\hspace{0.2cm}\parbox[t]{.7\linewidth}{
$\G(S, W)$ : a CADMG,\\
$v$ a vertex in $S$ with no children in $\G$,\\
$q_{S}(x_{S} \cmid x_W)$ : a kernel associated with $\G$, \\
}}\\
\algorithmicensure{\hspace{0.2cm}\parbox[t]{.7\linewidth}{
A list of constraints on $q(x_{S} \cmid x_W)$.}}
\begin{algorithmic}[1]
\Procedure{Node-Constraints}{$v, \G(S, W), q_{S}$}
\State ${\bf L} \gets \{\}$;
\For{ every $\emptyset \subset D \subset S$
	closed under descendants in $\G$, s.t.
	$v \not\in D$}
	\State {\bf let} $D' \gets S \setminus D$; $q_{D'} \gets \phi_{D}(q_{S}; \G(S, W))$;
	\If{$
(\pa_\G(S) \setminus S) \setminus \pa_\G(D')
\neq \emptyset$}
	\State ${\bf L} \gets {\bf L}
		\cup
%		\parbox[t]{1\linewidth}{
		``(X_{D'} \ci X_{
(\pa_\G(S) \setminus S) \setminus \pa_\G(D')
			}
\mid X_{
\pa_\G(D') \setminus D'
%W \setminus
%((\pa_\G(S) \setminus S) \setminus \pa_\G(D'))
}
			) \, {[q_{D'}]}$'';
	\EndIf
	\State {\bf let} $E \in
		{\cal D}(
		%\phi_{S\setminus D'}(\G)
		\phi_{D}(\G(S,W))
		)$, be the unique district
		s.t.~$v \in E$;
	\If{$|{\cal D}(
	%\phi_{S \setminus D'}(\G)
	\phi_{D}(\G(S,W))
	)| > 1$
	{\bf and} $(D' \cup \pa_\G(D'))
		\setminus (\mb_\G(v, E) \cup \{ v \})
		\neq \emptyset$}
	\State ${\bf L} \gets {\bf L}
		\cup
%		\parbox[t]{1\linewidth}{
		``
%\begin{array}{c}
		(X_{v} \ci X_{(D' \cup \pa_\G(D'))
		\setminus (\mb_\G(v, E) \cup \{ v \})}
		\mid X_{\mb_\G(v, E)})
%\text{ in} \\
 \, {[q_{D'}]}
%\\
%		\end{array}
		$'';
	\EndIf
	\State ${\bf L} \gets {\bf L}$
		$\cup$
%		\parbox[t]{1\linewidth}{
		\Call{Node-Constraints}{$v,
		\phi_{S \setminus E}(\G),
%		\phi_{S \setminus E}(q_{S}; \G)
		\phi_{D' \setminus E}(q_{D'}; \phi_{D}(\G(S,W)))
		$};
	
\EndFor
\State {\bf return} {\bf L}
\EndProcedure
\end{algorithmic}
\caption{A subroutine of Algorithm \ref{alg:tian} which recursively finds constraints
associated with a particular vertex $v$.}
\label{alg:tian_sub}
\end{algorithm}

\section{Connections with Tian's Constraint Algorithm}
\label{app:sec:tian}

\citet{tian02on} gave an algorithm for deriving constraints from a latent variable model.
In this section we show that a translation of this algorithm into the kernel framework used in this paper may be used to derive a set of constraints that define the nested Markov model. However, as we illustrate below in \S \ref{sec:constraint-example} this set of constraints
is  longer and hence more redundant than the ordered local nested property. At the same time
there are also constraints arising from the global nested property that are not given by Tian's algorithm.
%In subsequent discussions,
Below, we will call the algorithm in \citep{tian02on}
``Tian's algorithm,'' and our reformulation ``Algorithm \ref{alg:tian}.''

Whereas Tian's algorithm takes a latent variable DAG $\G(L \cup V)$ as input, Algorithm \ref{alg:tian} takes an
% latent projection
ADMG $\G(V)$ (without necessarily presupposing the existence of an underlying hidden variable DAG corresponding to this ADMG ${\cal G}(V)$);
both algorithms require as inputs a topological order $\prec$ on $V$, and the observed
marginal distribution $p(V)$.
%will take three inputs modified from the three
%inputs taken by Tian's algorithm, as follows:
%rather than a hidden variable DAG
%$\G(V \cup H)$, Algorithm \ref{alg:tian} will take as input a latent
%projection $\G(V)$; a density $p(x_V)$ which is a margin of $p(x_{V \cup H})$;
%rather than one that does not necessarily correspond to a margin of any
%hidden variable DAG; 
%and a topological total ordering $\prec$ on $V$ (this
%argument remains unchanged).
Given these inputs, both Tian's algorithm, and Algorithm \ref{alg:tian} 
inductively construct a list of constraints that is initially empty. % (line 2).

The first set of steps of Tian's algorithm lists constraints for
every $v \in V$ among the set $\Pre_v \equiv \{ v \} \cup \pre_{\G(V),\prec}(v)$.
The constraints enumerated in the step (A1) are that each $v$ is independent of
variables outside its Markov blanket in $\G(\Pre_v)$, conditional on
its Markov blanket.
These steps correspond exactly to those in
Algorithm \ref{alg:tian}, which enumerates the same list of constraints.

The second part (A2) of Tian's algorithm is recursive, taking
as input a subgraph $\G'(L' \cup S)$ obtained from the original hidden variable
graph $\G(L \cup V)$, as well a q-factor $Q[S]$ obtained from the original marginal
distribution $p(V)$.  Here $L'$ is a subset of $L$ consisting of hidden variables relevant for $S$, defined as
ancestors of $S$ in the subgraph $\G'(L \cup V)_{S \cup L}$.
Each step (A2) considers constraints %involving $v$
in certain q-factors obtained from the q-factor $Q[S]$. % = q_S$.

(A2) iterates over all observable subsets $D$ of $S$ closed under
descendants,\footnote{\citet{tian02on} use $D$ for sets that are closed under descendants, \emph{not} 
to indicate a district.} and possibly adds a constraint on the margin of $Q[S]$ involving variables $D' \equiv S \setminus D$; i.e.\ 
$Q[D'] = Q[S \setminus D] = \sum_{x_D} Q[S]$.
%Note that though $v \in S$, the constraint is actually stated in terms of the margin $Q[S \setminus D]$.
%Lemma \ref{lem:graph_map} shows that this iteration corresponds to an iteration over ancestral subsets $D' = S \setminus D$ of the
%corresponding CADMG containing $v$ (shown in line 3 of Algorithm \ref{alg:tian_sub}). 
This constraint %associated with $D'$ found by step (A2)
%is described by noting 
is based on the observation
that some observed parents of $D$ in the graph $\G(V)$,
which \citet{tian02on} call ``effective parents'' may not be
parents of $D'$, which implies that the q-factor 
$q_{D'} = Q[D'] = \sum_{x_D} Q[S]$ is independent of these missing parents.
This step, outlined in the first paragraph of the description of (A2) in \citep{tian02on}, corresponds to lines 3-7 of Algorithm \ref{alg:tian_sub}.
Note that in our reformulation the marginalization operation above is instead represented by the fixing operator $\phi_{D}(q_S; \G(S, W))$.

%\citet{tian02on} describe this part of the algorithm in terms of``effective parents''
%in the DAG \ilya{$\G'(L' \cup S)$ associated with the current recursive call}; these are just the parents in the latent projection \ilya{$\G'(S)$}.
%That is to say, $q_{S}(x_{D'} \,|\, x_{\pa_\G(S) \setminus S})$ does not depend upon $x_{(\pa_\G(S) \setminus S) \setminus (\pa_\G(D') \setminus D')}$,
%or equivalently $x_{(\pa_\G(S) \setminus S) \setminus \pa_\G(D')}$.

Every call to (A2) of Tian's algorithm potentially adds another constraint associated with
the set $D'$.
This constraint is based on examining the latent projection $\G'(D')$ obtained from the DAG subgraph
$\G'(L' \cup S)$ of the original DAG $\G(L \cup V)$.  Specifically, if $\G'(D')$ has more than one district,
then $Q[D']/\sum_{x_{v}} Q[D']$ is only a function of $x_{\mb(v,E)}$, where $E$ is the district in
$\G'(D')$ containing $v$.
This step, outlined in the second paragraph of the description of (A2) in \citep{tian02on}, corresponds to lines 8-11 of Algorithm \ref{alg:tian_sub}.

%We consider $Q[D']$ associated with the subgraph 
%$\phi_{V \setminus D'}(\G(V))$; 
%if this graph has more than one district and $v \in E \in \mathcal{D}(\phi_{V \setminus D'}(\G(V)))$, 
%then
%$Q[D']/\sum_{x_{v}} Q[D']$ is a function only of $x_{\mb(v, E)}$.
%In our notation and using Lemma \ref{lem:kernel_map} it is clear that
%\begin{align*}
%Q[D']/\sum_{x_{v}} Q[D'] = q_{D'}(x_v \,|\, x_{(D' \cup \pa_\G(D')) \setminus \{v\}}) = q_{D'}(x_v \,|\, x_{(E \cup \pa_\G(E)) \setminus \{v\}}).
%\end{align*}
%Thus, the constraint added  by line 10 of Algorithm \ref{alg:tian_sub} is %therefore
%\[
%{X_v  \ci X_{(D' \cup \pa_\G(D')) \setminus
%(\mb_\G(v, E) \cup \{v\})} \mid X_{\mb_\G(v, E)}} \, [q_{S}],
%\]
%by definition of $\mb_\G(v,E)$.
Finally, Tian's algorithm is called recursively with $v$,
\begin{align*}
Q[E] =
\prod_{z \in E}
\frac{
Q[D'_z]
}{
\sum_z Q[D'_z]
},
\end{align*}
where $Q[D'_z] \equiv \sum_{D' \setminus (\{ z \} \cup \pre_{\G'(D'),\prec}(z))} Q[D']$,
and the corresponding subgraph of $\G(L \cup V)$.
The final recursive call is outlined in the third paragraph of the description of (A2) in \citep{tian02on}, and corresponds to line 12 of Algorithm \ref{alg:tian_sub}.
Note that the above formula expressing $Q[E]$ in terms of $Q[D']$ is represented by the fixing operator $\phi_{D'\setminus E}(q_{D'}; \phi_{D}(\G(S, W)))$.

Table \ref{tab:tian_vs_us} displays the mapping between objects used in steps (A1) and (A2) of Tian's algorithm, and our reformulation via
Algorithms \ref{alg:tian} and \ref{alg:tian_sub}.

\begin{table}
{
\begin{center}
\begin{tabular}{l|l|l}
%\toprule
\multicolumn{1}{c|}{objects in} & \multicolumn{1}{c|}{\multirow{2}{*}{Tian's (A2)}} & \multicolumn{1}{c}{\multirow{2}{*}{Algorithm 2}}\\[-3pt]
\multicolumn{1}{c|}{call for $S \in {\cal I}(\G)$} &  & \\
\midrule[\heavyrulewidth]
%\midrule
%\hline
input subgraph & $\G'(L' \cup S)$ & $\begin{array}{l}\G(S, V \setminus S) = \G(S, W)\\[-4pt]
 \kern35pt = \phi_{V \setminus S}(\G(V))\end{array}$ \\
\midrule
input kernel & $Q[S]$ & $\begin{array}{l}q_S(x_S \,|\, x_{V \setminus S}) = q_S(x_S \,|\, x_W)\\
%\phantom{q_S(x_S \,|\, x_{V \setminus S})} 
\kern30pt = \phi_{V \setminus S}(p(x_V); \G(V))
\end{array}$ \\
%\hline
\midrule
$\begin{array}{l}\hbox{ancestral}\\[-4pt]
\hbox{margin kernel}
\end{array}$ & 
$\begin{array}{lr}
\!\!Q[D']\!\!\!\!&= Q[S \!\setminus\! D]\\
& = \sum\limits_{D} Q[S]
\end{array}$ &
												$\begin{array}{l}
												q_{D'}(x_{D'} \,|\, x_{V \setminus D'})\\
												\; = \phi_{S \setminus D'}(q_S(x_S \,|\, x_W); \G(S, W))\\
												\; = \phi_{V \setminus D'}(p(x_V); \G(V))
												\end{array}$
												\\
%\hline
\midrule
$\begin{array}{l}\hbox{kernel for}\\[-4pt] 
\hbox{subsequent call }
\end{array}$
 & $\begin{array}{l}Q[E] \\ 
= \displaystyle{\prod\limits_{z \in E}
\frac{
Q[D'_z]
}{
\sum\limits_z Q[D'_z]
}}\end{array}$ & $\begin{array}{l}
	q_{E}(x_{E} \,|\, x_{V \setminus E})\\ 
	\; = \phi_{D' \setminus E}(q_{D'}; \phi_{D}(\G(S,W))) \\
    \; = \phi_{V \setminus E}(p(x_V); \G(V))
	\end{array}$
	\\
%\hline
\midrule
$\begin{array}{l}\hbox{variables } Q[S] \hbox { or }\\[-4pt] 
 q_S(\cdot ) \hbox{ depends on}\end{array}$ & $\begin{array}{l}\text{Pa}^+(S)\\[-4pt]
  \equiv S \cup \pa_{\G(V)}(S)\end{array}$  & $S \cup \pa_{\phi_{V \setminus S}(\G(V))}(S)$ \\
%\bottomrule
\end{tabular}
\end{center}
}
\caption{
Differences between the notation used in Algorithms \ref{alg:tian} and \ref{alg:tian_sub}, and steps (A1) and (A2) of Tian's algorithm.
The input to Algorithm 1 is a latent projection $\G(V)$, and the corresponding marginal distribution $p(x_V)$.
The input to (A1) is a hidden variable DAG $\G(V \cup L)$, and the corresponding marginal distribution
$p(x_V) = \sum_{x_L} p(x_V, x_L)$.  % Here $\G(V)$ is the latent projection of $\G(V \cup L)$.
}
\label{tab:tian_vs_us}
\end{table}

While the graph input to the recursive part of Tian's algorithm is always some subgraph $\G'(L' \cup S)$
of the original hidden variable graph $\G(L \cup V)$, the graph input to Algorithm \ref{alg:tian_sub} is the
CADMG $\phi_{V \setminus S}(\G(V))$.  
Here, $L'$ is the subset of latent variables relevant to $S_i$, which are ancestors of $S_i$ in $\G_{S_i \cup L}$.
Lemma \ref{lem:graph_map} below shows that this reformulation
preserves the structure of Tian's algorithm.

\begin{lemma}
Let $\G'(L' \cup S)$ be a DAG obtained %\stkout{at some step of}
during some call to step (A2) of Tian's
algorithm from the original DAG $\G(L \cup V)$, where $V$ and $S$
are the sets of observable vertices in $\G$ and $\G'$,
respectively.  Then:

\begin{itemize}
\item[\rm (a)] $S$ is intrinsic in $\G(V)$.
%TODO:

\item[\rm (b)] $\G'(S)$, the latent projection of $\G'(L' \cup S)$ onto $S$, is equal to $\tilde{\G}_S$, where $\tilde{\G} \equiv \phi_{V \setminus S}(\G(V))$.
In other words, $\G'(S)$ is equal to the graph obtained from $\G(V)$ by fixing all vertices outside of $S$ and removing all fixed vertices from the resulting CADMG.

\item[\rm (c)] For any subset $D \subseteq S$ closed under descendants in
$\G'(S)$, $S \setminus D$ is ancestral in $\tilde{\G}_S$, where $\tilde{\G} \equiv \phi_{V \setminus S}(\G(V))$.

\item[{\rm (d)}] A district $E$ in $\G'(S \setminus D)$ is a district in $\phi_{V \setminus (S \setminus D)}(\G(V))$.

\item[{\rm (e)}] The set of ``effective parents'' of $S$ in $\G'(L \cup V)$ is equal to the set of parents of $S$ in $\phi_{V \setminus S}(\G(V))$.
\end{itemize}
\label{lem:graph_map}
\end{lemma}
\begin{proof}
We first note that (e) holds for every intrinsic set $S$ by definition of ``effective parents'' in \citep{tian02on} and definition of latent projections, and the fixing operator on ADMGs.
Thus, establishing (a) for an input set $S$ establishes (e) for $S$.

We prove the claim by induction on the recursive structure of Tian's algorithm.

We first establish that (a), (b), (c), and (d) hold in the base case, which consists of the initial recursive calls to step (A2) in step (A1).
Note that when Tian's algorithm calls step (A2) from step (A1), it is with sets $S_i \equiv S$ which are districts in %$\G(T)$,
$\G[T]$, where each $T$ is a subset of $V$ that is ancestral in $\G(V)$.  Thus, these sets $S$ are intrinsic by definition, which establishes (a).

The graph $\G'(L' \cup S_i)$ given as input to the initial recursive call of Tian's algorithm is equal to $\G(L \cup V)_{L' \cup S_i}$.  This graph is equal to $\tilde{\G}(L' \cup S_i, (V \cup L) \setminus (S_i \cup L'))_{L' \cup S_i}$ where $\tilde{\G}(L' \cup S_i, (V \cup L) \setminus (S_i \cup L')) \equiv \phi_{(V \cup L) \setminus (S_i \cup L')}(\G(L \cup V))$ by definition of graphical fixing in DAGs.

Since we established $S_i$ is intrinsic in $\G(V)$, $\tilde{\G}(S_i, V \setminus S_i) = \phi_{V \setminus S_i}(\G(V))$ by Lemma \ref{lem:fix_marg_commute_graph}.  This fact, coupled with the fact that the operations of taking latent projections and restricting the graph to random vertices commute in CADMGs establishes (b).
The link between graphs used in recursive calls of Algorithm \ref{alg:tian_sub}, and step (A2) of Tian's algorithm provided by (b) immediately implies (c) and (d).

%\ilya{Next,} (e) holds \ilya{for $S_i$ by definition of ``effective parents'' and latent projections.}
%for every initial input set $S$ to step (A2) by definition of latent projections.

%G'(S) = G~_S, where G~ = \phi_{V\S}(G(V)).

%G'(S,L') = \phi_{(V,L)\(Si,L')}(G(V,L))_{S,L'}.
%G'(S) = first fix all but Si,L' then get rid of fixed nodes, then project L'
%G~_S = first project L, then fix all but Si, then get rid of fixed nodes.

Assume we are in step (A2) with a set $S$, where (a), (b), (c), and (d) hold by the inductive hypothesis.
%\ilya{We aim to establish (a), (b), (c) and (d) hold for every set $E$ in the subsequent recursive calls.}
%This immediately establishes (c), namely that for any subset $D \subseteq S$ closed under descendants in $\G'(S)$,
%$S \setminus D$ is ancestral in $\tilde{\G}_S$, where $\tilde{\G} \equiv \phi_{V \setminus S}(\G(V))$.

Since (c) and (d) hold for $S$, (a) holds for $E$.
Since (a) holds for $E$, we can repeat the above argument for the base case to establish (b) for $E$.
This, in turn, establishes (c) and (d) for $E$.
%This then implies $S \setminus D$ is reachable in $\G(V)$, \ilya{which implies

%which implies by (b) for $\G(S \setminus D)$ by the above argument, which in turn implies (d) for $E$, the
%district in $\G'(S \setminus D)$ used for the next stage of the recursion.  This also implies (a) for $E$.
%Finally, (e) follows from (d) and the definition of latent projections and the definition of fixing.

This establishes the claim.
%All subsequent recursive calls in step (A2) correspond to districts $E_i$
%in $\G'(S \setminus D)$, for some $D$ above.  By Lemma \ref{lem:fix_marg_commute_graph},
%all such $E_i$ are districts of $\phi_{V \setminus (S \setminus D)}(\G(V))$ which establishes (d).
%It follows from (c) and (d) that any set $E$ arising in recursive applications of (A2) is reachable from $\G(V)$, thus establishing (a). 
\end{proof}

%We prove this by induction on the recursive structure of algorithm
%\ref{alg:tian_sub}, and step (A2) of Tian's algorithm.
%To establish the base case, note that when algorithm
%\ref{alg:tian} is called with an ADMG $\G(V)$, the set of
%outer calls to Algorithm \ref{alg:tian_sub} corresponds to districts in
%$\G(T)$, where $T$ are certain subsets of $V$,
%each of which is ancestral in $\G(V)$.  These sets reachable by definition.
%
%All subsequent calls to Algorithm \ref{alg:tian_sub} correspond to districts
%in $\G(A')$, where $A' \subseteq S$ is ancestral in
%$\phi_{V \setminus S}(\G(V))$, and $S$ is a reachable set.
%But this implies these districts are also reachable.  The other two inductive
%claims follow by the structure of Algorithm \ref{alg:tian_sub} coupled with
%Lemma \ref{lem:fix_marg_commute_graph}.

%Finally,
%The next result inductively establishes a correspondence between q-factors and kernels.

%$\intgr(\G)$

Step (A2) of Tian's algorithm takes q-factors as inputs.  The analogous inputs for our reformulated Algorithm \ref{alg:tian_sub} are kernels $q_S(x_S \cmid x_W)$ obtained from the original margin $p(x_V)$ by the fixing operator.
%The Algorithm \ref{alg:tian_sub} analogue of q-factors used as inputs to steps in (A2) of Tian's algorithm are kernels $q_S(x_S \cmid x_W)$ obtained from the original margin $p(x_V)$ by the fixing operator.
Lemma \ref{lem:graph_map} above ensures that every such kernel $q_S(x_S \cmid x_W)$ has the same form as the corresponding q-factor $Q[S]$.  The original formulation of Tian's algorithm did not %explicitly
fully describe the order of fixing operations used to construct q-factors.  {Instead, as described above, any particular recursive call with a set $S$ as input first marginalized out a subset $D \subseteq S$, and then constructed the q-factor $Q[E]$ for a subset $E \subseteq S \setminus D$ from the q-factor $\sum_{S \setminus D} Q[S]$.}
%Specifying an explicit order was unnecessary since the algorithm was formulated on margins of distributions that are Markov with respect to hidden variable DAGs, where commutativity of valid fixing sequences holds.  In our reformulation, we do not explicitly assume all fixing sequences on sets $V \setminus S$ valid in $\G(V)$ automatically commute when applied to the input distribution $p(x_V)$ to obtain $q_S$.
Thus, in our translation of Tian's procedure into Algorithms \ref{alg:tian} and \ref{alg:tian_sub}, the order of fixing operations is only partly specified, to be consistent with the specified sequence of q-factor constructions and marginalizations. Specifically, for every set $S$ that serves as input to Algorithm \ref{alg:tian_sub}, the variables in $S \setminus D$, for some set closed under descendants in $\G(S, W)$, are fixed first before variables in $(S \setminus E) \setminus (S \setminus D)$ are fixed, and further variables in $V \setminus S$ are fixed before anything in $S$.

%glue text.
Lemma \ref{lem:graph_map} establishes that Tian's algorithm, and our reformulation via Algorithms \ref{alg:tian} and \ref{alg:tian_sub} are structurally isomorphic.  We now turn to the connection between our reformulated algorithm and the local nested Markov property.

%Consider an intrinsic set $S$ and a fixing sequence ${\bf w}$ for $S$ consistent with the above restrictions.
%The resulting sequence of kernels constructed by the fixing operations in this sequence will contain a subsequnce 

%Next, we show that, for any intrinsic set $S$, if we specify a fixing sequence ${\bf w}$ for $S$.consistent with the above restrictions the set of kernels feature as inputs to Algorithm \ref{alg:tian_sub} obtained are a superset of the set of kernels featured in (\ref{eqn:localknew0}) and (\ref{eqn:localknew}).  We first prove the following utility lemma.

Consider an intrinsic set $S$, and a path $\langle S_0, S_1, \ldots, S \rangle$ in the intrinsic power DAG $\intgr(\G(V))$ starting from the root node $S_0$ corresponding to $S$.  Then there exists a sequence of nested recursive calls to
Algorithm \ref{alg:tian_sub} which use the sets in the above path as inputs in the specified order.  Further, there exists at least one fixing sequence ${\bf w}$ on $V \setminus S$ with the order of operations consistent with these calls, in the sense that all vertices in $S_i \setminus S_{i+1}$ are fixed before vertices in $S_{i+1} \setminus S_{i+2}$, and vertices in $S_i \setminus S_{i+1}$ are fixed according to the restriction described above.

%Given a fixing sequence ${\bf w}$ consistent with this path, in the above sense, the sequence of kernels constructed are a superset of 

\begin{lemma}
Fix an arbitrary ADMG $\G(V)$.
Then for every intrinsic set $S_k$ in $\G(V)$, and every sequence of intrinsic sets $\langle S_0, \ldots, S_{k-1}, S_k \rangle$ forming a path to $S_k$ from the corresponding root $S_0$ in $\intgr(\G(V))$, there exists a sequence of successive recursive calls to Algorithm \ref{alg:tian_sub} with a sequence of vertex set inputs equal to $\langle S_0, \ldots, S_{k-1}, S_k \rangle$.
\label{lem:power-dag-path}
\end{lemma}

Note that there are many sequences of successive recursive calls to Algorithm \ref{alg:tian_sub} that do not correspond to sequences of intrinsic sets forming a path in the intrinsic power DAG $\intgr(\G(V))$.  In particular, only recursive calls that marginalize a singleton vertex take part in such sequences of intrinsic sets.

\begin{proof}
We first show that every intrinsic set $S \in {\cal I}(\G(V))$ corresponds to an input of some call of Algorithm \ref{alg:tian_sub}.
Assume for a contradiction that the algorithm never reaches $S$.

Since $S$ is bidirected-connected, it is contained in some element in ${\cal D}(\G(V))$.   Since
Algorithm \ref{alg:tian_sub} is called from Algorithm \ref{alg:tian} with every element in ${\cal D}(\G(V))$, there exists a smallest
%Since $S$ is bidirected-connected
%this means that Algorithm \ref{alg:tian_sub} is called for some %$S' \in {\cal D}(\G(V))$ and
$S' \supset S$ called by Algorithm \ref{alg:tian_sub}.
However, by hypothesis
%but that
there is no
strict ancestral subset $D'$ of $S'$ in $\G(S',W)$ which contains $S$, where $W = V \setminus S'$.  Since $S'$ was called from Algorithm \ref{alg:tian_sub},
$S'$ is a single district in $\G(S',W)$, and since no ancestral subset contains $S$ it is the case that every $d \in S'$ is an ancestor
of some element of $S$ in $\G(S',W)$.  But then the only fixable  vertices in $S'$ are also in $S$; this contradicts the assumption that $S$ is intrinsic.

{
Next, fix a sequence $\langle S_0, \ldots, S_{k-1}, S_k \rangle$ forming a path to $S_k$ in $\intgr(\G(V))$, with the $\prec$-maximal vertex $v$ in $S_k$.
Note that the set $S_0$ is a district in the ancestral set $\overline{\hbox{pre}}_i$ for some $i$, and thus is one of the inputs to Algorithm \ref{alg:tian_sub}.

For every pair of sets $S_i, S_{i+1}$ adjacent in the sequence, there exists an element $w$ fixable in $\G[S_i]$ (and thus in $\phi_{V \setminus S_i}(\G(V))$) such that $S_{i+1}$ is a district in $\G[S_i \setminus \{ w \}]$ (and thus in $\phi_{V \setminus (S_i \setminus \{ w \})}(\G(V))$).
The conclusion follows since $\{ w \}$ is closed under descendants in $\G[S_i]$ and thus also in $\G(D, V \setminus S)) = \phi_{V \setminus S_i}(\G(V))$.  Thus, since, by induction, $S_i$ serves as an input to Algorithm \ref{alg:tian_sub}, so does $S_{i+1}$.
}
%The conclusion then follows from the fact that
%then $w$ is childless in $\G[S]$,
%the fact that the set $D$ is closed under descendants in $\G(S, W) = \phi_{V \setminus S}(\G(V))$,
%and Lemma \ref{lem:graph_map} (b).
\end{proof}

Given an intrinsic $S_k$ in $\G(V)$, we call a fixing sequence for $S_k$ \emph{consistent with} $\intgr(\G(V))$ if, given a sequence of intrinsic sets $\langle S_0, \ldots, S_{k-1}, S_k \rangle$ forming a path to $S$ from the appropriate root node in $\intgr(\G(V))$, every vertex in $S_i \setminus S_j$ is fixed before every vertex in $S_j$, if $i < j$, and further for each $S_i$, the vertex $w_i$ associated with the power DAG transition from $S_i$ to $S_{i+1}$ is fixed before any other element in $S_i$.

Lemma \ref{lem:power-dag-path} implies the following result.

%that for every $S_k$ intrinsic in $\G(V)$, and a fixing sequence for $S_k$ consistent with $\intgr(\G(V))$, there exists a sequence of recursive calls to Algorithm \ref{alg:tian_sub} that has $S_k$ as a vertex set input, and a corresponding kernel obtained by the same fixing sequence.

\begin{lemma}
For any kernel $q_S$ mentioned in {\rm(\ref{eqn:localknew0})} and {\rm(\ref{eqn:localknew})} for an ADMG $\G(V)$, obtained by some fixing sequence ${\bf w}$ consistent with $\intgr(\G(V))$,
there exists an implementation of Algorithms \ref{alg:tian} and \ref{alg:tian_sub} with a recursive call that has $q_S$ as input, and $q_S$ is implemented by Algorithm \ref{alg:tian_sub} by fixing from $p(V)$ in the order specified by ${\bf w}$.
\label{lem:kernel_map}
\end{lemma}
\begin{proof}
This is an immediate corollary of Lemma \ref{lem:power-dag-path}, and an induction on the recursive structure of Algorithm \ref{alg:tian_sub}.
%All kernels mentioned in (\ref{eqn:localknew0}), defined using some valid fixing sequence, are used as inputs to Algorithm \ref{alg:tian_sub} on line 9.
%
%Finally, for any $S$ intrinsic in $\G(V)$, consider any path from the appropriate root node of the power DAG $\intgr(\G)$ to $S$, and fix a fixing sequence $\sigma_{V \setminus S}$ for $V \setminus S$ consistent with (\ref{eqn:localknew}).  By this we mean that given a set of transitions $C_0 \to C_1 \to C_2 \to \ldots \to S$, with corresponding vertices $w_0, w_1, \ldots$,
%the fixing sequence fixes all elements in $V \setminus C_0$ according to some valid order, then fixes $w_0$ first, and all elements in $C_0 \setminus (C_1 \cup \{ w_0 \})$ according to some valid sequence, and so on.
%
%Then there is a version of Algorithms \ref{alg:tian} and \ref{alg:tian_sub} that fix $V \setminus S$ using $\sigma_{V \setminus S}$.  This follows from the fact that every $w_i$ is closed under descendants in the graph $\phi_{{V \setminus C_i}}(\G(V))$, that every intrinsic set is reached by Algorithm \ref{alg:tian_sub}. and that Algorithm \ref{alg:tian_sub} quantifies over all subsets $D$ of every input $S$ that are closed under descendants in $\phi_{V \setminus S}(\G(V))$.
\end{proof}

Let ${\cal P}_t(\G, V, \prec)$ be the set of distributions obeying 
restrictions in the list returned by Algorithms \ref{alg:tian} and \ref{alg:tian_sub}.

\begin{proposition}\label{thm:nested_implies_tian}
Let $\G(V)$ be an ADMG over vertex set $V$.  Then
\[
{\cal P}^{\rm n}(\G(V))
\subseteq
{\cal P}_t(\G, V, \prec).
\]
\end{proposition}
%Lemmas \ref{lem:graph_map} and \ref{lem:kernel_map} imply that for
%every $Q[R]$ where a constraint was enumerated, $R$ is reachable, and moreover,
%$\phi_{V \setminus C}(p(x_V);\G(V)) = Q[R]$.

\begin{proof}
It suffices to show that every constraint found by Algorithms \ref{alg:tian}
and \ref{alg:tian_sub}, if given a graph $\G(V)$ as one of the inputs,
is implied by some constraint given by the global nested Markov property
for $\G(V)$.

All constraints found in Algorithm \ref{alg:tian} on line 6 are ordinary
conditional independence constraints.  Moreover, they are easily seen to follow
from the m-separation criterion, which forms a part of the global
nested Markov property (since the sets of nodes $T$ in which these constraint
are found are all reachable in $\G(V)$).

Consider some $D'$ obtained during a recursive call of Algorithm \ref{alg:tian_sub}.
By Lemma \ref{lem:graph_map}, $D'$ is a set of random vertices in a set ancestral
in a CADMG corresponding to a set reachable in $\G(V)$.
Therefore $D'$ is itself reachable, so the global nested Markov property implies that the
kernel $q_{D'}(x_{D'} \cmid x_W) = \phi_{V \setminus D'}(p(x_V); \G(V))$,
where $W=V\setminus D'$,
is Markov with respect to $\phi_{V \setminus D'}(\G(V))$.

If $(\pa_\G(S) \setminus S) \setminus
(\pa_\G(D') \setminus D')$ is non-empty,
$D'$ is m-separated from
$(\pa_\G(S) \setminus S) \setminus (\pa_\G(D') \setminus D')$
by
$\pa_\G(D') \setminus D'$
in
the graph $\G(D', W)^{|W}$ obtained from
$\G(D', W) = \phi_{V \setminus D'}(\G)$
in the usual way.
This implies that $D'$ is m-separated from the smaller set
$(\pa_\G(S) \setminus S) \setminus
\pa_\G(D')$ by
$\pa_\G(D') \setminus D'$
in $\G(D', W)^{|W}$, which is precisely the constraint on line 6.

%{\it We need to give $\G^{|W}$ a name.  It is awkward to discuss
%the global property without a name for this important graph.
%Joke suggestion: Steffan had the ``moral graph,'' $\G^{|W}$ could
%be the ``polygamous'' graph (the set $W$ are all spouses).  -- Ilya}
%
Similarly, if the preconditions on line 9 hold,
$v$ is m-separated from a non-empty set
$\mb_\G(v, D') \setminus \mb_\G(v, E)$ given
$\mb_\G(v, E)$ in
$\G(D', W)^{|W}$.  This directly implies the constraint on line 10.
\end{proof}

%Before showing the other direction, we need to show that 
%Algorithm \ref{alg:tian_sub} reaches all intrinsic sets.
%Let $\G(V)$ be an ADMG.
%Denote by
%${\cal T}(\G(V))$
%the set of subsets of $V$ such that the corresponding CADMG and kernel
%are arguments to some call of Algorithm \ref{alg:tian_sub}, if
%Algorithm \ref{alg:tian} is invoked with $\G(V)$.

\begin{proposition}\label{thm:tian_implies_nested}
Let $\G(V)$ be an ADMG with a vertex set
$V$.  Then
\[
{\cal P}_t(\G, V, \prec)
\subseteq
%{\cal P}^{\rm n}_{\rm l}(\G(V), \prec).
{\cal P}^{\rm n}(\G(V)).
\]
\end{proposition}
\begin{proof}
Throughout the proof we will ignore trivial independences of the form $X_A \ci X_{\emptyset} \,|\, X_C$.
%We will first show that given 

Consider the list of restrictions arising from (\ref{eqn:localknew0}) and (\ref{eqn:localknew}) under an ordering $\prec$, where each kernel $q_S$ is constructed according to some sequence ${\bf w}$ that corresponds to a path from a root node to $S$ in $\intgr(\G(V))$.  We will show that this list is implied by the list of restrictions output by Algorithms \ref{alg:tian} and \ref{alg:tian_sub}.
%, where each kernel input $q_S$ obtained by a set of recursive calls corresponding to a path from a root note to $S$ in $\intgr(\G(V))$ obtained by the sequence ${\bf w}$.
That this list of restrictions exists follows by Lemma \ref{lem:kernel_map}.
%Having done so, we will conclude that in fact this version will imply (\ref{eqn:localknew0}) and (\ref{eqn:localknew}) where the kernels are obtained by \emph{any} consistent sequence.

That Algorithm \ref{alg:tian} implies all restrictions in (\ref{eqn:localknew0}) follows immediately by line 7.

By Lemma \ref{lem:kernel_map}, every kernel with a restriction in (\ref{eqn:localknew}) serves as an input to some call of Algorithm \ref{alg:tian_sub}.  For such a kernel $q_S$, we obtain
the following (possibly trivial) independence restrictions:
\begin{align}
\label{eqn:tian-1}
X_{D'} &\ci X_{(\pa(S) \setminus S) \setminus \pa(D')}\mid X_{\pa(D') \setminus D'} \quad {[q_{D'}]},\\
\label{eqn:tian-2}
X_{v} &\ci X_{(D' \cup \pa(D')) \setminus (\mb(v, E) \cup \{ v \})} \mid X_{\mb(v, E)} \quad {[q_{D'}]}.
\end{align}

It is sufficient to show that these restrictions imply:
\begin{align}
X_v \ci X_{(\mb(v, S) \setminus (S \setminus D')) \setminus \mb(v, E)} \mid X_{\mb(v, E)} \;\; [
\phi_{\langle \{ w \}, S \setminus ({E} \cup \{w\}) \rangle}(q_S, \G[S])
%\robin{\phi_{S \setminus ({E} \cup \{w\})} \circ\phi_w({q}_{S}; \G[S])}
%\ilya{\phi_{S \setminus ({E} \cup \{w\})}(\phi_w({q}_{S}; \G[S]); \phi_w(\G[S]))}
],
%X_v \ci X_{(\pa(S) \cup S \setminus \{w\}) \setminus (C \cup \pa(C))} \mid X_{(C \setminus\{v\})\cup \pa(C)} \;\; [\phi_{S \setminus (C \cup \{w\})} \circ\phi_w({q}_{S}; \G[S])],
\label{eqn:local-consequence}
\end{align}
which is equivalent to (\ref{eqn:localknew}) by equating the following left hand side sets in the above equation with those in (\ref{eqn:localknew}) on the right hand side: $S = D$, $E = C$,
%$\mb_{\G}(v, E) = (C \setminus \{ v \}) \cup \pa(C)$,
$\mb_{\G}(v, E) = \fm_\G(C) \setminus \{ v \}$,
$S \setminus D' = \{ w \}$, and
%$\mb_{\G}(v, S) \setminus D = (\pa(D) \cup D \setminus \{ w \})$.
$\mb_{\G}(v, S) \setminus (S \setminus D') = \fm_\G(D) \setminus \{ w \}$.

%We inductively apply the axiom of ordering to (\ref{eqn:tian-2}) to conclude that 
%\begin{align}
%\label{eqn:tian-3}
%(X_{v} &\ci X_{((S \setminus D') \cup D' \cup \pa(D')) \setminus (\mb(v, E) \cup \{ v \})} \mid X_{\mb(v, E)}) \, {[q_{D'}]}.
%\end{align}

By the weak union graphoid axiom and (\ref{eqn:tian-1}), since $v \in D'$, we conclude that:
\begin{align*}
X_{D'} &\ci X_{(\pa(S) \setminus S) \setminus \pa(D')}\mid X_{\pa(D') \setminus D'} \quad {[q_{D'}]}\\
& \implies\quad X_v \ci X_{(\pa(S) \setminus S) \setminus \pa(D')}\mid X_{(\pa(D') \cup D'){\setminus\{v\}}} \quad {[q_{D'}]}.
\end{align*}

Combining (\ref{eqn:tian-2}) with this independence
using the contraction graphoid axiom ($(Q \ci Y \mid Z) \land (Q \ci T \mid Y{\cup}Z) \Rightarrow (Q \ci Y{\cup}T \mid Z)$), we have:
%($(X_{v} \ci X_{(D' \cup \pa(D')) \setminus (\mb(v, E) \cup \{ v \})} \mid X_{\mb(v, E)}) \, {[q_{D'}]}$)
\begin{align*}
X_v \ci X_{((\pa(S) \setminus S) \setminus \pa(D')) \cup ((\pa(D') \cup D') \setminus (\mb(v, E) {\cup \{ v \}}))} \mid X_{\mb(v, E)} \quad [q_{D'}]
%(X_{v} &\ci X_{(D' \cup \pa(D')) \setminus (\mb(v, E) \cup \{ v \})} \mid X_{\mb(v, E)}) \, {[q_{D'}]}.
\end{align*}
setting $Q = \{ v \}$, $Z = \mb_\G(v, E)$, $Y = %(S \setminus D') \cup
(\pa(D') \cup D') \setminus (\mb_\G(v, E){\cup \{ v \}})$, $T = (\pa_\G(S) \setminus S) \setminus \pa_\G(D')$.

Note that since $D' = S \setminus \{ w \}$, $\mb(v, S) = \left( \left( \pa_{\G}(D') \cup D' \right) \setminus \{ v \}\right) \cup \left( \{ w \} \cup \pa_{\G}(w) \right)$,
and $\left( \pa_{\G}(S) \setminus S \right) \setminus \pa_{\G}(D') = \pa_{\G}(w) \setminus \left( S \cup \pa_{\G}(D') \right)$.

Thus, the above independence is equivalent to
\begin{align*}
%X_v &\ci X_{(\pa(S) \setminus S) \setminus \pa(D') \cup (S \cup \pa(D')) \setminus (\mb(v, E))} \mid X_{\mb(v, E)}) \, [q_{D'}] \text{ or}\\
X_v &\ci X_{(\mb(v, S) \setminus (S \setminus D')) \setminus \mb(v, E)} \mid X_{\mb(v, E)} \quad [q_{D'}].
\end{align*}

%pa(S)\S \ pa(D') = ``all strict parents of S that are not also strict parents of ancestral subset D'.''
%pa(D') cup D' \ mb(v,E)\cup{v} = ``fam(D') without v and mb(v,E).''

%mb(v,S) = ``fam(S).''
%mb(v,S) \ D = ``fam(S) where we keep only S' and pa(S)\S.
%{mb(v,S) \ D}\mb(v,E) = ``fam(S), where we keep only D' and pa(S)\S that isn't in mb(v,E).''

That this independence also holds in
\[
q_{E} = \phi_{\langle \{ w \}, S \setminus ({E} \cup \{w\}) \rangle}(q_S; \G[S]) = \phi_{S \setminus (E \cup \{ w \})}(q_{D'}; \G[D'])
\]
%\robin{$q_{E} = \phi_{S \setminus ({E} \cup \{w\})} \circ\phi_w({q}_{S}; \G[S])$}
%\ilya{$q_{E} = \phi_{S \setminus ({E} \cup \{w\})}(\phi_w({q}_{S}; \G[S]); \phi_w(\G[S]))$}
follows by the inductive application of the axiom of modularity (Proposition \ref{prop:ind-preserve-2}).
This establishes (\ref{eqn:local-consequence}).

%Finally, we note that by assumption, the fixing sequence yielding $q_{E} \equiv \phi_{D' \setminus E}(\phi_{D}(q_S; \G(S, W)))$ in the call to Algorithm \ref{alg:tian_sub} is consistent with the fixing sequence in (\ref{eqn:local-consequence}).  This then proves the claim.
\end{proof}

\begin{thma}{\ref{thm:tian_equals_nested}}
%\[
%${\cal P}_t(\G, V, \prec) = {\cal P}^{\rm n}(\G(V))$.
%\]
For an ADMG $\G(V)$, let ${\cal P}_{\rm t}(\G, V, \prec)$ be the set of densities
$p(x_V)$ in which the list of constraints found by Algorithm \ref{alg:tian} holds.
Then
%\[
${\cal P}_{\rm t}(\G, V, \prec) = {\cal P}^{\rm n}(\G(V))$.
%\]
\end{thma}
\begin{proof}
This is an immediate corollary of Propositions
\ref{thm:nested_implies_tian}, \ref{thm:tian_implies_nested}, and
Theorem \ref{thm:global_local}.
\end{proof}

%\subsection{Connections With R-Factorization}

\begin{thma}{\ref{r-fact}}
If $p(x_V) \in {\cal P}^{\rm n}(\G(V))$, then $p(x_V)$ r-factorizes with
respect to $\G$ and $\{ \phi_{V \setminus C}(p(x_V);
\G) \mid C \in {\cal I}(\G) \}$.
\end{thma}

\begin{proof}
This follows directly from the definition of r-factorization, the definition of
${\cal P}^{\rm n}(\G)$ and Theorem \ref{invariance}.
\end{proof}

\subsection{Comparison of Tian's Algorithm and Local Nested Markov Property}\label{sec:constraint-example}

Here we give an example, for the graph in Figure \ref{fig:counter1} of the
constraints found by Tian's Algorithm with the final vertex 6, and show how 
deeply nested the calls are.
For each constraint we give the independence that is found, together with the 
kernel that it appears in.

\eject

% \begin{enumerate}
    % \item 
    $\{1,2,3,4,5,6\}$: no constraint.
    
    Consider ancestral subgraphs:
    \begin{enumerate}
        \item $\{1,2,3,5,6\}$: no constraint.

    new c-component of $6$ is $\{1,3,5,6\}$, so 
    consider ancestral subgraphs:
    \begin{enumerate}
    \item $\{3,5,6\}$: $6 \ci 2,3,5 \, [
    %(\phi_1 \circ \phi_2 \circ \phi_4)(q_V; \G)
    \phi_{\langle 4, 2, 1 \rangle}(q_V; \G)
    ]$;
    
    % c-component of $6$ is $\{6\}$ so consider ancestral subgraphs:
    
    \item $\{1,5,6\}$: $1,5,6 \ci 2 \, [
    %(\phi_3 \circ \phi_2 \circ \phi_4)(q_V; \G)
    \phi_{\langle 4, 2, 3 \rangle}(q_V; \G)
    ]$;
    
    % Consider ancestral subgraphs:
    % \begin{enumerate}
    % \item $\{5,6\}$: $6 \ci 5$;
    % \item $\{1,6\}$: no constraint.
    % \end{enumerate}
    \item $\{1,3,6\}$: $1, 6 \ci 2,3 \, [
    %(\phi_5 \circ \phi_2 \circ \phi_4)(q_V; \G)
    \phi_{\langle 4, 2, 5 \rangle}(q_V; \G)
    ]$;
    
    % c-component of $6$ is $\{1,6\}$ so consider ancestral subgraphs:
    % \begin{enumerate}
    %     \item $\{6\}$ no constraint.
    % \end{enumerate}
    \item $\{5,6\}$: $6 \ci 5 \, [
    %(\phi_{1,3} \circ \phi_2 \circ \phi_4)(q_V; \G)
    \phi_{\langle 4, 2, 3, 1 \rangle}(q_V; \G)
    ]$;
    
    \item $\{3,6\}$: $6 \ci 2,3 \, [
    %(\phi_{1,5} \circ \phi_2 \circ \phi_4)(q_V; \G)
    \phi_{\langle 4, 2, 5, 1 \rangle}(q_V; \G)
    ]$.
    
    % \item $\{1,6\}$: no constraint.
    
    % \item $\{6\}$: no constraint.
    
    \end{enumerate}
    % \item $\{1,2,4,5,6\}$: no constraint.
    
    \item $\{1,2,3,4,6\}$: $6 \ci 3 \mid 1,2,4\, [\phi_{ \langle 5 \rangle }(q_V; \G)]$.
    
    new c-component of 6 is $\{1,2,4,6\}$, so consider ancestral subgraphs:
    \begin{enumerate}
        \item $\{1,2,6\}$: $6 \ci 2 \mid 1 \, [
        %(\phi_{4} \circ \phi_{3} \circ \phi_5)(q_V; \G)
        \phi_{\langle 5, 3, 4 \rangle}(q_V; \G)
        ]$;
        
        % new c-component of 6 is $\{1,6\}$ so consider ancestral subsets:
        % \begin{enumerate}
        %     \item $\{6\}$: no constraint.
        % \end{enumerate}
        
        % \item $\{1,4,6\}$: no constraint.
        % \item $\{1,6\}$: no constraint.
        \item $\{4,6\}$: $6 \ci 4 \, [
        %(\phi_{1,2} \circ \phi_{3} \circ \phi_5)(q_V; \G)
        \phi_{\langle 5, 3, 2, 1 \rangle}(q_V; \G)
        ]$.
        
        % new c-component $\{6\}$.
        
        % \item $\{1,6\}$: no constraint.
        % \item $\{6\}$: no constraint.
    \end{enumerate}
    
    % \item $\{1,4,5,6\}$: no constraint.
    
    \item $\{1,2,5,6\}$: $1,5,6 \ci 2 \, [
    %\phi_{3,4}(q_V; \G)
    \phi_{\langle 4, 3 \rangle}(q_V; \G)
    ]$.
    
    new c-component of 6 is $\{1,5,6\}$, so consider ancestral subgraphs:
    \begin{enumerate}
        \item $\{5,6\}$: $6 \ci 5 \, [
        %(\phi_{1} \circ \phi_2 \circ \phi_{3,4})(q_V; \G)
        \phi_{\langle 4, 3, 2, 1 \rangle}(q_V; \G)
        ]$.
        
        % \item $\{1,6\}$: no constraint.
        
        % \item $\{6\}$: no constraint.
    \end{enumerate}
    
    % \item $\{1,2,4,6\}$: no constraint.
    
    \item $\{1,2,3,6\}$: $6 \ci 2,3 \mid 1 \, [
    %\phi_{4,5}(q_V; \G)
    \phi_{\langle 5, 4 \rangle}(q_V; \G)
    ]$.
    
    % new c-component of 6 is $\{1,6\}$, consider ancestral subsets:
    % \begin{enumerate}
    %     \item $\{6\}$: no constraint. 
    % \end{enumerate}

    % \item $\{1,5,6\}$: no constraint.
    
    % \begin{enumerate}
    % \item $\{1,6\}$: no constraint.
    % \item $\{5,6\}$: $6 \ci 5$.
    % \end{enumerate}
    \item $\{4,5,6\}$: $6 \ci 4,5\, [
    %\phi_{1,2,3}(q_V; \G)
    \phi_{\langle 3, 2, 1 \rangle}(q_V; \G)
    ]$.
    
        % New c-component is $\{6\}$:
        %     consider ancestral subgraphs: 
        %     \begin{enumerate}
        %         \item $\{6\}$: no constraint.
        %     \end{enumerate}
    % \item $\{1,4,6\}$: no constraint.
    
    \item $\{1,2,6\}$: $6 \ci 2 \mid 1 \, [
    %\phi_{3,4,5}(q_V; \G)
    \phi_{\langle 5, 4, 3 \rangle}(q_V; \G)
    ]$.
    
    % new c-component of 6 is $\{1,6\}$:
    % consider ancestral subgraphs: 
    %             \begin{enumerate}
    %             \item $\{6\}$: no constraint.
    %         \end{enumerate}

    \item $\{5,6\}$: $6 \ci 5 \, [
    %\phi_{1,2,3,4}(q_V; \G)
    \phi_{\langle 4, 3, 2, 1 \rangle}(q_V; \G)
    ]$.
    
    % new c-component is $\{6\}$, consider ancestral subsets
    % \begin{enumerate}
    %     \item $\{6\}$: no constraint
    % \end{enumerate}

        \item $\{4,6\}$: $6 \ci 4\, [
        %\phi_{1,2,3,5}(q_V; \G)
        \phi_{\langle 5, 3, 2, 1 \rangle}(q_V; \G)
        ]$.
    
    % new c-component is $\{6\}$, consider ancestral subsets
    % \begin{enumerate}
    %     \item $\{6\}$: no constraint.
    % \end{enumerate}
    
        % \item $\{1,6\}$: no constraint.
    
    % \item $\{6\}$: no constraint
    \end{enumerate}
% \end{enumerate}

Note that, although many of these constraints are redundant under the 
model, this is not immediately apparent without our Theorem \ref{thm:invariant}
that shows the order of fixing is unimportant. 
In particular, we highlight that Tian's algorithm gives the constraints 1(d), 
3(a) and 7 all of which give a marginal constraint between $X_5$ and $X_6$ under
kernels that are not \emph{a priori} the same.  In 
contrast, our algorithm gives many fewer syntactically equivalent independences in 
different kernels.

We also note that Tian's algorithm under the ordering where 3 is placed after 
6 does not, for example, give the constraint $6 \ci 2,3,5$ directly, but 
rather it has to be obtained by applying our graphoid axioms to the independences
provided.

It follows that Tian's algorithm corresponds neither to a global property, 
since it does not give all possible independences, nor a local property, 
as it repeats many equivalent constraints in different subgraphs.  

Indeed, we can note that when considering an edge in the power DAG, it will only 
possibly lead to a constraint if the intrinsic set we move to and its parents
lose more than one vertex.  The subgraph of the power DAG in Figure \ref{fig:powerdag1}
that corresponds to having removed the edges that do not imply a constraint is 
shown in Figure \ref{fig:powerdag1a}; it only has nine edges, corresponding 
to the constraints in Table \ref{tab:counter1constraints}.

\begin{table}
\begin{center}
\begin{tabular}{l|l|l}
\multicolumn{1}{c|}{edge in $\mathfrak{I}(\G)$}&\multicolumn{1}{c|}{constraint}&\multicolumn{1}{c}{kernel}\\
\midrule
123456 $\to$ 1246 & $X_6 \ci X_3 \mid X_{1,2,4}$ & $\phi_{\langle 5 \rangle}(q_V; \G)$\\
12456 $\to$ 156 & $X_6 \ci X_2 \mid X_{1,5}$ &
%$(\phi_4 \circ \phi_3)(q_V; \G)$
$\phi_{\langle 3, 4 \rangle}(q_V; \G)$
\\
1356 $\to$ 156 & $X_6 \ci X_2 \mid X_{1,5}$ &
%$(\phi_3 \circ \phi_4)(q_V; \G)$
$\phi_{\langle 4, 3 \rangle}(q_V; \G)$\\
1246 $\to$ 16 & $X_6 \ci X_2 \mid X_{1}$ &
%$(\phi_4 \circ \phi_{3,5})(q_V; \G)$
$\phi_{\langle 5, 3, 4 \rangle}(q_V; \G)$\\
1356 $\to$ 16 & $X_6 \ci X_{2,3} \mid X_{1}$ &
%$(\phi_{5} \circ \phi_2 \circ \phi_4)(q_V; \G)$\\
$\phi_{\langle 4, 2, 5 \rangle}(q_V; \G)$\\
1356 $\to$ 6 & $X_6 \ci X_{2,3,5}$ &
%$(\phi_{1} \circ \phi_2 \circ \phi_4)(q_V; \G)$\\
$\phi_{\langle 4, 2, 1 \rangle}(q_V; \G)$\\
1456 $\to$ 6 & $X_6 \ci X_{4,5}$ &
%$(\phi_1 \circ \phi_{2,3})(q_V; \G)$\\
$\phi_{\langle 3, 2, 1 \rangle}(q_V; \G)$\\
146 $\to$ 6 & $X_6 \ci X_{4}$ &
%$(\phi_1 \circ \phi_{2,3,5})(q_V; \G)$\\
$\phi_{\langle 5, 3, 2, 1 \rangle}(q_V; \G)$\\
156 $\to$ 6 & $X_6 \ci X_{5}$ &
%$(\phi_1 \circ \phi_{2,3,4})(q_V; \G)$\\
$\phi_{\langle 4, 3, 2, 1 \rangle}(q_V; \G)$\\
\end{tabular}
\caption{The constraints implied by the power DAG in Figure \ref{fig:powerdag1a}.}
\label{tab:counter1constraints}
\end{center}
\end{table}

Although there is some syntactic redundancy between the constraints 
corresponding to the edges $12456 \to 156$ and $1356 \to 156$, this only 
holds because we already know that the marginal distribution over the vertices 
prior to 6 is Markov with respect to the corresponding induced subgraph; in 
particular, because $X_5 \ci X_2 \mid X_1$ in
%$(\phi_{3,4}\circ \phi_6)(q_V; \G)$. 
$\phi_{\langle 6, 4, 3 \rangle}(q_V; \G)$.

%, and though there is still some redundancy, it is
%hard to see how we could reduce the number of constraints any further.

\begin{figure}
\begin{center}
%  \begin{tikzpicture}[>=stealth, node distance=1.2cm]
\begin{tikzpicture}[>=stealth, node distance=1.5cm,
pre/.style={->,>=stealth,very thick,line width = 1.4pt}]
%   \begin{scope}
%    \tikzstyle{rv} = [minimum size=5mm, inner sep=.5mm]
%\node[rv] (123456) {$123456$};
%\node[rv, below of=123456] (12456) {$12456$};
%\node[rv, below of=12456] (1456) {$1456$};
%\node[rv, left of=1456, xshift=-1cm] (1356) {$1356$};
%\node[rv, right of=1456, xshift=1cm] (1246) {$1246$};
%\node[rv, below of=1456, xshift=-1.25cm] (156) {$156$};
%\node[rv, below of=1456, xshift=1.25cm] (146) {$146$};
%\node[rv, below of=156, xshift=1.25cm] (16) {$16$};
%\node[rv, below of=16, xshift=0cm] (6) {$6$};
%\draw[pre] (123456) -- (12456);
%\draw[pre] (12456) -- (1456);
%\draw[pre] (12456) -- (1246);
%\draw[pre] (123456) -- (1356);
%\draw[pre] (123456) -- (1246);
%\draw[pre] (12456) -- (156);
%\draw[pre] (1456) -- (156);
%\draw[pre] (1456) -- (146);
%\draw[pre] (1356) -- (156);
%\draw[pre] (1246) -- (146);
%\draw[pre] (1356) to[bend right] (16);
%\draw[pre] (1246) to[bend left] (16);
%\draw[pre] (146) -- (16);
%\draw[pre] (146) to[bend left] (6);
%\draw[pre] (156) -- (16);
%\draw[pre, color=green] (1456) to[bend left=20] (6);
%\draw[pre] (156) -- (6);
%\draw[pre, color=orange] (1356) to[bend right=20] (6);
%\draw[pre] (16) -- (6);
%  \end{scope}
    \begin{scope}
    \tikzstyle{rv} = [minimum size=5mm, inner sep=.5mm]
\node[rv] (123456) {$123456$};
\node[rv, below of=123456] (12456) {$12456$};
\node[rv, below of=12456] (1456) {$1456$};
\node[rv, left of=1456, xshift=-1cm] (1356) {$1356$};
\node[rv, right of=1456, xshift=1cm] (1246) {$1246$};
\node[rv, below of=1456, xshift=-1.25cm] (156) {$156$};
\node[rv, below of=1456, xshift=1.25cm] (146) {$146$};
\node[rv, below of=156, xshift=1.25cm] (16) {$16$};
\node[rv, below of=16, xshift=0cm] (6) {$6$};
% \draw[pre] (123456) -- (12456);
% \draw[pre] (12456) -- (1456);
% \draw[pre] (12456) -- (1246);
% \draw[pre] (123456) -- (1356);
\draw[pre] (123456) -- (1246);
\draw[pre] (12456) -- (156);
% \draw[pre] (1456) -- (156);
% \draw[pre] (1456) -- (146);
\draw[pre] (1356) -- (156);
% \draw[pre] (1246) -- (146);
\draw[pre] (1356) to[bend right] (16);
\draw[pre] (1246) to[bend left] (16);
% \draw[pre] (146) -- (16);
\draw[pre] (146) to[bend left] (6);
% \draw[pre] (156) -- (16);
\draw[pre, color=green] (1456) to[bend left=20] (6);
\draw[pre] (156) -- (6);
\draw[pre, color=orange] (1356) to[bend right=20] (6);
% \draw[pre] (16) -- (6);
  \end{scope}
  \end{tikzpicture}
\end{center}
\caption{The subgraph of the power DAG for the vertex 6 in $\G$ in Figure 
\ref{fig:counter1}, where the only edges included are those that could 
logically lead to a constraint.}
\label{fig:powerdag1a}
\end{figure}

{
\small
\bibliographystyle{chicago}
\bibliography{references}
}

\makeatletter\@input{xx_main.tex}\makeatother